\documentclass[10pt]{amsart}
\usepackage{latexsym,amssymb,amsmath,amsthm,graphicx,listings,comment}
\usepackage[section]{placeins}
\newtheorem{thm}{Theorem} 
\newtheorem{lemma}{Lemma} 
\newtheorem*{satz}{Theorem}
\newtheorem{coro}{Corollary} 
\def\D{\mathcal{D}} \def\V{\mathcal{V}}
\newcommand{\R}{\mathbb{R}} \newcommand{\Z}{\mathbb{Z}} \def\G{\mathcal{G}} 
\def\B#1#2{{#1\choose #2}}

\title{Gauss-Bonnet for multi-linear valuations}
\author{Oliver Knill}
\date{January 17, 2015}
\address{Department of Mathematics \\ Harvard University \\ Cambridge, MA, 02138, USA }
\subjclass{53A55, 05C99, 52C99, 57M15, 68R99, 53C65}

\keywords{Graph theory, Euler characteristic, Wu characteristic, Valuations, 
          Barycentric characteristics, Discrete Intersection Theory}

\begin{document}
\maketitle

\begin{abstract}
We prove Gauss-Bonnet and Poincar\'e-Hopf formulas for multi-linear valuations on finite
simple graphs $G=(V,E)$ and answer affirmatively a 
conjecture of Gr\"unbaum from 1970 by constructing higher order Dehn-Sommerville valuations which vanish for 
all $d$-graphs without boundary. A first example of a higher degree valuations was introduced by Wu in 1959. 
It is the Wu characteristic $\omega(G) = \sum_{x \cap y \neq \emptyset} \sigma(x) \sigma(y)$ with 
$\sigma(x)=(-1)^{{\rm dim}(x)}$ which sums over all ordered intersecting pairs of complete subgraphs of a 
finite simple graph $G$. It more generally defines an intersection number 
$\omega(A,B) =  \sum_{x \cap y \neq \emptyset} \sigma(x) \sigma(y)$,
where $x \subset A,y \subset B$ are the simplices in two subgraphs $A,B$ of a given graph. 
The self intersection number $\omega(G)$ is a higher order
Euler characteristic. The later is the linear valuation $\chi(G) = \sum_x \sigma(x)$ which sums over all 
complete subgraphs of $G$. We prove that all these characteristics share the multiplicative property of Euler characteristic: 
for any pair $G,H$ of finite simple graphs, we have $\omega(G \times H) = \omega(G) \omega(H)$ so that 
all Wu characteristics like Euler characteristic are multiplicative on the Stanley-Reisner ring.
The Wu characteristics are invariant under Barycentric refinements and are so combinatorial invariants in 
the terminology of Bott. By constructing a curvature $K:V \to R$ satisfying Gauss-Bonnet $\omega(G) = \sum_a K(a)$, 
where $a$ runs over all vertices we prove $\omega(G) = \chi(G) - \chi(\delta(G))$ which holds for any 
$d$-graph $G$ with boundary $\delta G$. There also prove higher order Poincar\'e-Hopf formulas:
similarly as for Euler characteristic $\chi$  and a scalar function $f$, where the index $i_f(a)=1-\chi(S^-_f(a))$ 
with $S_f^-(a)=\{b \in S(a) \; | \; f(b)<f(a) \}$ satisfies $\sum_a i_f(a) = \chi(G)$, there is for
every multi-linear valuation $X$ and function $f$ an index $i_{X,f}(a)$ such that $\sum_{a \in V} i_{X,f}(a)=X(G)$.
For $d$-graphs $G$ and $X=\omega$ it agrees with the Euler curvature. For the vanishing multi-valuations which were conjectured to exist,
like for the quadratic valuation $X(G) = \sum_{i,j} \chi(i) V_{ij}(G) \psi(j)=\langle \chi,V,\psi \rangle$
with $\chi=(1,-1,1,-1,1),\psi=(0, -2, 3, -4, 5)$ on $4$-graphs, discrete 4 manifolds, where $V_{ij}(G)$ is the 
$f$-matrix counting the number of $i$-simplices in $G$ intersecting with $j$-simplices in $G$,
the curvature is constant zero. For general graphs and higher multi-linear Dehn-Sommerville relations, 
the Dehn-Sommerville curvature $K(v)$ at a vertex is a Dehn-Sommerville valuation on the unit sphere $S(v)$.
We show $\chi \cdot V(G)\psi = v(G) \cdot \psi$ for any linear valuation $\psi$ of a $d$-graph $G$ with
$f$-vector $v(G)$. This leads to multi-linear Dehn-Sommerville valuations which vanish on $d$-graphs.
\end{abstract}

\section{Introduction}

Given a finite simple graph $G$, a {\bf valuation} is a real-valued map $X$ on the set of 
subgraphs of $G$, so that $X(A \cup B) = X(A) + X(B)-X(A \cap B)$ holds for any two subgraphs
$A,B$ of $G$ and $X(\emptyset)=0$. Here $A \cup B=(V \cup W,E \cup F)$  and $A \cap B=(V \cap W,E \cap F)$,
if $A=(V,E),B=(W,F)$ are finite simple graphs with vertex sets $V,W$ and edge sets $E,F$.
With the empty graph $\emptyset$, a graph with no vertices and no edges, the set of all subgraphs is a lattice.
If one requires additionally that $X(A)=X(B)$ holds for any two isomorphic subgraphs $A,B$ of $G$,
then $X$ is called an {\bf invariant valuation}. 
A {\bf quadratic valuation} $X$ is a map which attaches to a pair 
$A,B$ of subgraphs of $G$ a real number $X(A,B)$ such that for fixed $A$, 
the map $B \to X(A,B)$ and for fixed $B$
the map $A \to X(A,B)$ are valuations and such that $X(A,B)=0$ if $A \cap B= \emptyset$.
It is an {\bf intersection number} which extends to a bilinear 
form on the module of chains defined on the abstract Whitney simplicial complex of the graph $G$.  \\

More generally, a {\bf $k$-linear valuation} is a map which attaches to an ordered $k$-tuple 
of subgraphs $A_1,\dots,A_k$ of $G$ a real intersection number which is ``multi-linear" in the sense that 
any of the maps $A_j \to X(A_1,\dots ,A_j, \dots A_k)$ is a valuation. We also assume they are localized
in the sense that $X(A_1,\dots,A_k) \neq 0$ is only possible only if $\bigcap_{j=1}^k A_j \neq \emptyset$. 
For a $k$-linear valuation $X$ one obtains
the self intersection number $X(A)=X(A,A,\dots,A)$. We will see that some multi-linear valuations still 
honor the product property $X(A \times B) = X(A) \cdot X(B)$, where $A \times B$ is the product graph,
the incidence graph of the product $f_A f_B$ of the representations of $A,B$ in the Stanley-Reisner ring. 
When seen as integer-valued functions on this ring, valuations with this property are multiplicative functions. \\

\begin{tiny}
Here are footnotes on the choice of the definitions. \\
(i) The terminology ``linear valuation" and ``multi-linear valuation"
come from fact that the valuation property can be written as 
$X(A+B) =X(A) +X(B)$ using the symmetric difference $A+B =A \cup B \setminus A \cap B$. 
But this addition throws us out of the category of graphs 
into the class of chains. Valuations are in this larger ambient class compatible with the algebra.
For graphs, we have to to stick to the lattice description when defining the
valuation. An example: with the given definitions of union and intersection of graphs,
the graphs $A=(V,E)= (\{1,2\},\{(1,2)\})$ and $B=(W,F)=(\{2,3\},\{(2,3)\})$ 
define $A \cup B = (\{1,2,3\},\{(1,2),(1,3)\})$ and $A \cap B = \{ \{2\},\emptyset\})$.
The valuation property is satisfied, but with symmetric differences on vertices
and edges we get $(V+W, E+F ) = \{ \{2\}, \{(1,2),(2,3)\})$ which is only a chain and no more a graph.
The absence of a linear algebra structure on the category of graphs is a reason for the complexity for
computing valuations (it is an NP complete problem) and the reason to stick with the 
{\bf Boolean distributive lattice} of subgraphs using union and intersection rather than 
symmetric difference and intersection which is not defined in the category of graphs. The usual
graph complement is not compatible. Also in the continuum, valuations use the lattice structure 
using that unions and intersections of convex sets rather than allowing
complements of convex sets which would allow to build a {\bf measure theory for valuations}. 
In the discrete, the difficulty is transparent: for a valuation, the complement of the graph
$\delta G= a+b$ within $G=a+b+ab$ is only a chain $ab$ which we call
the virtual interior. It has the Euler characteristic $-1$, which is the Wu characteristic of $K_2$. \\
(ii) Without the assumption $X(\emptyset)=0$,
we would also count constant functions $X(A)=c$ as valuation, but it is tradition not to 
include them. It would render the valuation {\bf ``affine"} rather than ``linear" because linearity 
requires $X(\emptyset)=0$ for the empty graph $\emptyset$.  \\
(iii) The invariance property is also a common assumption, both in the continuum and in the discrete. 
In the continuum case, where valuations are mainly studied on unions
of {\bf convex sets} in Euclidean space, the invariance assumption is that the valuation is invariant under Euclidean
motion. The reason for assuming this is that the theory is already rich enough with that assumption and
that it also also in the continuum leads to a finite dimensional space of valuations. Flat Euclidean space
seems confining. This is however not the case as one can for the purpose of Valuations, embed Riemannian
manifolds into Euclidean or projective spaces. This induces a valuation theory on Riemannian manifolds and
beyond. More about this in the Hadwiger appendix. \\
(iv) The assumption $X(A_1,\dots A_k)=0$ if $\bigcup_j A_= \emptyset$
was done so that the theorems work. Without it, much would fail. The assumption
excludes cases like $X(A,B)=Y(A) Z(B)$ for two linear valuations $Y,Z$. 
The assumption will imply that the curvatures and indices are localized as is the 
case in differential geometry. One could look at cubic valuations $X(A,B,C)$ for
which $X$ is only zero if the nerve graph defined by $A,B,C$ is not connected. Widening the localization
as such would make the curvature functions $K_X(v)$ depend on a disk of radius $k$, if $X$ is
a $k$-linear valuation. With the assumption, the curvature functions always depend on a disk of radius $2$. 
The assumption of insisting all simplices to intersect simultaneously (as Wu did) 
implies that cubic and higher order characteristics agree with Euler characteristic. Otherwise, the boundary
formula as well as the product property fails. We would also expect that getting rid of the localization 
assumption would lead to additional ``connection terms" due to chains which are not boundaries.
\end{tiny}

\begin{figure}[!htpb]
\scalebox{0.24}{\includegraphics{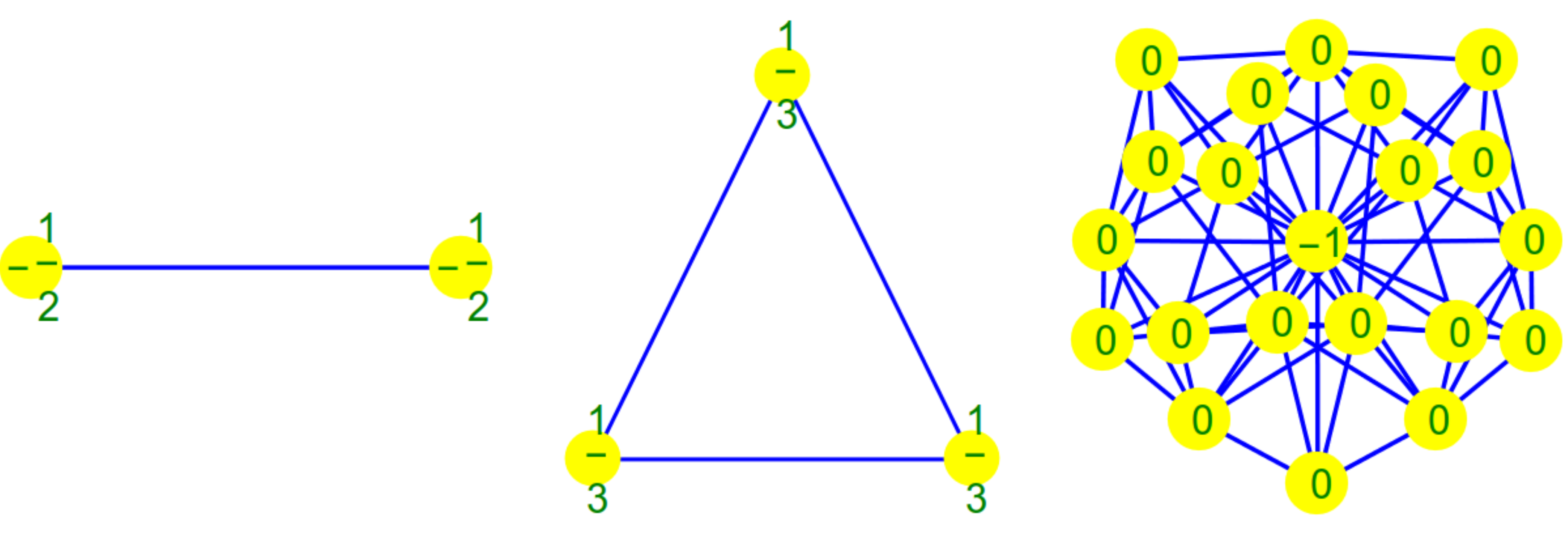}}
\caption{
We see two graphs $A=K_2,B=K_3$ and its product $G=A \times B$ which 
has all pairs of simplices $x \subset A,y \subset B$ as vertices and two
vertices connected if one is contained in the other. 
The vertices are labeled by the Wu curvature $K(v)$ which add up to the
Wu characteristic $\omega$. The picture illustrates some results of this paper
{\bf 1.} There is a curvature $K(v)$ such that
its sum is the Wu characteristic. {\bf 2.} $\omega(G) = \omega(A) \times \omega(B)$.
{\bf 3.} For a $d$-graphs $G$ with boundary, $\omega(G)=\chi(G)-\chi(\delta G)$. 
In this case, $G$ is the product of an interval $A$ with a disc $B$ and so a
$3$-ball $G$ with a boundary $\delta G$ which is a $2$-sphere. As $\chi(G)=1, \chi(\delta G)=2$,
we have the boundary formula $\omega(G)=\chi(G)-\chi(\delta G)=-1$ which is true for any 
$(2k+1)$-ball.
}
\end{figure}

An example of a non-invariant valuation is the map $A \to X_a(A)={\rm deg}_A(a)$ which assigns to a sub graph $A$ of $G$ 
the vertex degree of a fixed vertex within $A$. When summing this valuations $X_a(A)$ over all vertices $a$, we
get by the Euler handshake formula the invariant valuation $\sum_{a \in V} X_a(A) = 2 v_1(A)$, counting twice the number 
of edges in $A$. The Euler handshake is already a Gauss-Bonnet formula adding up local quantities to
get a global quantity. We see it now as adding up local valuations to a global valuation.
The prototype of an invariant valuation is the {\bf Euler characteristic} of a graph $G$. It is
$$  \chi(A) = \sum_x (-1)^{{\rm dim}(x)} = v_0(A)-v_1(A)+v_2(A)- \dots  \; ,   $$ 
where $x$ runs over all complete subgraphs of $A$ of non-negative dimension. It is an alternating
sum of the components $v_i(A)$ of the $f$-vector of $A$. Each $v_i(A)$ is of course an invariant valuation. 
By Gauss-Bonnet, the Euler characteristic can be written as a sum of generalized local valuations or curvatures
$A \to K_A(a) = \sum_{k=0} (-1)^k V_{k-1}(a)/(k+1)$, 
where $V_k(a)$ counts the number of $k$-simplices in the unit sphere of $A$ centered at $a$ including $V_{-1}(a)=1$
counting the empty graph as a $-1$ dimensional simplex. The map $A \to K_A(a)$ is the {\bf curvature} 
of $A$ at $a$. It is generalized as it assigns to the empty graph 
the constant $1$. Any invariant valuation can be written as a sum of such local valuations using the curvature
$\sum_{k=0} X(k) V_{k-1}(a)/(k+1)$. This is Gauss-Bonnet for 
valuations. We rediscovered it in \cite{cherngaussbonnet} in the case of Euler characteristic but the 
formula has been discovered and rediscovered before \cite{Levitt1992,I94a,forman2000}. 
It appears that we have entered new ground when extending Gauss-Bonnet 
to general linear valuations and especially to multi-linear valuations. It is important to note that
Gauss-Bonnet holds for all finite simple graphs and any multi-linear valuation and no geometric assumptions
whatsoever like being a discretization of a manifold are needed. The theorem holds for all networks. \\

An example of a quadratic valuation is the {\bf Wu characteristic} \cite{Wu1959} given by 
$$ \omega(A) = \sum_{x \cap y \neq \emptyset} (-1)^{{\rm dim}(x) + {\rm dim}(y)} \; ,  $$
where $(x,y)$ runs over all ordered pairs of complete subgraphs of $A$. It can be written 
as a quadratic form for the {\bf $f$-quadratic form} $V_{ij}(A)$ or simply {\bf $f$-matrix} 
counting the number of pairs of $i$-simplices intersecting with $j$-simplices in $A$: 
$$ \omega(A) = \sum_{i,j}  (-1)^{i+j} V_{ij}(A)  \; . $$
It more generally produces an {\bf intersection number}
$$  \omega(A,B) = \sum_{x \cap y \neq \emptyset, x \subset A, y \subset B} 
 (-1)^{{\rm dim}(x) + {\rm dim}(y)}  = \sum_{i,j} (-1)^{i+j} V_{ij}(A,B) \; , $$
where $x$ runs over all complete subgraphs of the subgraph $A$ and $y$ runs over
all complete subgraphs of the subgraph $B$. In the second equivalent formula for $X(A,B)$,
the quadratic form $V_{ij}(A,B)$ counts the number of $i$-simplices in $A$ and $j$-simplices in $B$
with $x \cap y \neq \emptyset$. Let $A=\{a\}$ be a one point graph for example,
then $B \to \omega(A,B) = \sum_{y \subset B, a \in y} (-1)^{{\rm dim}(y)}$ is a 
linear valuation, but not an invariant valuation as it is local giving nonzero values 
only near the vertex $a$. \\

While the Euler characteristic $\chi=\omega_1$ defines a linear map on the {\bf $f$-vector} $v=(v_0,v_2, \dots)$ 
with $v_k$ counting the number of complete subgraphs $K_{k+1}$ of $G$, the Wu characteristic $\omega=\omega_2$ 
evaluates a {\bf quadratic $f$-form} $V$, where $V_{kl}$ counts
the number of complete subgraphs $K_{k+1},K_{l+1}$ of $G$ which intersect in a non-empty graph.
If we write $\chi_1=(1,-1,1,-1,\dots,(-1)^{d})$ then 
$$   \chi(G) = \omega_1(G) = \chi_1^T \cdot v,  \hspace{1cm} \omega(G)=\omega_2(G) = \chi_1^T \cdot V \chi_1  \;  $$
and $\omega(A,B) =  \chi_1^T V(A,B) \chi_1$, where $V_{ij}(A,B)$ counts the number of intersections
of $i$-simplices in $A$ with $j$-simplices in $B$. The quadratic form $V$ is not necessarily positive definite 
but always has a positive maximal eigenvalue by Perron-Frobenius. One can also look at cubic situations like 
$$ \omega_3(G) = \sum_{x \cap y \cap z \neq \emptyset} (-1)^{{\rm dim}(x) + {\rm dim}(y) + {\rm dim}(z)}  \; , $$
where the behavior on complete graphs is $\omega_3(K_d)=1$ again for all $d$ and $\omega(G)=\chi(G)$ holds
for all $d$-graphs even if $G$ has a boundary.
It can be written using a {\bf cubic $f$-form} $V_{ijk}(G)$ counting the number of $i,j$ or $k$-simplices in $G$ which 
simultaneously intersect: with $\chi_1=(1,-1,1,-1,\dots)$,
$$ \omega_3(G)  = V(G) \chi_1 \chi_1 \chi_1  = \sum_{ijk} V_{ijk}(G) \chi_1(i) \chi_1(j) \chi_1(k) \; . $$
The higher order Wu characteristics $\omega_k$ are defined similarly using $k$-linear $f$-forms $V(G)$ 
or {\bf $f$-tensors}. We have $\omega_3(K_d)=\chi(K_d)=1$ and $\omega_4(K_d)=\omega_2(K_d)=(-1)^d$.
All these characteristics turn out to be invariant under Barycentric refinement as well as have the multiplicative 
property and agree up to a sign with Euler characteristic $\chi$ for $d$-graphs, graphs for which every 
unit sphere is a $(d-1)$-sphere. \\

Euler proved $\chi(G)=v-e+f=2$ with $v=v_0,e=v_1,f=v_2$ for planar graphs, where $f$ counts 
not only triangles but any region defined by an embedding of $G$ into a 2-sphere. 
In our terminology, graphs like the dodecahedron or cube graph are $1$-dimensional,
as they have no triangles. Among the five Platonic solids, only the octahedron and icosahedron are
as 2-spheres, the tetrahedron is a 3-dimensional simplex, has Euler characteristic $1$ and is contractible.
We can reformulate Euler's result then that for a discrete 2-sphere $G$, the Euler characteristic is $2$. 
We know that $\chi$ is a homotopy invariant,
that it is the first Dehn-Sommerville invariant \cite{Klee1964} and also that it is a Barycentric characteristic 
number to the eigenvalue $1$, obtained by looking at the linear map $A$ on $f$-vector $v$ of the graph and taking an 
eigenvector $a$ of $A^T$ and defining $X(G)=a \cdot v$. 
As $\chi$ is the number of even-dimensional simplices minus the number of odd-dimensional simplices,
it is the super trace of the identity operator on the $\sum_i v_i(G)$-dimensional Hilbert space of 
discrete differential forms. 
It is robust under heat deformation: if $d$ is the exterior derivative, the discrete McKean-Singer formula 
\cite{McKeanSinger,knillmckeansinger} tells that 
$\chi(G) = {\rm str}(\exp(-tL))$, where $L=(d+d^*)^2$ is the Hodge Laplacian on discrete differential forms
and ${\rm str}(L) = \sum (-1)^k {\rm tr}(L_k)$ if $L_k$ is the part of $L$ acting on $k$-forms.  \\

By the way, the {\bf total dimension} $X(G)=\sum_i v_i(G)$ of all discrete differential forms is an example of an
invariant valuation. It is believed that for $d$-graphs it is bounded below by $3^d-1$ by a conjecture of 
Kalai of 1989 \cite{Stanley1970} (formulated for polytopes) as one can learn in \cite{BergerLadder}. 
This number is also interesting as it is hard to compute for general graphs, that it is the number of vertices 
of the Barycentric refinement of $G$ or the value $f_G(1,1, \dots,1)$ if $f_G$ is the Stanley-Reisner polynomial
of $G$. The smallest $d$-sphere is believed to be the sharp lower bound like 
$8$ attained for the $C_4$ graph, $26$ attained for the octahedron, $80$ for the 3-cross polytop or $242$ for the 4-cross polytop; 
the reason is that $1+\sum_k v_k x^{k+1} = (1+2x)^{d+1}$ for the smallest $d$-sphere. As the Kalai valuation $X$ has the curvature
$K(x) = \sum_{k=0}^{d+1} V_{k-1}(x)/(k+1)$ which is always positive, one could try to use Gauss-Bonnet and induction, 
as $\sum_{k=0}^{\infty} V_k(x) \geq 3^{d-1}$ implying $\sum_k V_{k-1}(x)/(k+1) \geq 3^{d-1}/(d+1)$. But since the number of 
vertices  of $G$ is only $\geq 2d$ we have $X \geq 3^{d-1} (2d/(d+1)) = 3^d (2/3) (d/(d+1))$. If 
one could verify the conjecture for all graphs $G$ with $\leq 3d+3$ vertices, then the induction assumption 
would give $X \geq (3d+3) 3^{d-1}/(d+1) = 3^d$. The Kalai conjecture is already by definition verifiable in finitely
many cases for each $d$ (in principle) but the just given argument reduces the verification range from $3^d$ to $3d+d$.
The argument still gives inductively $X(G) \geq a^d$ for any $a<2$ so that
$X(G) \geq a^d$ holds for any graph with clique number $d+1$. But since
for $G=K_{d+1}$ we have $X(G)=2^d$ and $X(H)<X(G)$ if $H$ is a subgraph of $G$, 
$X(G) \geq 2^d-1$ holds already trivially. We believe it should be possible to use Dehn-Sommerville relations to estimate 
the Kalai curvature better on this space. But this requires to understand quantitatively the projection $P$ onto the 
Dehn-Sommerville space and then the vector $(1-P) (1,1/2,1/3,\dots,1/d)$.\\ 

Due to its multi-linear nature, one can not expect much from the Wu characteristic $\omega$ at first. 
Actually, it starts with bad news: $\omega$ is not a
homotopy invariant, as $\omega(K_1)=1$ but $\omega(K_2)=-1$, despite that $K_2$ and $K_1$ are homotopic.
Interestingly, the Wu characteristic $\omega$ picks up the dimension of a simplex and $\omega(K_{d+1})=(-1)^d$ so 
that one can write more elegantly 
$$   \omega(G)=\sum_{x \cap y \neq \emptyset} \omega(x) \omega(y) \; . $$
Note that $\omega(x,y) \neq \omega(x) \omega(y)$ as the case $x=y=K_2$ shows where $\omega(x,x)=\omega(x)=-1$,
but $\omega(x \times y) = \omega(x) \omega(y)$. As $\omega$ is invariant under Barycentric refinement, it has
been called {\bf combinatorial invariant} (see section 1.2 of \cite{Bott52}), where 
polynomial functionals $R_G(z) = \sum_s (-1)^{|s|} z^{b_d(G_s)}$ 
and $S_G(z) = \sum_s i^{|s|+b_{d-1}(G)} z^{b_d(G_s)}$ were defined,
which sum over all possible subsets $s$ of $V$ and where $b_k$ is the $k$'th Betti number of $G_s$ 
the graph generated by the set $V \setminus s$. 
Like Euler characteristics, also the Wu characteristic $\omega$ is a functional $X$ on the class of simplicial 
sub-complexes of $G$ but $\omega$ is not a valuation: the kite graph $G$ has $\omega(G)=1$ but two triangular
subgraphs $A,B$ with $\omega(A)=\omega(B)=1$ but $\omega(A \cap B) = -1$ so that $\omega(A \cup B) =1$
but and $\omega(A) + \omega(B) - \omega(A \cap B) = 3$. The Wu characteristic is a quadratic valuation
as defined above. Just having mentioned polynomial invariants, one could combine all Wu characteristics and
define the {\bf Wu function} 
$$  \omega_G(z) = \sum_k \omega_k(G) z^k \; , $$ 
where $\chi=\omega_1$ is the Euler characteristic, $\omega=\omega_2$ the Wu characteristic, 
$\omega_3$ the cubic Wu characteristic etc. For a graph $K_{d+1}$ or $d$-ball, 
we have $\omega(z)=z/(1-(-1)^d z)$, for a $d$-graph $G$, the Wu function is $\omega(z)=\chi(G) z/(1-z)$. 
The figure $8$ graph $G$ with $\chi(G)=\omega_1(G)=-1$, $\omega(G)=\omega_2(G)=7$, 
$\omega_3(G)=25$, $\omega_4(G)=79$ is already a case where we don't yet have a closed form for $\omega(z)$. \\

We will see that any of the $\omega_k$ can be computed fast for most graphs as it satisfies
a Poincar\'e-Hopf formula and $\omega(G) = \chi(G)-\chi(\delta(G))$ for $d$-graphs with boundary.
Also $\omega_3(G)=\chi(G)$ etc. What happens is that the index entering in Poincar\'e-Hopf is 
a valuation on the unit sphere, allowing to apply Poincar\'e-Hopf again there etc. 
So, for most graphs in the Erd\"os-Renyi space of all graphs with $n$ elements,
also the Wu characteristic can be computed quickly. The complexity for computing $\omega$ is polynomial in $n$, 
if one insists for example that for any finite intersection of unit spheres, maximally 99\% of all possible 
connections between vertices in the unit sphere  are present. This especially applies for
$d$-graphs, where any finite intersection of unit spheres is a $k$-sphere. 
A consequence of the boundary formula is that for 2-spheres, we still have 
$$ \omega(G) = v-e+f \; , $$
where $v=v_0,e=v_1,f=v_2$. The proof of the reduction to Euler characteristic makes use of 
the Gauss-Bonnet result for multi-linear valuations which in particular holds for
the Wu characteristics. It turns out that in the interior of a $d$-graph, in distance 2 or larger from
the boundary, the curvature of the Wu characteristic $\omega$ is the same than the curvature of 
the Euler characteristic $\chi$. 
For odd dimensional graphs, the curvature of $\omega$ lives near the boundary of $G$, similarly 
than the curvature for Euler characteristic. 
For an even dimensional graph with boundary, the curvatures for Euler characteristic
or Wu characteristics are in the interior. 
The proof of the boundary formula $\omega(G) = \chi(G)-\chi(\delta(G))$ reveals what needs to be satisfied: 
for each unit sphere, we have to be able to use induction. 
The formula shows that $\omega$ measures the Euler characteristic of a 
``virtual interior" of a graph. This interior is no more a graph but a chain as forming complements throws
us out of the category of graphs into a larger category of chains: take the two boundary points of $K_2$ away, then
we end up with an edge which has no vertices attached and which is not a graph any more. It is the virtual interior
of $K_2$ and has Euler characteristic $-1$ which is the Wu characteristic of $K_2$. 
Because of invariance under Barycentric subdivision, the Wu characteristic 
is defined also in the continuum limit for compact manifolds $M$ 
with boundary $\delta M$, where it satisfies the same boundary formula $\omega(M)= \chi(M)-\chi(\delta M))$. 
While expressible by Euler characteristic for discrete manifolds, 
the Wu invariant becomes interesting for {\bf varieties} as we see already in simple examples 
like the lemniscate 
$$  (x^2+y^2)^2=(x^2-y^2) $$ which has the Wu characteristic $7$
and Euler characteristic $-1$. Indeed, we can compute the curvature of the Wu-invariant at a singularity 
of a projective variety by discretization using a graph. Quadratic valuations are useful too as they produce
intersection numbers $\omega(A,B)$ for any pairs of subgraphs and so an intersection
number of a pair of varieties. Two $1$-dimensional circular graphs $A,B$
for example crossing twice have intersection number $2$. \\

We restrict here to the language of graphs. This means that we look only at abstract simplicial complexes which are 
Whitney complexes of some graph. Like with topologies, measure structures or other constructs placed on a set, 
one could consider different simplicial complexes on the same graph. 
The $1$-dimensional skeleton-complex consisting of vertices and edges
is an example. It is too small for geometry as it treats any graph as a ``curve". The neighborhood complex is an other
example, but it renders the dimension too large; for a wheel graph with boundary $C_n$ for example, the 
neighborhood complex would be $n$-dimensional.. The sweet spot is the Whitney complex defined by all complete 
subgraphs of $G$. For the wheel graph $W_n$ for example, it renders it a discrete disc of dimension $2$, with Euler characteristic 
$1$, which is contractible and has a $1$-dimensional boundary $C_n$. The geometry, cohomology, homotopy, even spectral
theory all behave similarly as in the continuum. The language of graphs is intuitive and 
not much of generality is lost: while not every abstract finite simplicial 
complex on $G$ is a Whitney complex of a finite graph - the simplicial complex with algebraic 
representation $f_G=x+y+z+xy+yz+xz$ is the
smallest example which fails to be a Whitney complex - the Barycentric refinement of an arbitrary 
finite abstract simplicial complex is always the Whitney complex of a finite simple graph. 
In the above triangle without 2-simplex, the refinement of the complex is the Whitney complex of $C_6$. 
Valuations on graphs satisfy a Gauss-Bonnet formula. For the valuation $v_k(G)$, this is the fundamental
theorem of graph theory, a name sometimes applied only for the Euler handshaking lemma which is Gauss-Bonnet 
for $v_1(G)$. We will see that Gauss-Bonnet generalizes to multi-linear valuations but curvature depend now
on the ball of radius $2$. The curvature remains local but  
in some sense has become a second order difference operator similarly as in the continuum, where the curvature 
tensor uses second derivatives. The language of graphs is equivalent but more intuitive
especially when dealing with valuations as a classically, valuation are defined as functionals on 
simplicial sub complexes of a complex on $V$ so that classically; it is important to realize that a valuation
is not a map from subsets of $V$ to $\R$ but a map from sets of subsets of $V$. 
It is much easier to work with real-valued maps from the set of subgraphs of $G$ to $R$. Subgraphs have
the intuitive feel of subsets but encode simplicial complexes. The use of valuations as a functional
on the set of graphs is language which allows to work with valuations using the intuition
we know from measures. \\

On a graph $G=(V,E)$, the discrete Hadwiger theorem \cite{KlainRota}
assures that the dimension of the space of valuations 
on $G$ is the clique number $d+1$ of $G$. A basis is given by the functionals $v_i(G)$, counting the number 
$K_{i+1}$ subgraphs of $G$. The Euler characteristic $\chi(G) = \sum_k (-1)^k v_k(G)$ can be characterized as 
the only invariant valuation which stays the same when applying Barycentric subdivision and which assigns the value
$1$ to all simplices. The multi-linear Dehn-Sommerville invariants we are going to construct which 
assign the value $0$ to $d$-graphs are not invariant under Barycentric refinements on the class of all graphs,
but only vanish on $d$-graphs, a class, where the Barycentric invariance is clear. 
The valuation $\chi$ can also be characterized as the only invariant valuation which is a
homotopy invariant and which assigns $1$ to a $K_1$ subgraph. Similarly, an extension of the discrete Hadwiger 
result of Klein-Rota shows that the space of invariant quadratic valuations
is $(d+1)(d+2)/2$-dimensional. A basis is given by the functionals $G \to V_{ij}(G)$ with $j \geq i$, 
where $V_{ij}(G)$ counts the number of pairs $x,y$ of $i$-dimensional simplices $x$ and $j$ dimensional
simplices $y$ for which $x \cap y \neq \emptyset$. For one dimensional graphs, graphs without triangles
for example, the space of linear valuations has dimension $2$, and the space of quadratic valuations
has dimension $3$. Graphs with maximal dimension $1$ have a vector space of linear valuations which is
two dimensional. It has a basis $v_0,v_1$ counting the number of vertices and edges, a spanning set for the
vector space of quadratic valuations is $V_{00}=v_0$ counting the number of vertices, 
$V_{11}$ counting the number of pairs $(x,y)$ of edges intersecting and $V_{12}$
counting the number of pairs $(x,y)$, where $x$ is a vertex contains in an edge $y$. 
In this case, $V_{12}=2v_1$ by Euler handshake but that is a relation between linear and quadratic
valuations. In an appendix, we review the discrete Hadwiger theorem and prove the extension 
to quadratic valuations. \\

One can look at some valuations like Euler characteristic also from an algebraic point of view, as it
is possible to write $\chi(G)=-f_G(-1,-1,\dots,-1)$,
where $f_G$ is the element in the Stanley-Reisner ring representing the graph. The later
is a polynomial ring, where each vertex is represented by a variable and every monomial 
represents a complete subgraph. Each summand in the polynomial is a monomial.
For the kite graph $G$ for example obtained by gluing two triangles $abc$ and $bcd$ along an edge 
$bc$, we have $f_G(a,b,c,d) = a + b + a b + c + a c + b c + a b c + d + a d + c d + a c d$ and
$\chi(G)=-f_G(-1,-1,-1,-1)=1$. The Z-module of chains for $G$ is the $11$ dimensional set containing
elements like $f=3a-b+2ab+4c+ad+2cd+7abc-3acd$ in which the monoids of $f_G$ form the basis elements.
Valuations extend naturally to this module. We would for example have $v_2=-2-4-1$ and $v_3(f)=7-3$ and
$V_{13} = (3-1)7 + 3(-3)$. Quadratic valuations are then quadratic forms on this module. 
We will write the quadratic Wu characteristic algebraically as
$$  \omega(G) = (f^2)_G(-1,\dots,-1)-f_G(-1,\dots,-1)^2  \; ,$$
a representation which will together with Poincar\'e-Hopf imply that $\omega$ is multiplicative
when taking Cartesian products of graphs. It reminds of {\bf variance} ${\rm E}[X^2]-{\rm E}[X]^2$ in 
probability theory. Similarly,
$$  \omega(A,B) = (f_A f_B)(-1, \dots, -1) - f_A(-1 \dots,-1) f_B(-1 \dots,-1) $$
reminds of {\bf covariance} ${\rm E}[X Y] - {\rm E}[X] {\rm E}[Y]$ if the evaluation $f \to -f(-1,\dots,-1)$ is
identified with ``expectation". In some sense therefore, if Euler characteristic $\chi(G)$ is a topological 
analog of expectation, the quadratic Wu characteristic $\omega(G)$ is a topological analogue
of variance and the product formulas for two graphs an analogue of ${\rm E}[X^k Y^k] = {\rm E}[X^k] {\rm E}[Y^k]$ 
which all hold for independent random variables as $X^k,Y^k$ are then uncorrelated.
For general networks, there is no relation between Euler characteristic and Wu characteristic.
It fits the analogy as for general random variables where no relation between expectation and variance exists in general.
The fact that for discrete manifolds without boundary, the Euler and Wu characteristics agree, 
and for discrete manifolds with boundary, the Euler and cubic Wu characteristic agree, comes
unexpectedly. We were certainly surprised when discovering this experimentally. Note however that in the 
entire Erd\"os-Renyi space of networks with $n$ vertices, the geometric $d$-graphs form a very thin slice.
(While both growth rates are not known, one can expect that the number of non-isomorphic $d$-graphs with $n$ 
vertices is bounded above by $\exp(C \sqrt{n})$, an estimate coming from partitioning into connected 
components, while the number of non-isomorphic graphs with $n$ 
vertices should have a lower bound of the form $\exp(C n^2)$ for some positive $C$, an estimate expected to hold as 
most pairs of graphs are not isomorphic.) \\

The reason for the multiplicative property is partly algebraic but there is a topological twist required. 
This is the Poincar\'e-Hopf formula which allows to make the connection from algebra to geometry: 
lets look at Euler characteristic $\chi(G)$ which sums up $I(x)=(-1)^{{\rm dim}(x)}$ over all complete subgraphs 
$x$ of the graph $G$. Since $I(x)$ is an integer-valued function on the vertex set of the 
Barycentric refinement graph $G_1$, one can ask whether it is the {\bf Poincar\'e-Hopf index} $i_f(x)$ 
of a scalar function $f$ on the vertex set of $G_1$. As Poincar\'e-Hopf tells $\sum_x i_f(x) = \chi(G_1)$
and $\sum_x I(x)=\chi(G)$, it would be nice if there existed a function $f$ for which $i_f(x)=I(x)$. 
When investigating this experimentally we were surprised to see that this is indeed the case.
The scalar function which enumerates the monomials of the Stanley-Reisner polynomial
$f_G$ defined by $G$, where the monomials are ordered according to dimension and lexicographic order, does the job.
The proof reveals that the sub graph $S^-(x)$ of $G$ generated by 
$\{ y \; | \; f(y) < f(x) \}$ is a $(k-1)$-sphere
if the dimension of $x$ is $k$. This implies $i_f(x) = 1-\chi(S^-(x)) = 1-(1+(-1)^{k+1}) = (-1)^k = I(x)$. 
This computation works for any finite simple graph $G$ and no geometric assumption on $G$ is necessary.
It shows the remarkable fact that the dimension signature $\sigma(x)=(-1)^{{\rm dim}(x)}$ on the Barycentric refinement of any
finite simple graph is actually a Poincar\'e-Hopf index of a gradient vector field on the graph. There is
no continuum analogue of that; only a shadow of this result can be seen in the proof of the classical 
Poincar\'e-Hopf theorem in differential topology \cite{Spivak1999}, where one proofs the theorem first 
for a particular gradient vector field defined by a triangulation and then proves by a deformation
argument that the Poincar\'e-Hopf sum does not depend on the field. \\

As the Euler characteristic is the only multiplicative {\bf linear} invariant valuation on the set 
of graphs, we have to go beyond linear valuations to get more multiplicative invariant functionals of this type. 
Quadratic valuations are the next natural choice and the Wu characteristic is a natural
quadratic valuation as it is invariant under Barycentric refinements,
assigns the same value to isomorphic subgraphs and assigns the value $0$ to the empty graph and the value
$1$ to the one point graph $K_1$. Similarly, the cubic Wu characteristic is a cubic valuation with
this property. The Euler characteristic and the Wu characteristics (including cubic and higher order versions) 
more generally appear to be the only multiplicative invariant valuations among all multi-linear invariant valuations.
Gr\"unbaum objected to the claim of Wu that the Wu characteristics
are the only combinatorial invariants and pointed out the existence of Dehn-Sommerville invariants. 
This is a valid objection but {\bf Wu's hypothesis} could still hold when one looks at it as a
functionals on {\bf all graphs}. It is well known that the Euler characteristic is the only 
invariant valuation which assigns the value $1$ to $K_1$ and is a combinatorial invariant
\cite{forman2000,Levitt1992}.
The Wu characteristic appears to be the only quadratic invariant valuation {\bf on the class of all finite simple graphs} 
which is invariant under Barycentric refinement and assigns the value $1$ to $K_1$. 
If that is true, we could also say that $\omega$ is the only quadratic invariant valuation
which is multiplicative. However - and that is the point which Gr\"unbaum made in the case of Euler characteristic - 
is that in geometric situations like what we call $d$-graphs, there are other valuations, the Dehn-Sommerville
invariants, which are on this class also invariant under Barycentric subdivision because they are zero, even so 
the values of refinements explode in general when applying refinements in the class of general networks. 
At present, we must consider it an open problem, whether the Wu invariant is the only quadratic valuation
on the class of all graphs, which is invariant under Barycentric refinements and assigns $1$ to every $K_1$ subgraph. 
Lets call the statement of uniqueness of Wu characteristic the {\bf Wu hypothesis}. To understand it, we have 
to investigate the behavior of the $f$-matrix under Barycentric refinement, something we have only just started to look
at. The fact that $f$-vectors transform linearly and looking at the eigenvectors shows immediately that the 
Euler characteristic is unique. If the transformation on $f$-matrices $V(G) \to V(G_1)$ were linear, we expect a unique
eigenvector to the eigenvalue $1$. This would be the Wu eigenvector, leading to the uniqueness of Wu characteristics. 
If the quadratic case works, then most likely also more general $k$-linear valuations work and $\omega_k$ be unique
in the class of $k$-linear valuations which are invariant under Barycentric refinement as well as assigning $1$ to 
a single vertex. \\

In a finite simple graph $G$, a complete subgraph is also called a face or simplex. 
The set of all subsets of the vertex set of a complete graph can be seen as 
a simplicial sub-complex which is indecomposable in the sense that it can not be written as 
a union of two different simplicial subcomplexes of the Whitney complex. 
Simplices are the ``elementary particles" in the Boolean algebra of all simplicial subcomplexes
of a graph. The Wu characteristic takes count of ``interactions" between these particles.
The realization of the signature dimension $I(x)=(-1)^{{\rm dim}(x)} = \sigma(x)$ 
as a Poincar\'e-Hopf index of a ``wave function" $f$ indicate that $I(x)$ is something like a ``charge".
When deforming the wave function $f$, these charges change but their total sum does not. There
is ``charge conservation". For any multi-linear valuation and any scalar function $f$ there is a
Poincar\'e-Hopf index. For a quadratic valuation like the Wu characteristic there is an
index $i_f(a,b)$ which vanishes if the distance between $x$ and $y$ is larger than $1$. The index
function can be seen as an integer-valued function on vertex pairs. The
Poincar\'e-Hopf formula $\omega(G)  = \sum_{a,b \in V \times V} i_f(a,b)$ sums over the vertex set
and not the set of simplices. It is possible then to push this function from vertex pairs to vertices. 
We have to stress that graphs include higher dimensional structures 
without the need to digress to multi-graphs. Much of graph theory literature deals with 
graphs equipped with the {\bf $1$-dimensional skeleton simplicial complex} and ignores the 
two or higher dimensional simplices. The language of graphs alone however is quite
powerful to describe a large part of the mightier and fancier language of abstract simplicial 
complexes and so rather general topology. While the structure of simplicial complexes is 
more general as there are simplicial complexes which are not Whitney complexes of a 
graph (like some matroids), refinement rectifies this: the Barycentric refinement of any 
abstract simplicial complex $K$ is always the Whitney complex of a graph: given an arbitrary simplicial 
complex $K$, take the simplices in $K$ as the vertices and connect two if one is contained in the other. 
The Whitney complex of this graph is then the Barycentric refinement of $K$. 
Also for simple polytopes, where now faces are not necessarily triangles, the graph determines
the combinatorial structure of the polytope \cite{Kalai1989}. 
We don't lose much generality therefore if we stick to the language of graphs, at least if we
look for discrete differential geometric structures. The advantage is
not only of notational and of didactic advantage - the category of networks can be grasped very 
early on, as it is familiar from maps and diagrams -, it is also convenient from the 
computer science point of view as many general-purpose computer algebra languages have the language
of graphs hardwired into their language. In an appendix we have given programs which allow to compute
all the objects discussed in this article.  \\

The dimension of a simplex $K_{d+1}$ is $d$. There are various notions of dimensions known for graphs.
One is the {\bf maximal dimension} which is defined as $d$ if the clique number of $G$ is $d+1$. In other
words, the maximal dimension of $G$ is the maximal dimension which a simplex in $G$ can have. 
Nice triangulations of d-dimensional manifolds have dimension $d$ but there are triangulations of d-manifolds
where $G$ is higher dimensional: take an octahedron for example and attach a new central vertex in each triangle
connected to the vertices of the triangle. This is still a triangulation but its dimension is $3$ as it contains
many tetrahedra. We have defined dimension motivated from Menger-Uhryson as the average of the dimensions of
the unit spheres minus $1$. The induction assumption is that the empty graph has dimension $-1$.
The original inductive Menger-Uhryson dimension of a graph is $0$. The just defined inductive dimension
satisfies all the properties one can wish for and even behaves in many cases like the Hausdorff dimension
in the continuum like ${\rm dim}( A \times B) \geq {\rm dim}(A) + {\rm dim}(B)$ in full generality for all 
finite simple graphs. It is also possible to compute explicitly the average dimension in Erd\"os-Renyi spaces
$G(n,p)$ as it satisfies the recursion
$d_{n+1}(p) = 1+\sum_{k=0}^n \B{n}{k} p^k (1-p)^{n-k} d_k(p)$ with $d_0=-1$. 
Each $d_n$ is a polynomial in $p$ of degree $\B{n}{2}$. See \cite{randomgraph}.  \\

A $d$-graph is a finite simple graph for which every unit sphere is a 
$(d-1)$ graph which is a $d$-sphere. Being a $d$-sphere was defined recursively by Evako as the property that 
every unit sphere is a $(d-1)$-sphere and that removing
one vertex renders the graph contractible. We could characterize $d$-graphs also as 
graphs for which the Barycentric refinement limit is a smooth, compact $d$-manifold with
boundary. For general graphs or networks, there is a dimension which mathematically very much behaves like
the Hausdorff dimension in the continuum: 
the {\bf inductive dimension} of a graph is defined by setting the dimension of the empty graph 
to be $-1$ and in general by adding $1$ to the average of the dimensions of the unit spheres of the graph. 
It is a rational number which similarly as Hausdorff dimension satisfies
${\rm dim}(G \times H) \geq {\rm dim}(G) + {\rm dim}(H)$ for all finite simple graphs $G,H$
where $G \times H$ is the Cartesian product of graphs defined by taking the product
in the Stanley-Reisner ring and looking at the graph defined by that algebraic object.
We can also computed the expectation of the inductive dimension on Erd\"os-Renyi spaces of graphs.
Now, when looking at valuations, even the presence of a single simplex of dimension $d$ allows us
to look at valuations counting in such simplices: counting the largest dimension simplices 
is the analogue of volume.  Having the discarded the $1$-dimensional space of constant 
valuations which assigns to any graph a constant $c$, we get a $(d+1)$-dimensional space of 
linear valuations, a $(d+1)(d+2)/2$-dimensional space of quadratic valuations if $G$ has 
maximal dimension $d$. This is a generalization of discrete Hadwiger. \\

The quadratic valuations and intersection numbers we are going to look at, are geometric
and do not have much interpretation yet in the case of general networks
as they are not homotopy invariants. Here are some attempts for interpretations:
in the case of a graph without triangles, there is a physical interpretation in that the Wu characteristic
adds up interaction energies between different edges and vertices. Think of the graph as a molecule, the
vertices as atom centers and the edges as bonds between atoms, there are positive self-interactions
between the positively charged nuclei and positive self-interactions between negatively charged bonds,
then there are negative self-interaction energies between the nuclei and bonds. The Wu characteristic 
$\sum_{x,y} \sigma(x) \sigma(y)$ has now an interpretation as an {\bf interaction energy}. This 
H\"uckel type interpretation however fades if triangles are involved.
An algebro-geometric perspective comes in by seeing a quadratic valuation $X(A,B)$ as the
{\bf intersection number} of pairs $(A,B)$ of subgraphs of a given graph
so that they can serve to study intersections in a purely combinatorial way.
Two one dimensional graphs intersecting transversely in a point have intersection number $1$. 
A one dimensional graph intersecting transversely with a two-dimensional graph has intersection 
number $-1$. Two two dimensional graphs intersecting in a point has again intersection number $1$. 
An other interpretation of the Wu invariant can come by seeing ${\rm int}(G)=G - \delta G$ as an interior so that 
$\chi(G)-\chi(\delta G)$ measures the Euler characteristic of the interior of $G$ if we think
of the interior and boundary being disconnected. Of course it is not possible to define a subgraph of $G$
taking the role of the interior such that $\omega(G)$ is the Euler characteristic of the subgraph. Here
is the reason: when looking at star graphs $S_n$
the Euler characteristic of any subgraphs is bounded above by $n$ while the Wu characteristic 
of a star graph with $n+1$ rays is the value of the Fibonnacci polynomial $n^2-n-1$ which grows
quadratically with $n$. Still, in the continuum, some notions along these lines have been developed, 
like in \cite{StaeckerWright}, where a valuation $X$ of the interior is defined as such. The formula 
$X(M) = \sum_x X({\rm int}(x))$ for a simplicial complex given in Lemma 2 of \cite{StaeckerWright} can 
be seen an analogue of the formula $\omega(G)=\sum_x \omega(x)$. Whether the picture of seeing the Wu 
characteristic as the Euler characteristic of some {\bf ``virtual interior"} of $G$, remains to be seen. 
Anyway, as $\omega$ is of kinetic nature as it sums neighboring interactions in a quadratic manner, it kind of measures an 
interior energy similarly as models in statistical mechanics, the Ising model in particular; only that now
the interaction energy is not given by a additionally imposed spin value but geometrically defined by the dimension
of the various pieces of space. 
The interpretation of $\sigma(x) = (-1)^{{\rm dim}(x)}$ as a spin value is not so remote as we have identified
it as a Poincar\'e-Hopf index of a gradient vector field. \\

Finally, one could seriously look at the Wu characteristic as a functional in physics, especially  
for naive approaches to quantum gravity. The reasons are similar as for Euler characteristic, which in even dimensions 
like for 4-graphs has the index $i_f(x)$ which is expressible through the Euler characteristic of a 2-graph and so 
an average over sectional curvatures in a well defined sense so that there is a strong analogy with the Hilbert action 
in general relativity.  \\

Lets look at a discrete algebro-geometric connection:
any quadratic valuation can be seen as a divisor on the {\bf intersection graph} of $G$, the graph 
of all complete subgraphs as vertices and where two are connected, if they intersect.
A divisor means here an integer-valued map on the vertices of the graph as in Baker-Norine theory.
That theory sees graphs as $1$-dimensional objects where assigning integer values to vertices
is the analogue of what a divisor means in the continuum. The Poincar\'e-Hopf indices play an important
role in that theory. The intersection graph is the graph for which the complete subgraphs are the vertices and two are 
connected if they intersect. The intersection graph is obtained from the Barycentric refinement by 
adding more connections. If we ``chip-fire" fractions of the divisor to the vertices, we
get a rational number at each point which is the curvature. Already the curvature
of linear valuations like Euler characteristic can be understood like that: start with 
the divisor which attaches the value $(-1)^k$  to the $k$ simplices. If we break up this
value $(-1)^k$ into $k+1$ pieces and chip fire each part
to the vertices, we send $(-1)^k/(k+1)$ to the vertices. Doing that to all gives
the Euler curvature value $K(x) = 1+\sum_{k=1}^{\infty} (-1)^k V_{k-1}(x)/(k+1)$.  
For the Wu characteristic, things become nonlinear, as the divisor attached to
the simplices is no more just a constant but depends on the connections but the proof
remains the same. 

\section{The Wu characteristic}

Wenjun Wu introduced in 1959 \cite{Wu1959} (possibly already in \cite{Wu1953}, a reference
we could not get hold of yet) the functional 
$$ \omega(G) = \sum_{x,y} (-1)^{{\rm dim}(x)+{\rm dim}(y)}   \; , $$
where $x,y$ runs over all pairs of simplices which intersect. We call it the 
{\bf Wu characteristic}. To get closer to the
notation used in models of statistical mechanics like the Ising model, one could define the 
{\bf signature of a simplex} as $\sigma(x)=(-1)^{{\rm dim}(x)}$ and write
$$ \omega(G) = \sum_{x \cap y \neq \emptyset} \sigma(x) \sigma(y) \;  $$
which now looks like adding up an {\bf interaction energy}.
The invariant was originally formulated by Wu for convex polyhedra but
we will look at it in the case of arbitrary graphs equipped with the Whitney complex.
It can also be considered for more general simplicial complexes.
As explained in the introduction, looking at graphs only, is almost no 
loss of generality, as the Barycentric refinement of an arbitrary abstract 
simplicial complex is already the Whitney complex of a finite simple graph. \\

For example, if $G=K_2$, we have three simplices in $G$. They are $\{a,b,ab \;\}$.
There are 4 intersections and both give a value $-1$ and there are three intersections
which give value 1. The value is $-1$. Algebraically, $f_G=a+b+ab$ and
$f_G(-1,-1)^2-(f_G)^2(-1,-1)$ as $f_G^2=2ab$ so that $f_G(-1,-1)^2 - (f_G)^2(-1,-1)=-1$. 
For the kite graph,
$$ G=(V,E) = (\{a,b,c,d\},\{ (a,d),(a,b),a,c),(b,c),(c,d) \}) \; , $$
with ``Bosonic simplices" $\{ (a,b,c),(a,c,d),(a),(b),(c),(d) \}$ 
and ``Fermionic simplices" $\{ (a,d),(a,b),a,c),(b,c),(c,d)  \}$, 
the Wu characteristic is $1$. We can see this also by looking at the square free
part of $f_G^2 = (a + b + a b + c + a c + b c + a b c + d + a d + c d + a c d)^2
= 2a b + 2a c + 2b c + 6a b c + 2a d + 2b d + 4a b d +  2c d + 6a c d + 4b c d + 8a b c d$.  \\

The Wu characteristic $\omega$ is not a linear valuation:
the Kite graph with two $K_3$ subgraphs $A,B$ intersecting in a $K_2$ shows that 
$\omega(A \cap B) + \omega(A \cap B) =\omega(A) + \omega(G)$ does not hold as the left hand side is 
$1-1=0$, while the right hand side $1+1=2$. 
Indeed $\omega$ is an example of a multi-linear valuation and is in particular 
a {\bf quadratic valuation}. It is also not a homotopy invariant, as it 
is not the same for all complete graphs. It is equal to $-1$ for odd dimensional simplices 
and $1$ for even dimensional simplices.
All complete graphs however are clearly collapsible to a point and so homotopic.
Nevertheless, it turns out that the Wu characteristic is multiplicative.
We initially also investigated its relation with analytic torsion which is a spectrally defined number for
graphs and an other highly dimension and geometry sensitive topological invariant. 
Like analytic torsion, or Dehn-Sommerville invariants, 
the Wu characteristic is fragile if we move away from geometric graphs: growing a zero-dimensional
dendrite to an odd dimensional geometric structure for example does not change the homotopy 
but changes the quadratic valuation. We could build a connected graph with Wu characteristic $-1000$ for
example by growing $1$-dimensional 501 hairs to a 2-sphere. \\

If $x$ is a complete subgraph, then $\omega(x) = (-1)^{\rm dim}(x) = \sigma(x)$. This will follow
from one of the main results Barycentric refinement of a $d$-simplex $x$ produces a geometric $d$-ball
with boundary for which $\omega$ is the difference between the Euler characteristic of the
graph minus the Euler characteristic of the boundary. Having the Wu characteristic of a simplex 
expressed in terms of $\sigma(x)$, we can write 
$$ \chi(G) = \sum_{x} \omega(x) \; , $$
where $x$ runs over all simplices in $G$. In some sense, the self-interaction functional $\omega$ 
``explains" the signs in the sum of the Euler characteristic. And also the Wu characteristic $\omega$ can now
be expressed by itself: 
$$ \omega(G) = \sum_{x \cap y \neq \emptyset} \omega(x) \omega(y) \;,   $$
where the sum is again over all ordered pairs of simplices $x,y$ which intersect. In comparison, we have
the formula $\sum_{x,y} \omega(x) \omega(y)  = \chi(G)^2$, where $x,y$ runs over {\bf all} possible ordered pairs, 
(pairs which do not necessarily intersect), which 
follows from $\omega (x \times y) = \omega(x) \omega(y)$ and $\chi(G) = \sum_{x} \omega(x)$.  \\

\begin{figure}
\scalebox{0.12}{\includegraphics{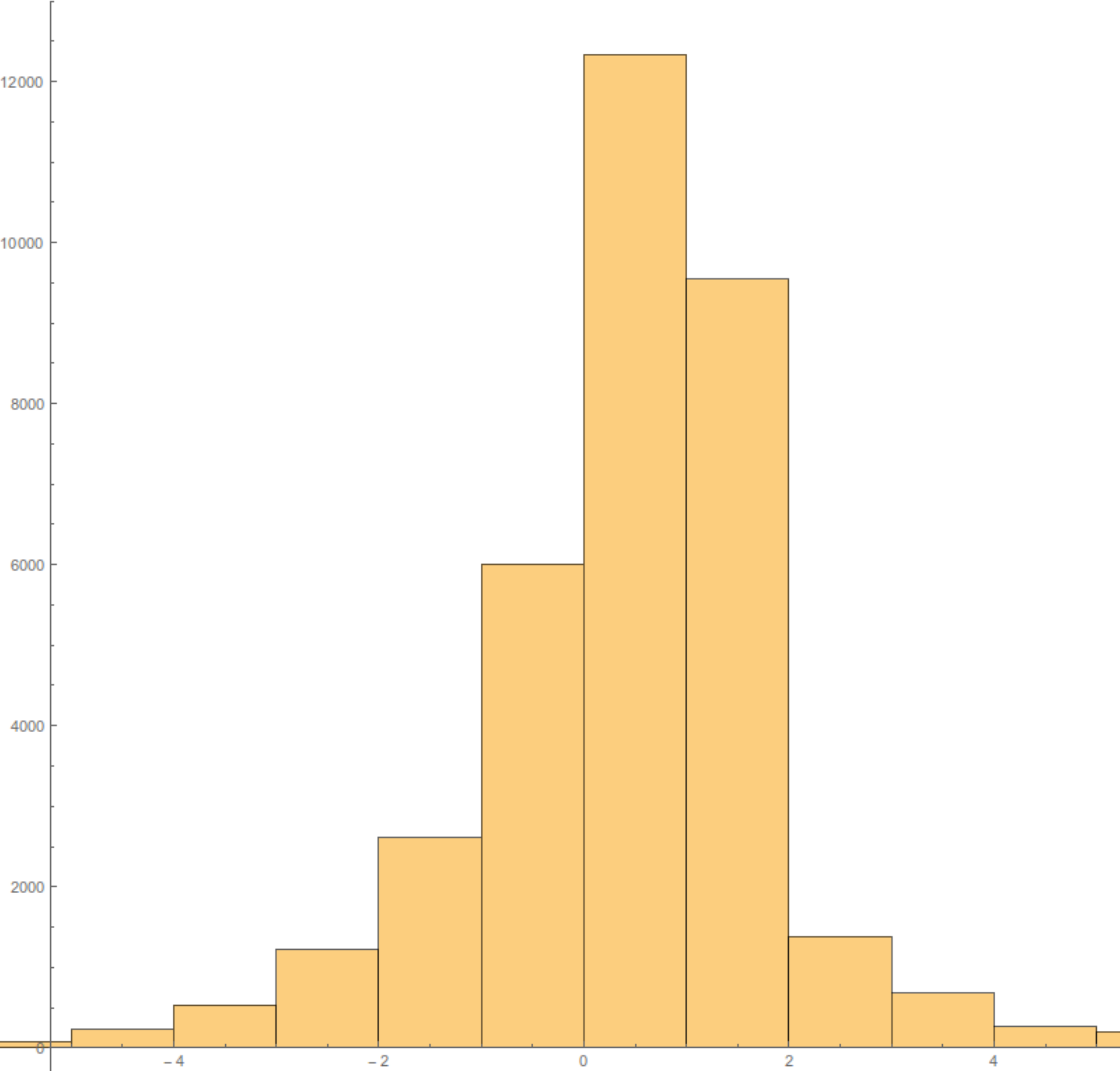}}
\scalebox{0.12}{\includegraphics{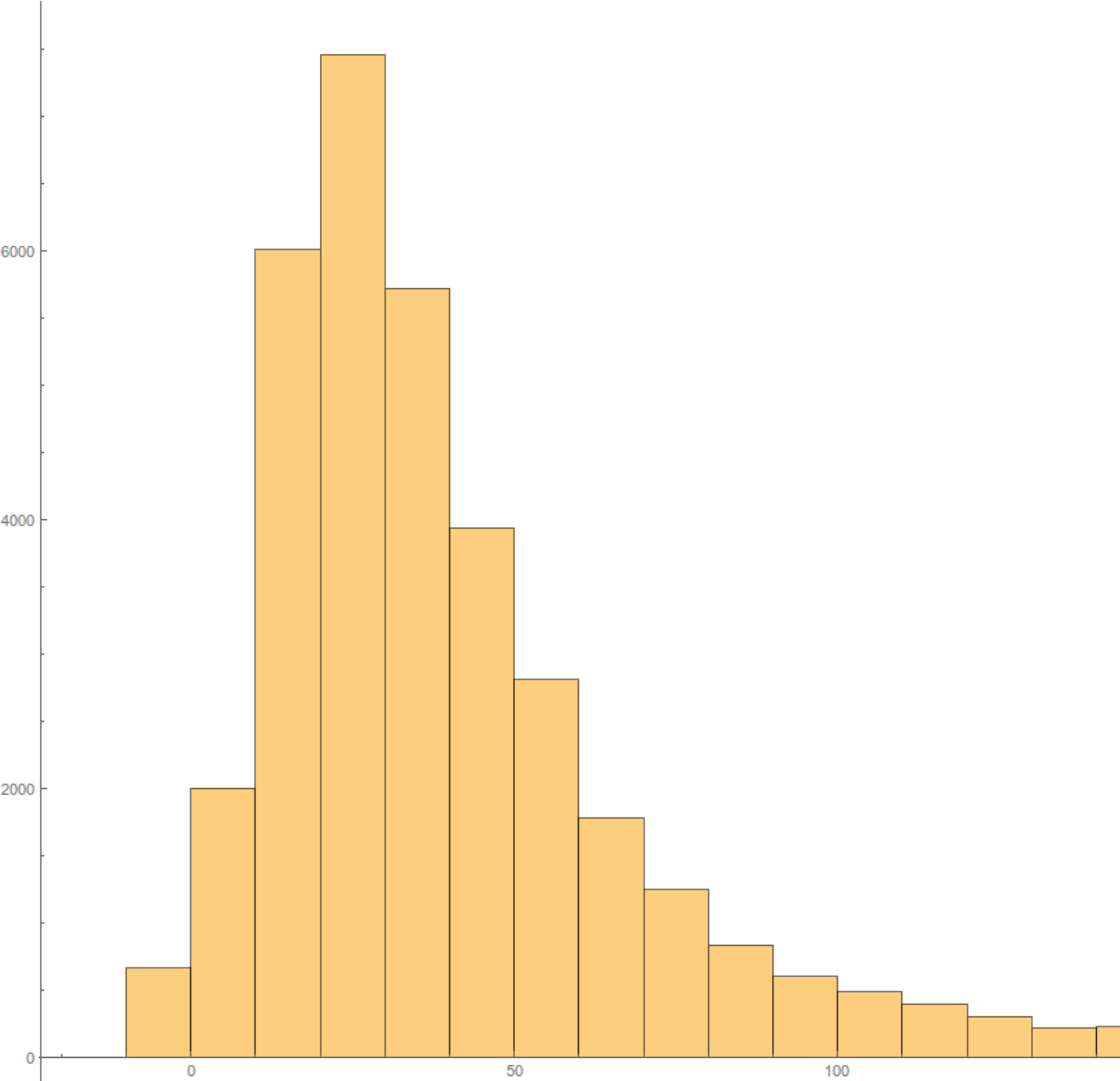}}
\scalebox{0.12}{\includegraphics{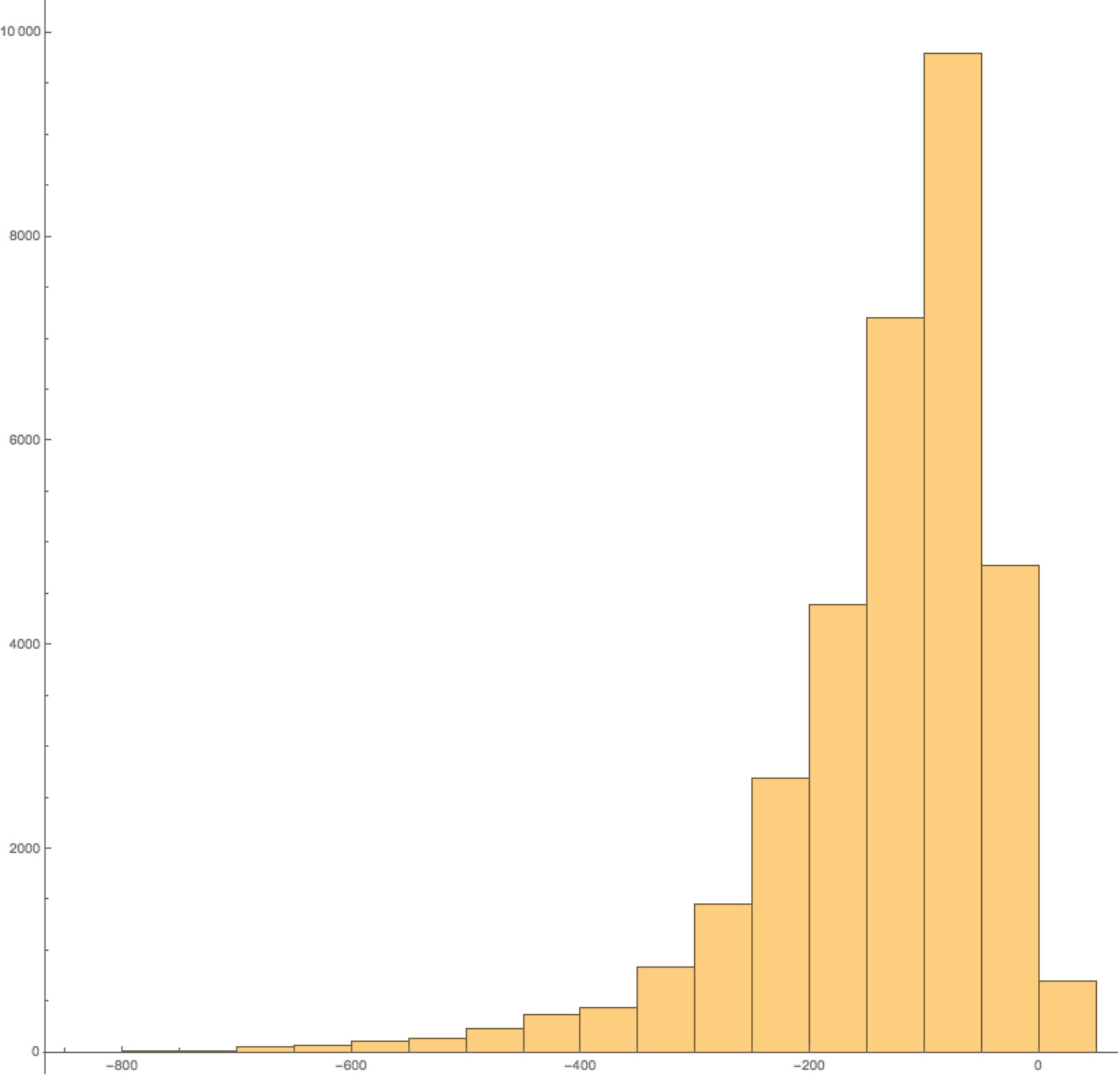}}
\caption{
We see the distribution of Euler, quadratic and cubic Wu characteristic on a list of 36'000 molecules
for which Mathematica has graphs provided.
The Wu characteristic ranges from $-16$ to $1405$. The Euler characteristic
ranges from $-37$ to $28$ on that list. The cubic Wu characteristic from -815 to 16 with a mean of 
$-133.97$.  The mean of the Euler characteristic is $-0.0368$
which is very close to $0$, the mean of the Wu characteristic is $42.58$. The maximum of $\omega$ 
is attained for an {\bf inulin} molecule with 801 atoms, the minimum of $\omega_2$ which appears here
as the maximum of $\omega$ is a disconnected graph containing 16 copies of $K_2$. }
\end{figure}

{\bf Examples.} \\
{\bf 1)} For any cyclic graph $C_n$ with $n \geq 3$, we have $\omega(G)=0$. 
For any 2-sphere like the octahedron or icosahedron $G$, one has $\omega(G)=2$. For 3-spheres like
the 16-cell, the 600 cell or a suspension of a 2-sphere, we have $\omega(G)=0$. For 4-spheres like
a suspension of a 3-sphere or the boundary of $K_2 \times K_2 \times K_2 \times K_2 \times K_2$ 
we have $\omega(G)=2$. For a $2$-torus graph or discrete Klein bottle, we have $\omega(G)=0$ 
again the same than the Euler characteristic. Also for a projective plane, we have Wu characteristic $1$. \\
{\bf 2)} For $G=K_{d+1}$ we have $\omega(G)=(-1)^d$. This remains so after Barycentric subdivision. 
We see that for a triangulation of a ball, $\omega(G)=(-1)^d$.  \\
{\bf 3)} For a figure 8 graph, $\omega(G)=7$. For star graph with $n$ rays, we have
$\omega(G) = n^2-3n+1$. For a sun graph, we have $\omega(G)=2n$. 
For example, for $n=4$, we get $\omega(G)=5$.  For two 2-spheres touching at a vertex, 
we have $\omega(G)=3$.  \\
{\bf 4)} The utility graph $G$ of Euler characteristic $\chi(G)=-3$ has the Wu characteristic 
$\omega(G)=15$. The utility graph is the only graph among all connected graphs with $6$ vertices
for which the Wu characteristic is that high. It is the graph with maximal Wu characteristic
in the class of graphs with $6$ vertices. \\
{\bf 5)} For a $k$-bouquet of 2-spheres glued together at one point, 
the Wu characteristic is $k+1$.  \\
{\bf 6)} For a $k$-bouquet of 1-spheres, there are no triangles. The Wu curvature
at the central vertex is $d=2k$ and zero at every other place. The Wu characteristic is 
$(k-1)(4k-1)$.  \\
{\bf 7)} For a sun graph with $k$ rays, the Wu characteristic is $2k$. Such graphs have no triangles. 
The total curvature contribution of each ray is $2$. \\
{\bf 8}  For a star graph with $n$ rays, the Wu characteristic is $n^2-3n+1$. For example,
for $n=0$, it is $1$, for $n=1$ it is $-1$ for $n=9$ it is $55$. \\
{\bf 9)} Adding a one dimensional hair to a 2 sphere reduces the Wu characteristic by $2$.  \\
{\bf 10)} The Wu characteristic of the cube graph is $20$, the Wu characteristic of
the dodecahedron is $50$. Both graphs have no triangles and constant vertex degree $d=3$
so that in both cases, the curvature is constant $5/2$. \\
{\bf 11)} The Wu characteristic of two crossing circles is $14$, the Wu curvature of a crossing
being $7$ and otherwise being zero.  \\  
{\bf 12)} The Wu characteristic of the tesseract is $112$. It is a graph without triangles
with constant Wu curvature $K=(1-d/2)(1-2d)=7$, where $d=4$ is the vertex degree. Since there are 16
vertices, the Wu characteristic is $112$. As the Euler curvature is $(1-d/2)=-1$, the 
Euler characteristic is $-16$. Of course, a triangulation of the tesseract, the boundary of
$K_2 \times K_2 \times K_2$ is a 3-sphere of Euler characteristic $0$ and Wu characteristic 
also equal to $0$. \\
{\bf 13)} For a suspension of a disjoint union of a circle (which has the Betti vector 
$(1,1,2)$ and so Euler characteristic $2$), the Wu characteristic is $2$. \\
{\bf 14)} For the Adenine, Guanine, Cytosine and Thymine graphs, the 
main bases in DNA and RNA, the Wu characteristics are 15,17,12 and 18. Since the 
Wu Characteristic measures an interaction between neighboring parts where equal charges
repell each other and unequal attract, the interaction energy makes some sense. The bonds
are mainly negatively charged, while the atom nuclei are positively charged. 

\begin{figure}
\scalebox{0.22}{\includegraphics{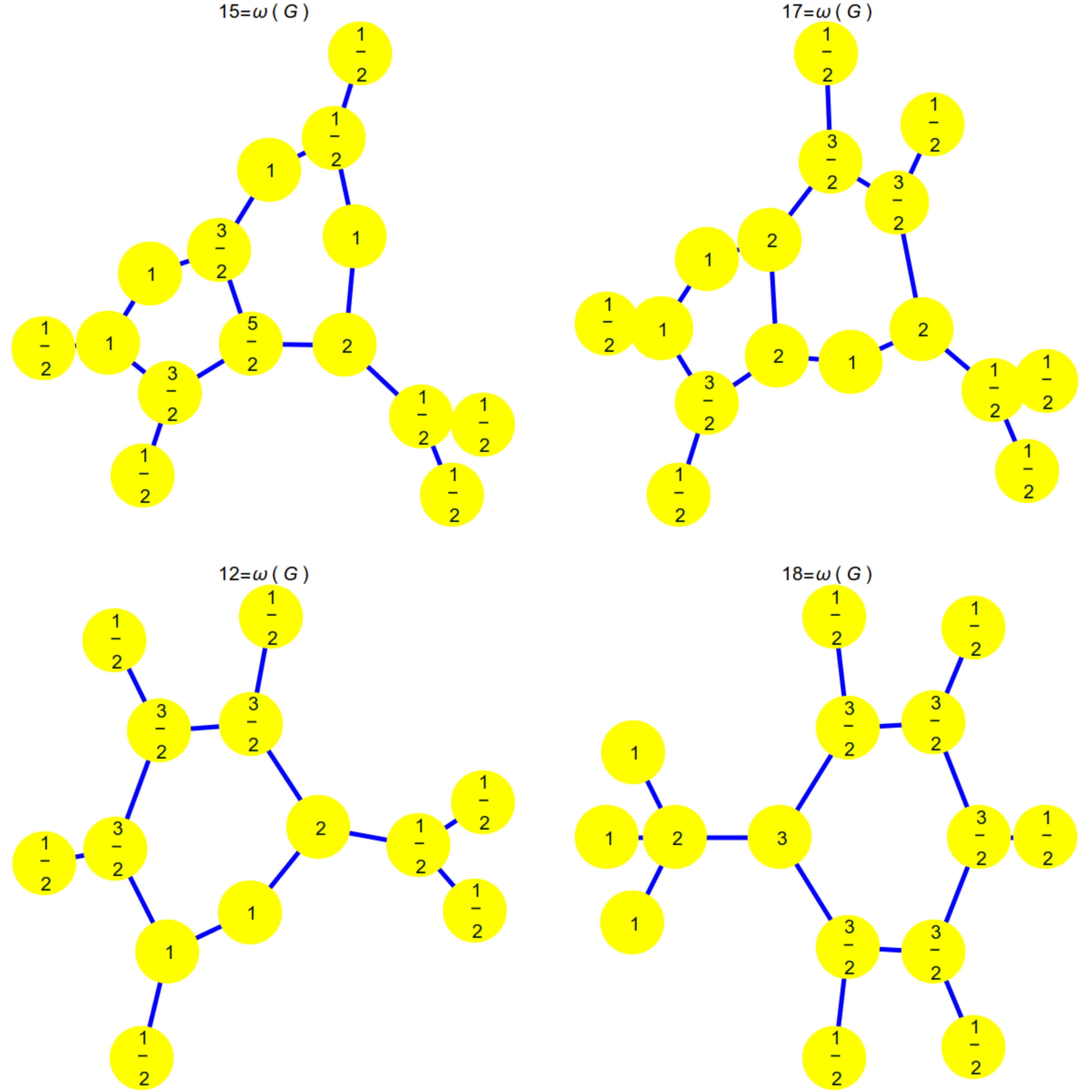}}
\caption{
The Wu curvatures of Adenine,Guanine, Cytosine and Thymine for which
the Wu characteristic are 15,17,12 and 18. For graphs without triangles,
the Wu characteristic gives an interaction energy, where bond-bond and vertex-vertex
interactions count positive and vertex-bond interactions are negative.
As bonds are mostly occupied by electrons and vertices by nuclei,
this interpretation Wu functional has some merit. The Wu functional is maybe a too
simple functional on molecules to be considered useful as it is even simpler than H\"uckel 
theory, which involves the eigenfunctions of the Laplacian. However, it could be
important on a more fundamental level when looking at the fabric of space.
}
\end{figure}

\section{Linear valuations} 

The {\bf $f$-vector} $v(G)$ of a finite simple graph $G$ is defined as 
$$  v(G) = (v_0,v_1,\dots, v_d) \; , $$
where $d$ is the maximal dimension of $G$. This means that $d+1$ is the clique number and $v_d$ 
the volume, counting the number of facets, maximal cliques in $G$. All the entries $v_k(G)$ are 
invariant valuations. Hadwiger's theorem shows that the list $v_0,\dots,v_d$ is a basis for
the linear space of invariant valuations in $G$. While one can see the $v_k$ 
as functionals on graphs, we look at it as a valuation, a functional on the set of subgraphs of $G$. 
It naturally defines a functional on the set of simplicial sub complexes of the Whitney complex of $G$,
which is the traditional way to look at valuations. The simplices in $G$ form the analogue of 
convex sets in integral geometry or geometric probability and subgraphs of $G$ are the analogue
of finite union of convex sets. Euler characteristic $\chi(G)= v_0-v_1+v_2-...$ is an important
functional. It is a valuation on $G$, assigning to every subgraph $A$ of $G$ the number $\chi(A)$. 
Every valuation on $G$ can be assigned a vector $\phi$ as $X(A) = \phi \cdot v(A)$. For the 
Euler characteristic, this vector is $\chi_1=(1,-1,1,-1 \dots, \pm 1)$.
Since we look at multi-linear valuations in a moment, we call classical valuations also 
linear valuations. A natural basis in the $d+1$ dimensional vector space of all linear
valuations of $G$ are the {\bf Barycentric vectors} $\chi_1,\dots \chi_{d+1}$, the 
eigenvectors of $A^T(G)$, where $A$ is the {\bf Barycentric refinement operator} which 
maps the $f$-vector of $G$ to the $f$-vector of its Barycentric refinement $G_1$. 
The {\bf Barycentric refinement matrix} is explicitly known as
$$  A_{ij} = i! S(j,i)  \; , $$
where $S(j,i)$ are the {\bf Stirling numbers} of the second kind. The {\bf Barycentric characteristic numbers}
which were algebraically defined like that are natural and especially singles out Euler characteristic. 
If we would not know about Euler characteristic, we would be forced to consider it now. \\ 

{\bf Examples.} \\
{\bf 1)} If $G$ has no triangles, then every edge gets mapped into two edges.
There are $|V|+|E|$ new vertices in the refinement. The matrix $A$ is
$$ A=\left[ \begin{array}{cc} 1 & 1 \\ 0 & 2 \\ \end{array} \right] \; . $$
{\bf 2)} If $G$ is two dimensional without tetrahedra, then every triangle becomes
$6$ triangles. Every edge becomes doubled and additionally there are 6
new edges for each of the triangles. The number of new vertices is the
sum of the number of vertices, edges and triangles. The matrix $A$ is
$$ A=\left[ \begin{array}{ccc} 1 & 1 & 1 \\ 0 & 2 & 6 \\ 0 & 0 & 6 \\ \end{array} \right] \; . $$
{\bf 3)} If $G$ is three dimensional without $K_5$ graphs, then every tetrahedron
splits into $24$. Every triangle gets split into $6$ and then there are 36 new triangles
coming from tetrahedra etc. The matrix $A$ is
$$ A= \left[ \begin{array}{cccc} 1 & 1 & 1 & 1 \\ 0 & 2 & 6 & 14 \\ 0 & 0 & 6 & 36 \\ 0 & 0 & 0 & 24 \\ \end{array} \right] \; .$$

In the case $d=4$ for example, this matrix is 
$$ A = \left[
                 \begin{array}{ccccc}
                  1 & 1 & 1 & 1 & 1 \\
                  0 & 2 & 6 & 14 & 30 \\
                  0 & 0 & 6 & 36 & 150 \\
                  0 & 0 & 0 & 24 & 240 \\
                  0 & 0 & 0 & 0 & 120 \\
                 \end{array}
                 \right] \; . $$
If $\chi$ is an eigenvector of $A^T$ to the eigenvalue $\lambda$, then 
$$  \chi v(G_1) = \chi A v(G) = v(G)^T A^T \chi^T = v(G)^T \lambda \chi^T = \lambda \chi v(G) $$
showing that the valuation scales by a factor $\lambda$ when applying the Barycentric refinement. 
Since the matrix $A$ is upper triangular, its eigenvalues $k!$ are all known and the eigenvectors
$\chi_k$ of $\lambda_k$ form an eigen-basis of the linear space of valuations. In the case $d=4$
for example, the basis is 
$$ \left\{ \left[ \begin{array}{c} 1 \\ -1 \\ 1 \\ -1 \\ 1  \\ \end{array} \right], 
           \left[ \begin{array}{c} 0 \\ -22 \\ 33 \\ -40 \\ 45 \\ \end{array} \right], 
           \left[ \begin{array}{c} 0 \\ 0 \\ 19 \\ -38 \\ 55 \\ \end{array} \right], 
           \left[ \begin{array}{c} 0 \\ 0 \\ 0 \\ -2 \\ 5 \\ \end{array} \right], 
           \left[ \begin{array}{c} 0 \\ 0 \\ 0 \\  0 \\ 1 \\ \end{array} \right] \right\} \; . $$
The first one is the eigenvector to the eigenvalue $1$ leads to Euler characteristic which manifests
itself as a {\bf Barycentric characteristic number}. The last one is the {\bf volume}, the number of facets
of a sub graph. A statement completely equivalent to the Dehn-Sommerville relations is:

\begin{thm}[Dehn-Sommerville]
If $d+k$ is even, then the Barycentric characteristic 
numbers satisfy $\chi_k(G)=0$ for every $d$-graph. 
\label{barycentric}
\end{thm}

\begin{figure}
\scalebox{0.2}{\includegraphics{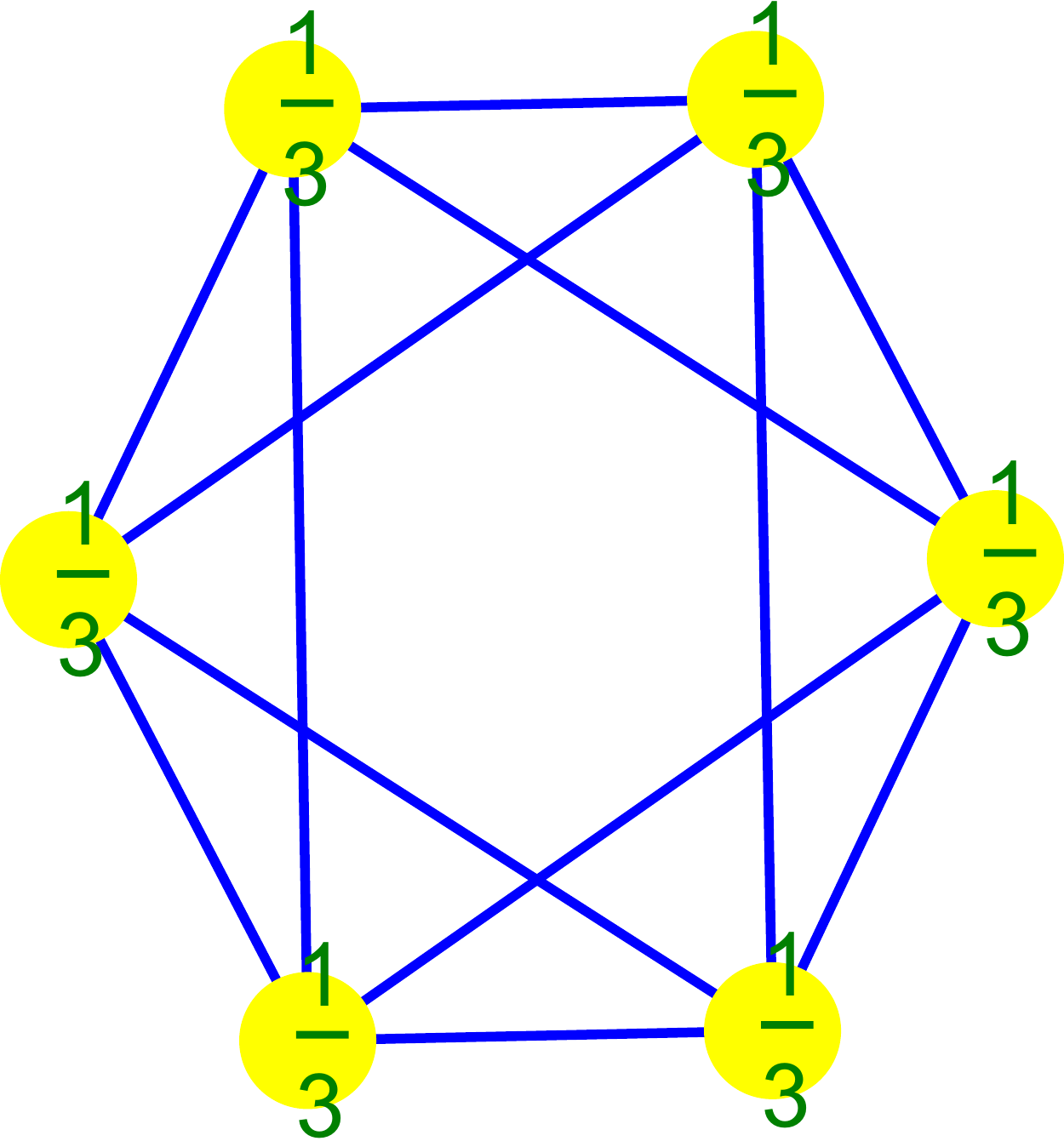}}
\scalebox{0.2}{\includegraphics{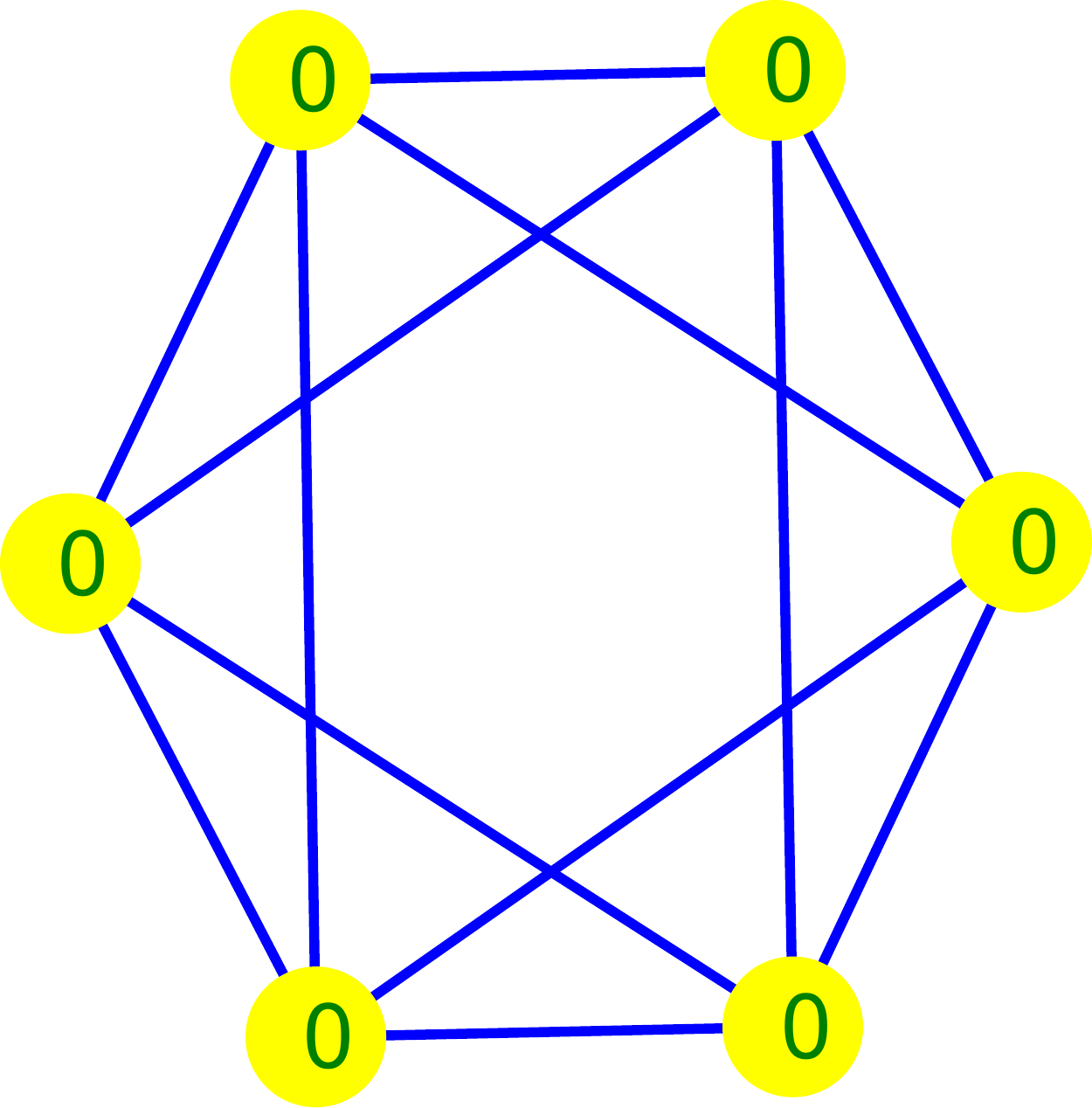}}
\scalebox{0.2}{\includegraphics{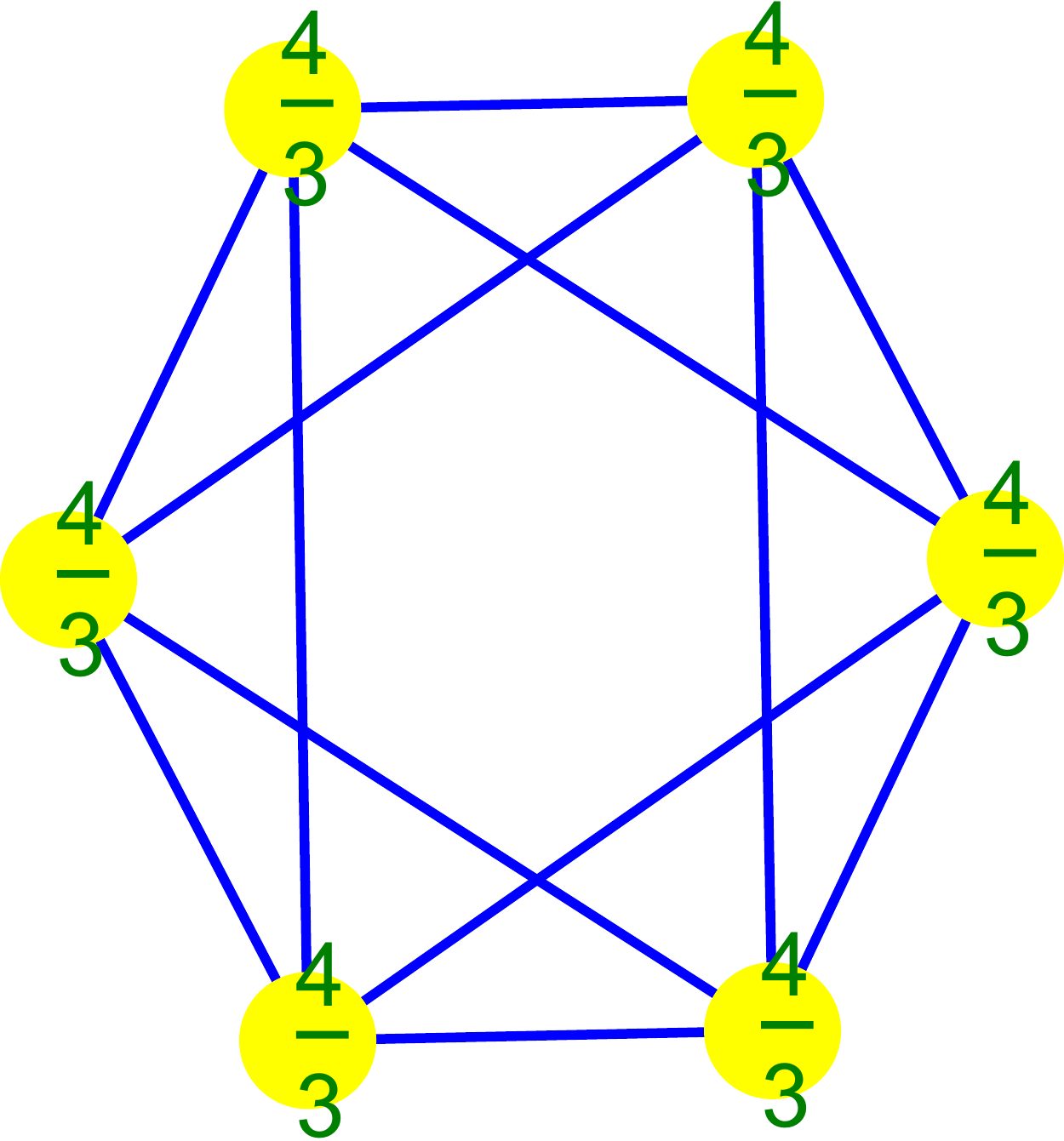}}
\caption{
The curvatures of the Barycentric numbers in the case $\chi_1=\chi$,
$\chi_2$ which is zero and $\chi_3$ which is the area.
}
\end{figure}

This gives more information than the usual Dehn-Sommerville relations
as it also proves immediately that the dimension of the Dehn-Sommerville space is $[(d+1)/2]$,
where $d+1$ is the clique number and $[t]$ is the largest integer smaller or equal to $t$.
If the dimension would be larger, then there would be an other invariant which is zero for
all geometric graphs. It would also be zero for the cross polytop, where we know the maximality.
It also removes any mystery about where these invariants come from or how they can be found. 
It is linear algebra which forces them on us. We actually discovered this theorem, not realizing 
first that they are the Dehn-Sommerville relations. \\

Here is the proof of Theorem~(\ref{barycentric}): It uses Gauss Bonnet and a {\bf suspension decent 
argument}: 
\begin{proof}
Let $\G$ be the class of $d$-graphs for which $\chi_k(G) \neq 0$ 
or for which there is vertex with nonzero curvature. We show that this class is empty by proving
that any graph $G$ in $\G$ for which some unit ball can be extended in $G$ remaining a ball
is either a cross-polytope or can be reduced to a smaller example. 
By definition, for $G \in \G$ there is always a vertex with nonzero curvature.
Take a graph $G$ in $\G$ with minimal vertex cardinality in $\G$. Now look at the suspension of the 
unit sphere $S(a)$. This graph is again in $\G$. It must be $G$ because as a subgraph it has to have less vertices and
therefore $\chi_k(G)=0$ with zero curvature everywhere contradicting the curvature at $a$ to be nonzero.
As $G$ agrees with the suspension of $S(a)$, take add an other vertex $b \in S(a)$ to the unit ball $B(a)$
and call it $H$. If no such other vertex would exist, then $G$ would be a cross polytope. Now the 
closure of $H$ is a sphere which is smaller than $G$ and so has everywhere zero curvature. That 
contradicts that the curvature at $a$ is nonzero. 
\end{proof}

Here is an other local necessary condition for $X(G)=0$ for a linear valuation. 

\begin{lemma}[Puiseux type formula]
a) If $X(G)=0$ for all $d$-graphs, then $2X(B(x))=X(S(x))$ for every $v \in V$. \\
b) If $2X(B(x))=X(S(x))$ for every $x \in V$, then $X(G)=0$.
\end{lemma}
\begin{proof}
a) Look at the suspension $U$ of $S(x)$ using a second vertex $y$. Since $U$ is again a $d$-graph, 
we have $X(U)=0$. The valuation condition shows 
$$  2 X(B(x) = X(B(x) + X(B(y)) = X(S(x)) + X(U) = X(S(x)) \; . $$
b) The condition $2X(B(x))=X(S(x))$ implies that $X(U)=0$ for any double suspension of $S(x)$. 
As in the proof above this means $K(x)=0$ for all $x$ so that $X(G)=0$.
\end{proof}

It implies for example that if a graph has all unit spheres of Euler characteristic $2$, then
$\chi(G)=0$. \\

{\bf Examples:}  \\
{\bf 1)} 2-sphere with a vertex of degree $6$ shows that $K(x)=0$ for a single
vertex does not necessarily imply $2X(B(x))=X(S(x))$. Its only the global condition of zero
curvature which implies it. \\
{\bf 2)} If $k=1$ and odd $d$, then ${\rm dim}(S(x))=d-1$ is even
and $\chi(S(x))=2$. \\

In the case $d=4$, the Dehn-Sommerville space is $2$-dimensional and spanned either 
by the {\bf Barycentric characteristic vectors} 
$$   \chi_2 = (0,22,-33,40,-45) \; ,\vspace{1cm} \chi_4=(0,0,0,2,-5) \;  $$
which are eigenvectors of $A_4^T$ for the Barycentric refinement operator on graphs with clique number $4$,
or then by the {\bf classical Dehn-Sommerville vectors}
$$   d_0 = (0,2,-3,4,-5)     \; ,\vspace{1cm} d_2-(0,0,0,4,-10) \; . $$
In the form defined here, we have $d_{-1}=(-1,1,-1,1,-1)$ and $d_3=(0,0,0,2,-5)$. 
The valuation $\chi_4$ (which is parallel to $d_2$) is a trivial boundary invariant, expressing that 
counting $5$ times the number of $4$-simplices is $2$ times the number of $3$-simplices. 
This is the Euler handshake in the dual graph of $G$, where the maximal $d$-simplices are the vertices and two 
such simplices are connected if they intersect in a $(d-1)$-simplex. Graphs with this
property are sometimes called {\bf pseudo manifolds}.  
The vector $\chi_2 +2 \chi_4 = (0,22,-33,44,-55)$ is parallel to the 
classical Dehn-Sommerville vector $d_0=(0,2,-3,4,-5)$.
Let $G$ be the $4$-sphere obtained by taking the suspension of the suspension 
of the octahedron graph. It is a graph $G$ with $10$ vertices and $40$ edges. Its $f$-vector is 
$$ v(G) = (10,40,80,80,32)  \; . $$
When taking the dot product of this with the above Barycentric basis vectors, we get the {\bf Barycentric 
invariants} $\chi_1(G)=2,\chi_2(G)=0, \chi_3(G)=240,\chi_4(G)=0,\chi_5(G)=32$. We see that
the Euler characteristic is $2$ as it has to be for any $4$-sphere. We also see the
two Dehn-Sommerville relations and last but not least that the volume is $32$. A special case of the 
above Dehn-Sommerville theorem is that for odd-dimensional $d$-graphs, the Euler characteristic
is zero. Also, in any dimension, the invariant to $\chi_d$, the boundary invariant is always
zero. It is a manifestation of the fact that we assumed that the graph $G$ has no boundary. 
The volume $\chi_{d+1}(G)=(0,0,\dots,0,1)$ is a valuation which is never zero for a $d$-graph or 
more generally for a graph with clique number $d+1$. \\

Lets look at an other $4$-graph $G$, the product $H \times H$ of two $2$-spheres given as octahedron graphs $H$.
Since the octahedron graph $H$ has the $f$-vector $v(H)=(6,12,8)$, which has $26 = \sum_i v_i$  simplices, 
the product graph has $676$ vertices. Its $f$-vector is 
$$ v = (676, 8928, 28992, 34560, 13824) \; . $$
Again, by taking the dot product of this $f$-vector with the basis vectors, we get the Barycentric invariants. 
They are $\chi_1 = 4, \chi_2=0,\chi_3=-2112, \chi_4=13824$. The Euler characteristic is $4$ as it has to be for
the product of two $2$-spheres, the two zero values are the Dehn-Sommerville relations and $13824$ is the volume.
The graph has volume $13824$ counting the complete subgraphs $K_5$. \\

As a third example, lets look at the suspension $G$ of $S^2 \times S^2$ just constructed before.
This is a discrete $5$-dimensional graph, a graph with $678$ vertices but it is
no more a $5$-graph, as the there are now by construction two vertices for which the unit sphere is 
not a sphere $S^4$. Indeed, the unit sphere is a graph whose topological realization is the standard
$S^2 \times S^2$. The $f$-vector of $G$ is
$$ v(G) = (678, 10280, 46848, 92544, 82944, 27648) \; . $$
The Euler characteristic is now $2$ and not zero it would have to be if it were a $5$-graph. 
Interestingly enough, the higher Barycentric invariants are still zero, as 
$$ (v \cdot \chi_1, v \cdot \chi_2, v \cdot \chi_3, v \cdot \chi_4, v \cdot \chi_5) 
   = (2, -231152, 0, 114432, 0, 27648) \; . $$
It will be interesting to study for which pseudo $d$-graphs of this type, higher Dehn-
Sommerville relations still hold.  \\

Valuations extend naturally from graphs to chains by 
linearity. This is the case for linear as well as multi-linear valuations.
For a chain $H=\sum_x a_x x$  on a graph with simplices $x$, 
define its $f$-vector $v(H)$ by $v_k(H) = \sum_{{\rm dim}(x)=k} a_x$. In particular, if $H=G$, then 
$v_k(G) = \sum_{{\rm dim}(x)=k} 1$. Given a valuation $X$ defined by a vector $\chi$, define 
$X(G) = \chi \cdot v(G)$. \\

Given a graph $f$ described in the Stanley-Reisner ring as $f=\sum_i p_i$ with quadratic 
free monoids $p_i$ in the variables $x_1,\dots,x_n$ representing the vertex set $V=\{ x_1,\dots,x_n \}$.
A chain over $G$ is an element $\sum_i a_i p_i$, where $a_i$ are integers. The set of all chains forms
an Abelian group. This and the corresponding construction of homology 
is one reason why chains were introduced by Poincar\'e. An
other reason for the need of chains is that the boundary of a graph is no more a graph in general,
nor are quotients of group actions. As we have noted in \cite{TuckerKnill}, for a group $A$ acting
as automorphisms on a graph $G$, the Riemann-Hurwitz formula 
$\chi(G) = n \chi(G/A) - \sum_{x} (e_x-1)$ holds, where $x$ sums over all simplices in $G$
and $e_x=1+\sum_{a \neq 1, a(x)=x} (-1)^{{\rm dim}(x)}$ is the {\bf ramification index}.
This formula holds also generally on the larger class of
chains as it just reduces to the Burnside lemma (which is the special case if $G$ has no edges).
In general, one first has to do Barycentric refinements before applying the quotient operation in order to stay 
within the class of graphs. Still, if $G$ is a $d$-graph, the quotient $G/A$ 
is a {\bf discrete orbifold} in general. By the way, the Riemann-Hurwitz
idea goes over from Euler characteristic to valuations. One just has to adapt 
$(-1)^{{\rm dim}(x)}$ to $\psi({\rm dim}(x))$ if $\psi$ is the vector definining the valuation 
$G(V) = v(G) \cdot \psi$.  
We have not yet investigated Riemann-Hurwitz for $k$-linear valuations but expect
things to work similarly, however to become more interesting. \\

Here are some examples showing the need to go from graphs to chains: 
lets take the star graph $S_3$ for example with $f_G = ab + ac + ad$. 
If the orientation on the simplices is chosen from the way the 
monomials ere written, the boundary $\delta f$ is $b-a + c-a + d-a = a+c+d -3a$ which is now only a 
chain and no more a graph. 
A second example is to let the group $A=Z_4$ act on $G=C_4$. The quotient $G/A$ is the chain $a+b+2 ab$
which no more a graph. Both the Euler characteristic of $G$ and the quotient are $0$, there are no ramification
points of the group action. 
Graphs with multiple connections, multi-graphs or graphs with selfloops must be considered examples of chains. 

\section{The $f$-matrix}

Given a graph $G$, define the {\bf $f$-matrix} or quadratic $f$-form as
$$   V_{ij}(G) = | \{ (x,y) \; | \; x \sim K_{i+1}, y \sim K_{j+1}, x \subset G, y \subset G, 
                                 x \cap y \neq \emptyset \; \} |  \; .  $$ 
It is a symmetric matrix counting the 
number of ordered pairs of $i$-simplices and $j$-simplices in $G$ which have non-empty intersection. 
For example, if $G$ is the star graph with $3$ spikes, its
$f$-vector is $v(G)=(4,3)$ as there are $4$ vertices and $3$ edges. The $f$-matrix $V(G)$ is 
$V=\left[ \begin{array}{cc} 4 & 6 \\ 6 & 9 \\ \end{array} \right]$ as there are $4$ self-intersections of vertices,
$3+6=9$ intersections of edges and $3+3$ intersections of vertices with edges.  \\

A quadratic valuation of a graph $G$ can now be written as
$$ X(G) = (V(G) \phi) \cdot \psi = V(G) \phi \psi \; , $$ 
where $\phi,\psi$ are two $(d+1)$-vectors $\phi,\psi$, if the clique number
of $G$ is $d+1$. For example, if $\phi=\psi=(1,-1,1,\dots)$, then $X$ is the Wu characteristic. 
In the case of the star graph $G$ above, we have 
$$  [1,-1] \left[ \begin{array}{cc} 4 & 6 \\ 6 & 9 \\ \end{array} \right]  
           \left[ \begin{array}{c} 1 \\ -1 \end{array} \right] = 1 \; . $$ 
The graph $G$ is one of the rare cases, where $V(G)$ has a zero eigenvalue. The Perron-Frobenius
eigenvector is $(2,3)$. As $V(G)$ is symmetric, the eigenvector to $0$ is perpendicular: $(-3,2)$. \\

Given two subgraphs $A,B$ of $G$, define the {\bf intersection form} 
$$V_{ij}(A,B)  =  | \{ (x,y) \; | \; x \sim K_{i+1}, y \sim K_{j+1}, x \subset A, y \subset B, 
                                 x \cap y \neq \emptyset \; \} |                                  $$ 
as the number of ordered pairs $(x,y)$, where $x$ is an $i$-simplex in $A$ and $y$ is a $j$-simplex in $B$ for which 
$x \cap y$ is a non-empty graph. A quadratic valuation $X$ can be written using two vectors $\phi,\psi$ as
$$ X(A,B) = \psi \cdot V(A,B) \phi  \; . $$
For example, if $A=(a+b+c+ab+ac)$ is a linear subgraph of the above star graph $G$ and 
$B=(a+c+d+ac+ad)$ is an other linear subgraph of $G$, then 
$$ V(A,B) = \left[ \begin{array}{cc} 2 & 2 \\ 2 & 1 \\ \end{array} \right]  $$
as there are are two matches $aa,cc$ for vertices, four matches $(ab)(ac)$,$(ab)(ad)$,$(ac)(ac)$,$(ac),(ad)$ 
for edges and three pairs $a (ac), a (ad), c (ac)$ of vertices in $A$ and edges in $B$. Now
$$ X(A,B) = [1,-1] \left[ \begin{array}{cc} 2 & 3 \\ 3 & 4 \\ \end{array} \right]  
           \left[ \begin{array}{c} 1 \\ -1 \end{array} \right] = 0  \; .   $$  

Let $G$ be the ``16-cell" again, the regular Platonic 3-sphere which is the suspension 
of the octahedron. Its $f$-vector is $(8, 24, 32, 16)$. Its $f$-matrix  is
$$  V(G) = \left[
                 \begin{array}{cccc}
                  8 & 48 & 96 & 64 \\
                  48 & 264 & 480 & 288 \\
                  96 & 480 & 800 & 448 \\
                  64 & 288 & 448 & 240 \\
                 \end{array}
                 \right] \; . $$
Lets look at the Barycentric eigenspace of the $4$-dimensional space of valuations on $G$: 
$$ \left\{ \left[ \begin{array}{c} 1 \\ -1 \\ 1 \\ -1 \end{array} \right], 
      \left[ \begin{array}{c} 0 \\  22 \\ -33 \\ 40 \end{array} \right], 
      \left[ \begin{array}{c} 0 \\ 0 \\ -1 \\ 2     \end{array} \right], 
      \left[ \begin{array}{c} 0 \\ 0 \\  0 \\ 1     \end{array} \right]  \right\} \; . $$
Lets call them $\{ \chi_1,\chi_2,\chi_3,\chi_4 \; \}$. 
Taking the dot product with the $f$-vector produces the Barycentric characteristic numbers
are $\chi_1(G) = 0, \chi_2(G)=112,\chi_3(G)=0,\chi_4(G)=16$. The Euler characteristic $\chi(G)=\chi_1(G)$ 
is zero on the graph $G$ as for all $3$-graphs. Lets now compute the quadratic Barycentric 
characteristic matrix. It is defined as 
$$ \Omega_{ij}(G) = \chi_i \cdot V(G) \chi_j \; .  $$
In this example it is given by 
$$  \Omega(G) = 
\left[ \begin{array}{cccc}
                  0 & 112 & 0 & 16 \\
                  112 & 10176 & 224 & 1152 \\
                  0 & 224 & -32 & 32 \\
                  16 & 1152 & 32 & 240 \\
                 \end{array} \right] \; . $$
The first entry $\chi_1 \cdot V(G) \chi_1$ is the Wu characteristic, which is also zero. 
The first row or column agree with the Barycentric characteristic numbers. We will prove 
that the first entry is the same and the zero entries in the first row and column are there. 
These zero entries are the quadratic valuations which were conjectured to be zero
by Gr\"unbaum. Establishing the relations
$$   \chi_1 V(G) \chi_k= v(G) \chi_k $$
for any $k$ will prove that and so prove the conjecture of Gr\"unbaum positively. 

\section{Gauss-Bonnet} 

For a linear valuation $X$, the curvature 
$$  K(x) = \sum_{k=0} X(k) \frac{V_{k-1}(x)}{(k+1)} $$ 
with $V_k(x)=v_k(S(x))$ and $V_{-1}(x)=1$ satisfies the Gauss-Bonnet formula 
$$  X(G) = \sum_{x \in V} K(x)=X(G) \; . $$
Each of the numbers $V_k(x)=v_k(S(x))$ are valuations applied to the unit sphere $S(x)$ counting the number of 
$k$-simplices present in $S(x)$. This can be called the ``fundamental theorem of graph theory"
as for $X(G)=(0,1,0,\dots,0)$ counting the number of edges the curvature $K(x)=V_0(x)/2$ is half the 
vertex degree and the Euler handshake lemma is sometimes called as such. The Euler Handshake is maybe
the simplest version of a Gauss-Bonnet result for graphs, where the sum of local properties, the degree,
adds up to a global property which is twice the length of the graph when seen as a curve. \\

Lets call $K(x)$ the {\bf Euler curvature} if $X$ is the Euler characteristic. Unlike in the continuum, 
where curvature is a notion involving second order derivatives, 
the Euler curvature of a linear valuation is a first order notion. We have experimented with second order
curvatures for Euler characteristic in
\cite{elemente11} and searched since for conditions in two dimensions for which a second order curvature 
would work. It turns out that we were too much obsessed with Puiseux formulas in differential
geometry and therefore searched in two dimensions for curvatures of the form 
$K(x) = 2 |S_1(x)| - |S_2(x)|$, where $|S_r(x)|$ is the vertex cardinality of the sphere $S_r$.
This search for second order curvatures using Puiseux type discrete formulas was fruitless even
in two dimensions. We have now a notion in the form of the Wu 
curvature which is defined as a second order curvature for general finite simple graphs and which 
happens to agree with the Euler curvature on $d$-graphs but manifests as higher order if evaluated
on more general spaces. 
Gauss-Bonnet for linear valuations  easily can be proven as follows (see the introduction in
\cite{cherngaussbonnet}). Look first at the curvature on the Barycentric refinement 
which assigns to a simplex $x$ the value $(-1)^{{\rm dim}(x)}$. Now distribute this curvature to 
vertices by moving to each vertex in $x$ the value $(-1)^{{\rm dim}(x)}/({\rm dim}(x)+1)$. 
For the valuation $X(G) = v_k(G)$ the same procedure gives the curvature 
$$ K(x)=V_{k-1}(x)=\frac{1}{k+1} \;  $$
and the theorem:

\begin{thm}[Gauss-Bonnet]
For any linear valuation $X$, we have
$$  X(G) = \sum_{v \in V} K(v)  \; . $$
\end{thm}

We can do the same thing for multi-linear valuations. \\

Lets define now the curvature for the Wu characteristic. Given a complete subgraph $x$ of $G$,
define 
$$  V_k(x) = \sum_l (-1)^{l} v_{kl}(x)  \; , $$ 
where $v_{kl}(x)$ counts the number of simplices 
$y$ of $G$ for which $x \cap y \neq \emptyset$. 
We have now an integer-valued function on the simplices of $G$ which is the 
sum of the interactions with neighboring simplices including the self interaction.
This {\bf simplex curvature} is
$$ \kappa(x) = \sum_{k=0} X(k) V_{k-1}(x)\;  $$
By definition, $X(G) = \sum_{x \subset G} \kappa(x)$.
If the value $\kappa(x)$ is broken up and distributed equally to the vertices of $x$, we get
a scalar valued function. It is
$$ K_X(v) = \sum_{v \in x} \kappa(x)/({\rm dim}(x)+1) \; , $$
where the sum is over all simplices $x$ in $G$ which contain $v$. The same construction works
in the quadratic as well as higher degree case. 

\begin{thm}[Gauss-Bonnet]
\label{theorem1}
For any multi-linear valuation $X$, we have
$$  X(G) = \sum_{v \in V} K_X(v)  \; . $$
\end{thm}
\begin{proof}
One can prove it by induction with respect to the degree $k$ and use Gauss-Bonnet for one dimensions. 
For $k=1$, we have the case of valuations. To make the induction step reducing it from $k$ to $k-1$, 
look for the valuation $A \to \omega(A_1,\dots,A_{k-1},A)$ and its curvature $k_{A_1,\dots,A_{k-1}}(v)$
which is a degree $k-1$ valuation. By induction it satisfies Gauss Bonnet $X(A_1,\dots,A_k) = \sum_w K(w)$ 
for a curvature $w \to K(w)=K_{A_1,\dots,A_{k-1},v}(w)$. This shows $K(A_1,\dots,A_k) = \sum_{v,w} K(v,w)$
for some curvature depending on two variables. Now move the value of $K(v,w)$ for any $v \neq w$ equally
onto the vertex $v$ and $w$ to get a scalar curvature for $X$. 
\end{proof} 

The curvature of a quadratic valuation is now a second order difference operator as the geometry 
of the ball $B_2(x)$ of radius $2$ matters. As we assumed $k$-valuations to be local in the sense
that we discard any contributions $X(x_1,\dots,x_d)$ if their mutual intersection $\bigcap_j x_j$ is
empty, the curvature is localized as such. If we require only that the nerve graph of intersections
of simplices is connected, then the curvature of a $k$-linear valuation has longer range too. For
a $3$-linear valuation for example, we would consider contributions of chains $x y z$, where $x,y,z$
are edges building a linear graph of length $2$. As for now, we don't count such connections in the 
valuation, the reason being that the theorems would not work. 
Including long range valuations could be useful when looking at a more exhausting list of
invariants. But currently, we want curvature to be local as this is the case in differential geometry. \\

{\bf Examples.} \\
{\bf 1)} If $G$ is a wheel graph $W_n$ with boundary $C_n$, then the Euler curvature is $1/6$ on 
the boundary and $1-n/6$ at the center. For any wheel graph, the Wu curvature is $1$ in the interior 
and $0$ on the boundary. \\
{\bf 2)} If $G$ is a $3$-ball like a pyramid construction over an icosahedron, then the
Euler curvature is supported on the boundary. The Wu curvature however is $-1$ in interior and $0$
on the boundary. \\

{\bf Remark.} \\
Also higher degree multi-linear valuations satisfy {\bf Poincar\'e-Hopf} and 
index averaging theorems as we will see later. 
For any function $f$ on the vertex set and any multilinear valuation, we have an {\bf index}
$i_f(v)$ on the vertex set $V$ such that $\sum_v i_f(v) = X(G)$. The index of $f$ for the 
Euler  characteristic is defined as
$$ i_f(x) = \chi(B^-_f(x))-\chi(S^-_f(x)) \; , $$
where $B^-_f(x)$ is the graph generated by $\{ y \; | \; f(y) \leq f(x) \; \}$ and
$S^-_f(x)$ is the graph generated by $\{ y \; | \; f(y) < f(x) \; \}$. 
There is a similar formula for the Wu characteristic. We also can show that the expectation is
curvature ${\rm E}[i_f]=K$, when integrating over a reasonable space of functions. 
Unlike for Euler characteristic, the indices can now be nonzero even at places which 
are usually regular. 

\section{The Gr\"unbaum conjecture} 

Gr\"unbaum \cite{Gruenbaum1970} conjectured in 1970 that multi-linear Dehn-Sommerville invariants 
like quadratic valuations
$$ X(G) = \sum_{i,j} a_{i,j}  V_{ij}(G) \; , $$
exist which vanish on geometric graphs. Here, $V_{ij}(G)$ is the number of ordered pairs $(x,y)$ 
of $i$-simplices and $j$-simplices which intersect in a nonempty simplex. We answer this positively:
for every degree and every classical Dehn-Sommerville invariant, there is a corresponding
multi-linear degree $d$ invariant which is zero on $d$-graphs, as Gr\"unbaum has suspected. \\

Lets recall the quadratic self-intersection form 
$V_{ij}(G)$ which counts the number of $i$-simplices intersecting with the number of $j$-simplices.
Given a graph $G$ with clique number $d+1$. 
Define {\bf Dehn-Sommerville space} $\D_d$ is the linear space of $1$-dimensional valuations
which are spanned by eigenvectors $\chi_k$ of the Barycentric operator $A^T$ 
for which $d+k$ is even. For every $X$ in $\D_d$ and any $d$-graph $G$ we have $X(G)=0$. 
In other words $\chi_k^T v(G) = 0$, if $v(G)$ is the $f$-vector of $G$.  \\

We look now at quadratic valuations of the form
$$ Y(G) = \sum_{ij} \psi(j) (-1)^i V_{ij}(G) = \psi \cdot V(G) \chi_1   $$
and compare this quadratic valuation with the linear valuation
$$ X(G) = \sum_i \psi(j) v_j(G) = \psi \cdot \chi_1 \; . $$

\begin{thm}
\label{theorem4}
For any $d$-graph, the linear valuation $X(G)=v(G) \psi$ evaluates on $G$
to the same value than the quadratic valuation $\chi_1^T V(G) \psi$ 
Especially, if $X(G)=0$, then $Y(G)=0$. 
\end{thm}

In particular, this holds for odd-dimensional $d$-graphs and $\chi=\chi_1$, where we obtain
that the Wu invariant $\omega(G)=0$. We will later prove the stronger claim that$\omega(G)=\chi(G)$
for such graphs, for which the Euler characteristic is zero. 

\begin{proof}
We use induction with respect to dimension. For $d=0$ it is trivial since no simplices can interact and $X=Y$ 
holds trivially. Given the vector $\Psi = (\Psi(0), \dots, \Psi(d))$, we write for a simplex $x$
$\varphi(x) = \Psi({\rm dim}(x))$ so that
$$ X(G) = \sum_{x} \varphi(x)  \;  $$
and
$$ Y(G) = \sum_{x \cap y \neq \emptyset} \sigma(x) \varphi(y) \; . $$ 
Every pair $(x,y)$ of simplices $x,y$ intersect in some simplex $z$. 
Partition the sum into subsets, for which the intersections are the simplex $z$. Then
$$ Y(G) = \sum_{z} \sum_{x,y, x \cap y = z} \sigma(x) \varphi(y)   \;,   $$
where the first sum is over all simplices of $G$. 
The induction assumption implies
$$ \sum_{x,y, x \cap y = \emptyset} \sigma(x) \varphi(y) = \sigma(G) Y(G) = 0    \;  $$
We claim that
$$ \sum_{x,y, x \cap y = z} \sigma(x) \varphi(y) =  \varphi(z)  $$
which immediately proves the theorem. To prove this, partition the sum
further. Let $z=(z_1,\dots,z_k)$. Any $m$-simplex $x$ different from $z$ defines
a simplex of dimension $m-k$ in the $(d-k)$-sphere
$$ S(z) = S(z_1) \cap S(z_2) \cap  \cdots \cap S(z_k) \; . $$
Since $x \cap y = z$, the simplices $x',y'$ defined by $x=z\cup x', y=z \cup z'$ do not
intersect. There are four possibilities: either $y'=z'=\emptyset$ or $y'=\emptyset$
or $x'=\emptyset$ or then - and that is the fourth case - that $x',y'$ are both not empty. 
In the first case we get $\sigma(z) \varphi(z) = \varphi(z)$. In order to show the result we
therefore have to show that the sum of the other three cases is zero.  \\
In the second or third case respectively, we have 
$\sigma(z) \varphi(y)$ or $\sigma(z) \varphi(x)$ and so a contribution $2 (-1)^{d+1} \sigma(z) Y(S(z))$
from the second and third case. \\
In the forth case we have two non-empty non-intersecting simplices in $S(z)$ and a contribution
\begin{eqnarray*}
 \sum_{x',y', x' \cap y' = \emptyset} \sigma(x) \varphi(y)  &=&
   \sum_{x',y'} \sigma(x) \varphi(y)  - Y(S(z))  \\ 
  &=& 2 \chi(S(z)) Y(S(z)) - Y(S(z))  \\
  &=& 2 (\chi(S(z))-1) Y(S(z)) \\
  &=& 2 (-1)^d \sigma(z) Y(S(z))  \; . 
\end{eqnarray*}
We have used that the unrestricted sum $\sum_{x',y'} \sigma(x) \varphi(y)$ (without 
intersections) is equal to $(\sum_{x'} \sigma(x)) (\sum_{y'} \varphi(y)) 
= \chi(S(z)) Y(S(z))$. 
\end{proof}

Similarly, higher order Dehn-Sommerville valuations are zero.  For example, in the cubic case, 
$$   V(G) (X,X,Z) \; , $$
where $X(j)=(-1)^j$ and $V(G)(i,j,k)$ 
counts the ordered lists of simplices $(x,y,z)$ of dimension $i,j,k$.  \\

{\bf Remark}. \\
If $A$ is the Barycentric refinement operator, one could think that since $A v(G)$ is the
$f$-vector of the Barycentric refinement $G_1$, also $A V(G) A^T$ is the $f$-matrix of the
Barycentric refinement $G_1$. This is not the case. It would only hold if $V_{ij}(G)$ counted
the number of {\bf all} pairs $x,y$ of $i$-simplices $x$ and $j$ simplices $y$, 
without selecting only the pairs which intersect in a nonempty graph.  \\

{\bf Examples.}   \\
{\bf 1)} Take the $4$-sphere $G$ with $f$-vector $v=(10, 40, 80, 80, 32)$. 
The Dehn-Sommerville space is $2$-dimensional for $d=4$ and spanned either by 
the classical Dehn-Sommerville vectors
$$ (0,0,0,-2,5), (0,-2,3,-4,5)     $$ 
or then by two eigenvectors of the transpose of the Barycentric refinement operator:
$$ A = \left[ \begin{array}{ccccc}
                   1 & 1 & 1 & 1 & 1 \\
                   0 & 2 & 6 & 14 & 30 \\
                   0 & 0 & 6 & 36 & 150 \\
                   0 & 0 & 0 & 24 & 240 \\
                   0 & 0 & 0 & 0 & 120 \\
                  \end{array} \right] \; . $$
The quadratic {\bf $f$-form} encoding the intersection cardinalities is
$$ V(G) = \left[
                  \begin{array}{ccccc}
                   10 & 80 & 240 & 320 & 160 \\
                   80 & 600 & 1680 & 2080 & 960 \\
                   240 & 1680 & 4400 & 5120 & 2240 \\
                   320 & 2080 & 5120 & 5680 & 2400 \\
                   160 & 960 & 2240 & 2400 & 992 \\
                  \end{array}
                  \right]   \; . $$
This quadratic form encodes in how many ways a $k$-simplex intersects with a $l$-simplex in $G$. 
The Barycentric quadratic valuations are 
$$  \chi_{k,l} = \chi_k^T V(G) \chi_l \; , $$
where $\chi_1=(1, -1, 1, -1, 1), \chi_2 = (0, -22, 33, -40, 45)$,
$\chi_3=(0, 0, 19, -38, 55)$, $(0, 0, 0, -2, 5)$ and $(0, 0, 0,0,1)$ 
are the eigenvectors of $A^T$. The Barycentric characteristic numbers are
$$ \chi = \left[
                  \begin{array}{ccccc}
                   2 & 0 & 240 & 0 & 32 \\
                   0 & -4560 & 7760 & -800 & 1440 \\
                   240 & 7760 & 47440 & 2720 & 5920 \\
                   0 & -800 & 2720 & -480 & 160 \\
                   32 & 1440 & 5920 & 160 & 992 \\
                  \end{array}
                  \right]  \; . $$
The first row are the standard Barycentric characteristic numbers in the eigenbasis
of $A^T$. The two zeros at the place where the classical Dehn-Sommerville space is.
But note that we have now evaluated a quadratic invariant. The fact that
the values are the same than the values of the valuations is only due to the 
fact that we deal with geometric graphs, where also Euler characteristic and 
Wu characteristic agree. The relation
$$ \chi_k^T \cdot V(G) \chi_0 = \chi_k^T \cdot v(G) \;  $$ 
is not true for all networks. \\

{\bf 2)} Here are all the quadratic invariants evaluated on a Barycentric basis for
the octahedron graph $G$ for which 
$$ V(G) = \left[ \begin{array}{ccc} 6 & 24 & 24 \\ 24 & 84 & 72 \\ 24 & 72 & 56 \\ \end{array} \right] $$
and the matrix $\chi_k V \chi_l$ is 
$$ \left[ \begin{array}{lll}
                  2 & 0 & 8 \\
                  0 & -24 & 24 \\
                  8 & 24 & 56 \\
                 \end{array} \right] \; . $$
And here are the cubic invariants: 
$$ \left[
                 \begin{array}{lll}
                  \left[2,0,8\right] & \left[0,-24,24\right] & \left[8,24,56\right] \\
                  \left[0,-24,24\right] & \left[-24,-120,120\right] & \left[24,120,264\right] \\
                  \left[8,24,56\right] & \left[24,120,264\right] & \left[56,264,344\right] \\
                 \end{array} \right] \; . $$
{\bf 3)} First the quadratic invariants for the three sphere, the suspension of the octahedron,
are
$$ \left[
                 \begin{array}{llll}
                  0 & 112 & 0 & 16 \\
                  112 & 10176 & 224 & 1152 \\
                  0 & 224 & -32 & 32 \\
                  16 & 1152 & 32 & 240 \\
                 \end{array} \right]  \; . $$
Here are all the cubic invariants for a three sphere:
\begin{tiny}
$$ \begin{array}{llll}
\left[ 0,112,0,16          \right]&\left[ 112,10176,224,1152      \right]&\left[ 0,224,-32,32         \right]&\left[ 16,1152,32,240       \right] \\
\left[ 112,10176,224,1152  \right]&\left[ 10176,703264,21216,76480\right]&\left[ 224,21216,-1056,2880 \right]&\left[ 1152,76480,2880,14848\right] \\
\left[ 0,224,-32,32        \right]&\left[ 224,21216,-1056,2880    \right]&\left[ -32,-1056,-288,0     \right]&\left[ 32,2880,0,864        \right] \\
\left[ 16,1152,32,240      \right]&\left[ 1152,76480,2880,14848   \right]&\left[ 32,2880,0,864        \right]&\left[ 240,14848,864,2800   \right] \\
   \end{array}  \; . $$
\end{tiny}

{\bf 4)} Here the quadratic invariants for a discrete $2$-dimensional projective plane:
$$ \left[
                  \begin{array}{ccc}
                   1 & 0 & 28 \\
                   0 & 56 & 224 \\
                   28 & 224 & 336 \\
                  \end{array} \right] $$
And here the cubic invariants
$$ \left[
                  \begin{array}{ccc}
\left[1,0,28\right]     & \left[0,56,224 \right]     & \left[28,224,336 \right] \\ 
\left[0,56,224\right]   & \left[56,798,1638\right]   & \left[224,1638,2142 \right] \\
\left[28,224,336\right] & \left[224,1638,2142\right] & \left[336,2142,2422 \right] \\
                  \end{array}
                  \right] \; . $$

\section{Poincar\'e-Hopf} 

Lets first look at the Poincar\'e-Hopf formula for linear valuations. 
The {\bf unit ball} $B(v)$ at a vertex $v$ is defined as the subgraph of $G$ 
generated by the set of vertices in distance $\leq 1$ from $v$.
The {\bf unit sphere} $S(x)=\{ y \in V \; | \; (x,y) \in E \;\}$ which
is the {\bf boundary} $\delta B(x)$ of the ball. For $f \in \Omega$ and a linear
valuation $X \in \V$ define the {\bf index} 
$$  i_{X,f}(x) = X(B^-(x))-X(S^-(x)) \; , $$
where $B^-(x) = S^-(x) \cup \{x\} = \{ y \in B(x) \; | \; f(y) \leq f(x) \; \}$
and $S^-(x) = \{f(y)< f(x) \; \}$. 

\begin{thm}[{\bf Poincar\'e-Hopf}]
$\sum_{x \in V} i_{X,f}(x) = X(G)$.
\label{poincarehopf1}
\end{thm}
\begin{proof}
Start with a single vertex $x$. Now $B^-(x)=\{x\}$ and $S^-(x)=\emptyset$.
Now $i_{X,f}(x)=X(\{x\})-0$. Now add recursively a new vertex to get a growing
set of graphs $G_k$ which covers eventually $G$. By the properties
of valuations, we have $X(G_n) = X(G_{n-1} \cup B^-(x)) 
= X(G_{n-1}) + X(B^-(x)) - X(S^-(x))$. 
\end{proof}

See \cite{poincarehopf} for our first proof in the case of Euler characteristic,
and for \cite{josellisknill} for a Morse theoretical inductive proof. We extended
the result to valuations in September 2015 while getting interested Barycentric characteristic
numbers, a topic which emerged from from the construction of a graph product. \\

Let $\Omega(G)$ denote the set of {\bf colorings} of $G$, locally injective function $f$ on $V(G)$.
Let $P$ be a {\bf Borel probability measure} on $\Omega(G) = \R^{v_0(G)}$ and
let $E[\cdot]$ its {\bf expectation}. Let $c(G)$ be the {\bf chromatic number} of $G$.
Assume either that $P$ is the counting measure on the finite set of colorings of $G$ with
$c \geq c(G)$ real colors or that $P$ is a product measure on $\Omega$ for which
functions $f \to f(y)$ with $y \in V$ are independent identically distributed
random variables with continuous probability density function. For all $G \in \G$ and
$X \in \V$:

\begin{thm}[{\bf Index expectation}]
For any finite simple graph $G$ and Euler characteristic $X$, we have 
$$ {\rm E}[i_{X,f}(x)] = K_X(x) \; . $$
\end{thm}

\begin{proof}
Let $V_k$ denote the number of $k$-dimensional
simplices in $S(x)$ and let $V_k^-$ the number of $k$-dimensional
simplices in $S^-_f(x)$. Given a vertex $x \in V(G)$ and a $k$-dimensional simplex
$K_k$ in $S(x)$, the event
$$  A = \{ f \; | \; f(x)>f(y), \forall y \in V(K_k) \; \} $$
has probability $1/(k+2)$. The reason is that the symmetric group of color permutations
acts as measure preserving automorphisms on the probability space
of functions, implying that for any $f$ which is in $A$ there are $k+1$ functions which
are in the complement so that $A$ has probability $1/(k+2)$. This implies
$$   {\rm E}[V_k^-(x)] = \frac{V_k(x)}{(k+2)}  \; .$$
The same identity holds for continuous probability spaces. Therefore,
\begin{eqnarray*}
   {\rm E}[ 1-\chi(S^-(x)) ] &=& 1-\sum_{k=0}^{\infty}  (-1)^k {\rm E}[V_k^-(x)] 
                         = 1+\sum_{k=1}^{\infty} (-1)^k {\rm E}[V_{k-1}^-(x)]  \\
                        &=& 1+\sum_{k=1}^{\infty} (-1)^k \frac{V_{k-1}(x)}{(k+1)} 
                         =  \sum_{k=0}^{\infty} (-1)^k \frac{V_{k-1}(x)}{(k+1)} = K(x)   \; .
\end{eqnarray*}
\end{proof}

See \cite{indexexpectation} for a proof in the case of continuous distributions.
In \cite{knillgraphcoloring} we adapted the result to finite probability spaces after
getting interested in graph colorings. \\

Assume now $X(A,B)$ is a quadratic valuation like the Wu intersection number.
For a fixed subgraph $B$ the map $A \to X(A,B)$ is a valuation and has a curvature $K_B(x)$
as well as an index $i_{f,B}(x)$. The Gauss-Bonnet and Poincar\'e-Hopf
results show that they both sum to $X(G,B)$. This especially applies for 
$B=G$, so that we have 
$$   X(G,G) =  \sum_x i_{f,G}(x)   $$
Now we use that for any $x$, $A \to i_{f,A}(x)$ is a valuation and apply 
Poincar\'e-Hopf  again to see, using the same function $f$ that 
$$  X(G,G)  = \sum_{x,y} i_f(x,y) \; . $$
We have $(v,w) \to i_f(v,w)$ is zero if $d(v,w)>1$ and that $i_f(v,w)=i_f(w,v)$.  \\

Define the {\bf index a quadratic valuation $X$} as 
$$  i_f(v) = i_f(v,v) + \sum_{w, (v,w) \in E} i_f(v,w)  \; . $$
More generally, if $X(A_1,\dots,A_k)$ is a degree $k$ valuation, we inductively 
define an index $i_f(v_1,\dots,v_k)$ first which can only be non-zero if $v_1,\dots, v_k$ are
contained in ball of radius $2$, then since this is symmetric in permutations of $k$, it is 
divisible by $k$ and can be distributed to the $k$ vertices to get an integer
value $i_f(v)$ on the vertex set which is the {\bf index} of $X$.

\begin{thm}[Poincar\'e-Hopf for quadratic valuations]
If $X$ is a degree $k$ valuation, then 
$$  \sum_v i_f(v) = X(G) \; . $$
\end{thm}
\begin{proof}
Use induction with respect to $k$. If $k=1$, it is 
Theorem~(\ref{poincarehopf1}). Having verified it for
degree $k-1$, use Theorem~(\ref{poincarehopf1}) for the 
valuation for 
$$ A \to X(A_1,\dots,A_{k-1},A) \;  $$
to get a function $i_f(v_1,\dots,v_k)$ whose total 
value $\sum_{\pi} i_f(v_{\pi(1)}, \dots, v_{\pi(k)})$
when summing over all permutations gives an index value
on each simplex which is divisible by $k$ so that one can 
assign $i_f(v_1,\dots,v_k)/k$ to each of the adjacent vertices. 
\end{proof}

Lets make this more explicit in the Wu case $X=\omega$. For a fixed vertex $v$ in $G$
and a fixed subgraph $G$, the valuation
$$ A \to X_v(A) = i_{f}(v)(A)   \;  $$
is by definition given by 
$$  \omega(A \cap B^-_{f}(v)) - \omega(A \cap S^-_{f}(v)) \; , $$
where $B^-_{f,B}$ is part of the ball $B_{f,B}(v)$ in $G$,
where $f$ takes values smaller than $f(v)$ and 
$S^-_{f}$ is part of the sphere $S_{f,B}(v)$ in $B$
where $f$ takes values smaller than $f(v)$. Now apply this to get
$$  i_f(v,w) =  \omega(B^-_f(v),B^-_f(w)) - \omega(B^-_f(v),S^-_f(w))
              - \omega(S^-_f(v),B^-_f(w)) + \omega(S^-_f(v),S^-_f(w)) \; . $$
In the case $v=w$, it is
$$  i_f(v,v) = \omega(B^-_f(v)) - \omega(S^-_f(v)) \; . $$
This splitting up into cases $v=w$ and $v \neq w$ can be done for a general
quadratic valuation and  leads to: 
                
\begin{thm}[The index]
The Poincar\'e-Hopf index $i_f(v)$ for a quadratic valuation $X$ is
\begin{eqnarray*}
  i_f(v) &=& X(B^-_f(v)) - 2 X(S^-_f(v),B^-_f(v)) + X(S^-_f(v))  \\
         &+& \sum_{w}  [X(B^-_f(v),B^-_f(w)) - X(B^-_f(v),S^-_f(w))  \\
         & &      - X(S^-_f(v),B^-_f(w)) + X(S^-_f(v),S^-_f(w))] \; .
\end{eqnarray*}
\end{thm}

{\bf Examples}. \\
{\bf 1)} Let $G$ be the icosahedron graph and $X$ the Wu characteristic valuation. We take
the function $f$ enumerating the indices so that the adjacency matrix of $G$ is
$$ \left[
                 \begin{array}{cccccccccccc}
                  0 & 1 & 1 & 1 & 1 & 1 & 0 & 0 & 0 & 0 & 0 & 0 \\
                  1 & 0 & 0 & 0 & 1 & 1 & 0 & 0 & 1 & 1 & 0 & 0 \\
                  1 & 0 & 0 & 1 & 1 & 0 & 0 & 1 & 0 & 0 & 1 & 0 \\
                  1 & 0 & 1 & 0 & 0 & 1 & 0 & 1 & 0 & 0 & 0 & 1 \\
                  1 & 1 & 1 & 0 & 0 & 0 & 0 & 0 & 1 & 0 & 1 & 0 \\
                  1 & 1 & 0 & 1 & 0 & 0 & 0 & 0 & 0 & 1 & 0 & 1 \\
                  0 & 0 & 0 & 0 & 0 & 0 & 0 & 1 & 1 & 1 & 1 & 1 \\
                  0 & 0 & 1 & 1 & 0 & 0 & 1 & 0 & 0 & 0 & 1 & 1 \\
                  0 & 1 & 0 & 0 & 1 & 0 & 1 & 0 & 0 & 1 & 1 & 0 \\
                  0 & 1 & 0 & 0 & 0 & 1 & 1 & 0 & 1 & 0 & 0 & 1 \\
                  0 & 0 & 1 & 0 & 1 & 0 & 1 & 1 & 1 & 0 & 0 & 0 \\
                  0 & 0 & 0 & 1 & 0 & 1 & 1 & 1 & 0 & 1 & 0 & 0 \\
                 \end{array}
                 \right]   \; . $$
Now $i_f(v,w)$ is the symmetric matrix
$$ i_f  = \left[
                 \begin{array}{cccccccccccc}
                  1 & -1 & -1 & 0 & 1 & 1 & 0 & 0 & 0 & 0 & 0 & 0 \\
                  -1 & 0 & 1 & 0 & 0 & 0 & 0 & 0 & 0 & 0 & 0 & 0 \\
                  -1 & 1 & 0 & 1 & 0 & -1 & 0 & 0 & 0 & 0 & 0 & 0 \\
                  0 & 0 & 1 & 0 & -1 & 0 & 0 & 0 & 0 & 0 & 0 & 0 \\
                  1 & 0 & 0 & -1 & 0 & 0 & 0 & 0 & 0 & 0 & 0 & 0 \\
                  1 & 0 & -1 & 0 & 0 & 0 & 0 & 0 & 0 & 0 & 0 & 0 \\
                  0 & 0 & 0 & 0 & 0 & 0 & 1 & -1 & -1 & 0 & 1 & 1 \\
                  0 & 0 & 0 & 0 & 0 & 0 & -1 & 1 & 1 & 0 & -1 & -1 \\
                  0 & 0 & 0 & 0 & 0 & 0 & -1 & 1 & 1 & 0 & -1 & -1 \\
                  0 & 0 & 0 & 0 & 0 & 0 & 0 & 0 & 0 & 0 & 0 & 0 \\
                  0 & 0 & 0 & 0 & 0 & 0 & 1 & -1 & -1 & 0 & 1 & 1 \\
                  0 & 0 & 0 & 0 & 0 & 0 & 1 & -1 & -1 & 0 & 1 & 1 \\
                 \end{array} \right] \; .$$
When summing all the entries, we get by Poincar\'e-Hopf the Wu
characteristic $\omega(G)=2$. The scalar function $i_f(v)$ on vertices $v$
is obtained by summing over each row. It gives
$$  i_f = \left( 1, 0, 0, 0, 0, 0, 1, -1, -1, 0, 1, 1 \right) \; . $$
These indices by the way are the same then the Euler index for the Euler characteristic. \\

2) For the House graph with adjacency matrix
$$ \left[ \begin{array}{ccccc}
                   0 & 1 & 0 & 1 & 0 \\
                   1 & 0 & 1 & 0 & 1 \\
                   0 & 1 & 0 & 1 & 1 \\
                   1 & 0 & 1 & 0 & 0 \\
                   0 & 1 & 1 & 0 & 0 \\
                  \end{array} \right]  \; , $$
vertex degrees $(2,3,3,2,2)$ and Euler curvatures
$(0,-1/5,-1/6,0,1/3)$. For the function $f(x_k)=k$, the Wu index matrix is
$$ i_f = \left[ \begin{array}{ccccc}
                   1 & -1 & 0 & -1 & 0 \\
                   -1 & 0 & 0 & 1 & 0 \\
                   0 & 0 & 0 & 0 & 1 \\
                   -1 & 1 & 0 & 1 & 0 \\
                   0 & 0 & 1 & 0 & 0 \\
                  \end{array} \right] \; , $$
the Wu indices are $(-1,0,1,1,1)$ adding up to the Wu characteristic $2$. 
The Wu curvatures are $(0,2/3,2/3,0,2/3)$. We see that some Wu indices
$i_f(v,w)$ are nonzero for $v,w$ of distance $2$. \\

We see that like curvature this is a second order notion. As we have in each induction 
step the expectation being curvature, we have the 

\begin{thm}[{\bf Index expectation}]
For any finite simple graph $G$ and any $k$-linear valuation $X$, we have 
$$  {\rm E}[i_{X,f}(x)] = K_X(x)  \; . $$
\end{thm}

Here is an important lemma which is a generalization of the fact that 
the function $\sigma(x)$ is a Poincar\'e-Hopf index on $G_1$. 

\begin{lemma}
If $G$ and $H$ are two finite simple graphs and $f$ is a function
on the vertices of $G \times H$ ordered according to degree and given lexigographic orders
of vertices of $G$ and $H$, then the Poincar\'e-index of any valuation $X$
assigning the value $1+(-1)^{k}$ to a sphere of dimension $k$ is
$i_{f}(x,y) = \sigma(x) \sigma(y) =\sigma(x \times y)$.
\end{lemma}
\begin{proof}
By definition, 
$$ i_f(x,y) = 1-\omega(S^-_f(x,y)) \; . $$
The set of vertices $w \times z$ in $S^-_f(x,y)$
is the set of subsimplices of $x \times y$. This is the boundary 
part of the Barycentric refinement of the boundary of $x \times y$.
It is therefore a $k={\rm dim}(x) + {\rm dim}(y) -1$ dimensional sphere
and has for any of the Wu characteristics, the value $1+(-1)^k$. This is $2$ for even $k$ and
$0$ for odd $k$. This means 
$$  i_f(x,y) = 1-(1+(-1)^k) = -(-1)^k = (-1)^{{\rm dim}(x) + {\rm dim}(y)} = \sigma(x \times y) \; . $$
\end{proof}

The Poincar\'e-Hopf index $i_f(x,y)$ is a function on the product $V \times W$,
if $V$ is the vertex set of $G$ and the $W$ the vertex set of $H$. 
The fact $\sum_{x,y} \sigma(x) \sigma(y) = \omega(G)$ is the definition of $\omega$. 
Applying Poincar\'e-Hopf shows that that this is $\omega(G_1)$,
the Wu characteristic of the Barycentric refinement. 

\begin{coro}
The Barycentric subdivision $G_1$ of $G$ satisfies
$\omega(G_1) = \omega(G)$ for any Wu characteristic
\label{barycentricinvariance}
\end{coro}
\begin{proof}
We will see later that all Wu characteristics of a sphere are
all the same and agree with the Euler characteristic. 
\end{proof}

Poincar\'e-Hopf generalizes to multi-linear valuations of higher degree: there is a function $i_f(x_1,\dots,x_k)$
which is symmetric in $x_1,\dots,x_k)$ and nonzero only if $x_1,\dots,x_k$ are in a simplex. 
The values of $i_f(x_1,\dots,x_k)$ can all be chip-fired to vertices. But now, 
for quadratic valuations, the index on vertices is a rational number of the form $p/2$. 
For refined graphs, it is an integer. For cubic valuations, when chip-fired onto vertices,
the charges are now fractions of the form $k/3$. 

\begin{thm}
The Wu characteristic $\omega_k$ satisfies
$$ \omega_k(G) = \sum_{V \times V \dots \times V} i_f(v_1,\dots,v_k) $$
for an index function $i_f$.
\end{thm}

\section{Product property}

The Cartesian product of two graphs $G,H$ is defined as the incidence graph of 
the product $f_G \cdot f_H$ in the Stanley-Reisner ring. We have introduced this product
in \cite{KnillKuenneth} for graphs. (We had at that time not been aware of the Stanley ring
\cite{Stanley86}. There are various products already known in graph theory. The ring
theoretically defined product does the right thing on cohomology, with respect to dimension and 
also with respect to valuations like Euler characteristic. It has been known 
for simplicial complexes, but was so far unused in graph theory. It has the properties 
known from the continuum: in full generality, for any finite simple graph, it satisfies
the {\bf Kuenneth formula} \cite{KnillKuenneth} 
$$  H^k(G \times H) = \oplus_{i+j=k} H^i(G) \otimes H^j(G) \; .  $$
It also satisfies the {\bf dimension inequality}
$$  {\rm dim}(G \times H)  \geq {\rm dim}(G) + {\rm dim}(H) \; . $$
These two results hold for general finite simple graphs which need not necessarily
have to be related to any geometric setup. Like cohomology, homotopy or calculus, 
the product works for any network. \\

A special case of a graph product is $G \times K_1$ which is the Barycentric refinement. 
As $\chi(G) = -f_G(-1)$, the Euler characteristic is multiplicative in 
general on the class of finite simple graphs. One can see the multiplicity of the
Euler characteristic also from the general Euler-Poincar\'e formula equating 
cohomological and combinatorial Euler characteristic
and the fact that the Poincar\'e polynomial $p_G(x)=\sum_{k=0} {\rm dim}(H^k(G)) x^k$
satisfies $p_{G \times H}(x) = p_G(x) p_H(y)$ and $\chi(G)=p_G(-1)$ which implies
also $\chi(G \times H) = \chi(G) \chi(H)$. \\

There is a also a direct algebraic way to see that the $\chi$ is preserved when taking
a Barycentric refinement $G \to G_1$. Given a graph $G$, the vertices of $G_1$
are the simplices of $G$. The function $I(x)=(-1)^{{\rm dim}(x)}$ is a $\{-1,1\}$-
valued function on the vertices of $G_1$. If $f_G$ is the polynomial in the 
Stanley ring and the monomials are ordered alphabetically, then this defines
a function $f$ on the vertex set of $G_1$. It turns out that $i_f(x)=I(x)$. 
Proof: $i_f(x) = 1-\chi(S^-(x))$. If $x$ is zero dimensional, then $S^-(x)$ is
empty and $i_f(x)=1$. If $x=(a,b)$ is one dimensional and represented by $ab$
in the Stanley ring, then $S^-(x)=\{a,b\}$ has two not connected elements
and $i_f(x)=-1$. If $x=(a,b,c)$ is two dimensional and represented by $abc$, then
$S^-(x)=\{ab,bc,ac,a,b,c\}$ has a circular graph as $S^-(x)$ and $i_f(x)=1$. 
If $x=(a,b,c,d)$ is three dimensional and represented by $abcd$, then
$S^-(x)=\{abc,acd,abc,bcd,ab,ac,ad,bc,bd,cd,a,b,c,d\}$ is a two dimensional 
sphere graph and $\chi(S^-(x))=2$ so that $i_f(x)=-1$. Etc.
Now, since $\sum I(x)$ is the Euler characteristic of $G$ by definition
and $\sum_x i_f(x)$ is the Euler characteristic of $G_1$ by Poincar\'e-Hopf
applied to the function $f$, the invariance of the Euler characteristic 
under Barycentric refinement has become a consequence of the Poincar\'e-Hopf formula. 
Also the product formula of Euler characteristic becomes clear as $f_G f_H$
defines the vertex set of the graph $G \times H$ on which the functions $I(x,y)=I_G(x) I_H(y)$ 
taking values in $\{-1,1\}$ are the Poincar\'e-Hopf index of a function and so 
the Euler characteristic of $G \times H$. \\

When computing the Wu characteristic of $G$, it is the sum $\sum_{x,y} \sigma(x) \sigma(y)$ which lives
on vertices and edges of $G_1$. We can look at the 1-dimensional skeleton simplicial complex
of $G_1$ and look at the corresponding element in the Stanley ring on $(G_1)'$. If $I(x,y)$
is the function, then $\sum I$ is the Wu characteristic. Now again, $I$ can be seen as
the Poincar\'e-Hopf index of a function. It has the property that for vertices $S^-(x)=\emptyset$
and for edges either $S^-(x)=\{a,b\}$ or $S^-(x) = \emptyset$. 

\begin{lemma}
For any finite simple graph $G$, we have $\chi(G)=-f_G(-1)$ and 
$$  \omega(G) = f_G(-1)^2 - (f_G^2)(-1) \; . $$
\end{lemma} 
\begin{proof}
We know that $\omega(G)$ sums over the monomials of $f_G(x)*f_G)(x)$ containing some square and
that $f_G^2(x)$ is the part of $f_G(x)*f_G(x)$ containing no squares. So, if we take
$f_G(x)*f_G(x)$ and subtract $f_G^2(x)$ we get the sum of the monomials which contain some
square. Evaluating at $(-1,-1,\dots,-1)$ gives the formula.
\end{proof}

\begin{thm}
\label{theorem3}
If $G,H$ are two arbitrary finite simple graphs, then 
$$  \omega(G \times H) = \omega(G) \cdot \omega(H) \;  $$
for any of the Wu characteristics. 
\end{thm} 
\begin{proof}
We know that $f_{G \times H} = f_G \cdot f_H$. Now, if we take the part of $f_G^2$ 
which contains a square and the part of $f_H^2$ which contains a square
and multiply them, we get the part of $(f_G f_H)^2$ which contains a 
square in each part of the variables. Any pair of intersecting simplices $x,y$ in $G \times H$ is of the form 
$x=(x_1,x_2), y=(y_1,y_2)$, where both $x_1$ and $y_1$ and $x_2$ and $y_2$ intersect. \\
We have also combinatorially $\omega(G) = \sum_{x,y} \omega(x) \omega(y)$ and
$\omega(H) = \sum_{u,v} \omega(u) \omega(y)$ and
$\omega(x u) = \omega(x) \omega(u)$ and $\omega(y v) = \omega(y) \omega(v)$ so that
The vertices of $G \times H$ are of the form $(xu, yv)$, the Poincar\'e-Hopf index sum of $G \times H$ is
$$ \omega(G) \cdot \omega(H) = \sum_{(x,u,y,v}  \omega(x) \omega(u) \omega(y) \omega(v) \; . $$
This is a function on pairs of vertices in $G \times H$ from which we can get, by summing over
one of the variables, a scalar function on the vertices of the $G \times H$. This is a Wu Poincar\'e-Hopf
index for $\omega$ and the result is $\omega(G \times H)$. 
\end{proof} 

{\bf Examples.} \\

{\bf 1)} If $G=K_2$, then $f_G(x,y) = x+y+xy$. The usual product in $Z[x,y,z]$ gives 
$f_G*f_G = x^2 + 2xy + 2x^2y + y^2 + 2xy^2 + x^2y^2$ and the Stanley product gives
$f_G f_G = 2x y$. We have $f_G*f_G(-1,-1) = 1$ and $f_G^2(-1,-1) = 2$, and so
$\omega(G)  = f_G(-1,-1)^2  - (f_G^2)(-1,-1) = -1$. \\
{\bf 2)} If $G=K_3$, then $f_G(x,y,z)=x + y + x y + z + x z + y z + x y z$ which satisfies $\chi(G)=-f_G(-1,-1,-1)$.
The usual product is $f_G*f_G =x^2+2xy+2x^2y+y^2+2xy^2+x^2y^2+2xz+2x^2z+2yz+6xyz+4x^2yz+2y^2z+4xy^2z+2x^2y^2z
+z^2+2xz^2+x^2z^2+2yz^2+4xyz^2+2x^2yz^2+y^2z^2+2xy^2z^2+x^2y^2z^2$ which satisfies $f_G(-1,-1,-1)*f_G(-1,-1,1)=1$.
In the Stanley ring, the product is $f_G^2 = f_G f_G = 2xy+2xz+2yz+6xyz$ which satisfies $f_G^2(-1,-1,-1)=0$.
We have $\omega(G) = 1-0=1$. \\
{\bf 3)} If $G=C_4$, then  $f_G=w + x + w x + y + x y + z + w z + y z$ and $\chi(G)(-1,-1,-1,-1)=0$ 
and $f_G*f_G = w^2+2wx+2w^2x+x^2+2wx^2+w^2x^2+2wy+2xy+4wxy+2x^2y+2wx^2y+y^2+2xy^2+x^2y^2+2wz+2w^2z+2xz+4wxz+2w^2xz+
2yz+4wyz+4xyz+4wxyz+2y^2z+2xy^2z+z^2+2wz^2+w^2z^2+2yz^2+2wyz^2+y^2z^2$ and
$f_G^2 = 2wx+2xy+4wxy+2wz+4wxz+2yz+4wyz+4xyz+4wxyz$. We have
$f_G(-1,-1,-1,-1)^2 - (f_G^2)(-1,-1,-1,-1) = 0-0=0$.

\section{Graphs with boundary} 

A graph $G$ is called a {\bf $d$-graph with boundary} if every unit sphere is either
a $(d-1)$-sphere or a $(d-1)$-ball with $(d-2)$ sphere as boundary. 
The set of vertices, for which $S(x)$ is a sphere
forms a subgraph called the {\bf interior} ${\rm int}(G)$, 
the set of vertices, for which $S(x)$ is a ball form the {\bf boundary} $\delta G$. 
We ask that the boundary $\delta G$ is either a $(d-1)$-graph or that it is empty. 
The class of graphs with boundary plays the role of {\bf compact manifolds
with boundary}. The class of $d$-graphs with boundary are invariant
under Barycentric refinement because the boundary operation
commutes with the process of taking Barycentric refinements. \\

Lets look first at a Gauss-Bonnet proof of the {\bf Dehn-Sommerville identities}, which tell
that $X(G)=0$ for the Dehn-Sommerville valuations:
$$ X_{k,d}(v) = \sum_{j=k}^{d-1} (-1)^{j+d} \B{j+1}{k+1} v_j(G) + v_k(G)  \; . $$

We have changed the signs slightly. The classical way to write the Dehn-Sommerville 
valuations is 
$$ \tilde{X}_{k,d}(v) =  \sum_{j=k}^{d-1} (-1)^{j-1} \B{j+1}{k+1} v_j(G) - (-1)^d v_k(G) \; . $$
The reason for our choice of the sign is that with this choice, the 
valuations are mostly non-negative for any finite simple graph. \\
The fact that they are zero for $d$-graphs follows immediately from Gauss-Bonnet and the fact
that their curvature is a Dehn-Sommerville valuation of a unit sphere: 

\begin{lemma}[Curvature of Dehn-Sommerville is Dehn-Sommerville]
Given a $d$-graph and $X$ a Dehn-Sommerville valuation, then 
its curvature is $K(x) = X_{k-1,d-1}(S(x))$. 
\end{lemma}
\begin{proof}
The curvature of the valuation is a local valuation
is a shifted functional $X_{k-1,d-1}/d$.
To see this, note the identity
$$   X_{k+1,d+1}(l+1)/(l+1) = X(k,d)(l)/(k+2)  \; .  $$
\end{proof}

Dehn-Sommerville relations are usually considered only for polytopes 
or topological manifolds and not for arbitrary 
finite simple graphs. Interestingly, for most $G$ with clique number $d$, we see
$X_{k,d}(G) \geq 0$ for $k \geq 0$. We see very rare instances of $X_{0,d}<0$ and
have not seen an example with $X_{1,d}(G)<0$.
It is not a surprise as the curvature is a Dehn-Sommerville invariant of a 
lower dimensional graph, which if rarely negative makes it unlikely that the higher 
invariant is negative too. The Euler characteristic $X_{-1,d}$ on the other hand is
negative pretty frequently. And then there is the classical result of Dehn-Sommerville-Klee 
on the vanishing of Dehn-Sommerville relations.  
Note however also here that we never refer to the continuum. The theorem is 
entirely graph theoretical: 

\begin{coro}[Dehn-Sommerville]
If $G$ is a $d$-graph, then each Dehn Sommerville valuation $X_k$ with $k \geq 0$
satisfies $X(G)=0$ and the curvature of $X$ is constant zero. 
\end{coro}
\begin{proof}
We use Gauss-Bonnet and induction with respect to dimension as well as the 
previous lemma.  The curvature is zero because it is a Dehn-Sommervile
invariant of a sphere, which is a smaller dimensional graph. 
\end{proof}

As before, we write $\sigma(x) = (-1)^{{\rm dim}(x)}$ for a complete subgraph $x$ of $G$.  \\

For a complete subgraph $z$ of $G$ define 
$$ \omega_z(G) = \sum_{z=x\ cap y} \sigma(x) \sigma(y)  \;  $$
which sums over all interaction pairs $(x,y)$ of simplices in $G$ which have $z$ as 
their intersection. Because the intersection of two complete subgraphs $x,y$ is a complete 
subgraph of $G$, we have partitioned all intersections and can write
$$ \omega(G) = \sum_z \omega_z(G)  \; , $$
where the summation is over all complete subgraphs $z$ of $G$.  \\

Given a complete subgraph $z$ of $G$ with vertex set $a_1,\dots,a_k$, 
define $S(a_1,\dots,a_k)$ as the intersection of the spheres
$S(a_1) \cap S(a_2) \dots \cap S(a_k)$. 
If $G$ is a $d$-graph, then $S(a_1) \cap S(a_2) \dots \cap S(a_k)$ is a $(d-k)$-sphere.
This can be seen by induction: $S(a_1)$ is a $(d-1)$-sphere by definition. 
Since $S(a_2) \cap S(a_1)$ is a $(d-2)$-sphere in $S(a_1)$, it is also a $d-2$ sphere 
in $G$ etc. 

\begin{lemma}
If $G$ is a $d$-graph, then 
$\omega_z(G) = \sigma(z)$.
\end{lemma}
\begin{proof}
Look at all the interactions in which both
$x,y$ contain $z$ properly: this gives $\omega(S_z) \sigma(z)$. 
Then look at all the interactions in which one is $z$. This gives $-2 \omega(S_z) \sigma(z)$.
Then there is the case when both $x=z,y=z$ which gives $1$. The sum is either $1$ or $-1$. 
\end{proof}

\begin{thm}
Given a $d$-graph $G$, then $\omega_k(G) = \chi(G)$  for any Wu characteristics $\omega_k$. 
\end{thm}
\begin{proof}
We can follow the same proof which worked in the case of a quadratic valuation
using the Euler characteristic and a Dehn-Sommerville valuation. Again use induction.
What happens still is that $\omega(z) = \sigma(z) = (-1)^{{\rm dim}(z)}$ which then implies
$$ \omega(G) =  \sum_{z} \sum_{x \cap y = z} \sigma(x) \sigma(y)  = \sum_z \omega(z) = \chi(z) \; . $$
To see this, there are again four cases, $x=y=z$, the cases when one $x$ or $y$ is equal to $z$
or then if both are not equal to $z$. The contribution of the first is $1$, the contribution 
of the second and third is $2\chi(S(z))$, The contribution of the last $\chi(S(z))^2 - \chi(S(z))$. 
If $\chi(S(z))=2$, we get the sum $1-\chi(S(z))=-1$. If $\chi(S(z))=0$, the contribution is $1$. 
We see that the $\omega(z) = \sigma(z)$. 
The higher order case follows inductively. 
\end{proof}

\begin{thm}
For a $d$-graph $G$ and a vertex $v$, then the Wu curvature $K_{\omega}(v)$ is equal to 
the Euler curvature $K_{\chi}(v)$. 
\end{thm}
\begin{proof}
Fix a vertex $v$ and fix a simplex $x$ of dimension $k$ containing $v$. 
It is enough to show that for this $x$, each contribution
$$  \sum_{y \cap x \neq \emptyset} \sigma(x) \sigma(y) $$
to the curvature is zero. We can write this as 
$$ \sigma(x) \sum_{z \subset x} [ \sum_{y \cap x = z} \sigma(y) ]  \; . $$
As before, for fixed $z=(a_1,\dots,a_k)$, we have 
$$  \sum_{y \cap x = z} \sigma(y) = \sigma(z) \; . $$
Every $y$ defines a simplex in $S(a_1) \cap \cdots \cap S(a_k) \subset x$ 
of dimension ${\rm dim}(y)-k$. The sum $\sum_{z \subset x} \sigma(z)$ 
is the Euler characteristic of a contractible set which is $1$.
This shows that the contribution to the curvature is $\sigma(x)/(dim(x)+1)$
is the same than the contribution for the Euler curvature. 
\end{proof}

{\bf Examples.} \\
$d=1$: there are no other interactions. Curvature is zero. \\
$d=2$: every triangle has $d(a)+d(b)+d(c)-6$ additional triangle-triangle
connections and $d(a)+d(b)+d(c)-6$ additional triangle-edge connections. 
Every edge has additional $d(a)+d(b)-2$ edge edge connections and
$d(a)+d(b)-2$ edge-triangle connections. 

\begin{thm}[Boundary formula]
\label{boundaryformula}
The Wu characteristic of a $d$-graph $G$ with boundary $\delta G$ satisfies
$$  \omega(G) = \chi(G)-\chi(\delta G)) \;  $$
where $\chi$ is the Euler characteristic. 
\end{thm}
\begin{proof}
Again, we prove this by induction with respect to dimension $d$. For $d=1$, 
a $1$-graph with boundary is a line graph which has $\omega(G)=-1 = \chi(G)-\chi(\delta G)$. 
Again, write the sum as
$$ \sum_{x,y} \sigma(x) \sigma(y)  = \sum_z \sum_{x \cap y = z} \sigma(x) \sigma(y)  \; . $$
By induction, For a boundary simplex $z$, the sum 
$S(z) = \sum_{x \cap y = z} \sigma(x) \sigma(y)$ is either $-1$ or $1$ depending on the 
dimension of $z$. The reason is that this sum $S(z)$ is $\chi(B(z)) - \Omega_z(G)$ 
which is $1$ or $-1$. But this is the Euler characteristic of a thickened boundary of $G$
which, since homotopic to the boundary has the same Euler characteristic than the boundary. 
\end{proof}

Li Yu \cite{Yu2010}, it is shown that any real valued (not necessarily linear) 
invariant of a compact combinatorial manifold with boundary which is invariant under Barycentric subdivision 
is determined by the two numbers $\chi(G)$ and $\chi(\delta G)$. The Wu invariant is 
such a real-valued invariant and the boundary formula gives an example for the Yu theorem. \\

An immediate corollary of this formula is:

\begin{coro}
For even dimensional $d$-graphs, $\omega(G)=\chi(G)$.
For odd dimensional $d$-graphs, $\omega(G)=-\chi(\delta G)$. 
\end{coro} 
\begin{proof}
For odd dimensional $d$-graphs $G$, we have $\chi(G)=0$ and for even dimensional
$d$-graphs $G$ the boundary $\delta G$ is odd dimensional and $\chi(\delta G)=0$.  
\end{proof}

{\bf Remarks.} \\
{\bf 1)} It again follows that the Wu invariant of a $d$-ball is $(-1)^d$.  \\
{\bf 2)} For a general finite simple graph, let $\delta G$ denote the {\bf subset} of $V$
where $K_{\omega}$ and $K_{\chi}$ do not agree. 
If we define $\omega(\delta G)) = \sum_{x \in \delta G} K_{\omega}(x)$
and $\chi(\delta G)) = \sum_{x \in \delta G} K_{\chi}(x)$, then 
the formula $\omega(G)- \omega(\delta G)=\chi(G)-\chi(\delta G)$ would hold.
for any finite simple graph. \\
{\bf 3)} Its follows that the Wu characteristic is             
a topological invariant for manifolds with boundary.  \\
{\bf 4)} It also follows that the Wu characteristic is
a cobordism invariant like Euler characteristic: if $H,K$ are cobordant
using an even dimensional graph, then $H,K$ both have Wu and Euler characteristic
zero. If $G$ is odd dimensional, then $H,K$ are odd dimensional without boundary
and $\omega(H)=\chi(H) = \chi(K) = \omega(K)$. \\

It seems that $\chi$ and $\omega$ are essential prototypes as $\omega_3$ behaves
again in the same way than $\chi$: 

\begin{thm}
\label{cubicformula}
If $G$ is a $d$-graph with boundary, then 
$$  \omega_3(G) = \chi(G) \; . $$
\end{thm}

{\bf Remarks.} 
{\bf 1)} Again, this illustrates the theorem of Yu \cite{Yu2010} that any possibly nonlinear
topological invariant depends only on the Euler characteristic of $G$ and $\delta G$. 
But as the quadratic Wu characteristic, also the cubic characteristic not only involves
the $f$-vector but also higher $f$-tensors. \\
{\bf 2)} Euler characteristic for geometric even dimensional graphs has a nice
``Hilbert-action" type interpretation as it is the average over a
naturally defined probability space of two dimensional subgraphs
and so an average of ``scalar curvatures" obtained by averaging
all sectional curvatures through a point. \cite{indexformula}. The Wu-invariant makes
the Euler characteristic of even dimensional graphs look even more
``interaction like". It not only can be seen as a super count of the
``indecomposable parts" of space given in terms of simplices; it is also
a super count of the interactions between these indecomposable 
parts. The interactions of equal type (Fermionic-Fermionic pairs
or Bosonic-Bosonic pairs) are counted positive, the
interactions of opposite type (Bosonic-Fermionic) are counted 
negative. The result that this number still has a geometric interpretation
as is remarkable even in the odd dimensional case, where the Euler
characteristic is zero.

\section{More examples}

Graphs without triangles can be seen as {\bf one dimensional curves}. One can force on any
finite simple graph $(V,E)$ a simplicial structure which is one dimensional and ignore the Whitney complex.
This is done by taking the $1$-dimensional skeleton complex $V \cup E$. Graphs without triangles only have
cohomologies $b_0=H^0(G)$ counting the number of connectivity components and 
$b_1=H^1(G)$ counting the number of generators for the fundamental group.
The Euler characteristic of this complex can then be given
by Euler-Poincar\'e in two ways as $v_0-v_1 = b_0-b_1$. For example, if $G$ is the cube graph where 
$v_0=8,v_1=12$, and $b_0=1,b_1=5$. As for the Euler curvature, the Wu curvature is local but
while Euler curvature depends on the disc of radius $1$, the Wu curvature depends on a disk
of radius $2$. 

\begin{lemma}
The Wu curvature of a graph $G$ without triangles at a vertex $x$
is $K(x) = 1-5d/2+d^2/2+\sum_i d_i/2$, where $d$ is the
vertex degree at $x$ and $d_i$ are the vertex degrees of vertices neighboring $x$. 
\end{lemma}
\begin{proof}
The vertices contribute $1-d$, where $1$ is the self interaction and $-d$ the interaction
with the $d$ neighboring edges.  The edges contribute $1$ for the self interaction,
$-2$ for the interaction with the neighboring vertices, and then 
$d_i-1$ for the interaction with the neighboring edges. We have so the 
curvature contribution of the edge $(a,b)$ with vertex degree $d(a)=d,d(b)=d'$: 
$$ (1-2+(d-1) + (d'-1))/2  = -3/2 + d/2 + d'/2 $$
from each edge. The vertex contribution $1$ plus the sum over all these edge contributions
gives 
$$ 1 + d+ \sum_i (-3/2 + d/2 + d_i/2) = 1-5d/2+d^2/2+\sum_i d_i/2 \; . $$
\end{proof}

\begin{figure}
\scalebox{0.12}{\includegraphics{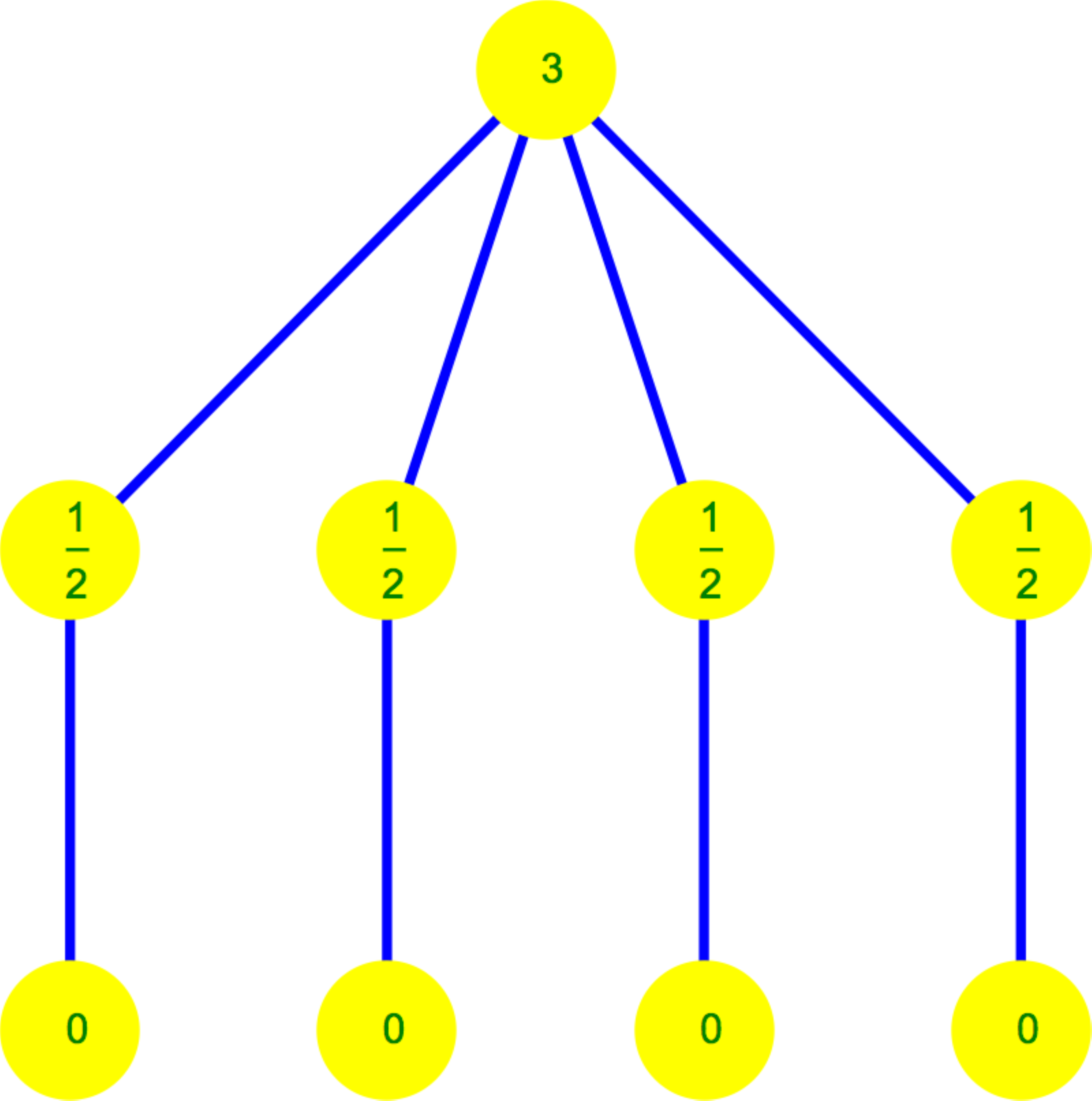}}
\scalebox{0.12}{\includegraphics{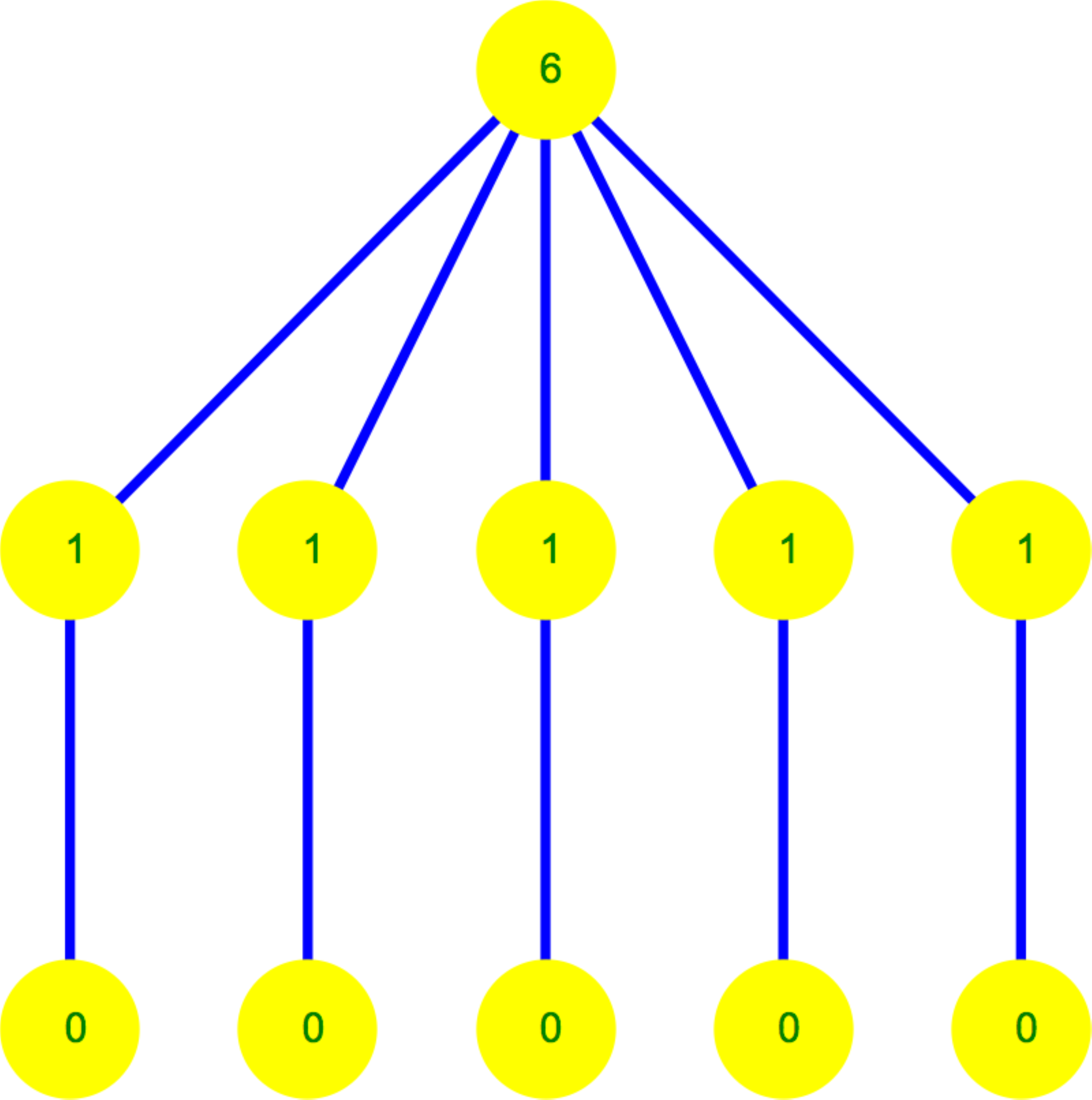}}
\scalebox{0.12}{\includegraphics{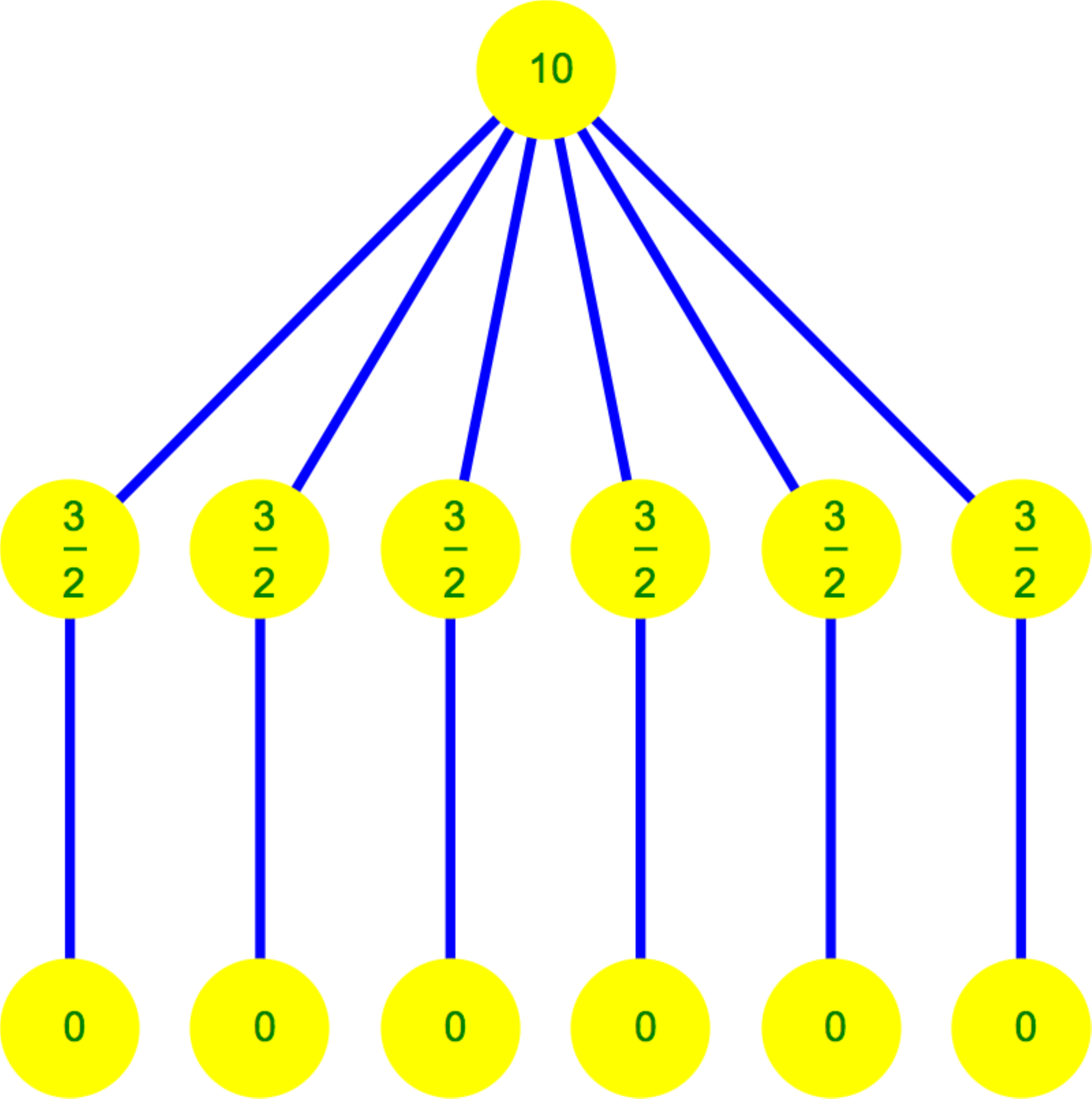}}
\scalebox{0.12}{\includegraphics{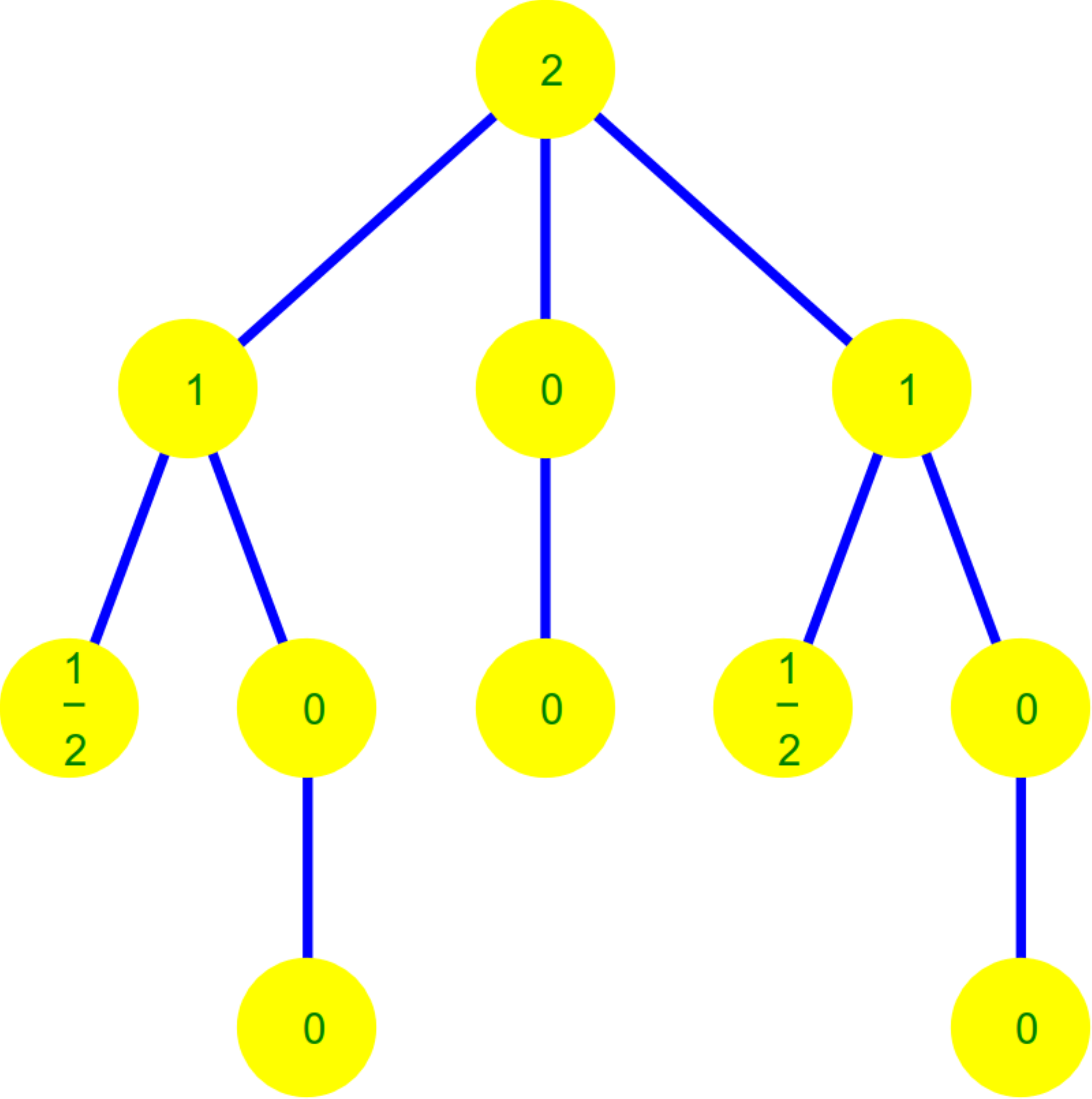}}
\scalebox{0.12}{\includegraphics{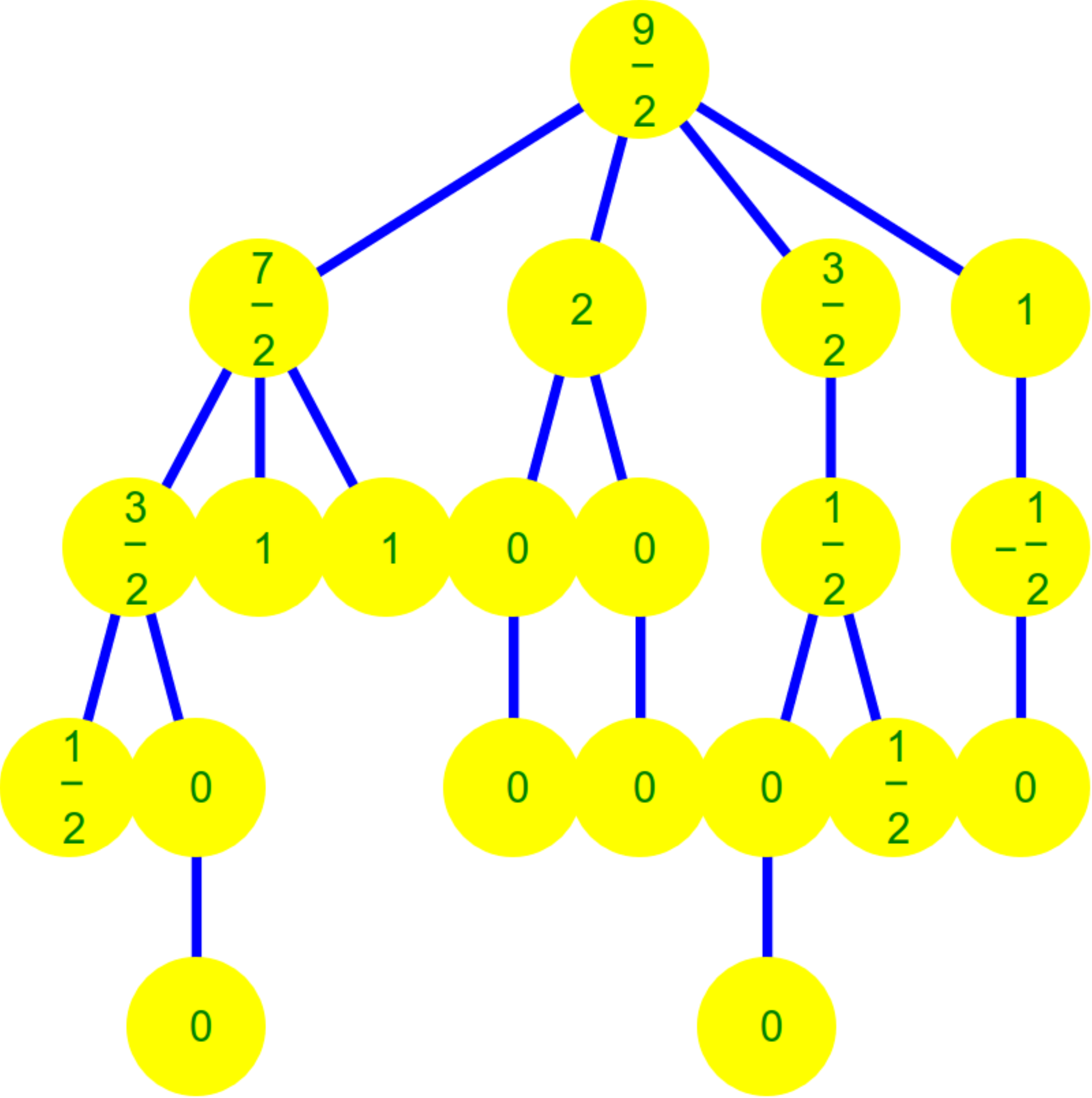}}
\scalebox{0.12}{\includegraphics{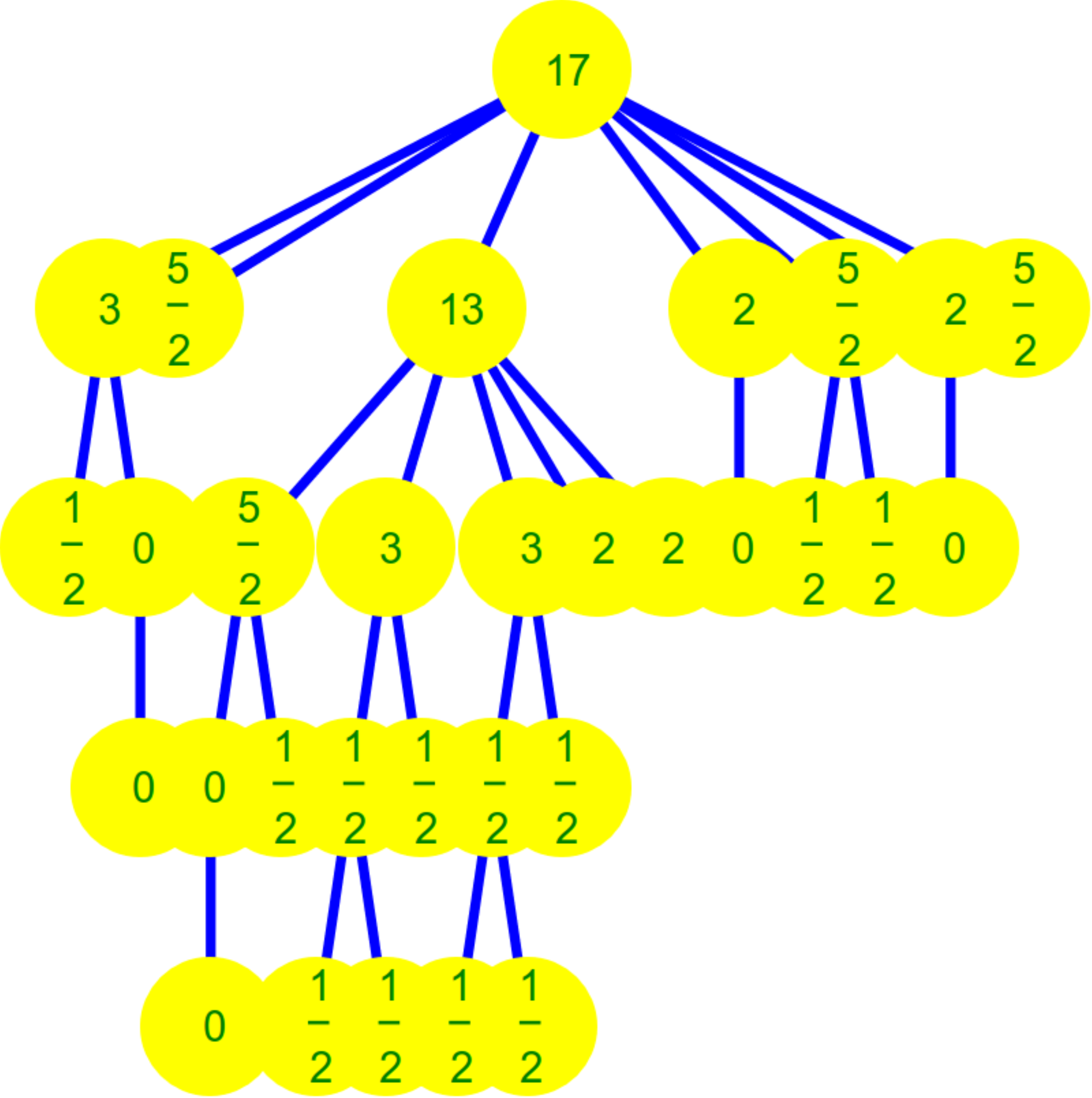}}
\scalebox{0.12}{\includegraphics{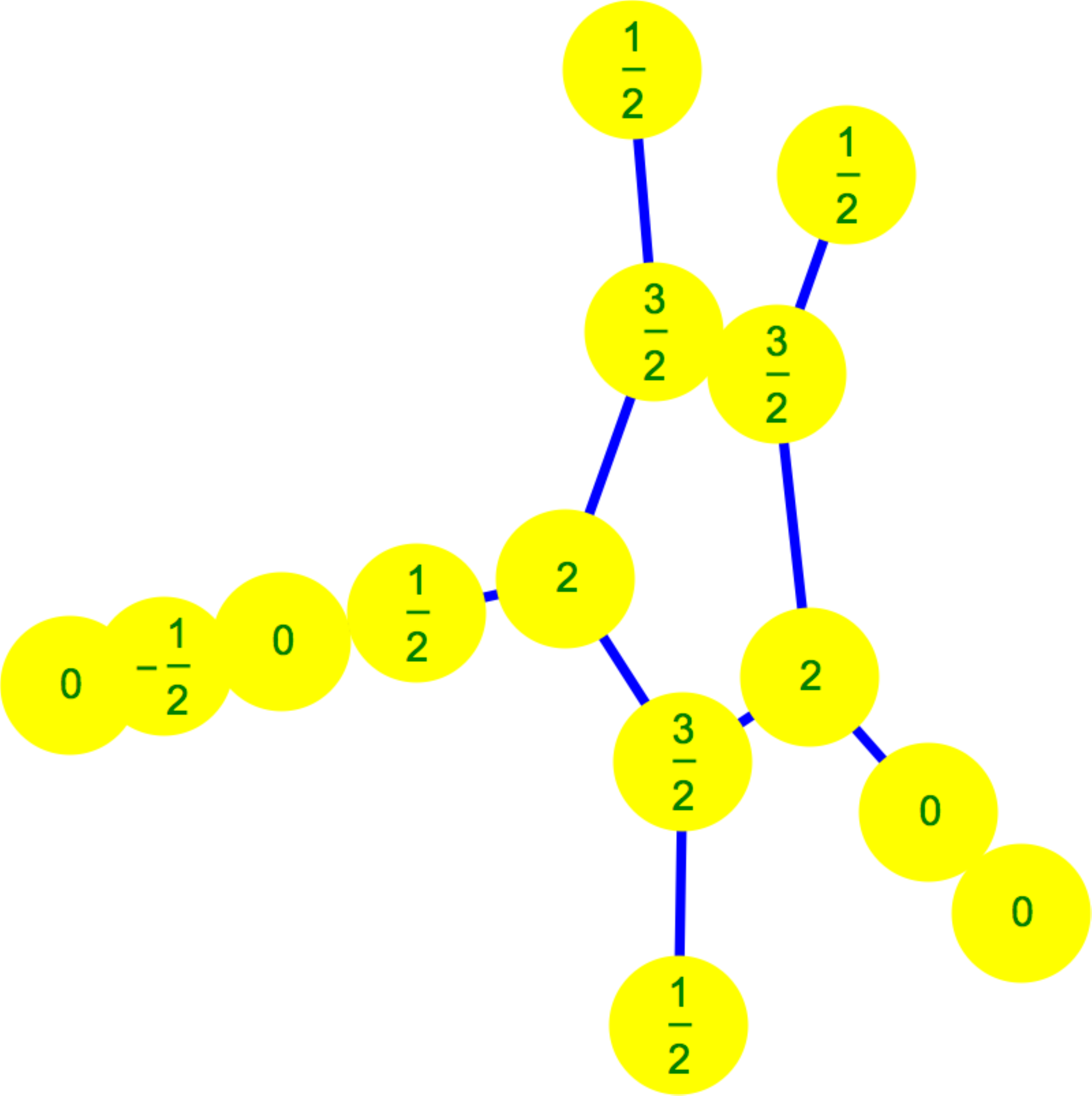}}
\scalebox{0.12}{\includegraphics{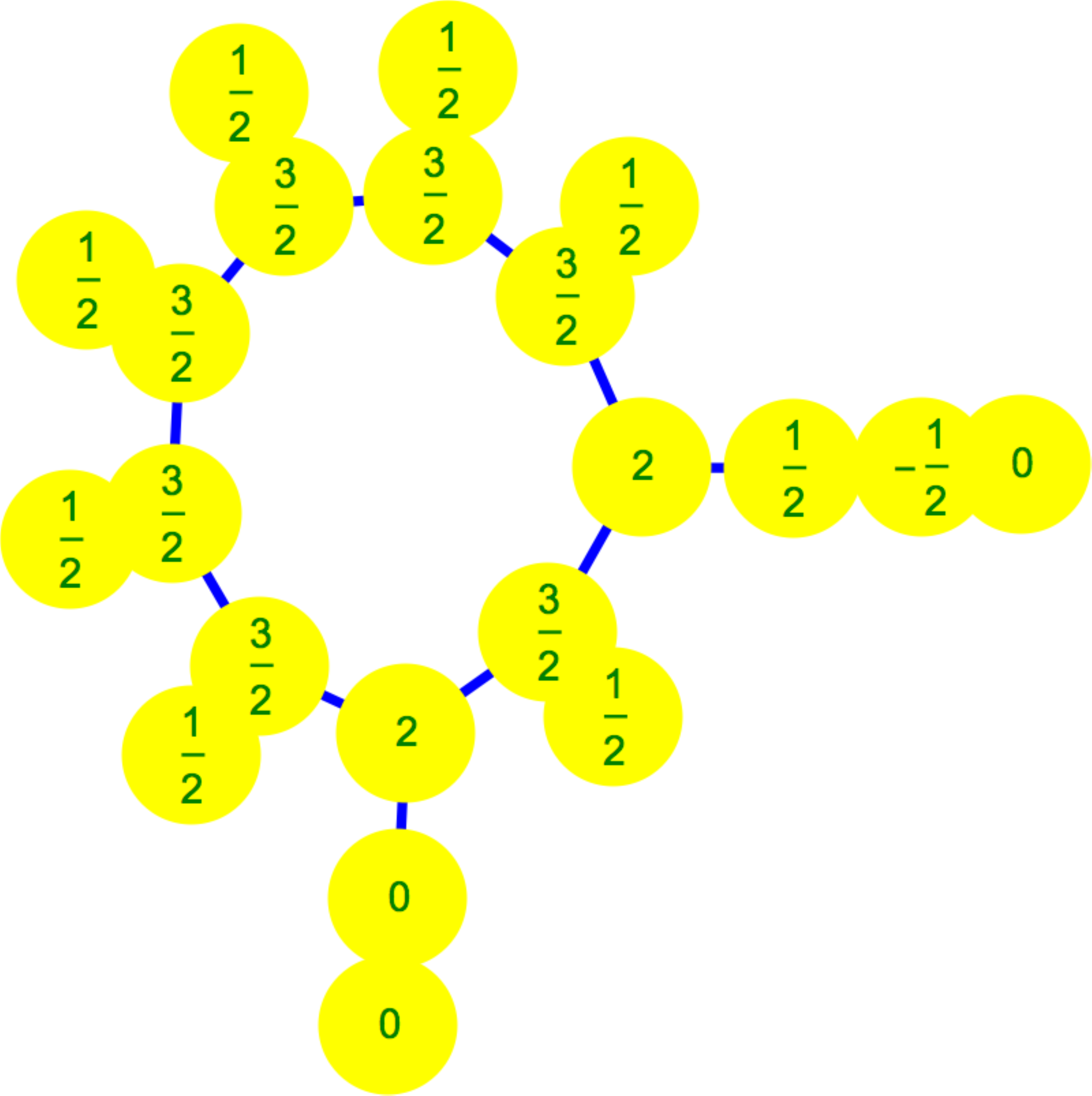}}
\scalebox{0.12}{\includegraphics{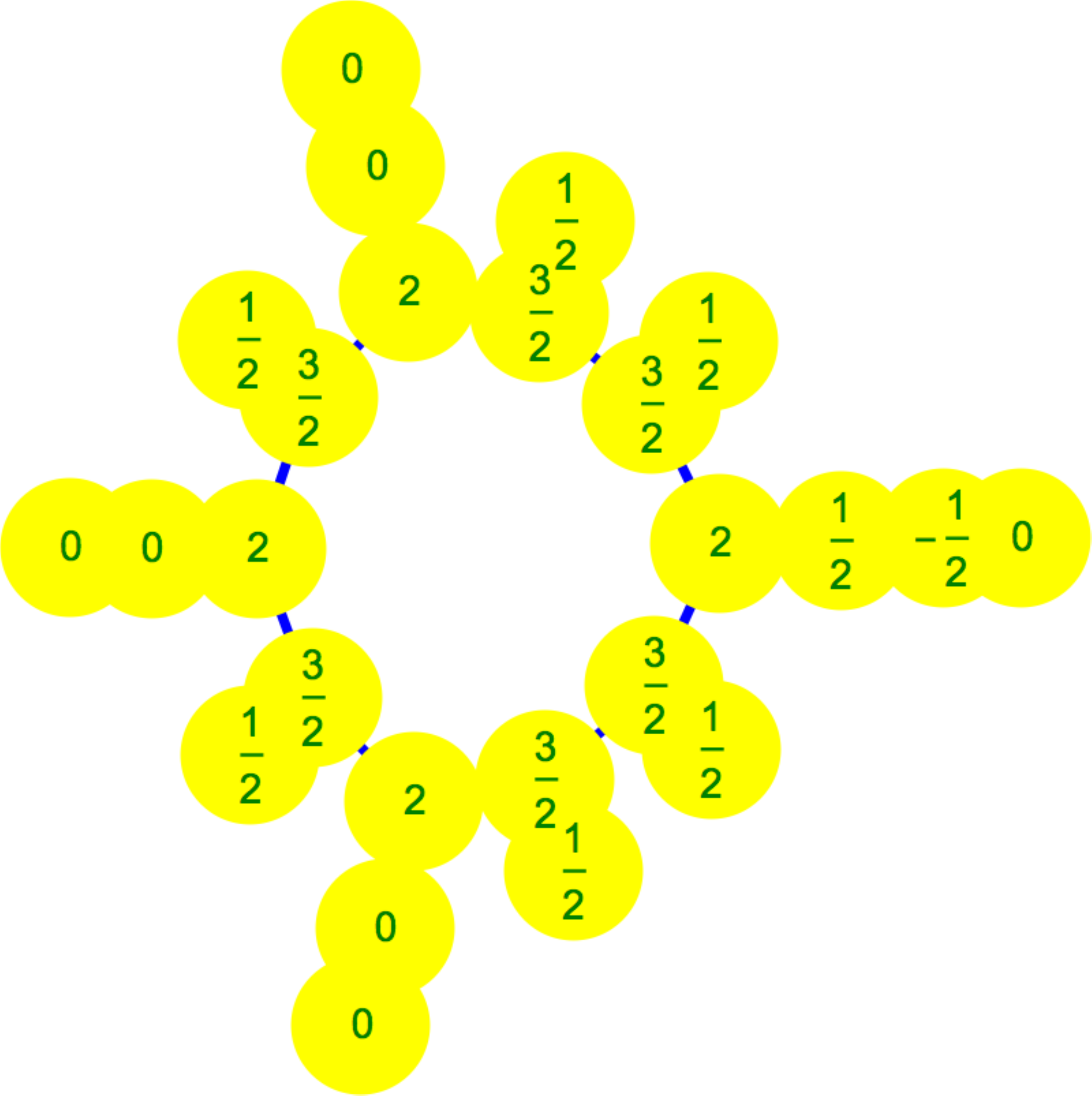}}
\caption{
Examples of Wu curvatures of graphs without triangles. 
First we see the Wu curvatures of three star graphs,  then 
we see the Wu curvatures for three random trees and finally for three 
random sun graphs.  }
\end{figure}

{\bf Remarks.} \\
{\bf 1)} If $G$ is a Barycentric refinement of a graph without triangles, 
then $K(x) = 1-5d/2+d^2/2+d$ at every vertex not adjacent to a non
flat point. Including the curvatures $(3d-8)/2$ at the adjacent points 
to the degree $d$ branch point, we get the total curvature contributing to 
a branch point 
$$  K(x)=(1-d/2)(1-2d) \; . $$
This is the number we should assign to a branch point in the continuum limit. The curvature attached
to an intersection of $d$ rays is $(1-d/2)(1-2d)$. \\
{\bf 2)} Motivated from the structure of {\bf spin networks} in quantum loop 
gravity, where the edges are equipped with curvature, one could also attach the 
curvature to the edges and give each edge $d=(a,b)$ a spin value $(d(a)+d(b)-5)$,
where $d(a),d(b)$ are the degrees at the vertices $a$ and $b$.   \\

While the variational problem of maximizing the Wu characteristic has to be confined
to graphs with a fixed number of vertices to be interesting, the global minimum 
on the class of connected graphs without triangles is known:

\begin{coro}
Among the class of all connected graphs without triangles, the line graphs 
minimize the Wu characteristic. 
\end{coro}
\begin{proof}
We see that the only way to get a negative curvature is to have $d=1$. 
It follows that the Wu characteristic for a graph without triangles is only negative for
a line graph like $K_2$. It is zero for a union of circular graphs and otherwise positive. 
\end{proof}

{\bf Examples:}  \\
{\bf 1)} For a circle bouquet graph, a collection of $k$-circles hitting a common point, the
Wu characteristic is $4k^2-5k+1$. \\

\begin{figure}[h]
\scalebox{0.24}{\includegraphics{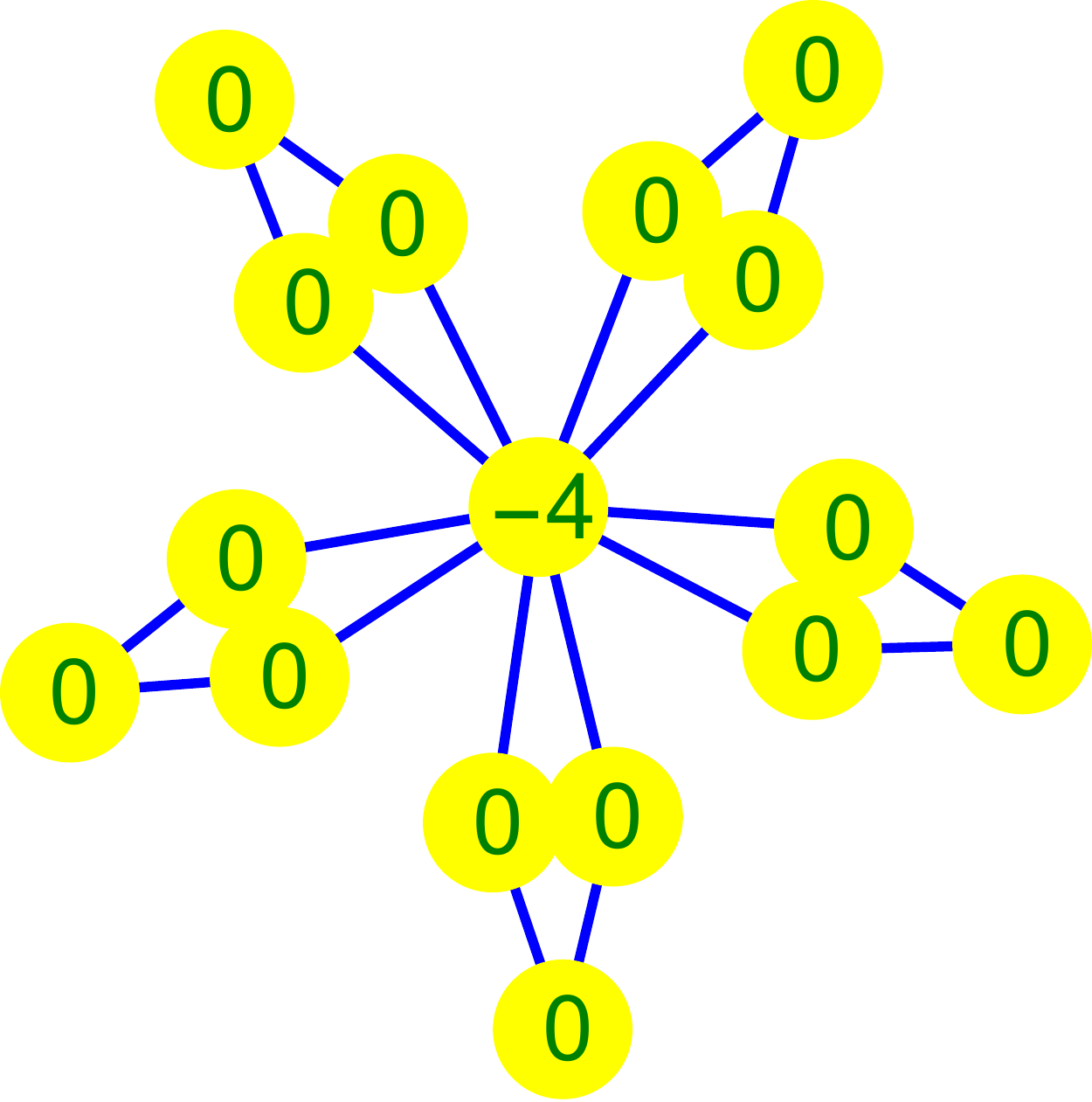}}
\scalebox{0.24}{\includegraphics{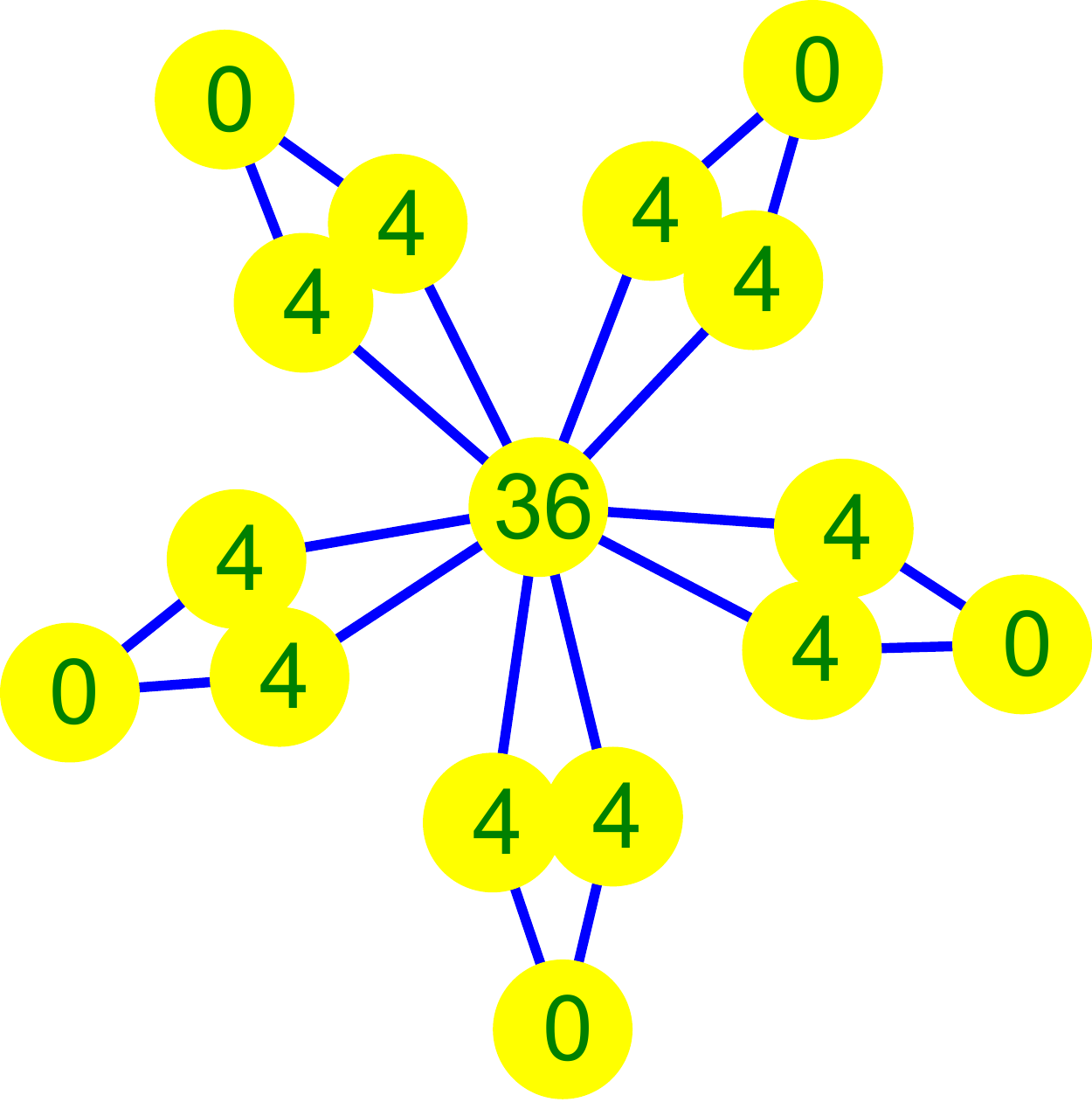}}
\caption{
The Euler and Wu curvatures of a $1$-dimensional sphere bouquet of 5 spheres. 
The Euler characteristic is $-4$, the Wu characteristic is $76$ which 
is $4k^2-5k+1$ for $k=5$. 
}
\end{figure}

\begin{figure}[h]
\scalebox{0.24}{\includegraphics{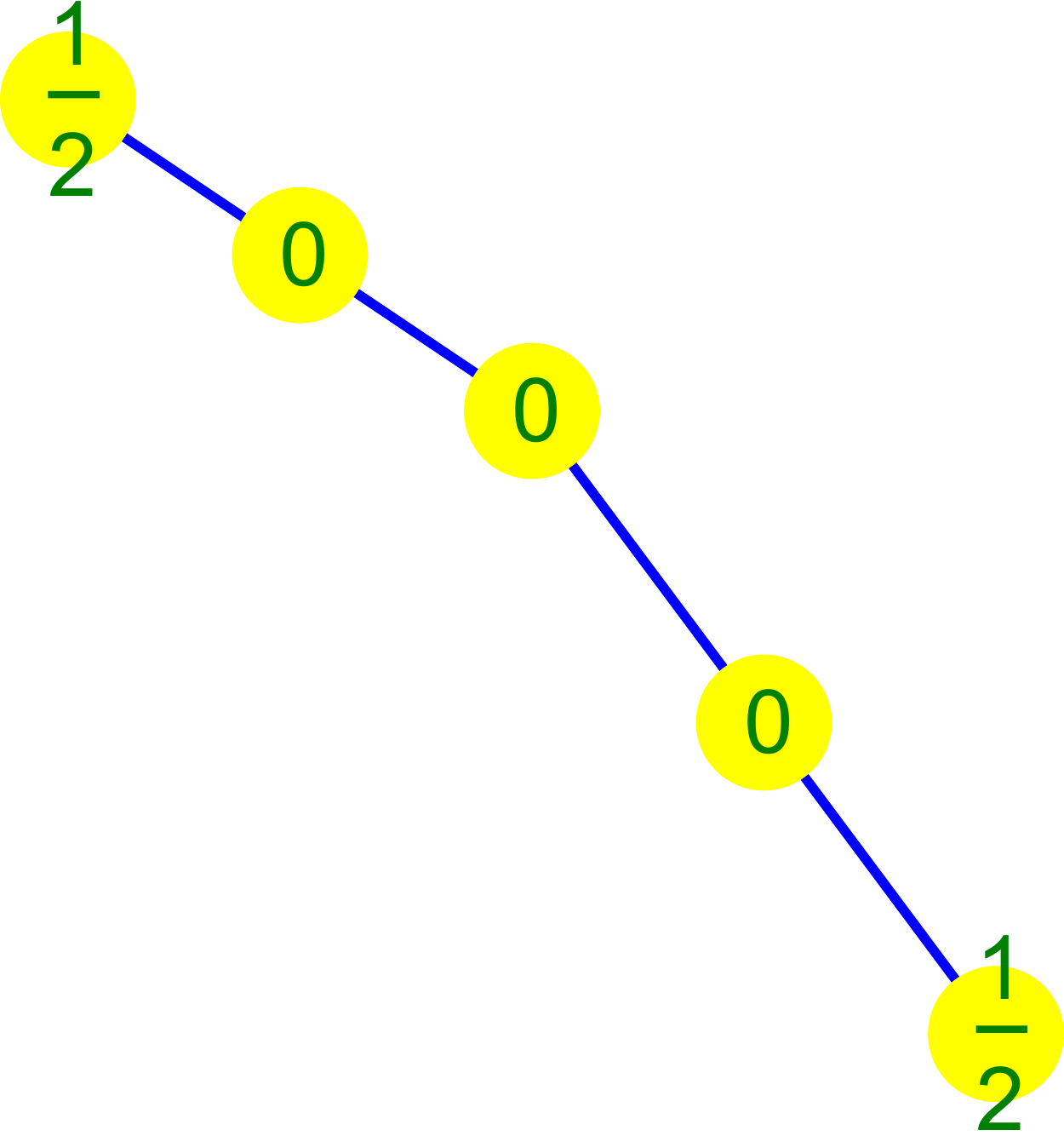}}
\scalebox{0.24}{\includegraphics{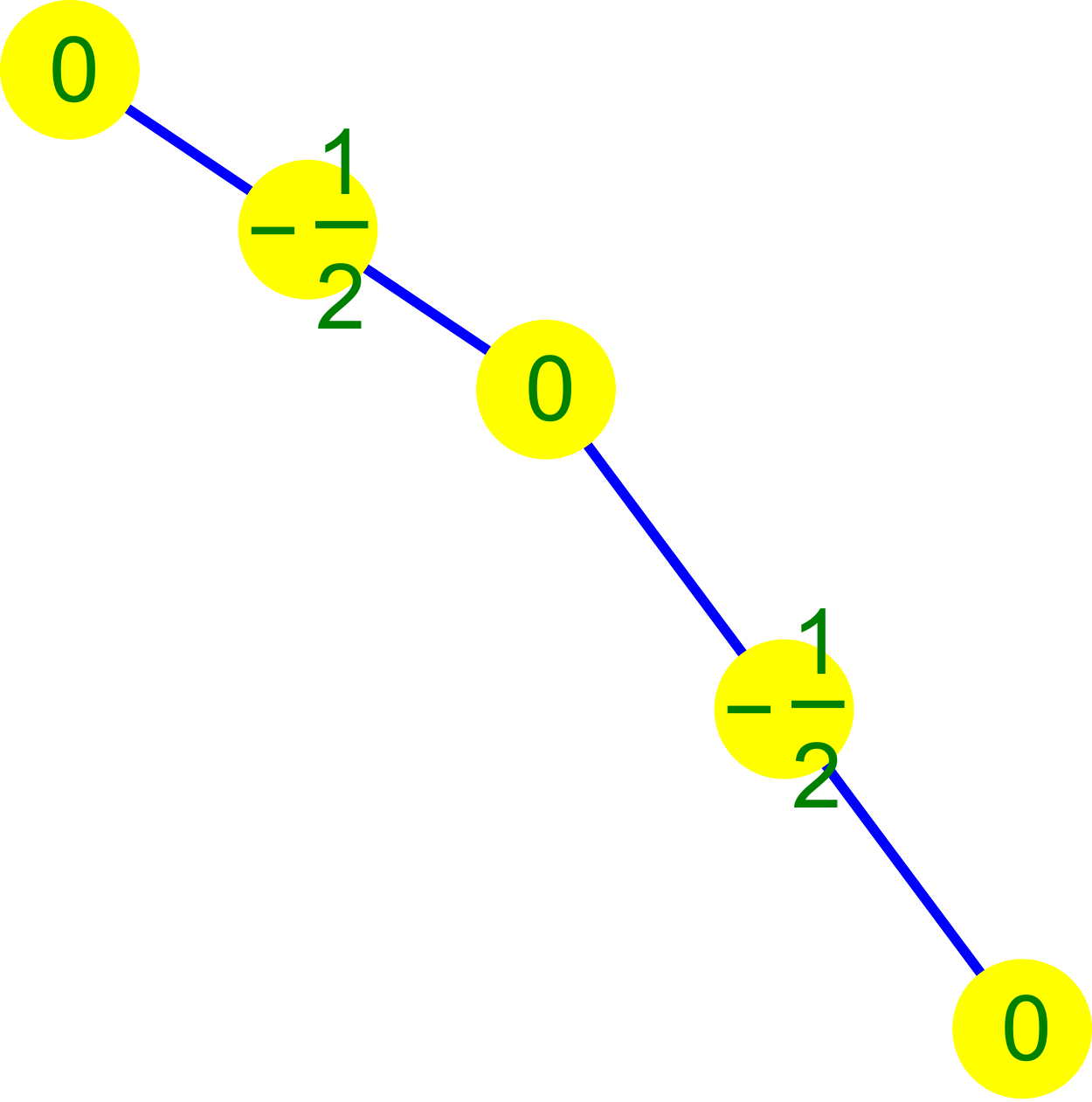}}
\caption{
The Euler and Wu curvatures of a line graph, a $1$-dimensional disk. The curvatures start to
be different in distance 1 to the boundary. The Euler characteristic is $1$, the
Wu characteristic of the disk is $-1$. 
The one dimensional line graphs are the only connected graphs without triangle for which the 
Wu characteristic is negative. 
}
\end{figure}

Here are some constructions.  \\

{\bf 2)} adding a circular loop at a two dimensional graph decreases the curvature at the glue
point by $1$.  \\
{\bf 3)} adding a two dimensional sphere to a two dimensional graph increases the Wu characteristic
by $1$. This means that it decreases the curvature by $1$. \\
{\bf 4)} Gluing a two and three dimensional sphere along a point produces a graph of Wu characteristic $1$.  \\


It follows from $\chi(G)-\chi(\delta G)$ that for 2-graphs with boundary, 
the Wu curvature is $1-d(x)/6$ in distance $2$ to the boundary. At the boundary,
the Wu curvature is $0$. 

\begin{figure}[h]
\scalebox{0.24}{\includegraphics{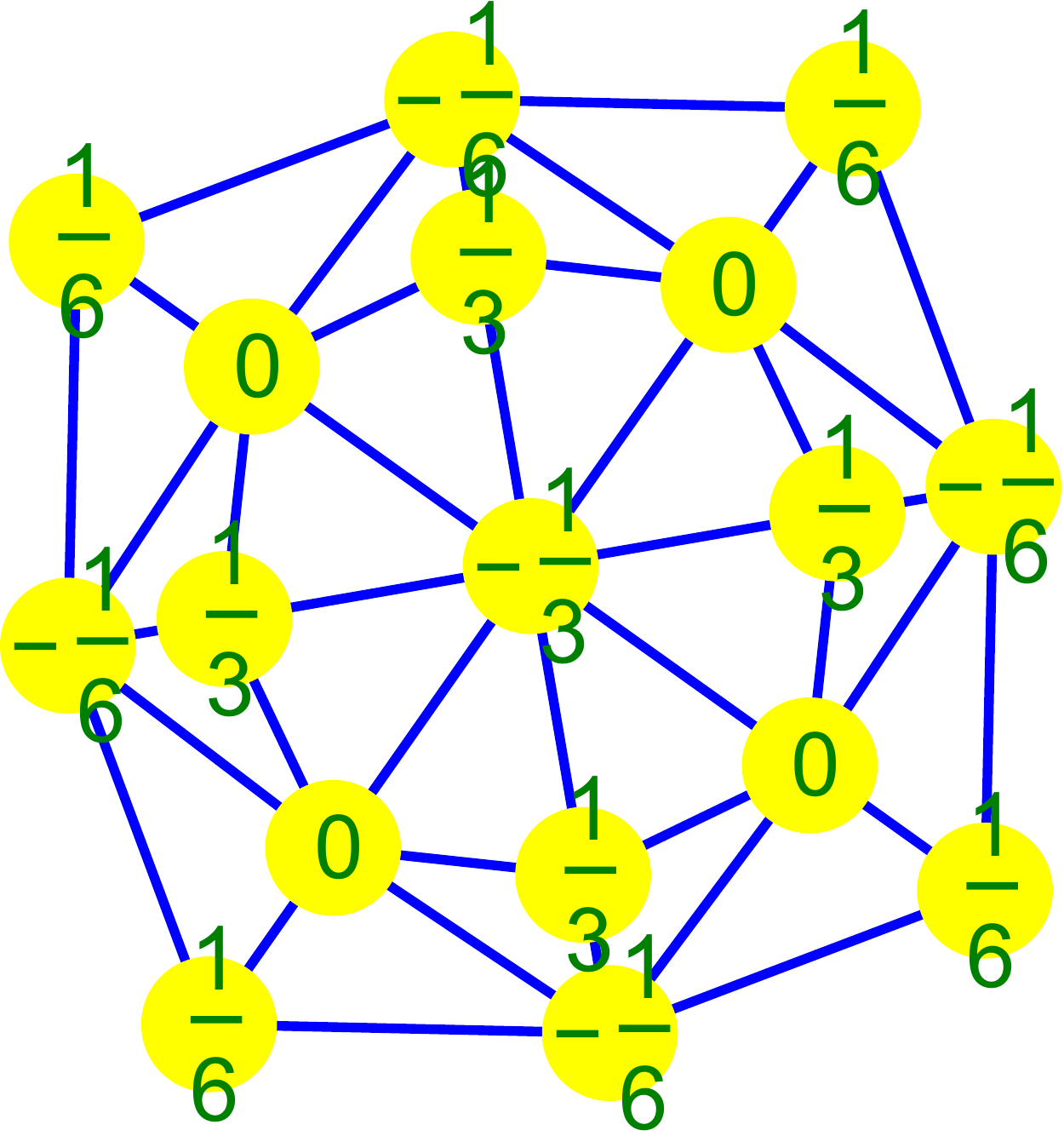}}
\scalebox{0.24}{\includegraphics{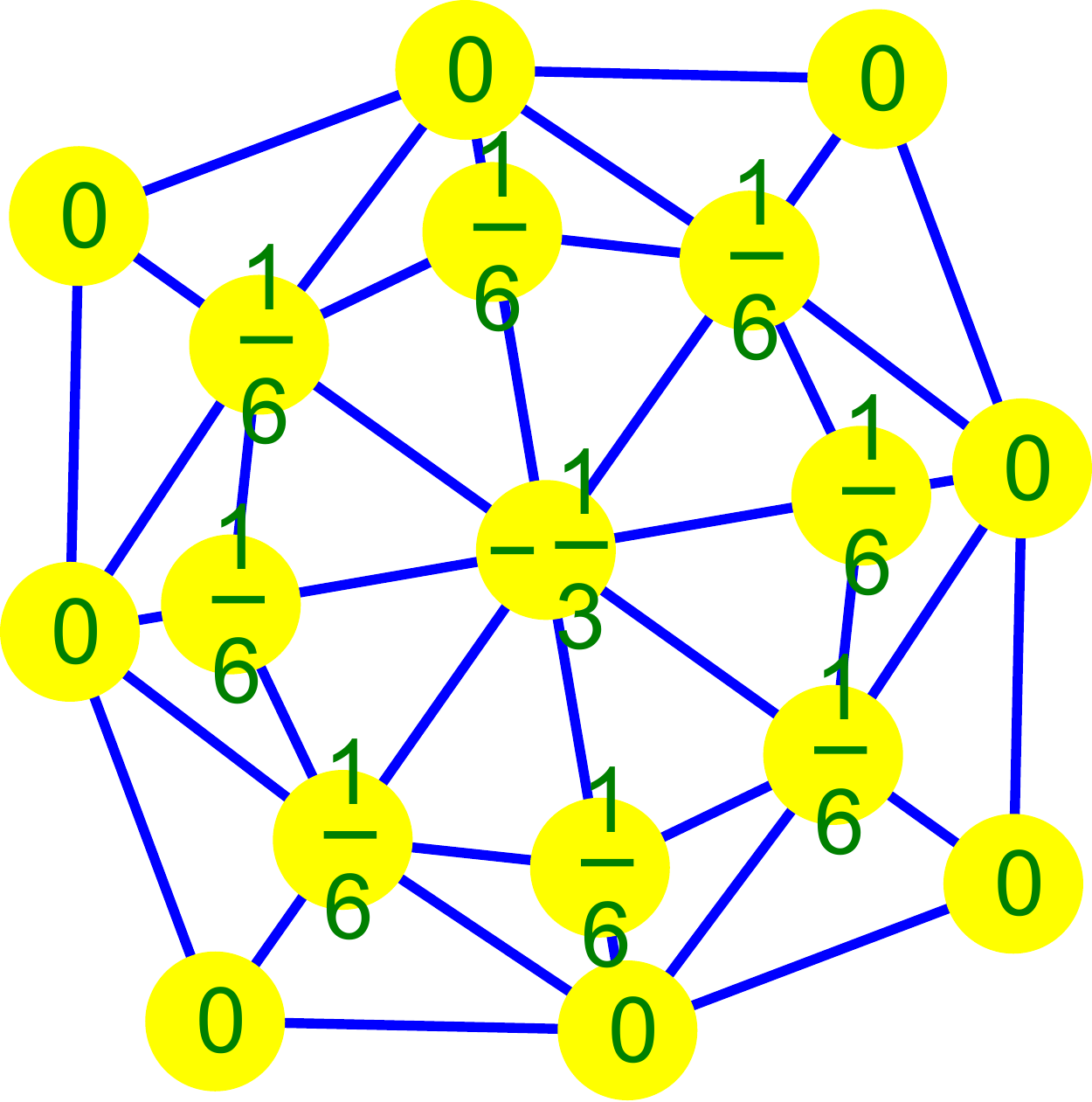}}
\caption{
The Euler curvatures and Wu curvatures of a $2$-disk. Both,
Euler characteristic and Wu characteristic are $1$. The 
later follows because the boundary, a circular graph, has
zero Euler characteristic. 
}
\end{figure}

For 3-graphs with boundary, the curvatures in the interior is 
$K_1(x)=(1-V_1(x)/2+V_2(x)/3-V_3(x)/4)$. Also the Dehn-Sommeville
invariant and $K_{d=1}(x)=(2 V_2-3 V_3)$ is zero in the interior. 
If we take $2K_1(x)$ as the curvature on a ball, then this gives the
right curvature for a ball and for a torus.  \\

More examples 

\begin{lemma}
If $G,H$ are two $d$-graphs and $x$ is a vertex in $G$ and $y$ a vertex in $H$,
then the connected sum $G \cup_{xy} H$ of $G$ and $H$ connected along $S(x)$ and $S(y)$
satisfies $\omega(G \cup_{xy} H) = \omega(G) + \omega(H) - \omega(S)$, where $S$ is a $d$-sphere.
\end{lemma}

\section{Intersection numbers}

Given a graph $G$ and two subgraph graphs $K,H$ of $G$, we can look at
$$ \omega(H,K) = \sum_{x,y} (-1)^{{\rm dim}(x) + {\rm dim}(y)}  $$
where $x$ is a subsimplex of $K$ and $y$ is a subsimplex of $H$ and
the sum is over all pairs for which the intersection is not empty. 
When seen like this, the Wu characteristic is $\omega(G) = \omega(G,G)$.
But we can look at the number $\omega(H,K)$ as an 
{\bf intersection number}. \\

{\bf Examples:}  \\
{\bf 1)} if $K,H$ are two simple curves on a graph intersecting simply in 
a point, then $\omega(K,H)=1$. Here is the computation: assume $x y z$
is the first arc and $a b c$ the second. The interactions are $by$ counting $1$ and
$xyb, yzb, aby, bcy$ counting $-1$ and then $xy ab, yz ab, yz ab,yz bc$ counting $1$.
The total sum is $1$. \\

\begin{figure}[h]
\scalebox{0.14}{\includegraphics{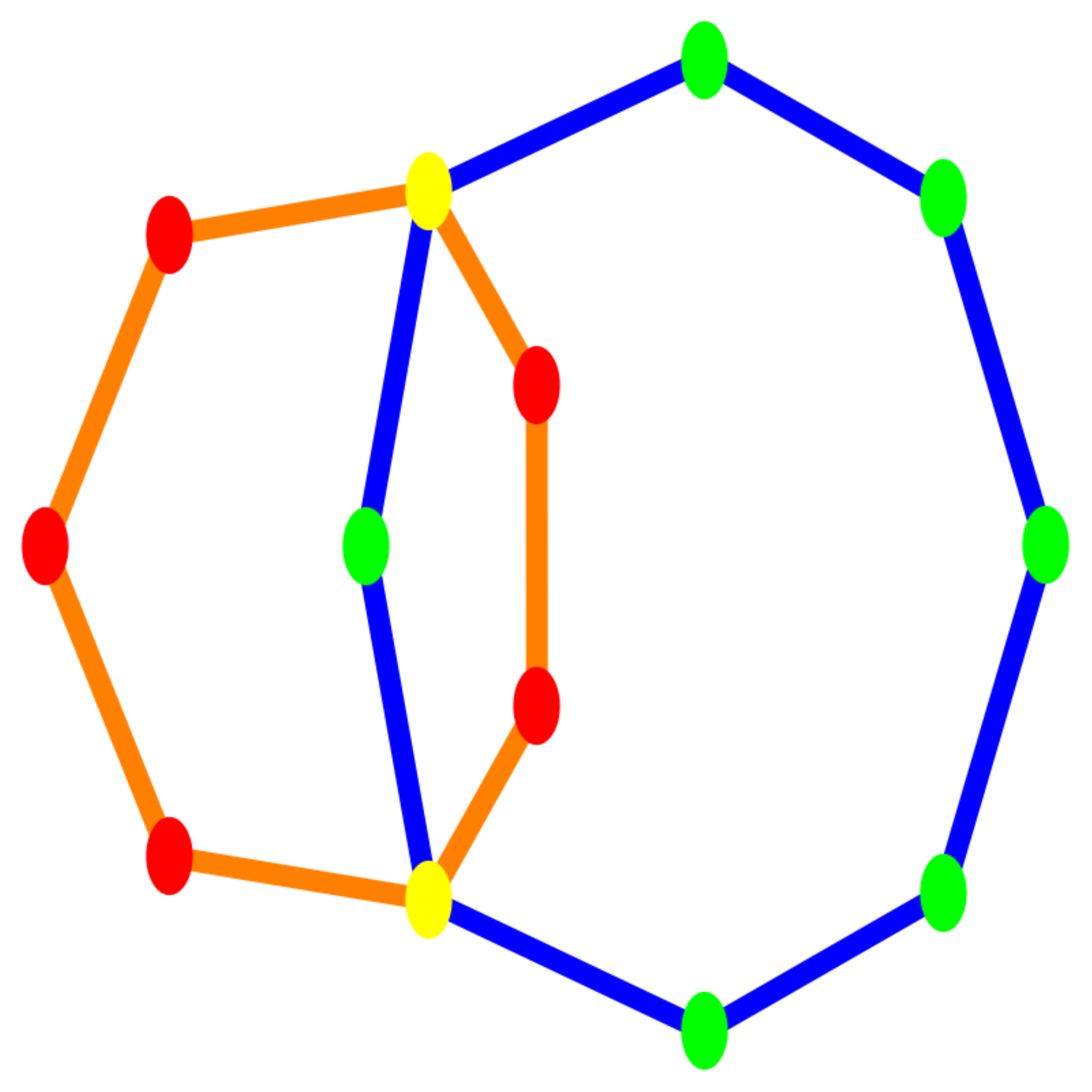}}
\scalebox{0.14}{\includegraphics{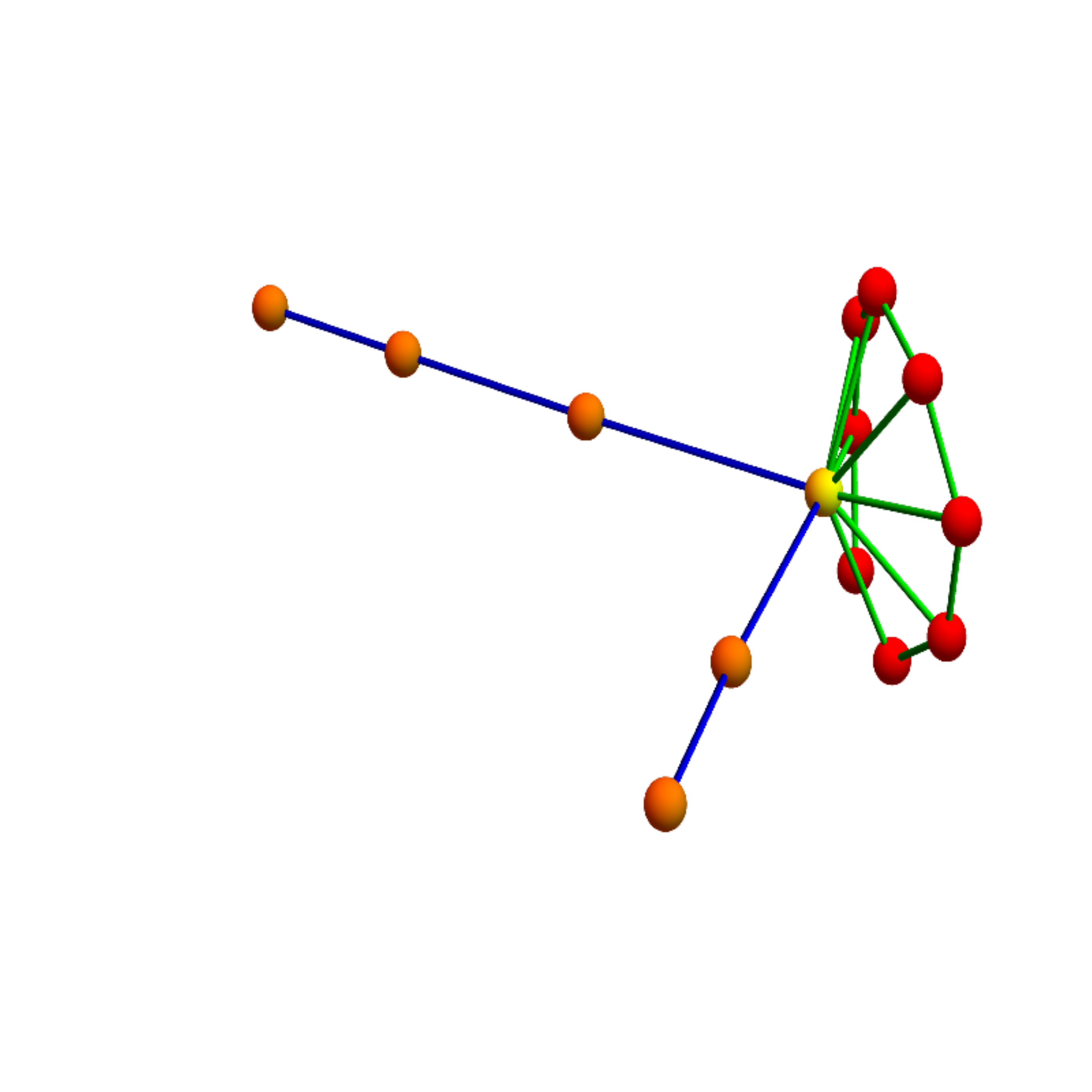}}
\scalebox{0.14}{\includegraphics{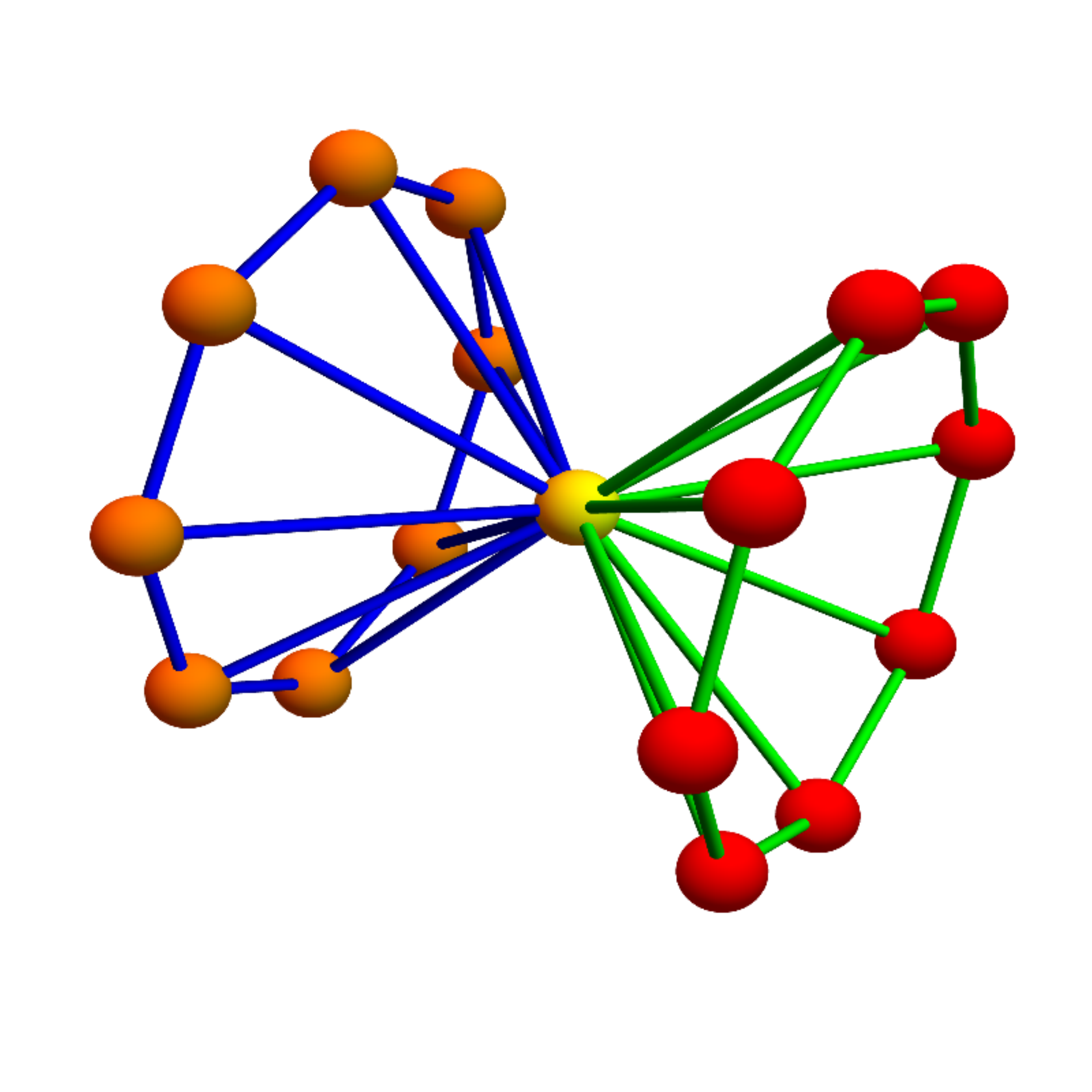}}
\caption{
Intersection examples: for two one dimensional graphs $A,B$ intersecting 
in $n$ points, the intersection number is $\omega(A,B)=n$. For a 
$1$-dimensional graph intersecting with a two dimensional surface (here a 
disc), the intersection number is $-1$. For two two dimensional 
surfaces intersecting in a point, the intersection number is $1$. 
}
\end{figure}

Lets look at the cubic valuation
$\omega_3(G) = \sum_{x,y,z} \sigma(x) \sigma(y) \sigma(z)$
summing over all triples $x,y,z$ of simplices which have a nonempty common intersection.  \\

{\bf Examples.} \\
{\bf 1)} In the case of $G=K_3$ with $f_G = x+y+z+xy+yz+xz+xyz$,
there are already 175 possible ordered triples of non-intersecting complete
sub graphs like
$$  \{x,xy,xyz\}, \{x,x,x\}, \{x,x,y\}, \{xyz,xy,xy\}  \; . $$ 

{\bf 2)} The behavior of cubic valuations on complete graphs $K_k$ is the same than 
for $\chi$. We have $\omega_3(K_d)=1$. 

\begin{thm}
For $d$-graphs with boundary and any {\bf even} $k>1$:
$$  \omega_k(G) = \chi(G) - \chi(\delta G)  $$
For any {\bf odd} $k$ we have
$$  \omega_k(G) = \chi(G)  $$
\end{thm}
\begin{proof}
The proof bootstraps and uses that for a unit sphere $\omega_2(S(x)) = \chi(S(x))$. 
\end{proof}

In other words, we know the Wu polynomial for $d$ graphs.

The examples of star illustrates how branch points produces quadratic growth in the 
number of branches. These numbers are of algebra-geometric nature as they go over
to the continuum.  Lets explain a bit more algebro-geometric relation.

\section{Questions}

{\bf A)} Is the Wu characteristic the only multiplicative quadratic valuation on graphs
which assigns the value $1$ to points? Such a result was suggested by Wu in the 50ies and
Gr\"unbaum cautioned in the 70ies to consider Dehn-Sommerville cases. 
We know that the Euler characteristic is the only linear valuation on the category of finite 
simple graphs which assigns the value $1$ to points. This fact has been proven a couple of times and
can be seen from the Barycentric refinement operator $A$, which maps $v(G)$ to $v(G_1)$ 
as it has only one eigenvalue $1$ and the eigenvector of $A^T$ is $(1,-1,1,-1,\dots)$ 
leading to Euler characteristic. 
While we understand the behavior of the $f$-vector under Barycentric refinement,
we don't understand yet the behavior of the $f$-matrix $V$ under Barycentric refinement but
we have not looked yet really. In the case of Barycentric refinements, we have searched for
the law first by data fitting, proved it and then found it in the literature.
Any possible law $V(G) \to V(G_1)$ should be detectable like that if it exists. \\

{\bf B)} Assume we know all the multi-linear valuations of a graph. How much of the graph $G$ is 
determined by these numbers? Are there non-graph isomorphic graphs with the property that all
$k$-linear valuations are the same? Yes, if we restrict to the local valuations, we can already find
trees which are not isomorphic but have all valuations the same. 
More difficult is to answer the question modulo homeomorphism equivalence \cite{KnillTopology} or even 
homotopy. Can we get a complete set of invariants by relaxing the local property?
As the story of invariants for geometric objects has shown repetitively, it would be
surprising to have an exhaustive set of generalized valuations which are invariants. The 
complexity of the graph isomorphism problem is not settled but it is not excluded yet that a 
polynomial number of valuations together with maybe other numbers like spectra 
suffice to determine the isomorphism type. History has shown that these are tough questions
and that even the search for counter examples can be computationally hard. Here is a quick 
experiment: on all connected graphs of $4$ and $5$ vertices, linear and quadratic valuations determine 
the isomorphism class uniquely (there are $6$ in the case of $4$ vertices and $21$ in the case of $5$ vertices).
But for $6$ vertices already, there are $112$ isomorphism classes of connected graphs but only 
101 different linear and quadratic valuation patterns so that some graphs have the same pattern. 
Including the cubic valuations already allows to distinguish $108$ isomorphism classes still missing $4$.
Including quartic valuations still does not resolve more. Indeed, the following 
two graphs $G,H$ with Stanley-Reiner ring elements $f_G = a+b+c+d+e+ab+bc+cd+de+df$ and
$f_H=a+b+c+d+e+ab+bc+cd+de+cf$ have the same $f$-forms for all degrees. We have 
$v(G)=v(H)=(-1,5)$ and $V(G)=V(H)=$ $\left[ \begin{array}{cc} 1 & 5 \\ 5 & 15 \\ \end{array} \right]$ 
and the $f$-cubic form is $V_3(G)=V_3(H) 
= \left[ \begin{array}{cc} \{5,11\} & \{11,21\} \\ \{11,21\} & \{21,41\} \\ \end{array} \right]$. 
By the way, the two graphs have different Kirchhoff spectra. \\


{\bf C)} Is there any significance of the {\bf $f$-spectrum}, the spectrum of the quadratic $f$-form 
$V$ defined by $G$? This spectrum is of course the same for isomorphic graph. Are there isospectral
graphs in the sense that they are not isomorphic but have the same $f$-spectrum? 
Is there a significance to the number of negative eigenvalues of $V$? Is there
a significance of the Perron-Frobenius eigenvector which always exist as $V$ is a positive matrix
if the graph is connected? What happens with the spectrum when applying 
Barycentric refinements repetitively?
We have seen that the spectrum of the Laplacian etc converges 
universally \cite{KnillBarycentric2,KnillBarycentric}. The number of negative eigenvalues of the 
$f$-matrix varies from graph to graph. The star graph $S_3$ with 4 vertices
has the $f$-matrix $V(G)=\left[ \begin{array}{cc} 4 & 6 \\ 6 & 9 \\ \end{array} \right]$ with 
eigenvalues $13$ and $0$. For all other star graphs $V(G)$ is positive definite. 
For circular graphs $C_n$ with $n \geq 4$, there is always one negative and positive eigenvalue.
For complete graph $K_k$, the number of negative eigenvalues of $V(G)$ is $[k/2]$. 
We see that in general the number of negative eigenvalues in the $f$-spectrum can change if $G$ 
undergoes a Barycentric refinement.  For the $f$-vector already, the inverse problem of characterizing 
the possible $f$-vectors is interesting. In the quadratic case we don't know anything about the possible
$f$-matrices.  \\

{\bf D)} For each valuation functional one can look at its variational problem: it is
the problem to maximize or minimize among all graphs with $n$
vertices. When restricting to $d$-graphs of a fixed number of vertices, maximizing 
volume is related to the upper bound conjecture. This shows that the problem of maximizing
a general linear or quadratic valuation can be hard even in very intuitive cases. 
As for Euler characteristic, one can look at the expectation of the 
functional on Erd\"os-Renyi graphs. For any $n$, we get so a function on $p$. 
We can ask to maximize or minimize this. For $n=6$, the utility graph is the only maximum
and the maximal value of the Wu characteristic is $15$. 
We have still to find a connected graph $G$ for which $\omega(G)<-v_0$. Do such examples exist? 
Attaching hairs to 2-spheres shows that we can for any $C<1$ find graphs $G$ of $n$ vertices 
for which $\omega(G)<-C \sqrt{n}$. Also for any $C<1$, there are graphs of order $n$ with 
$\omega(G) \geq C v_0^2$. Is it true that for any connected graph $G$ with $n$ vertices, the bound
$$   -\sqrt{n} \leq \omega(G) \leq n^2  \;  $$
holds for Wu characteristic? Monte Carlo experiments with smaller random graphs suggest 
such an estimate to hold, but this can be misleading and be a case for the law of small 
numbers. For Euler characteristic, where we know the expectation exactly on the Erd\"os-Renyi space
$E(n,p)$ explicitly \cite{randomgraph}. The expectation is
${\rm E}_{n,p}[\chi] = \sum_{k=1}^n (-1)^{k+1} \B{n}{k} p^{\B{k}{2}}$ showing that there
is $C>0$ and arbitrary large graphs with $n$ vertices for which $\chi(G) \geq \exp(C n)$
despite the fact that we can not explicitly construct such examples. This might also happen
for the Wu characteristic. \\

{\bf E)} Valuations can be studied from a statistical mechanics point of view. Define for every graph 
and a fixed function $J$ on edges of the intersection graph of simplices and fixed 0-form $h$ 
on the simplices the interaction energy
$$   H(G) = - \sum_{(x,y)} J(x,y) \sigma(x) \sigma(y) - \sum_x h(x) \sigma(x)  \; .  $$
For a particular graph theoretical approach, see \cite{Biggs1977}. 
A simple case to look at $H(G) = \omega(G)$, where $J=1$ and $h=0$. 
As notation from the Ising model is close, the spin values $\sigma(x)=\omega(x)$
of a simplex $x$ is geometrically defined and $(x,y)$ are all pairs of simplices which 
intersect. As custom, one can then define a probability distribution function 
$e^{-\beta H(G)}/Z_{\beta}$ on the Erd\"os-R\'enyi space of all graphs on a fixed vertex set, 
where $Z_{\beta} = \sum_G e^{-\beta H(G)}$ is the partition function. For any functional $f$ on graphs, there is a
mean value $\sum_G f(G) P_{\beta}(G)$ which could be studied in the limit $n \to \infty$.
One can also consider models where a host graph $E$ is fixed and exhausts it using natural sequences $E_n$
of graphs. But unlike the Ising model, where the underlying graph is fixed
and the spin configurations $\sigma$ are changed, the geometry alone determines the Wu energy. This renders
the story different. If we look at random Erd\"os-R\'enyi models, then we can look at the average Wu characteristic.
Unlike in the case of Euler characteristic, we were able to give an explicit formula for
the expectation value ${\rm E}_{p,n}[\chi]$ we don't have an expectation for the Wu characteristic yet.  \\

{\bf F)}  The Euler characteristic satisfies in full generality the 
Euler-Poincar\'e formula equating cohomological and combinatorial
versions of the Euler characteristic 
$$ b_0-b_1+b_2-b_3 \dots = v_0-v_1+v_2- \dots . $$
As noticed in \cite{brouwergraph}, it was Benno Eckmann \cite{Eckmann1} appeared have
been the first to point out this connection in a purely discrete setting without 
digressing to the continuum. An intriguing question is whether there is a cohomology 
defined for general finite simple graphs which produces an Euler-Poincar\'e formula for 
the quadratic Wu characteristic. It would have to be a sort of discrete intersection cohomology,
where a general finite simple graph now plays the role of a perverse sheaf. This is not so 
remote, as finite simple graphs are an Abelian category which contain discretizations of 
perverse sheaves. If a cohomology should exist which produces the Euler-Poincar\'e formula
for the Wu characteristic, then one can expect that its exterior derivative $d$ defines
an operator $D=d+d^*$ whose Laplacian $L=D^2$ produces a McKean-Singer formula equating
the Wu characteristic with the super trace of the heat kernel $\exp(-tL)$.  \\

{\bf G)} For $d$-graphs with boundary, the formula $\omega(G) = \chi(G) - \chi(\delta G)$ 
computes the Euler characteristic of the ``virtual interior" chain of the graph. 
This also works for complete graphs $K_{d+1}$ for which the Barycentric refinement is a d-ball.
For the triangle $G=K_3 = xyz+xy+yz+xz+x+y+z$ for example,
$\omega(G)$ computes the Euler characteristic of the chain  $xyz$. Unlike graphs, 
chains form a Boolean algebra and valuations are linear and nice. The category of graphs
is somehow like a manifold in the linear space of chains as adding two graphs throws us
out of the graph category. Lets write $\oplus$ for the Boolean addition on the Boolean algebra 
$2^(2^V)$ of chains. The addition is the symmetric difference on each simplex level. 
The example $(xy+x+y)  \oplus (yz+y+z) = xy+yz + x + z$ shows that the sum is no more a graph. 
The product $(xy+x+y) \star (yz+y+z)=y$ which is the intersection on each dimension however is a graph. 
The example $f=xyz +xy+yz+xz+x+y+z = (xy+yz+yx + x + y + z) \oplus xyz=g+h$ shows how to
break up a triangle into two chains. The first is the boundary chain, the second is the interior chain. 
But this decomposition is only possible in the ambient Boolean ring and not in the category of graphs.
As in the Stanley-Reisner picture (even so the multiplication in that ring is different), 
the Euler characteristic the triangle as $-f(-1,-1,-1)=1$,
the Euler characteristic of the boundary is $-g(-1,-1,-1) = 0$, 
and the interior is $-h(-1,-1,-1)=1$. We see that the Wu characteristic
of $K_n$ is the Euler characteristic of the interior of $K_n$ and the same holds for a d-ball.
In \cite{KnillTopology} we defined a notion of homeomorphism for graphs which treats graphs as higher 
dimensional structures which is compatible with homotopy, cohomology and dimension by design.
The topology of a graph is given by a base of subgraphs whose nerve is homotopic to the graph. We
have seen that such a topology always works and is given by the star graph.
This picture can now be generalized a bit by allowing the base to consist of virtual open balls. \\

{\bf H)} This question is still under investigation: what happens with the $f$-matrix $V(G)$ of a graph 
$G$ if it undergoes Barycentric refinement $G \to G_1$? As the $f$-vector $v(G)$ is transformed with
a universal linear transformation, one can expect something similar to happen for the $f$-matrix, or more
generally for the intersection matrix $V(A,B)$ for two intersecting graphs. If $G=C_n$, then 
$V(G)=\left[ \begin{array}{cc} n & 2n \\ 2n & 3n \\ \end{array} \right]$. 
For graphs without triangles, the $f$-matrix is a symmetric $2 \times 2$ matrix and there are
three independent components. We can get the Barycentric refinement using linear algebra: 
as for a cycle graph $C_4$ we have $(V_{11},V_{12},V_{22})=(4,8,12) \to (8,16,24)$,
$K_2$, we get $(2,2,21) \to (3,4,4)$ and for the figure $8$ graph, we have
$(7,16,32) \to (15,32,56)$. The transformation matrix can be found easily as 
$$  A = \left[ \begin{array}{ccc} 1 & \frac{1}{2} & 0 \\
              0 & 2 & 0 \\ 0 & \frac{3}{2} & 1 \\ \end{array} \right ] \; . $$
It has eigenavalues $2,1,1$. At the moment, it looks as if this does not work any more
for $2$-dimensional graphs with clique number $3$, where we look for a $6 \times 6$ transformation matrix $A$.
For $K_3$ for example, we have $(3,6,3,9,3,1)$ mapping to $(7,24,18,78,54,36)$
reflecting the fact that the $f$-matrix of $K_3$ is 
$\left[ \begin{array}{ccc} 3 & 6 & 3 \\ 6 & 9 & 3 \\ 3 & 3 & 1 \\ \end{array} \right]$ 
and the $f$-matrix of its Barycentric refinement (a wheel graph with $C_6$ boundary) is 
$\left[ \begin{array}{ccc} 7 & 24 & 18 \\ 24 & 78 & 54 \\ 18 & 54 & 36 \\ \end{array} \right]$.
When trying to find a linear transformation relating the two $f$-matrices using a set of 
6 independent graphs, it does not work any more in general. 
There could be a possible on some subclasses of graphs, as there appear linear relations 
necessary. Also a nonlinear transformation rule $V(G) \to V(G_1)$ can not yet be ruled out. 

\section{Conclusion}

According to \cite{BiggsRoots}, there are two basic counting principle which underlie
most of arithmetic. It is that for two disjoint finite sets, $|A \cup B| = |A| + |B|$ and 
$|A \times B| = |A| \times |B|$. This fundamental structure is also present for graphs 
and there is one quantity $|A|$ which in particular satisfies these properties. It is the 
Euler characteristic.  In this light, the Wu characteristics and its higher order versions
which were discussed here, are also fundamental as they satisfy both counting principles. \\

A quadratic valuation attaches to a pair of simplicial complexes a number
$\omega(A,B)$ such that $A \to \omega(A,B)$ and $B \to \omega(A,B)$ are both
valuations. Similarly, as valuations are linear in the $f$-vector $v$
where $v_i$ is the number of simplices of dimension $i$, a
quadratic valuation is a linear function on $f$-quadratic forms $V$, where $V_{ij}$
is the number of simplices which are intersections of $i$ and $j$ dimensional simplices.
The function $G \to \omega(G,G)$ in particular defines a functional on
simplicial complexes which measures a self interaction energy.
It does so especially for finite simple graphs equipped with the Whitney complex.
An example of a quadratic valuation is the Wu characteristic $\omega(G)$ \cite{Wu1959}
which sums over $\omega(x \times y)$ where $(x,y)$ is an ordered
pair of complete subgraph of $G$ which have non-empty intersection. It is more subtle than Euler
characteristic as unlike the later, the higher order Wu characteristics are not 
a homotopy invariant, and so is more subtle like the Bott invariant defined in
or Reidemeister (analytic) torsion. We prove that $\omega$ is invariant under edge or Barycentric
refinements, that it behaves additively with respect to connected sums and multiplicatively
with respect to Cartesian products. We prove a Gauss-Bonnet formula for multi-linear valuations
and use this to show that for a geometric graph $G$ with boundary
$\delta G$, all the higher order Wu characteristic satisfy $\omega(G) = \chi(G)-\chi(\delta G)$, where $\chi$ is the Euler
characteristic. Simple examples like discrete curves show that the curvature is now nonlinear in the vertex degree.
The setup allows to construct higher order Dehn-Sommerville relations, answering so positively a conjecture
of Gr\"unbaum \cite{Gruenbaum1970} from 1970. The simplest example is the vanishing
of the quadratic valuation $X_{1,3}(M) = \sum_{k,l} \chi_1(k) \chi_3(l) V_{kl}(G) = \chi_1^T V \chi_3$
with $\chi_1=(1,-1,1,-1,1), \chi_3 = (0, -22, 33, -40, 45)$, where
the quadratic $f$-form $V_{kl}(G)$ counts the
number of intersecting $K_k$ and $K_l$ subgraphs in a triangulation of a 4-manifold $M$.
More generally, we show that if $Y$ is a Dehn-Sommerville invariant then the quadratic valuation
$X(G) = \sum_{k,l} Y(l) (-1)^k V_{kl}(G)$ is zero on $d$-graphs. 
The lack of homotopy invariance is not a surprise, as most linear valuations already are not homotopy invariant.
It is still not known whether the quadratic Wu invariant is the only quadratic invariant valuation
on general graphs which is invariant under Barycentric refinement.
For discrete curves, graphs without triangles for example, the curvature for Euler characteristic is
$1-d(x)/2$ while the curvature of the Wu characteristic is $K(x) = 1-5d/2+d^2/2+\sum_i d_i/2$, where $d$ is the
vertex degree at $x$ and $d_i$ the neighboring vertex degrees. 
Quadratic valuations also produce intersection numbers and are interesting for algebraic
sets, where they produces combinatorial invariants for varieties which go beyond Euler characteristic
as the quadratic Dehn-Sommerville curvatures at isolated singularities report on their internal structure.

\begin{figure}[!htpb]
\scalebox{0.12}{\includegraphics{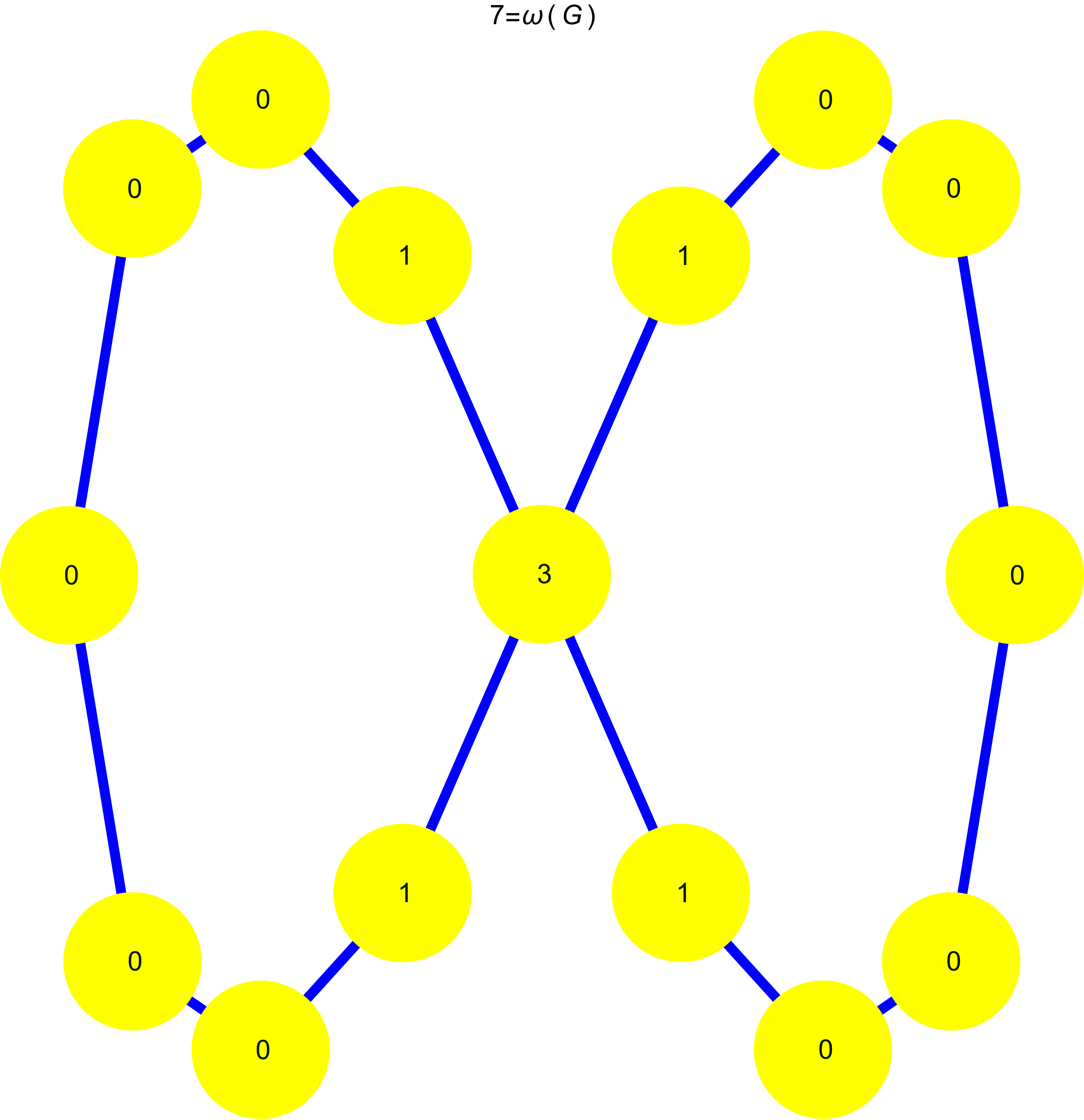}}
\scalebox{0.12}{\includegraphics{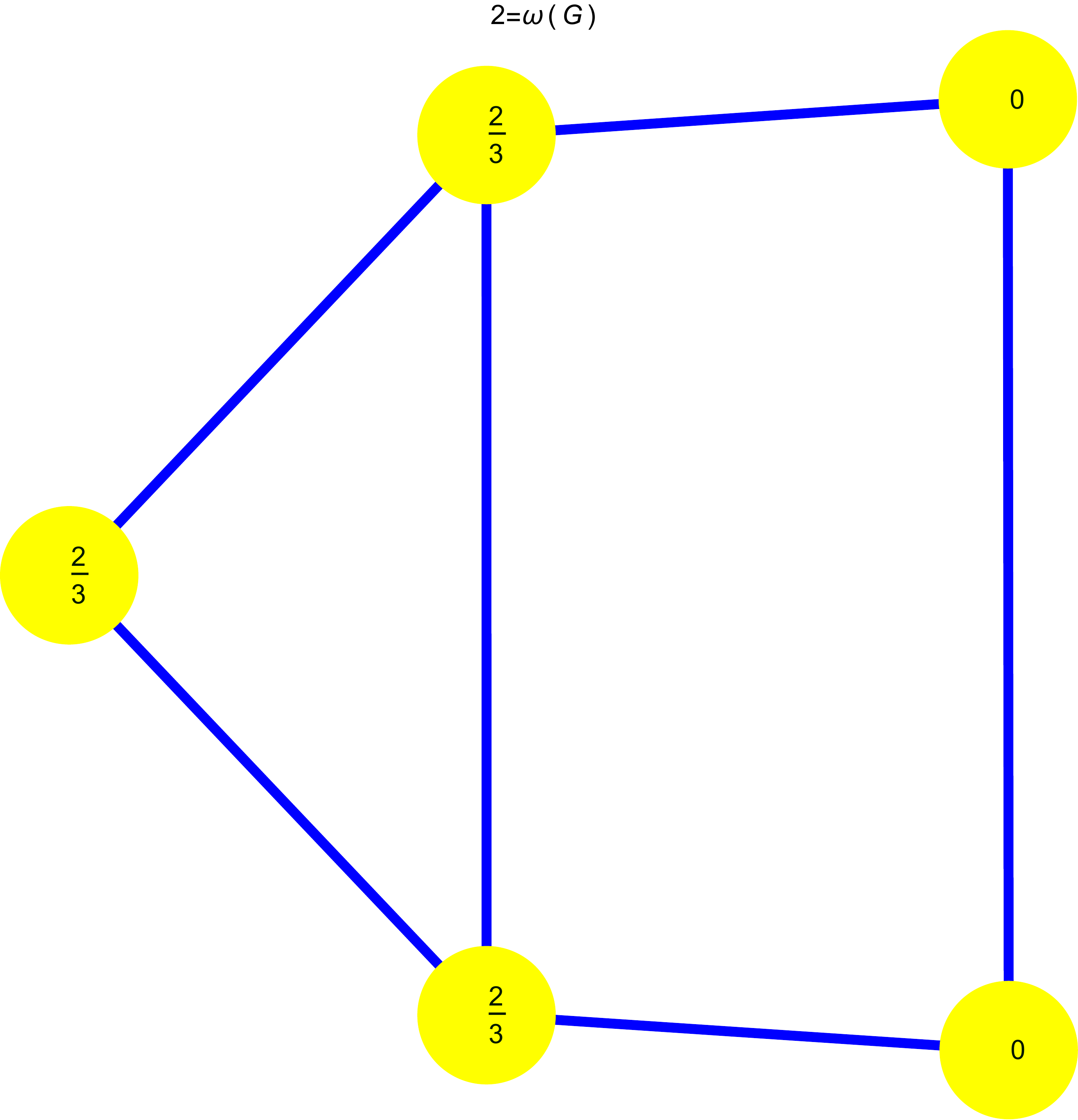}}
\scalebox{0.12}{\includegraphics{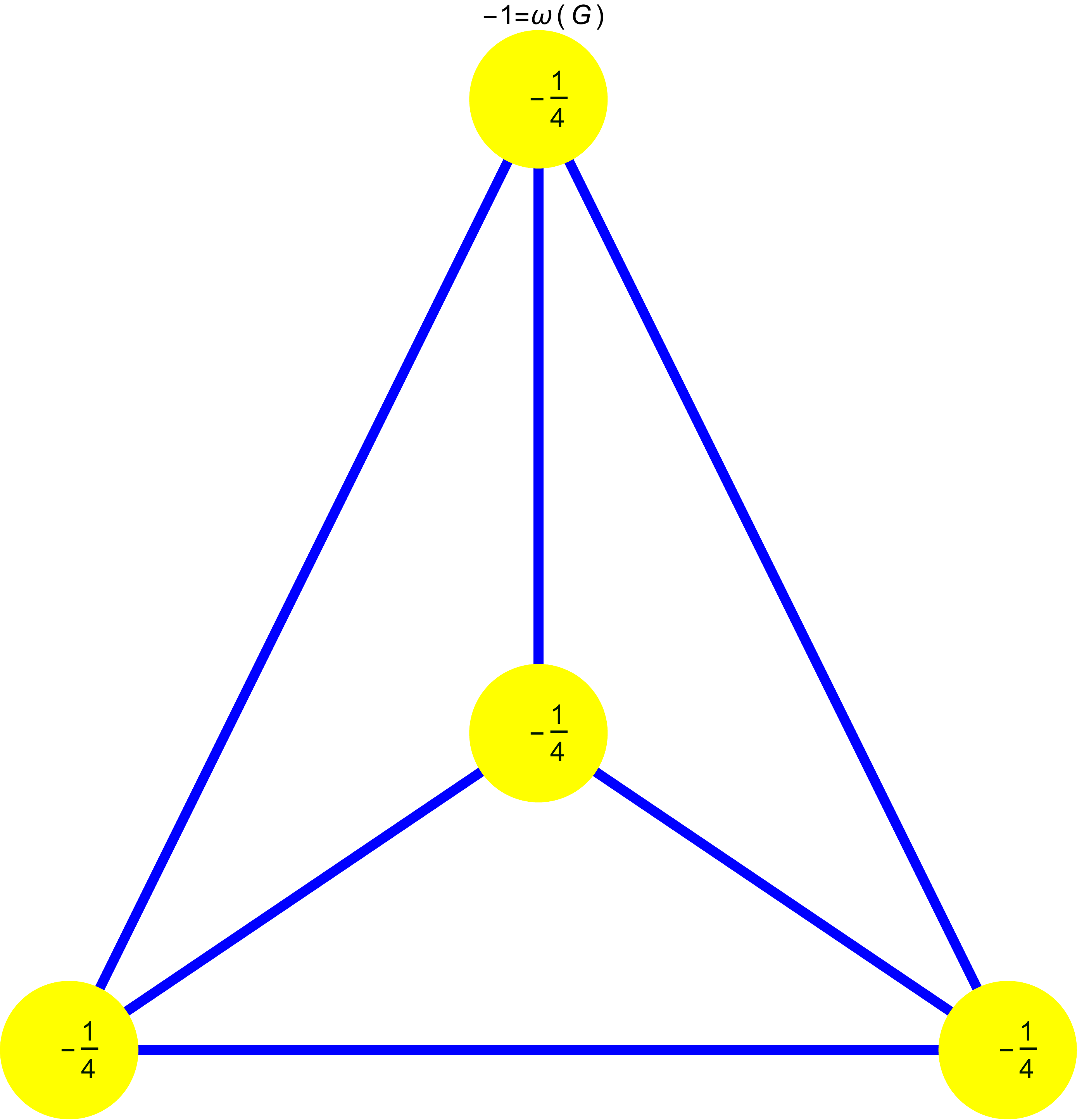}}
\scalebox{0.12}{\includegraphics{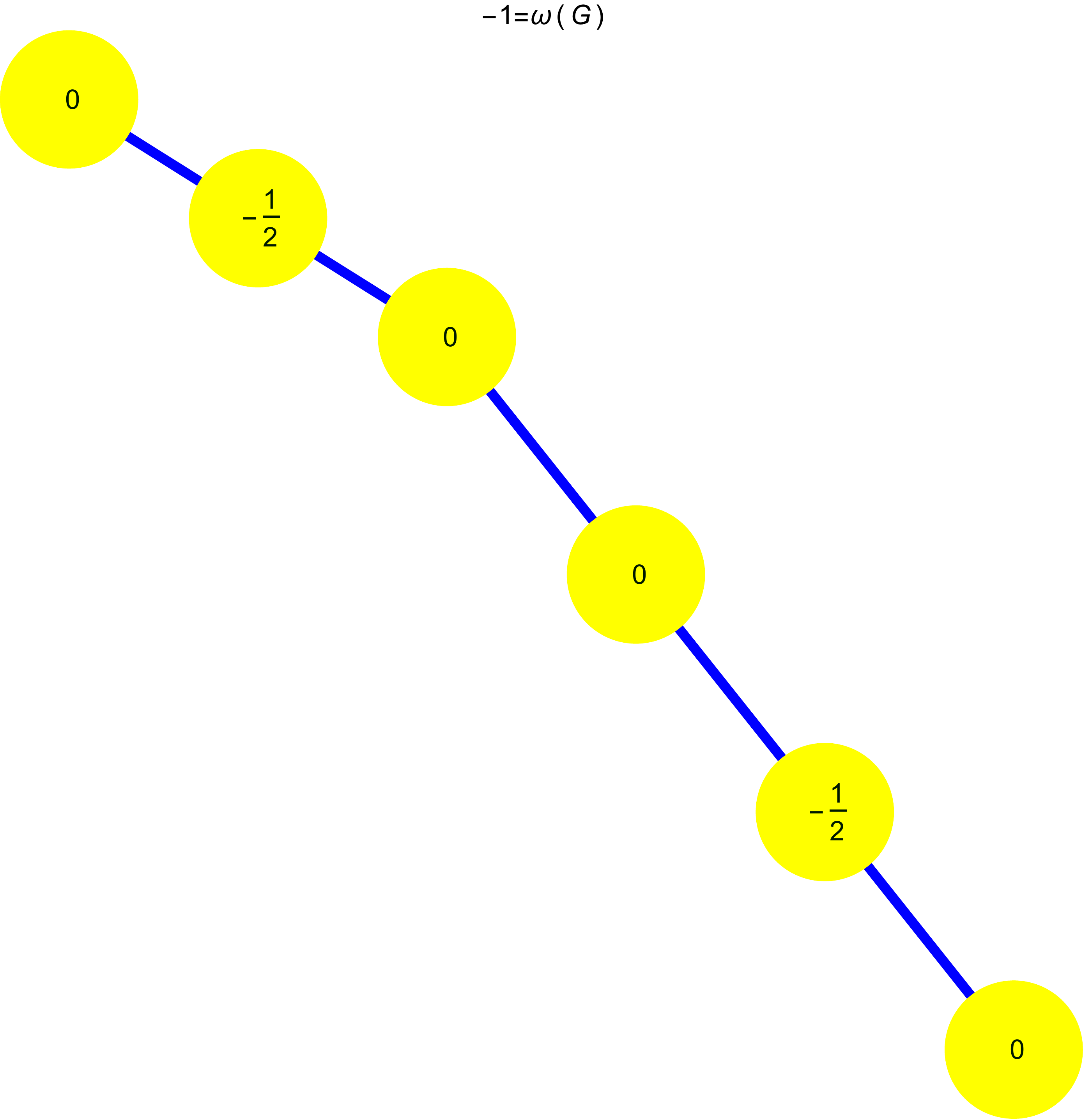}}
\scalebox{0.12}{\includegraphics{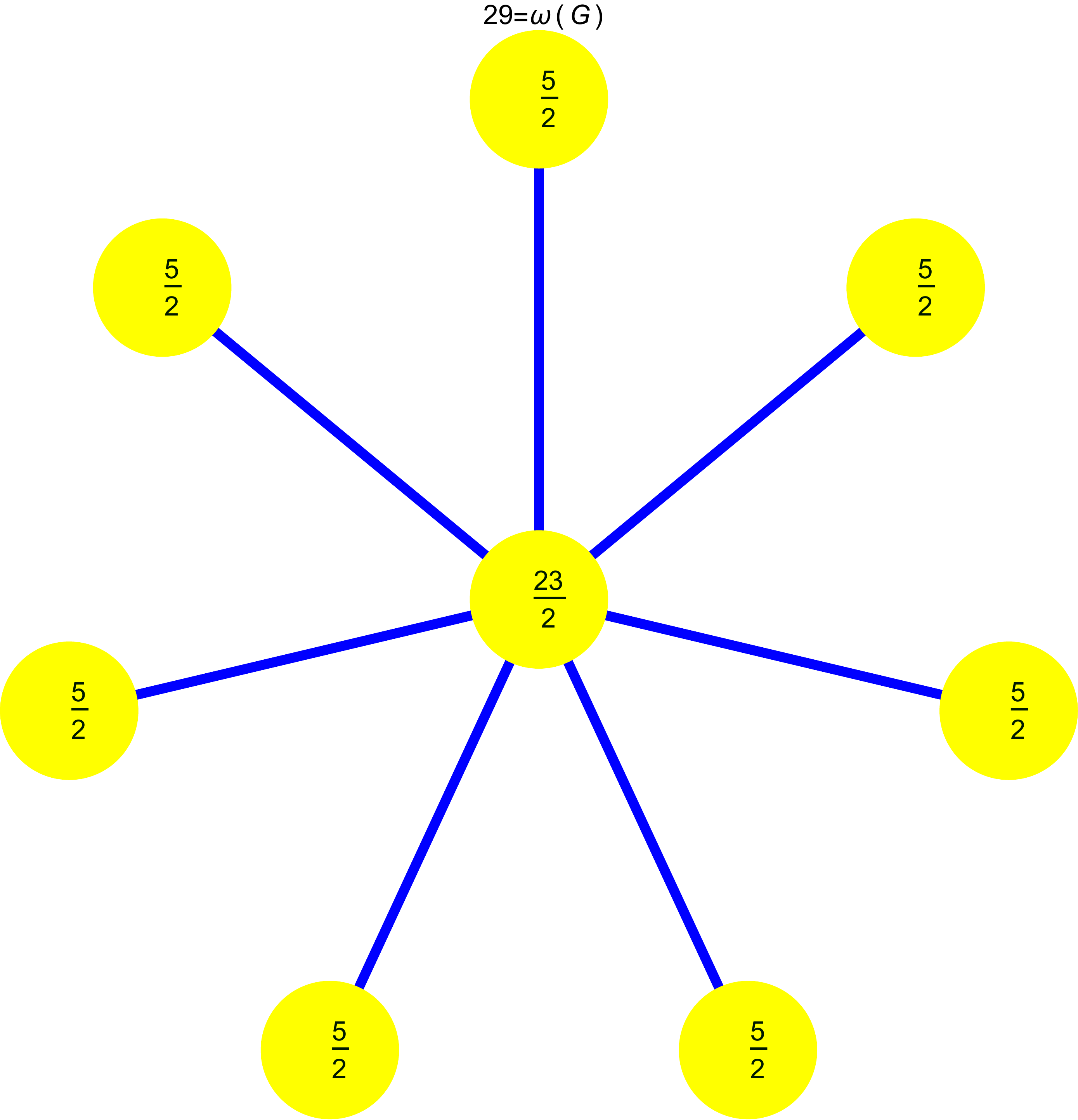}}
\scalebox{0.12}{\includegraphics{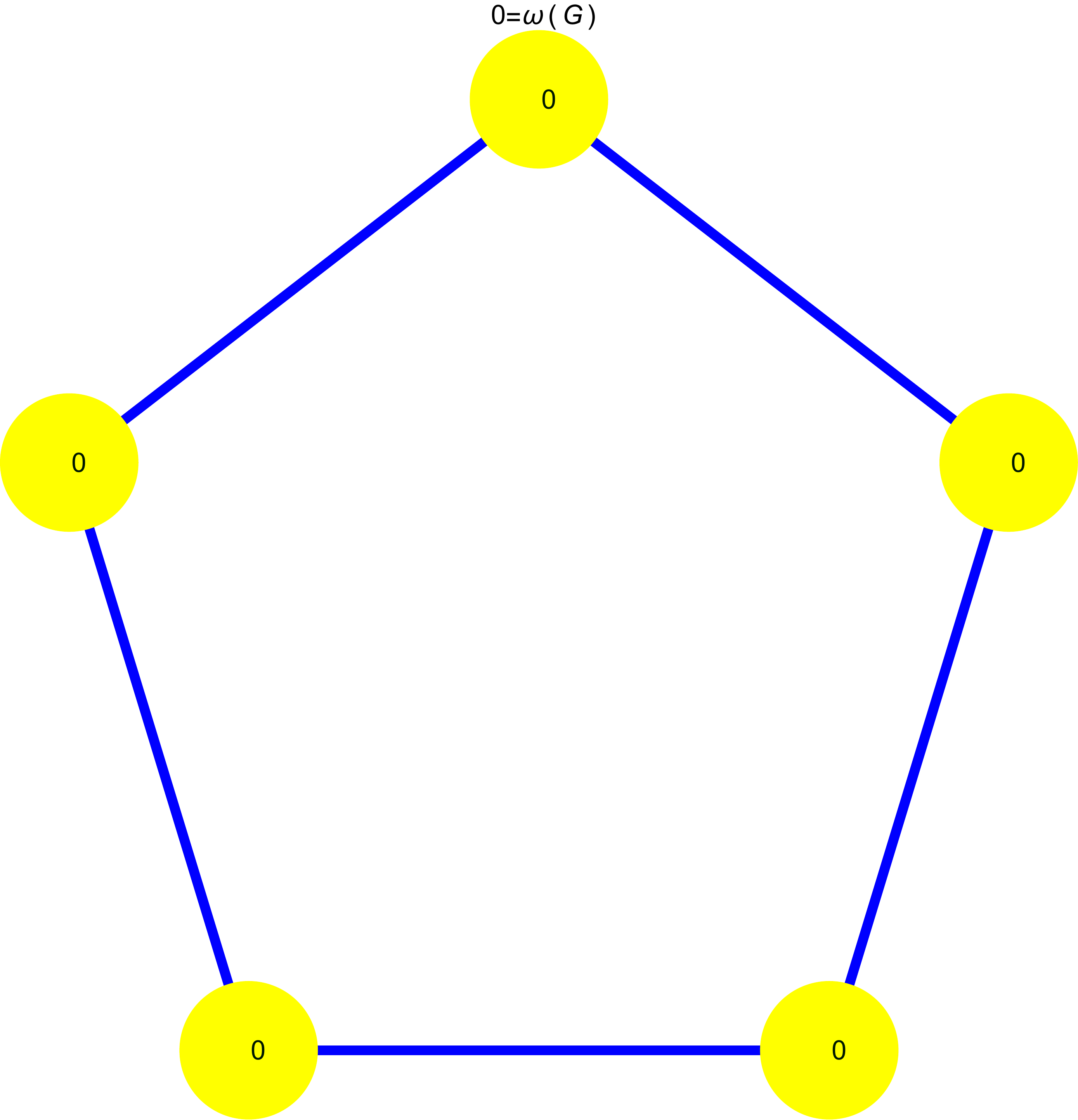}}
\scalebox{0.12}{\includegraphics{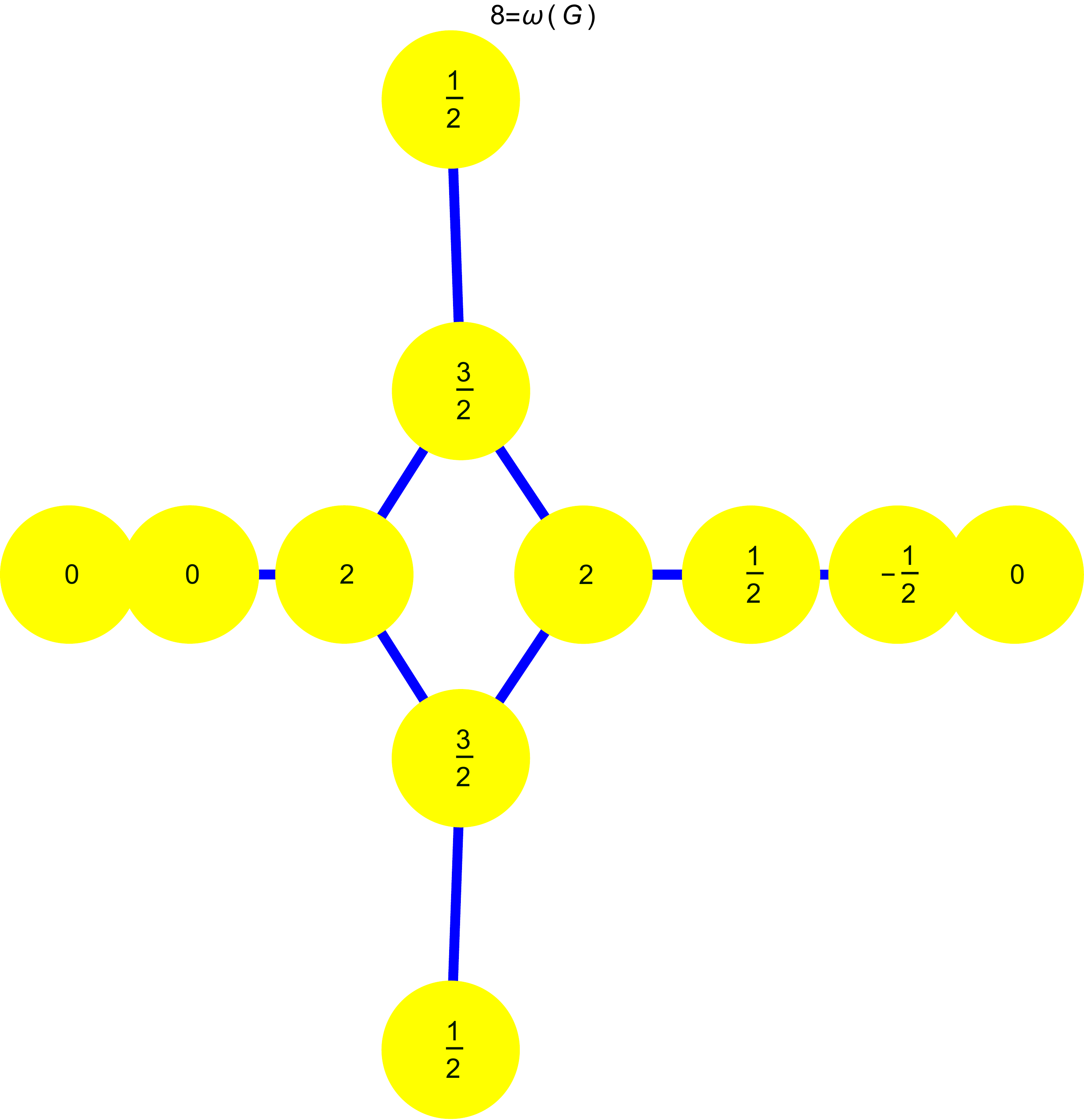}}
\scalebox{0.12}{\includegraphics{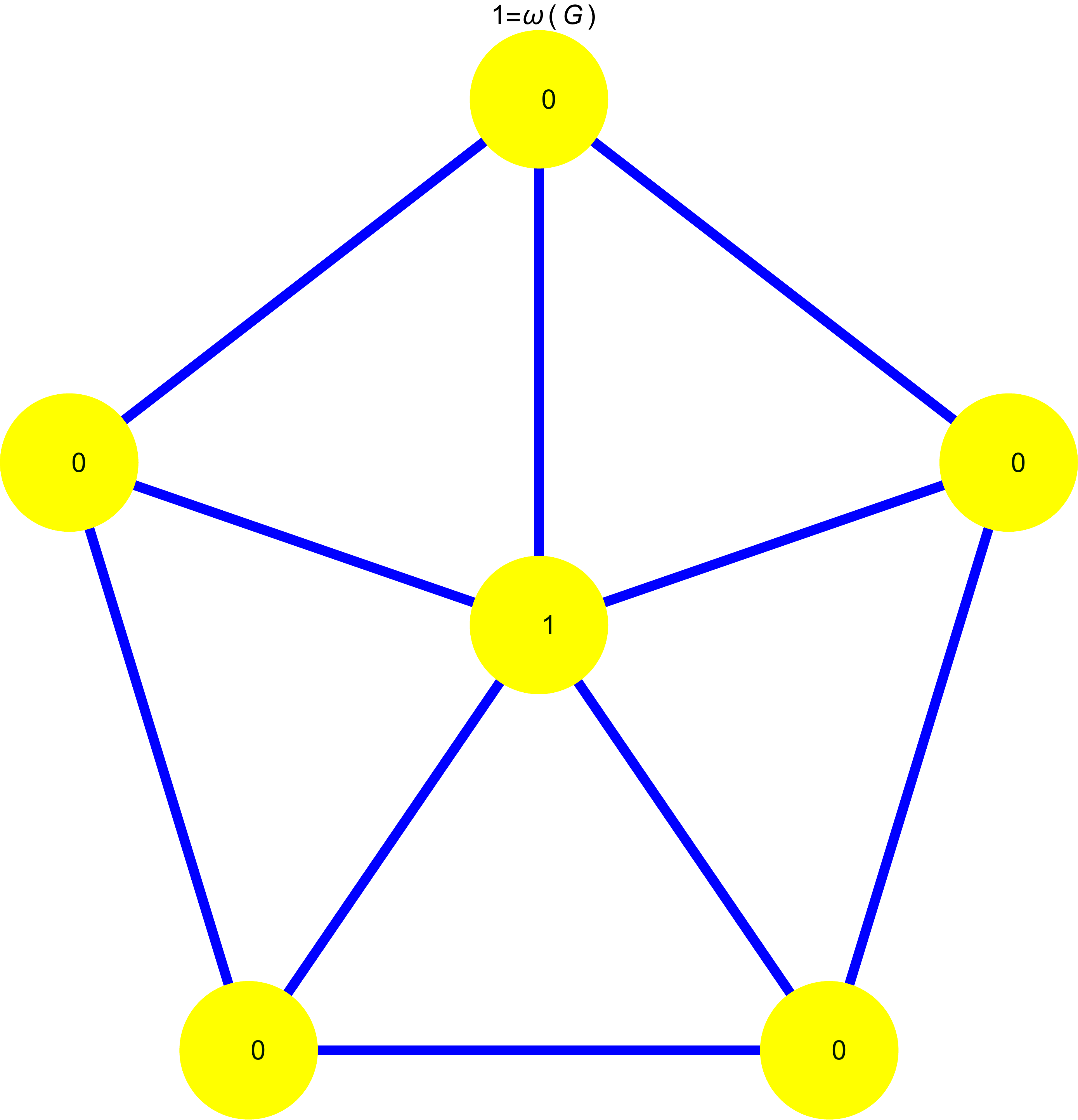}}
\scalebox{0.12}{\includegraphics{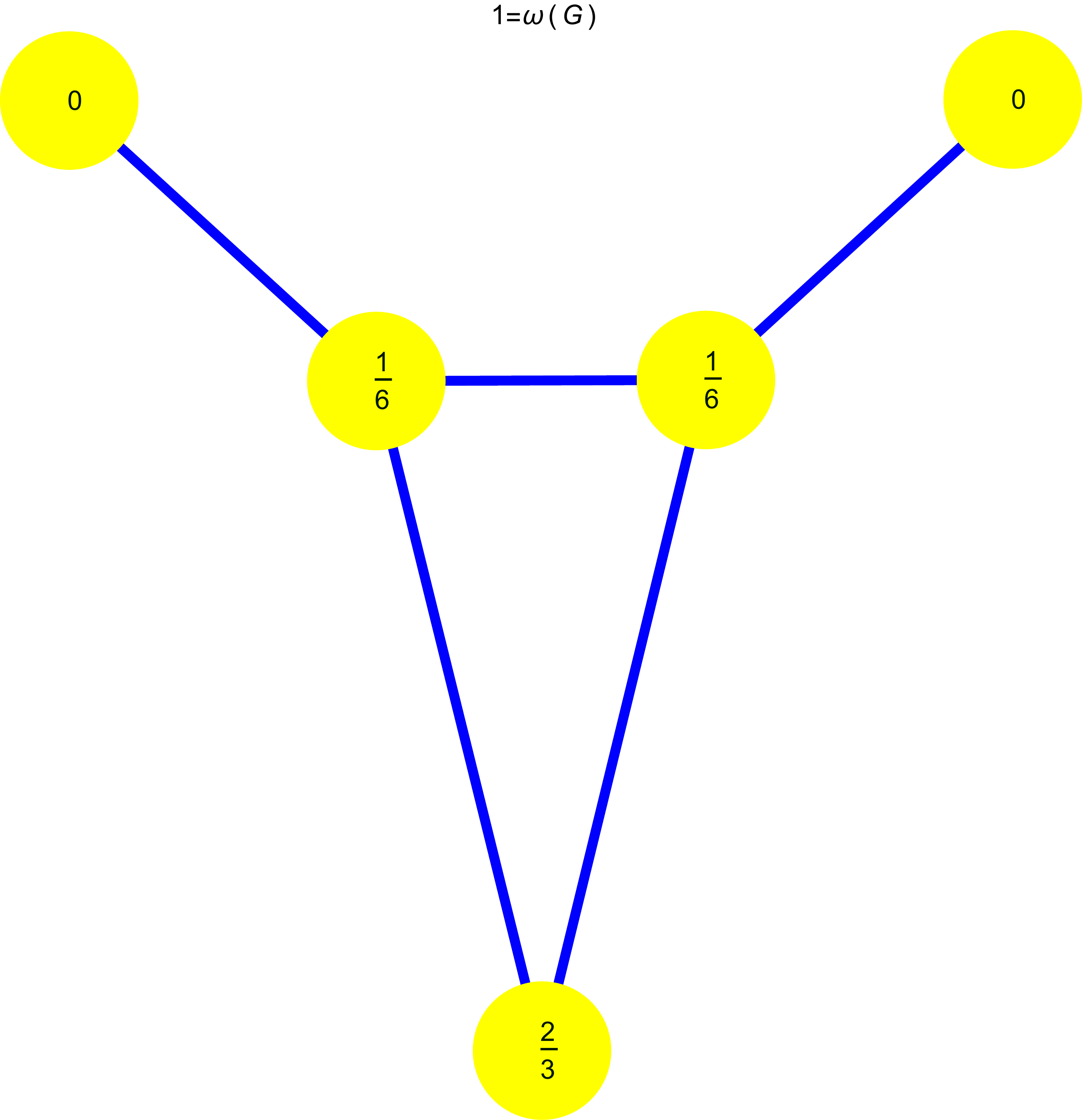}}
\scalebox{0.12}{\includegraphics{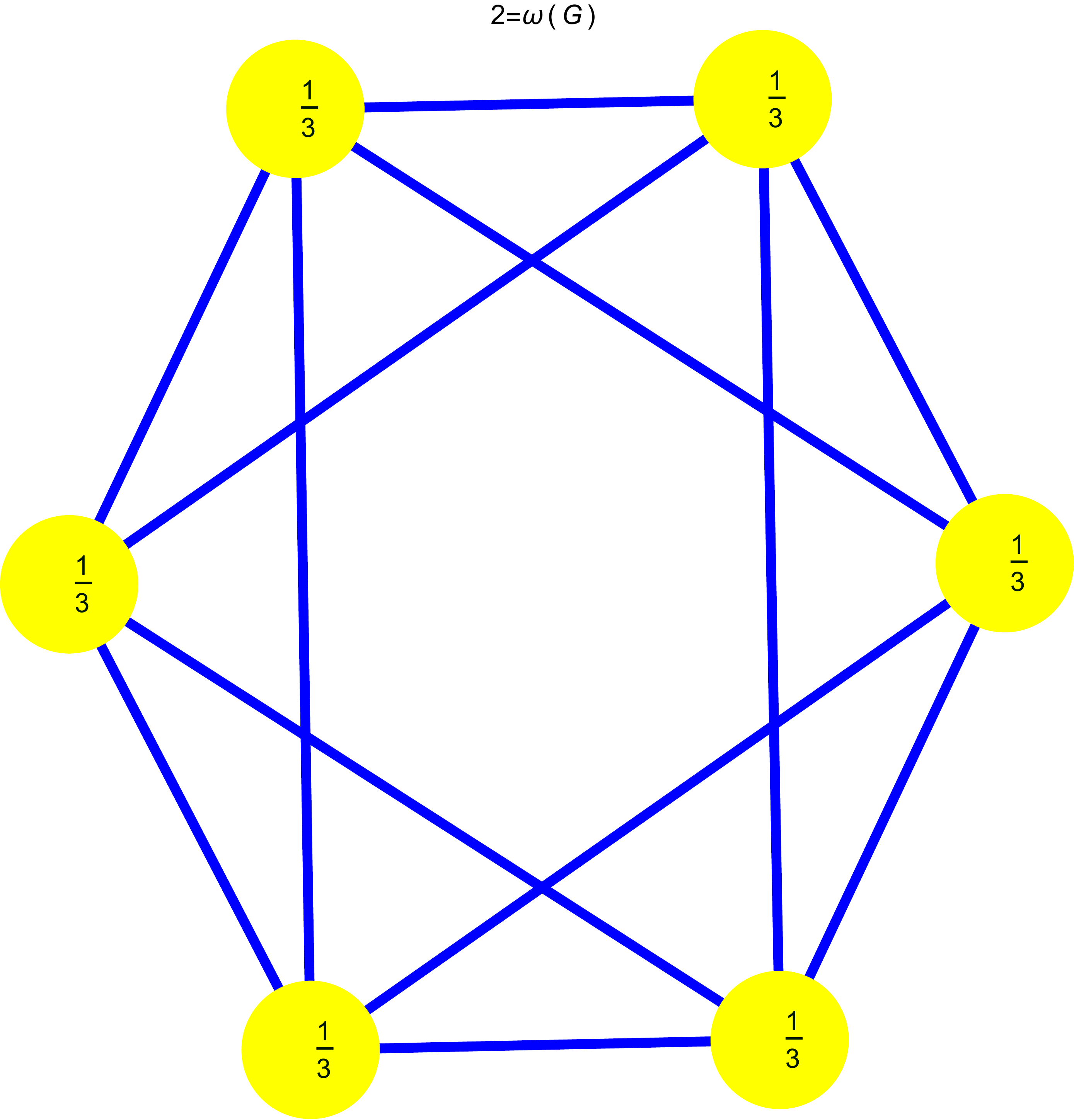}}
\scalebox{0.12}{\includegraphics{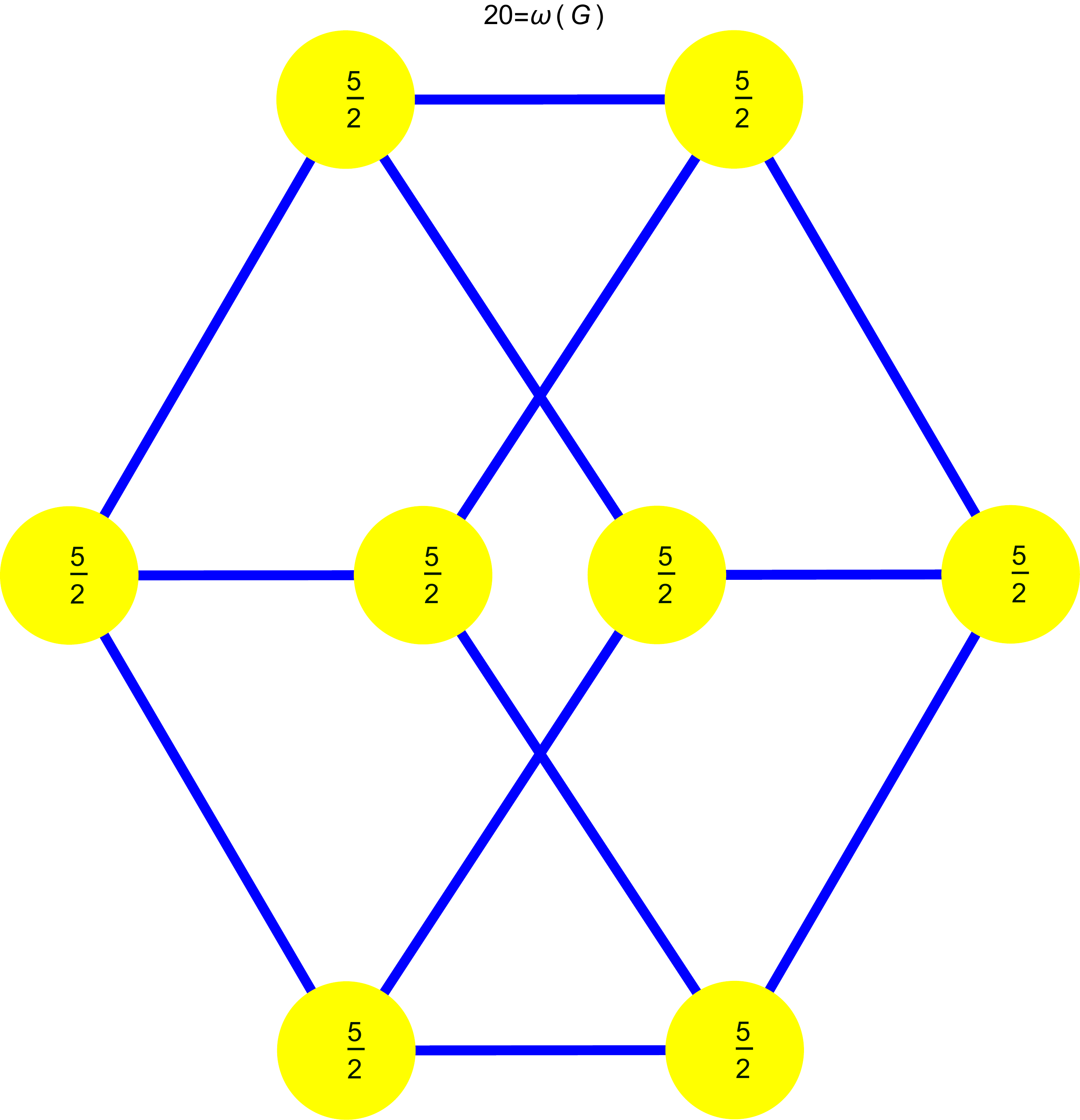}}
\scalebox{0.12}{\includegraphics{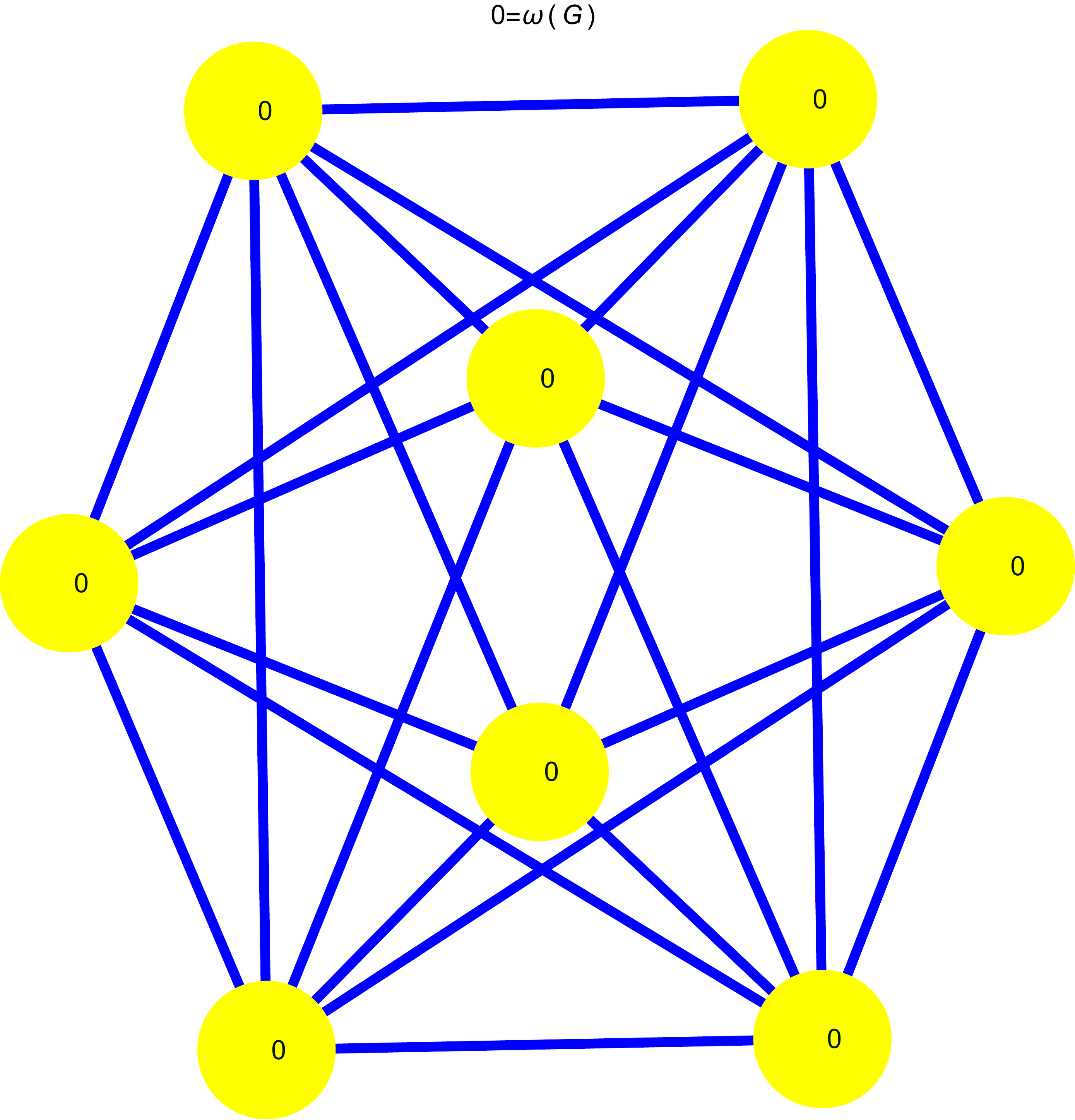}}
\scalebox{0.12}{\includegraphics{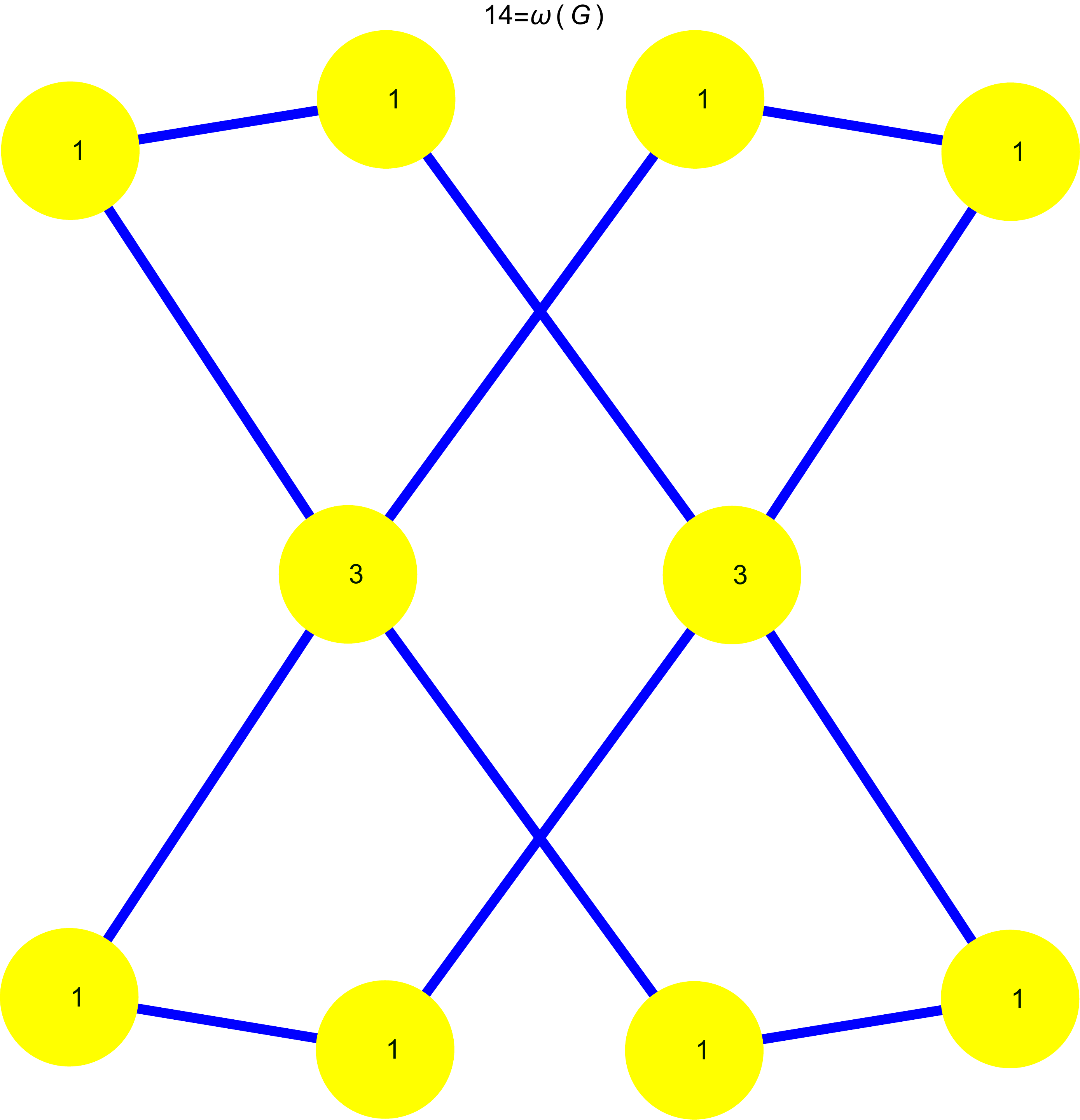}}
\scalebox{0.12}{\includegraphics{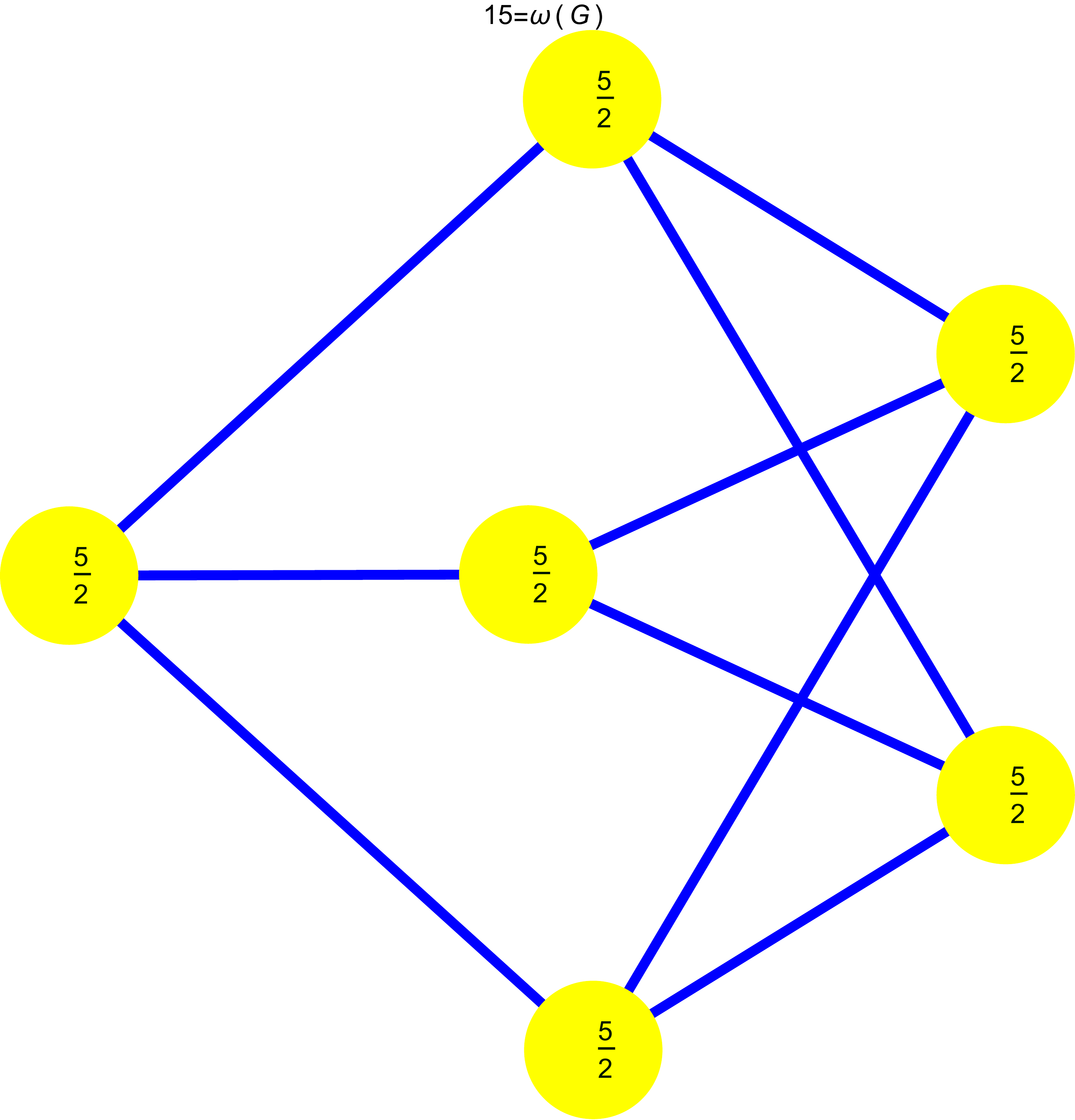}}
\scalebox{0.12}{\includegraphics{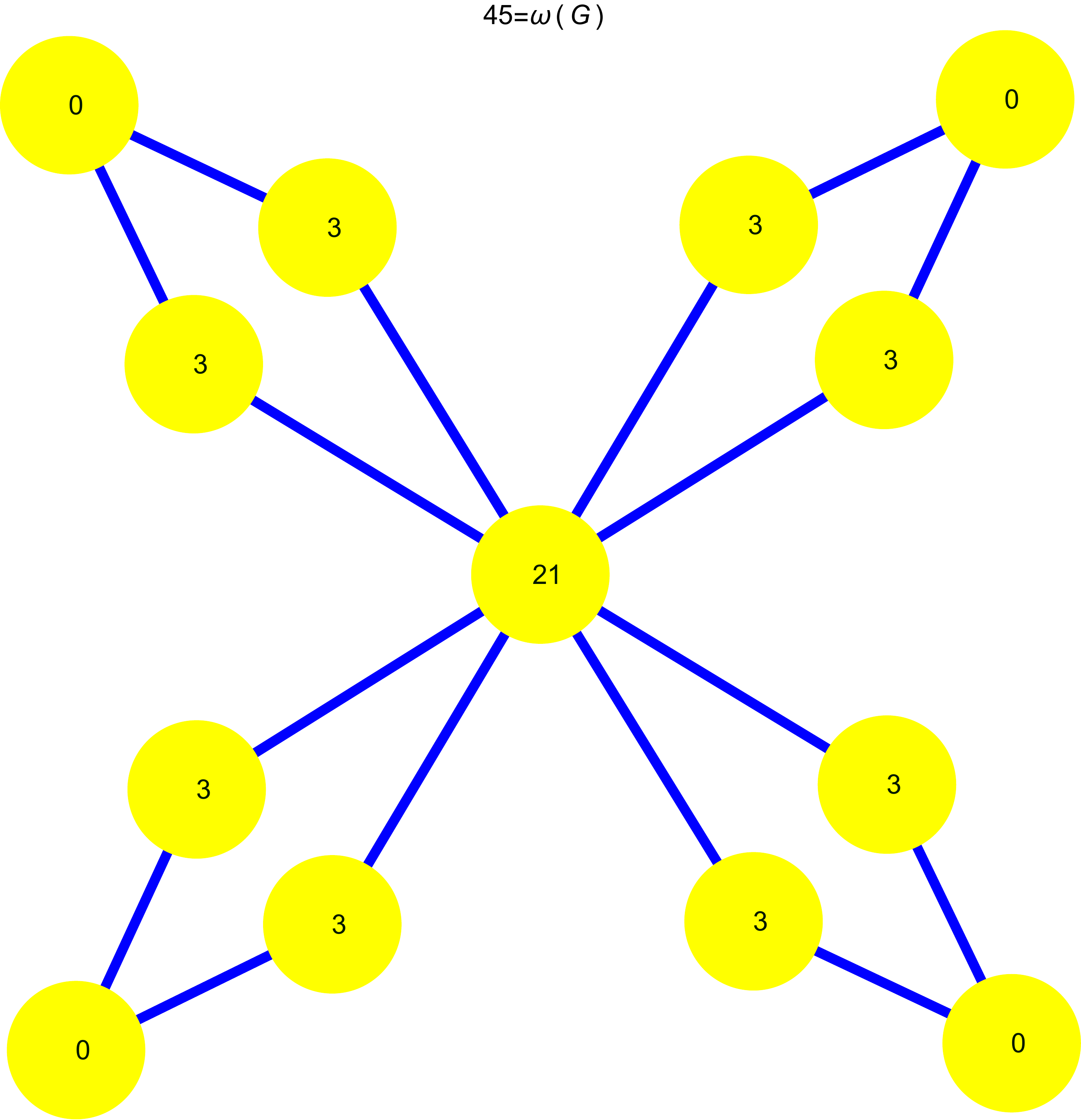}}
\scalebox{0.12}{\includegraphics{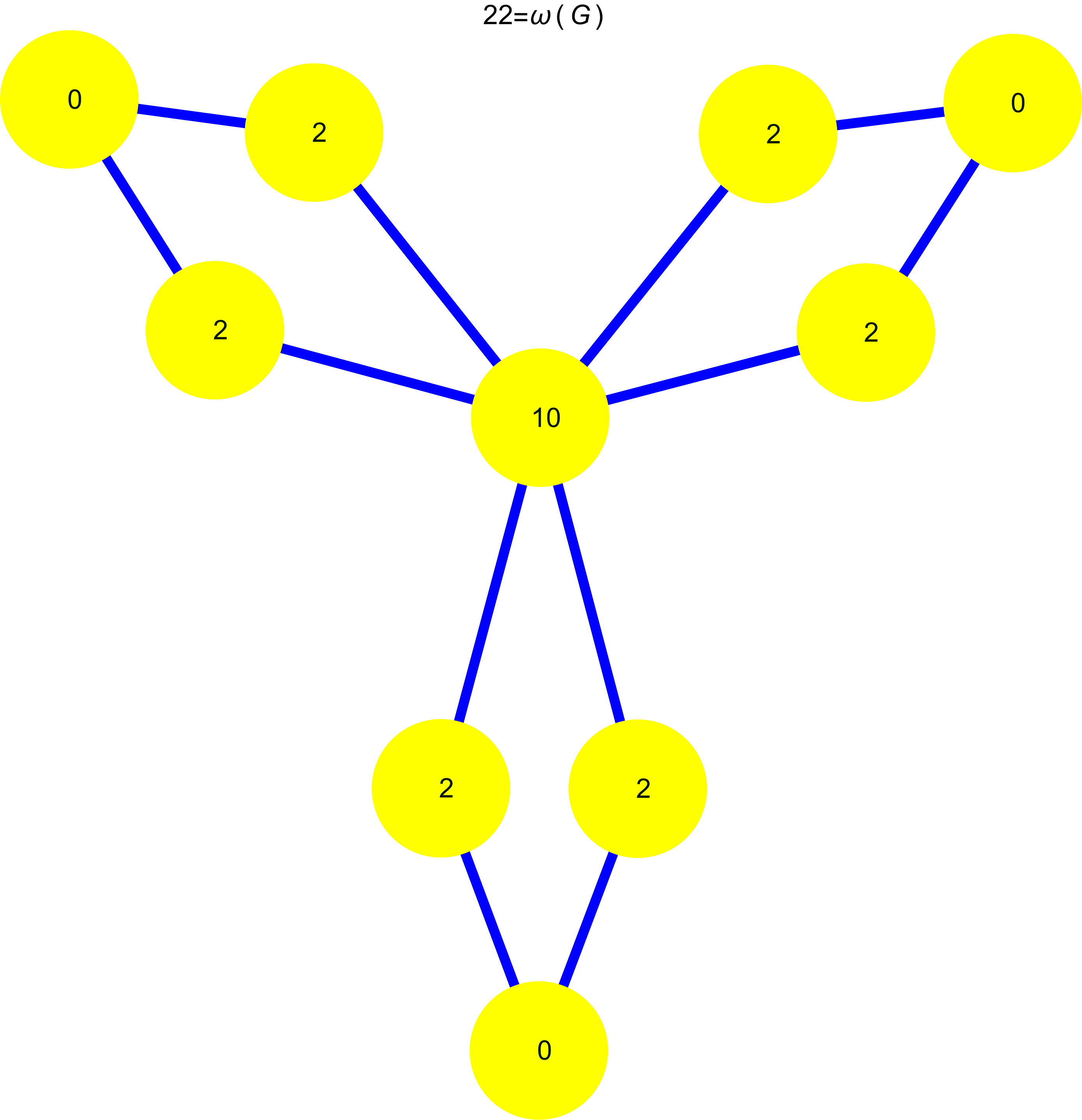}}
\scalebox{0.12}{\includegraphics{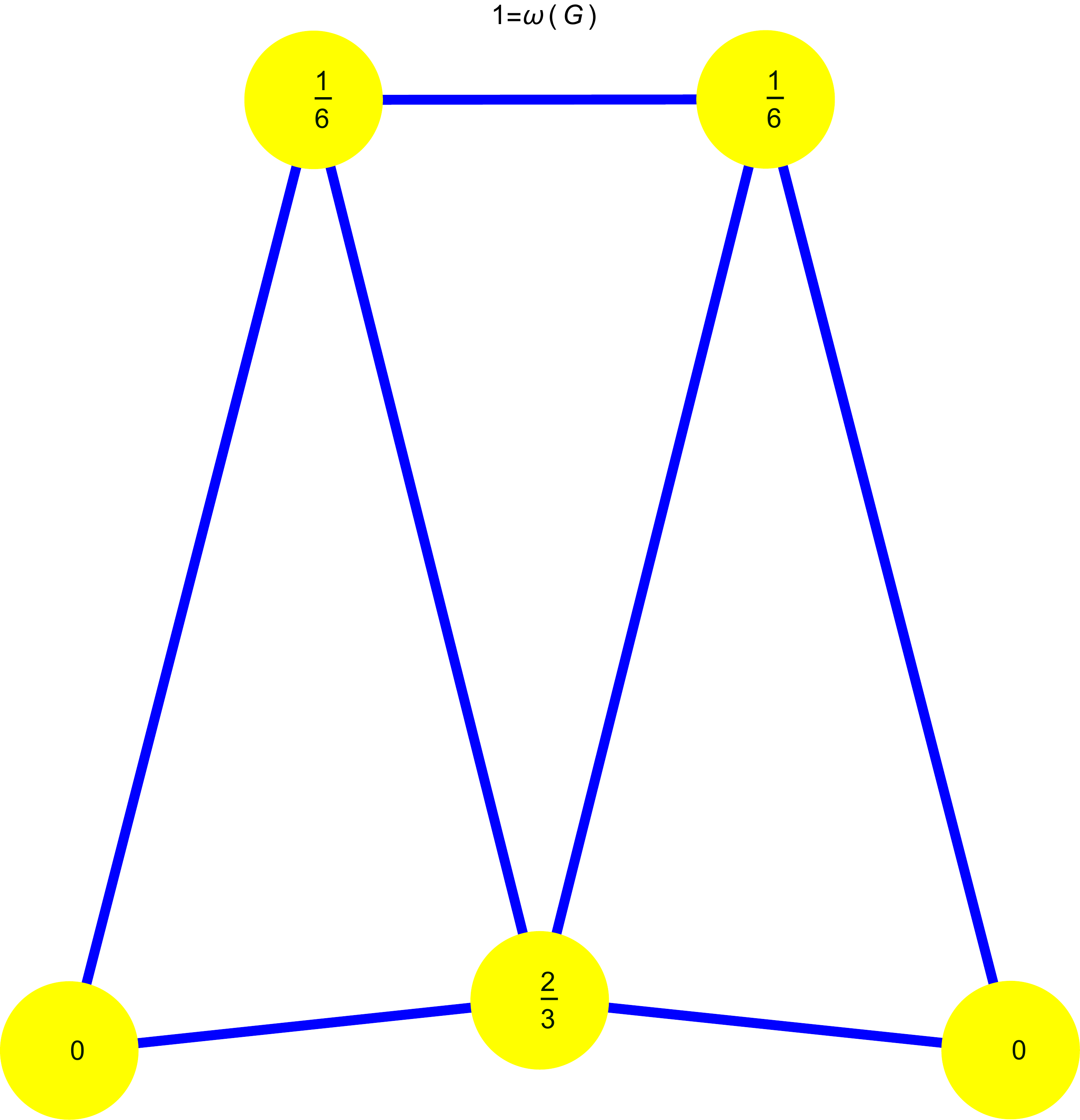}}
\scalebox{0.12}{\includegraphics{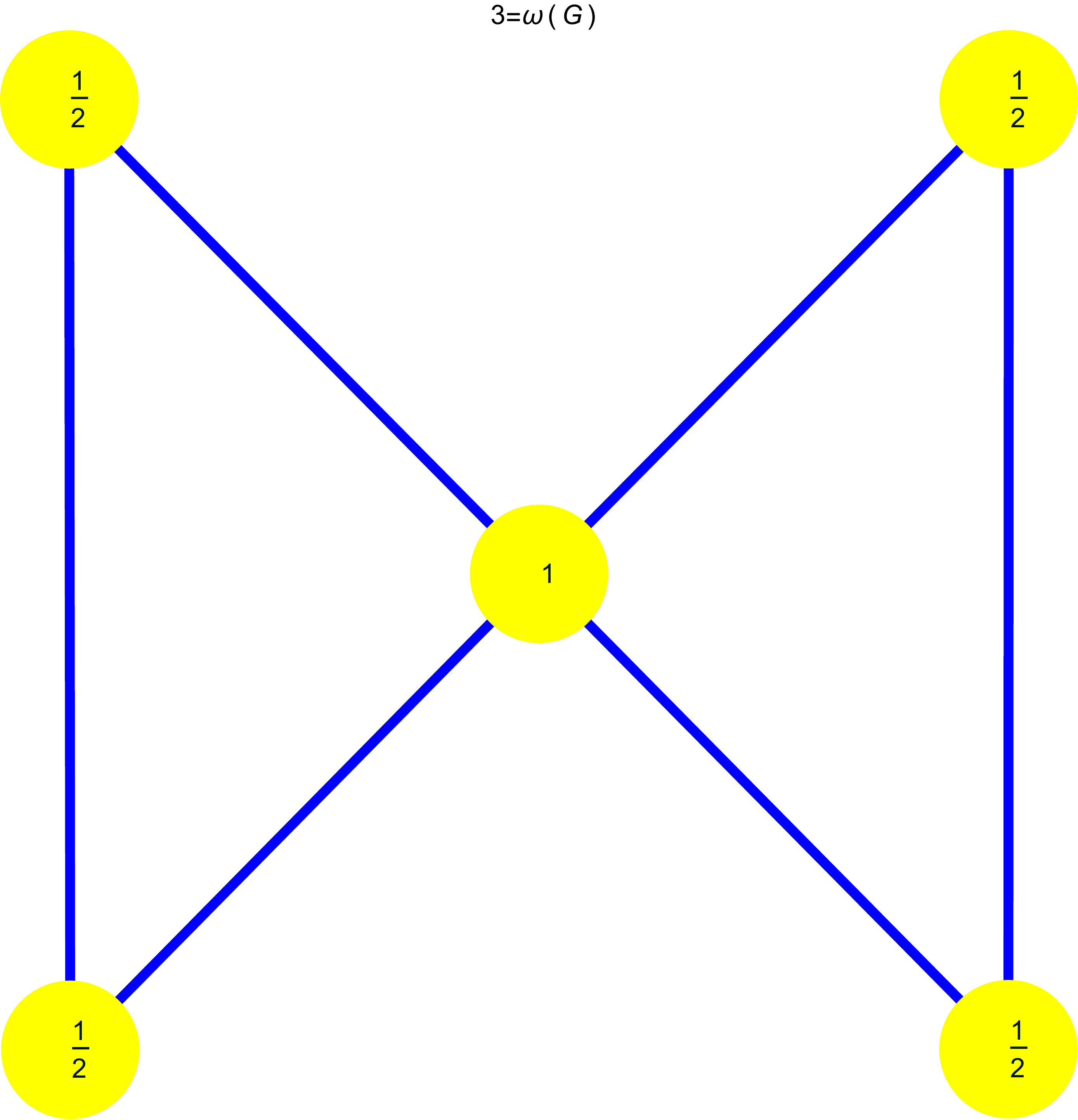}}
\caption{
Examples of graphs with quadratic Wu curvatures. 
}
\end{figure}

\begin{figure}[!htpb]
\scalebox{0.12}{\includegraphics{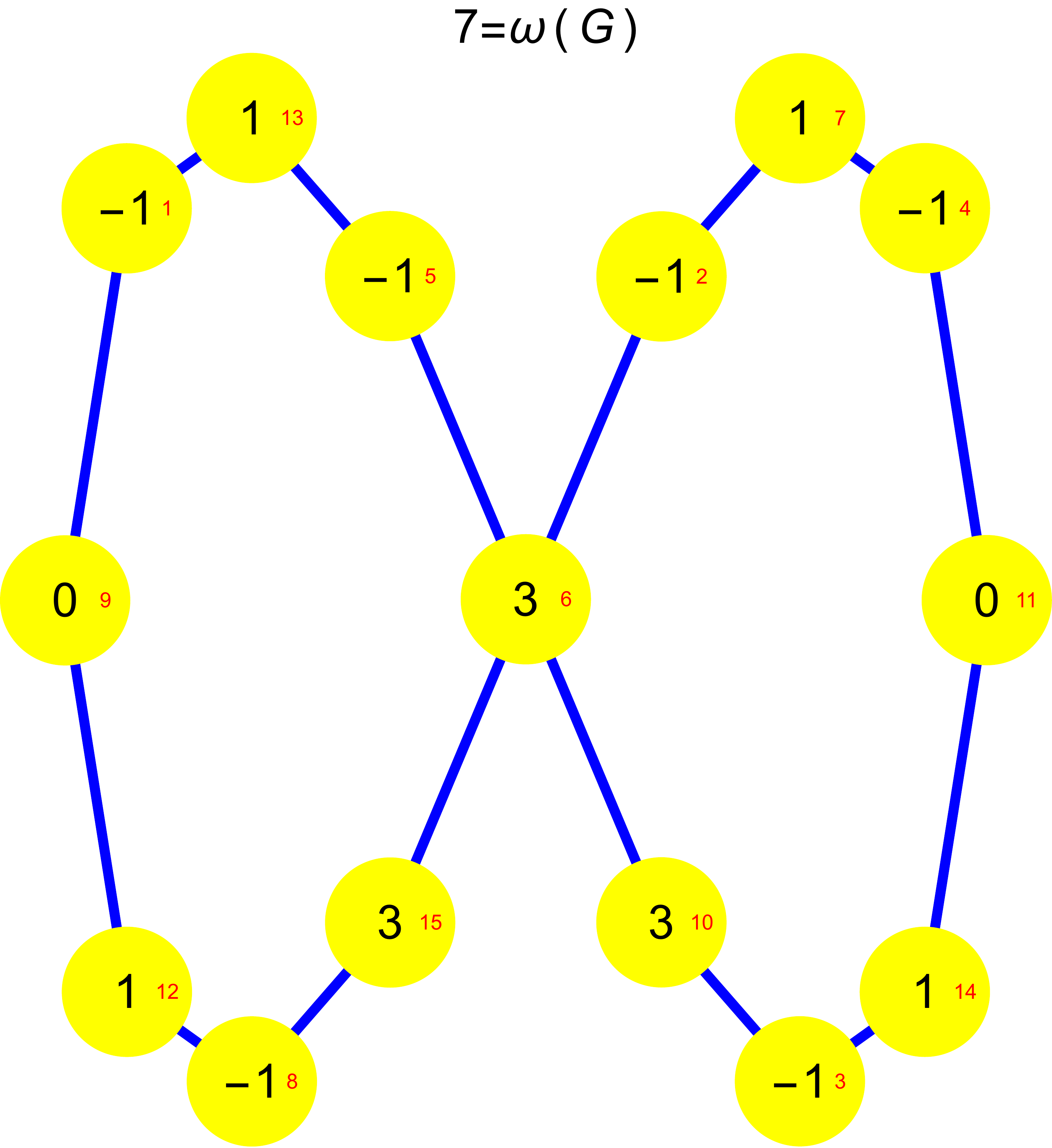}}
\scalebox{0.12}{\includegraphics{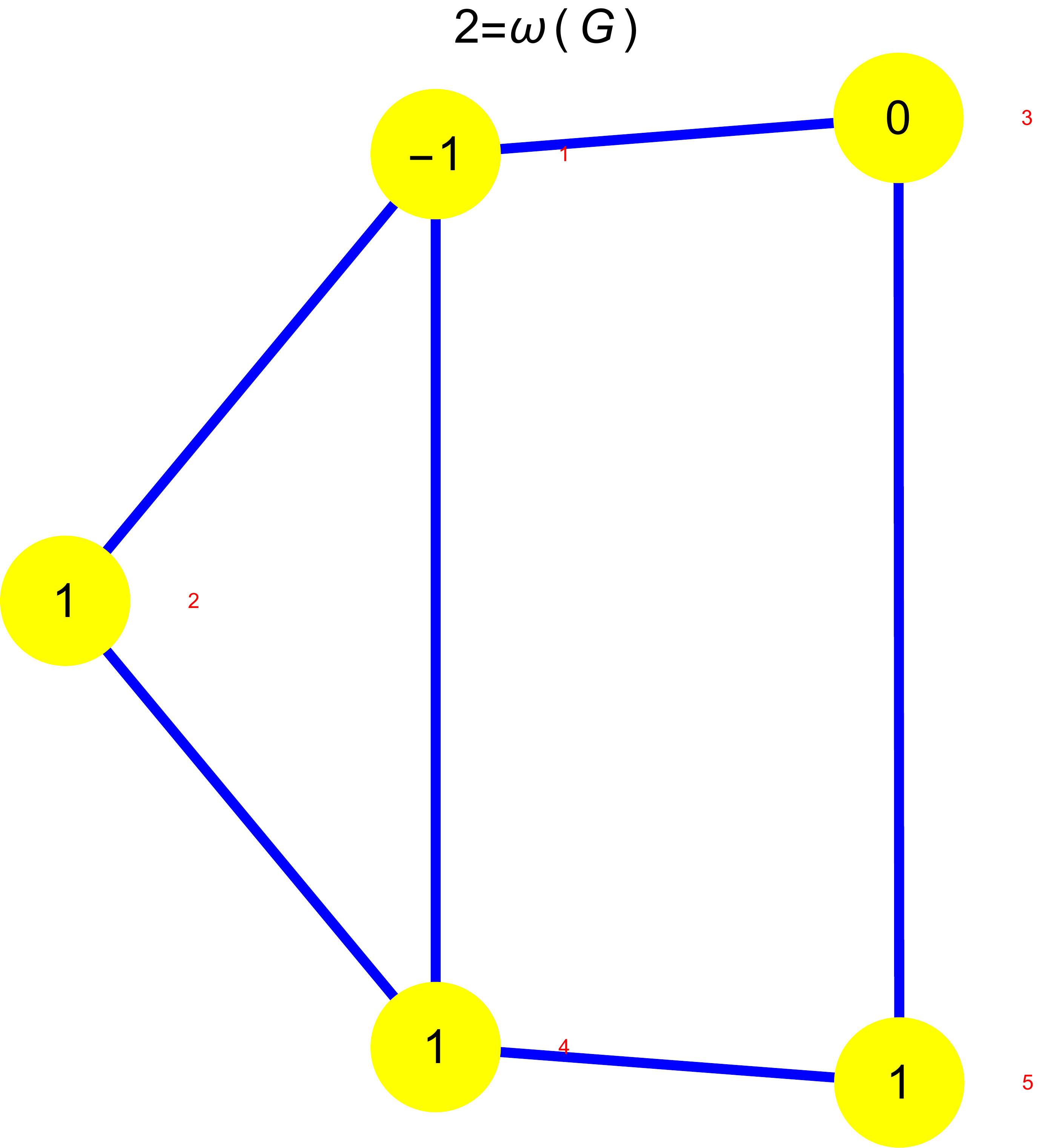}}
\scalebox{0.12}{\includegraphics{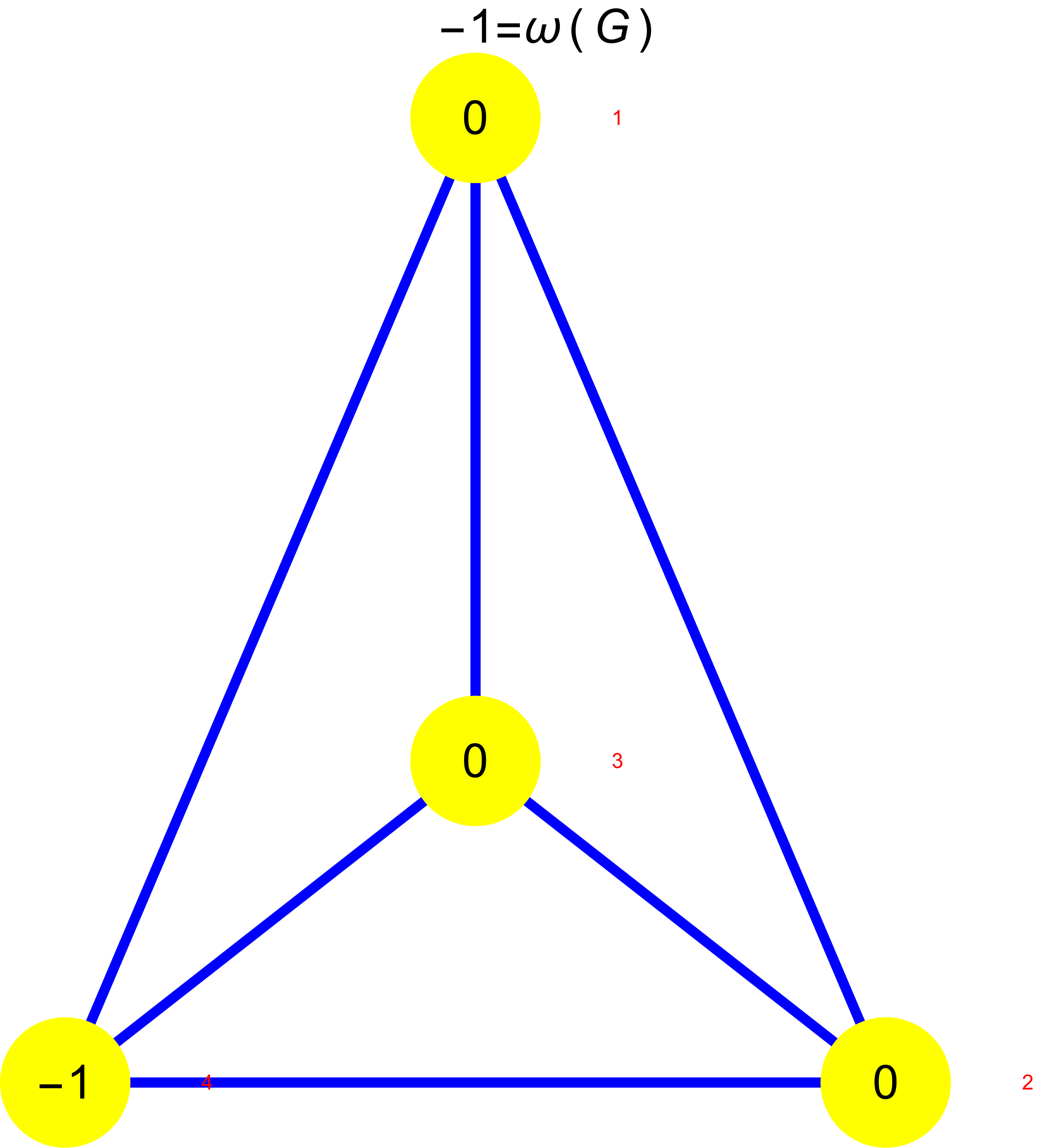}}
\scalebox{0.12}{\includegraphics{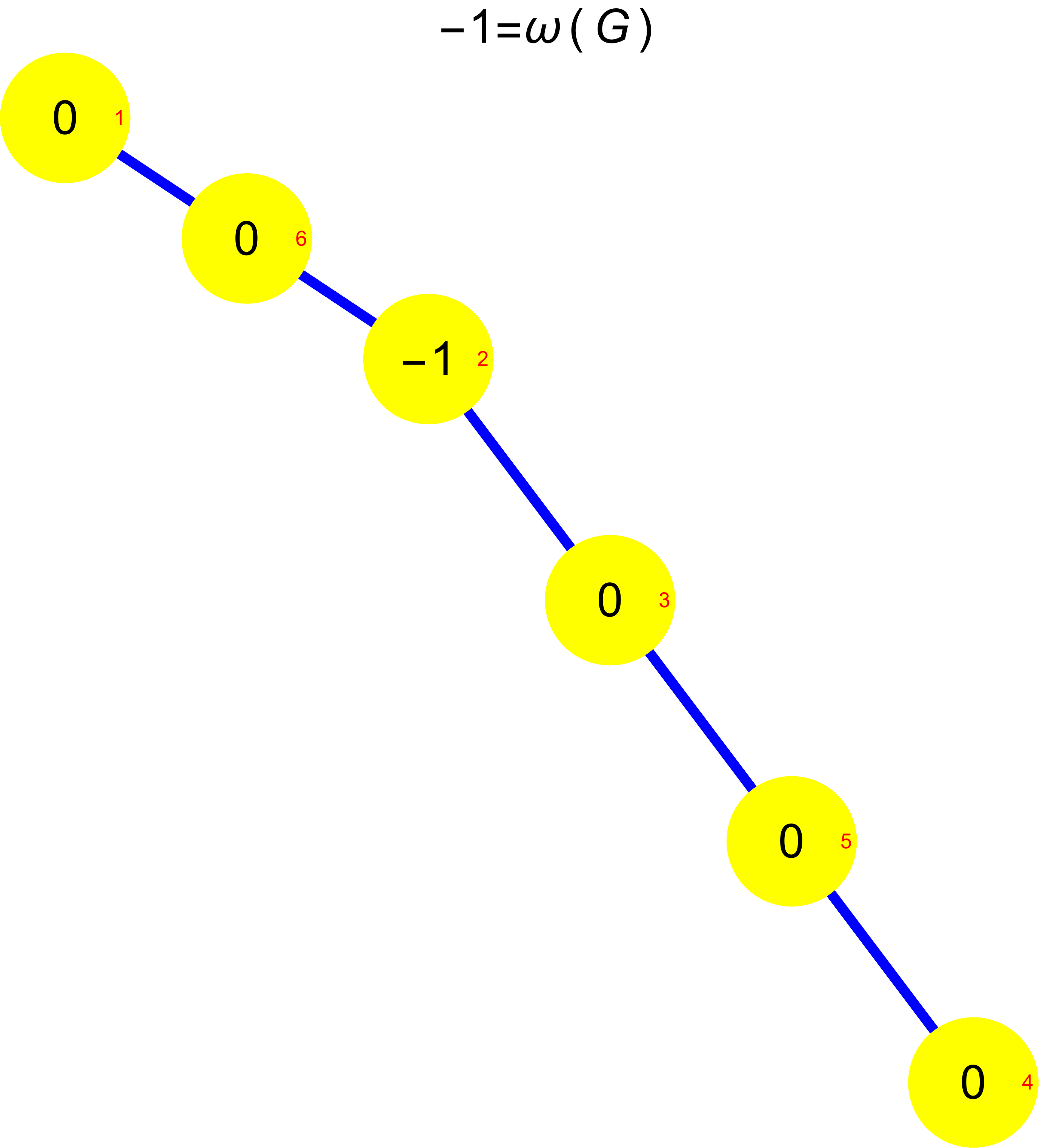}}
\scalebox{0.12}{\includegraphics{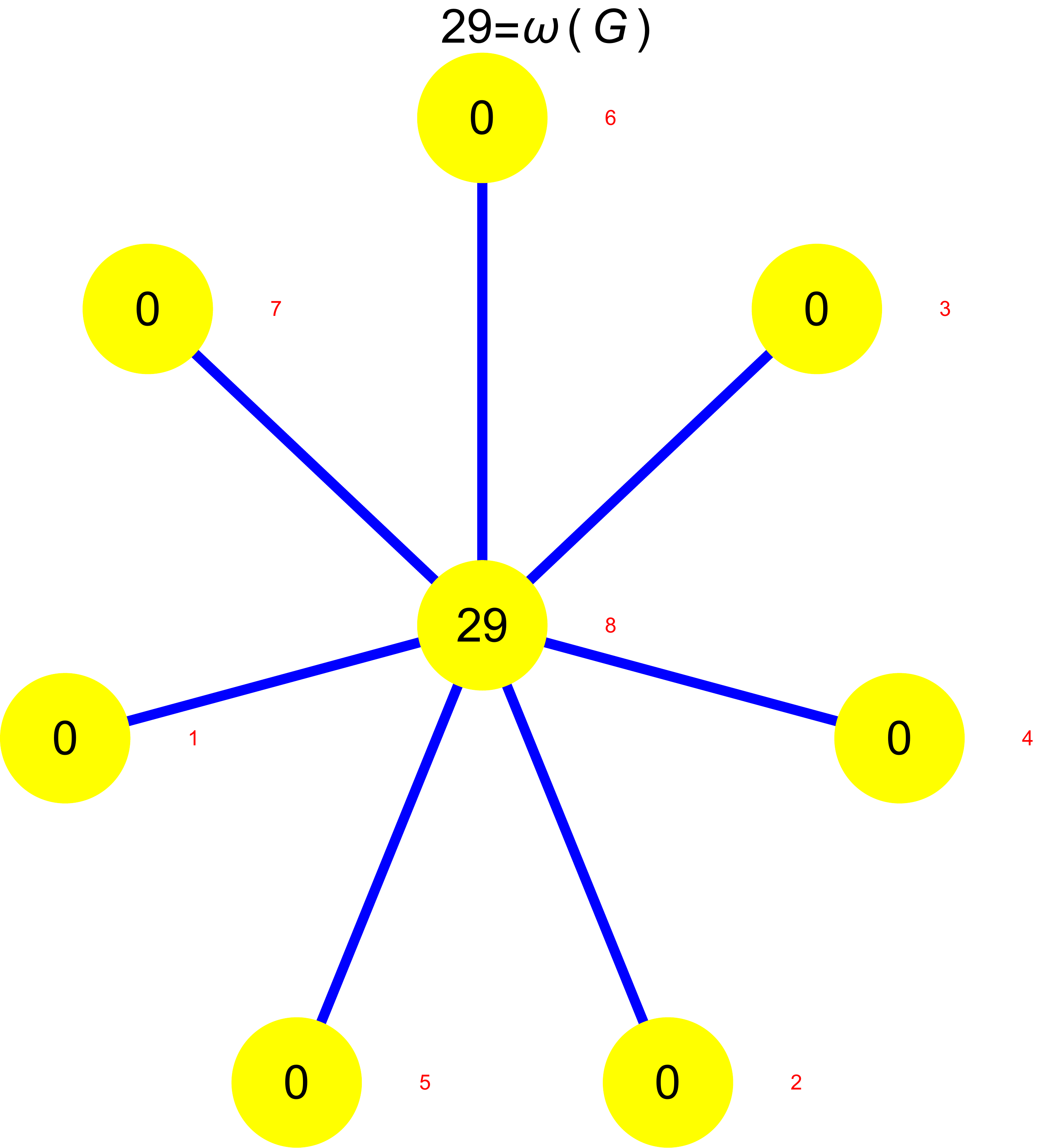}}
\scalebox{0.12}{\includegraphics{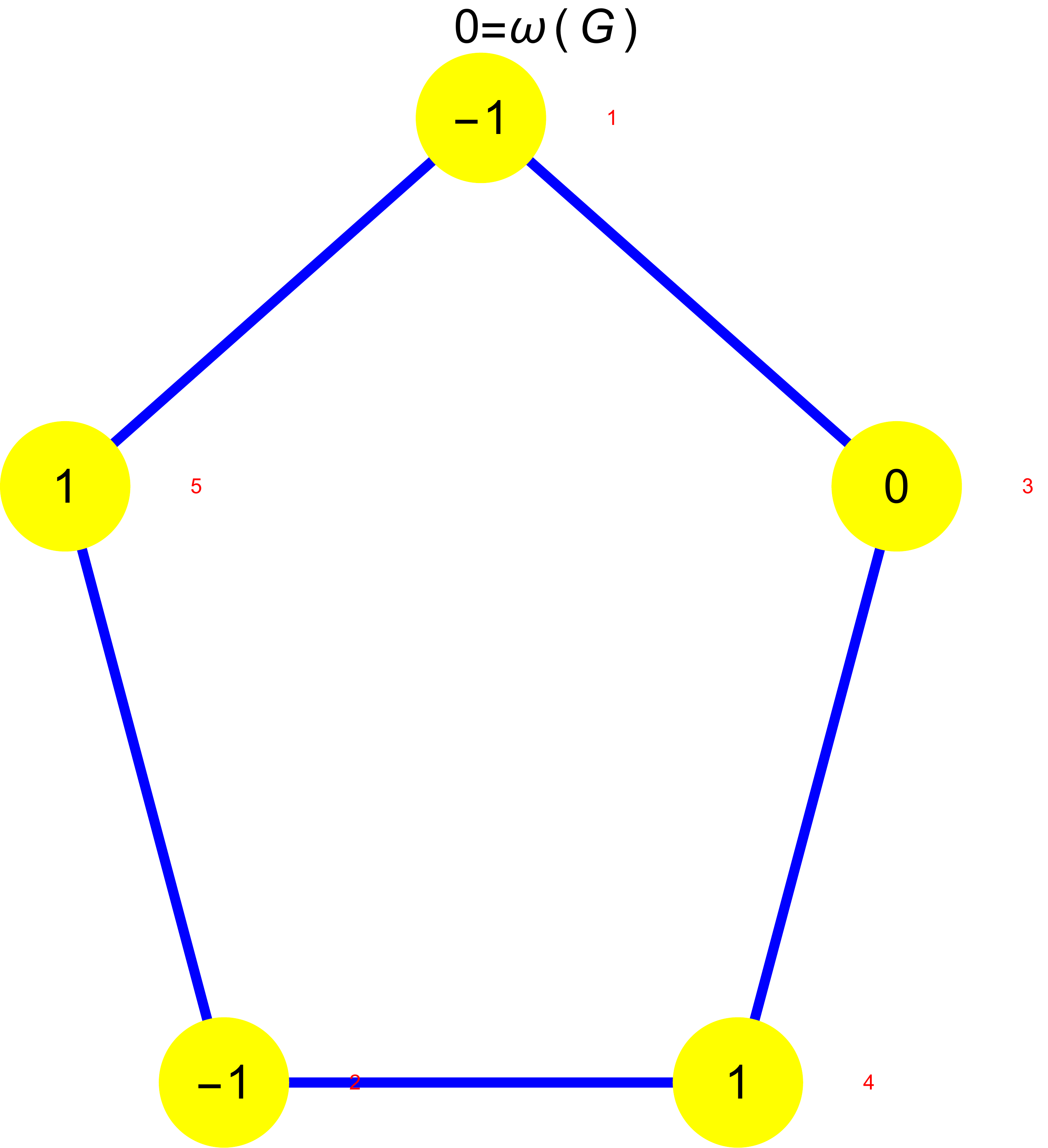}}
\scalebox{0.12}{\includegraphics{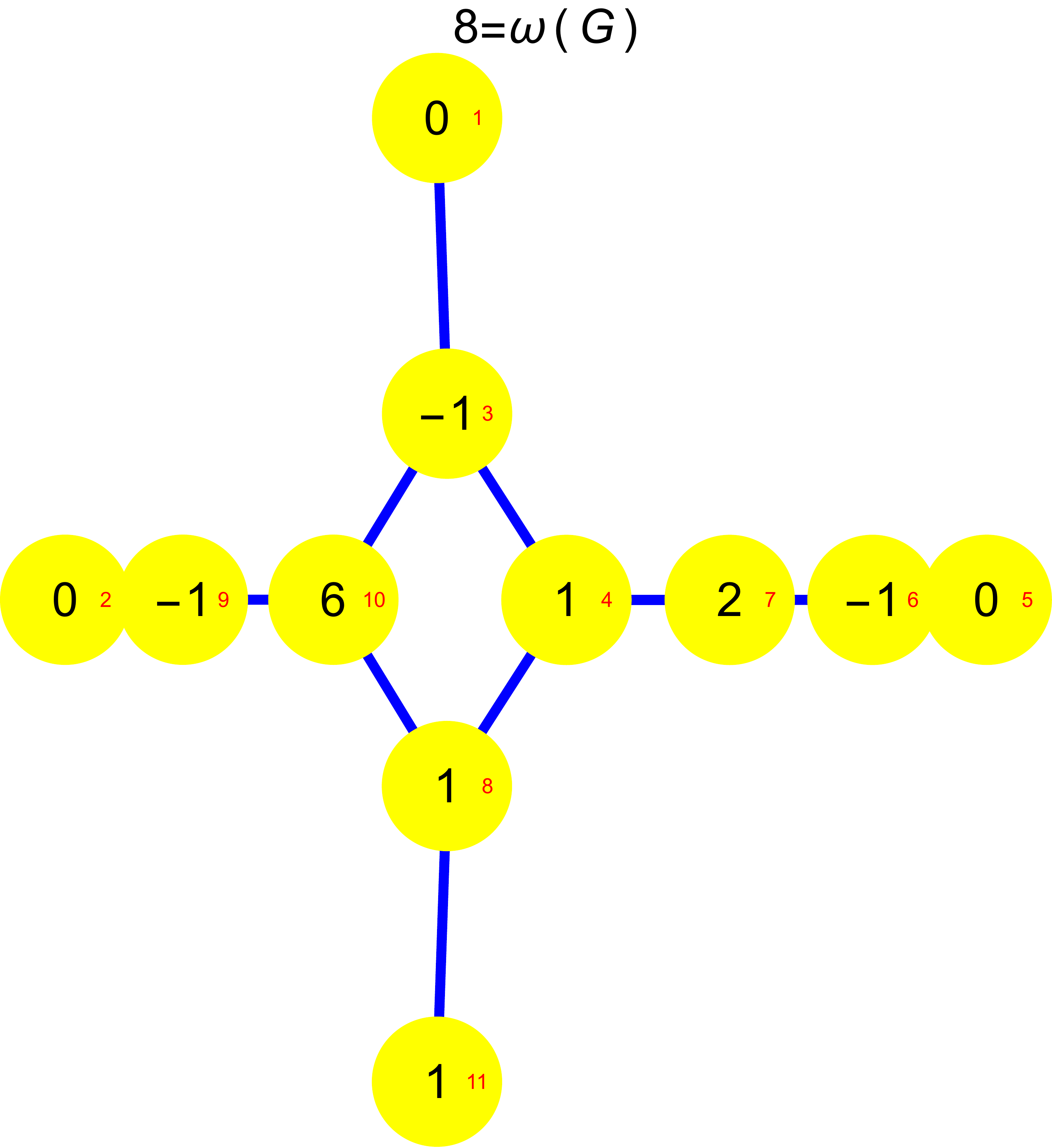}}
\scalebox{0.12}{\includegraphics{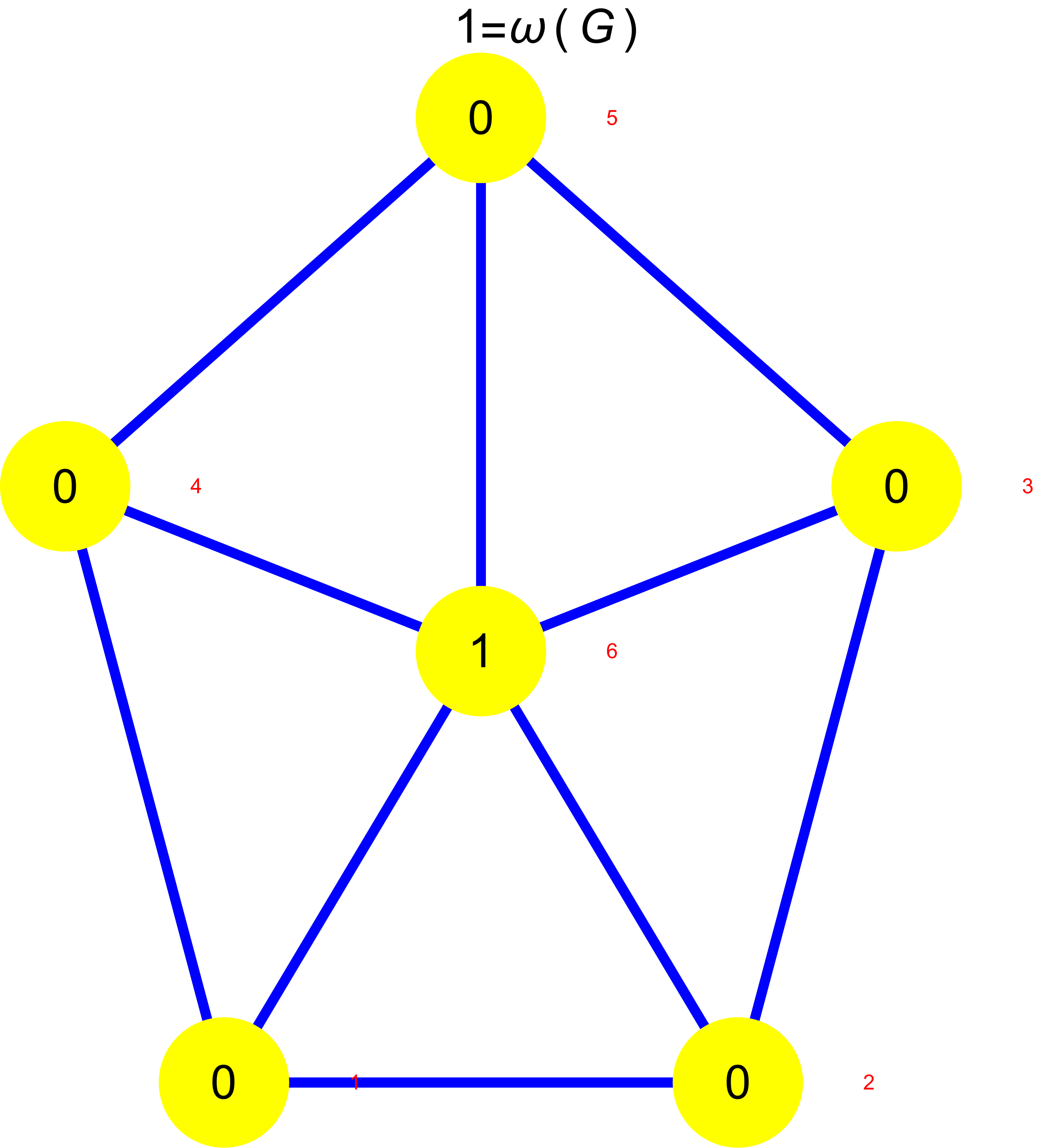}}
\scalebox{0.12}{\includegraphics{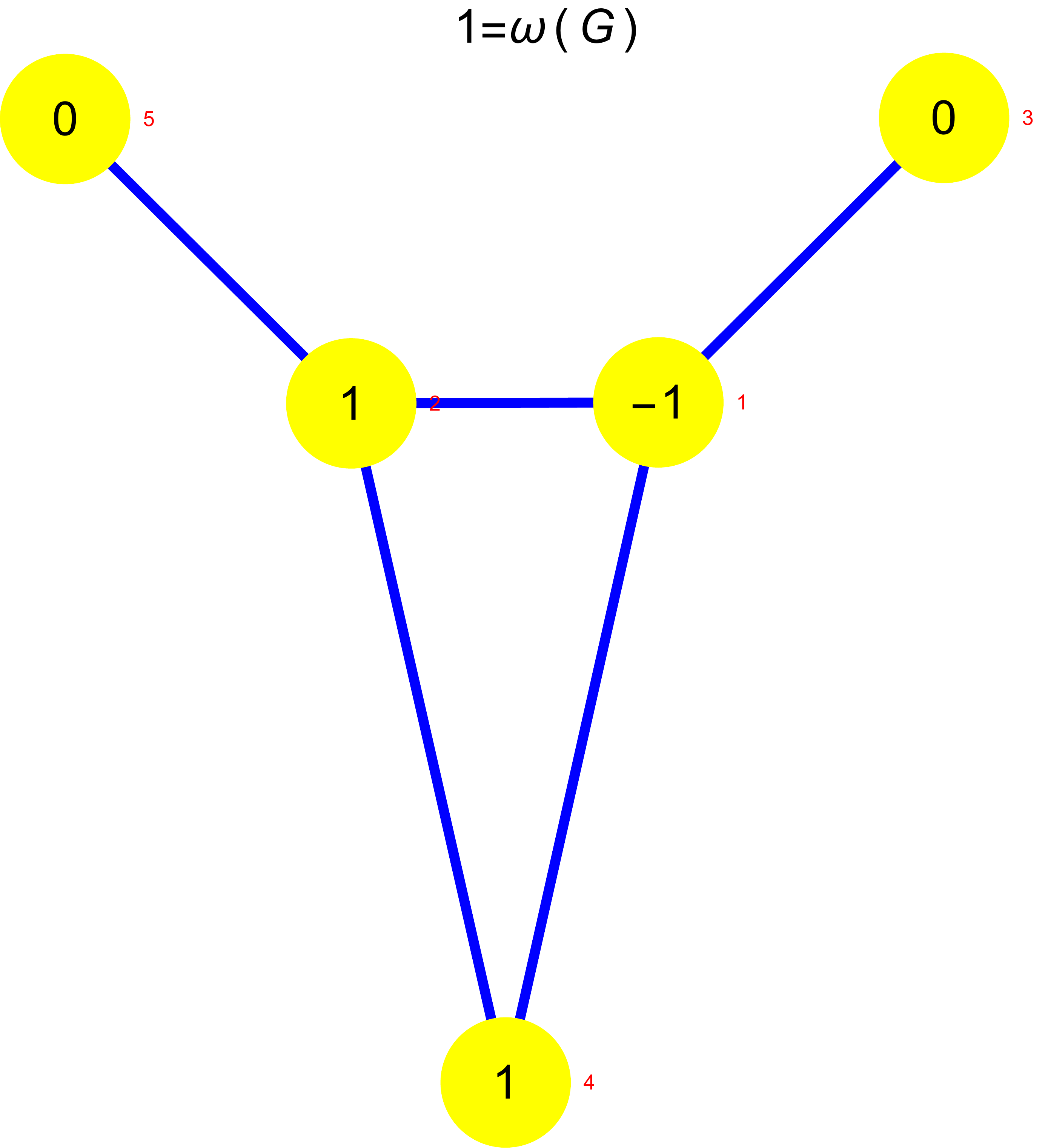}}
\scalebox{0.12}{\includegraphics{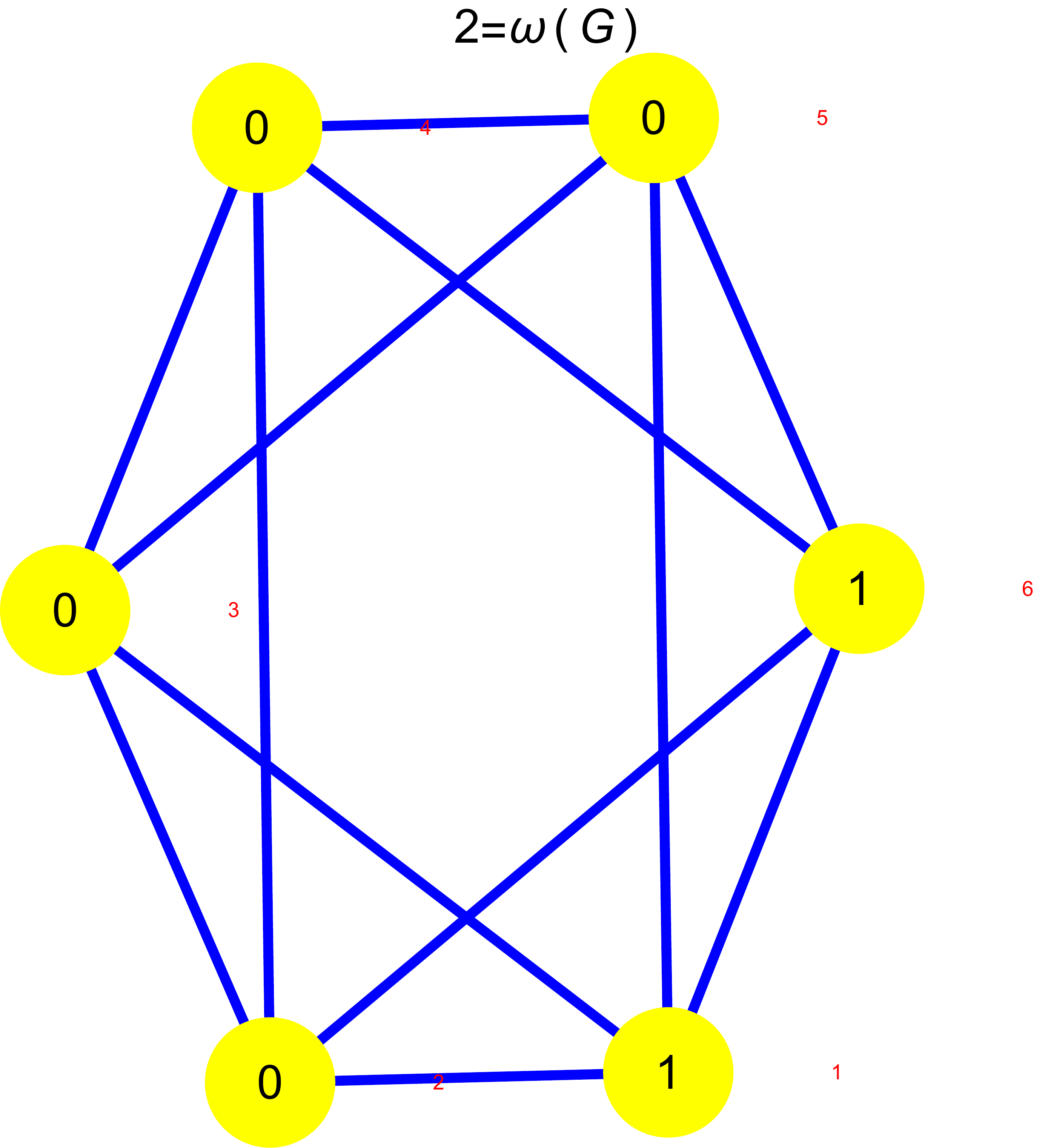}}
\scalebox{0.12}{\includegraphics{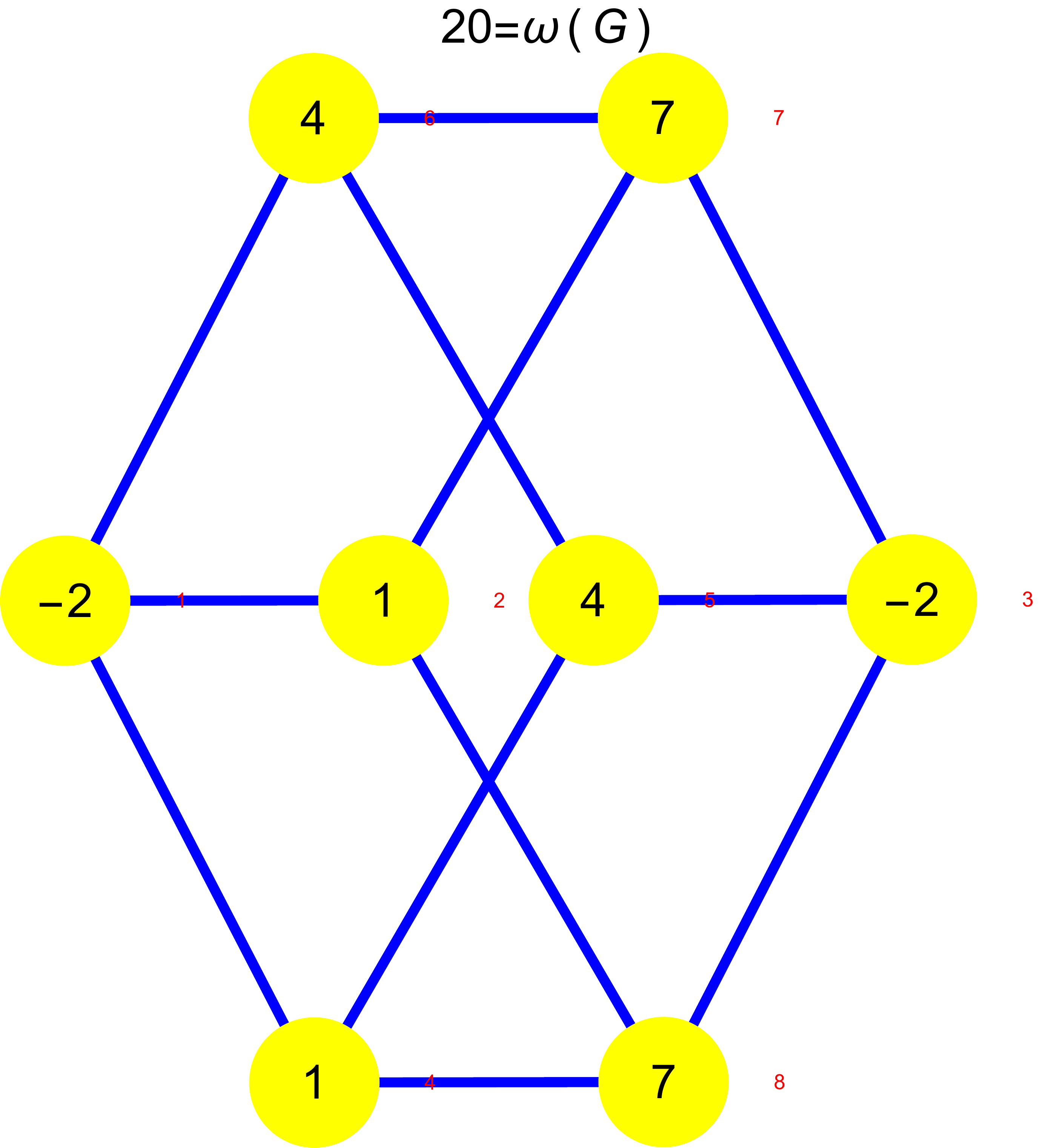}}
\scalebox{0.12}{\includegraphics{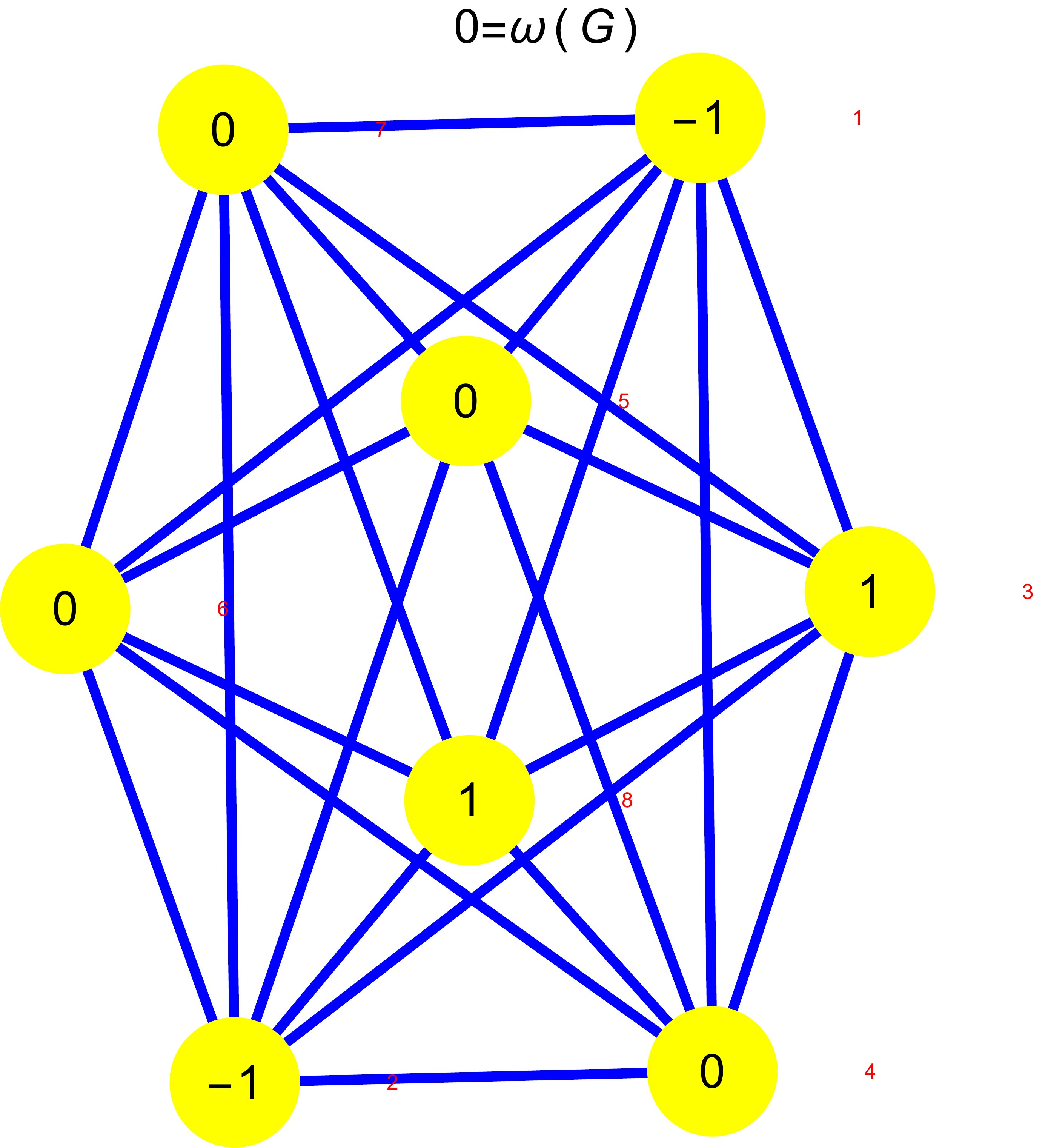}}
\scalebox{0.12}{\includegraphics{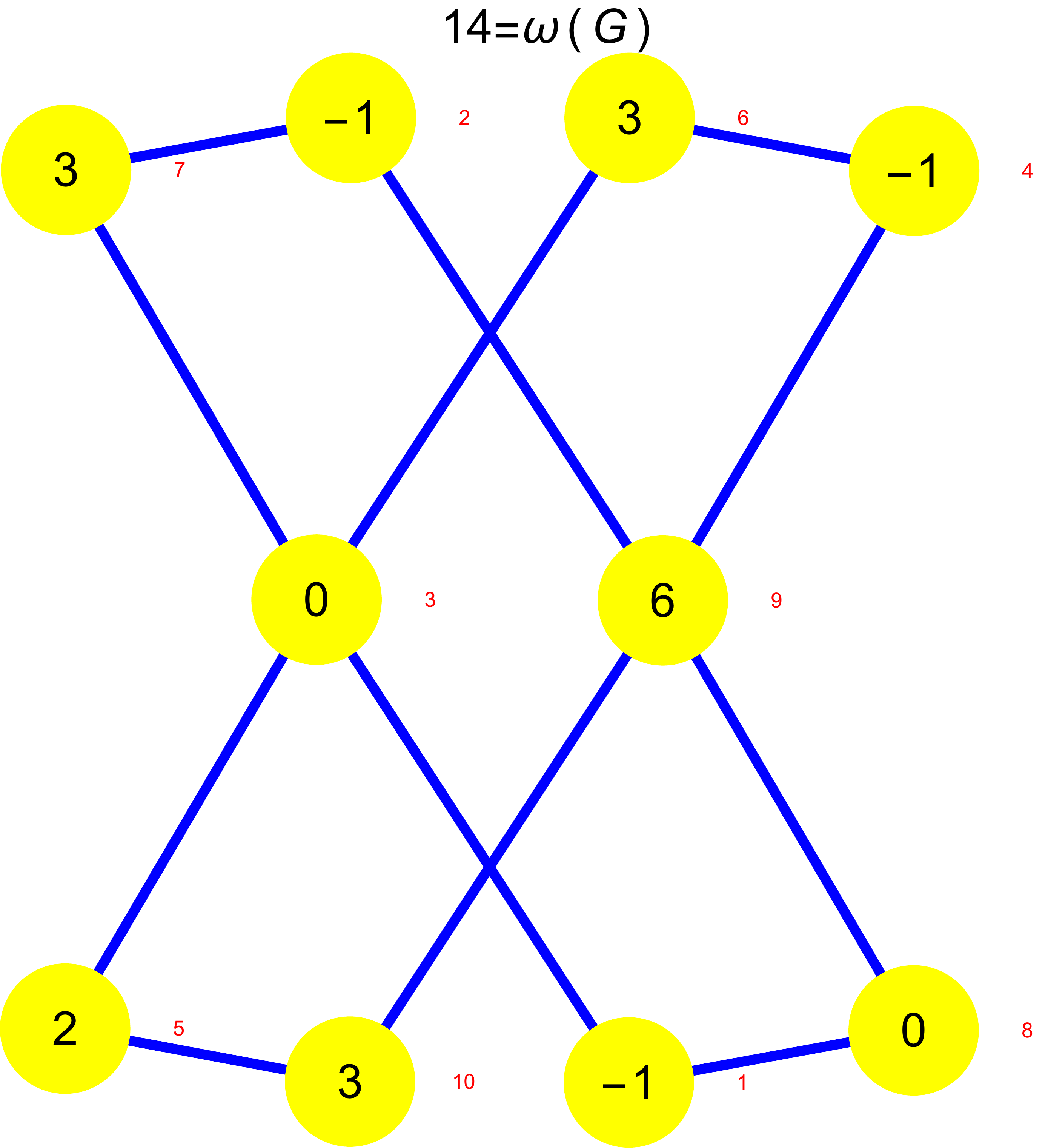}}
\scalebox{0.12}{\includegraphics{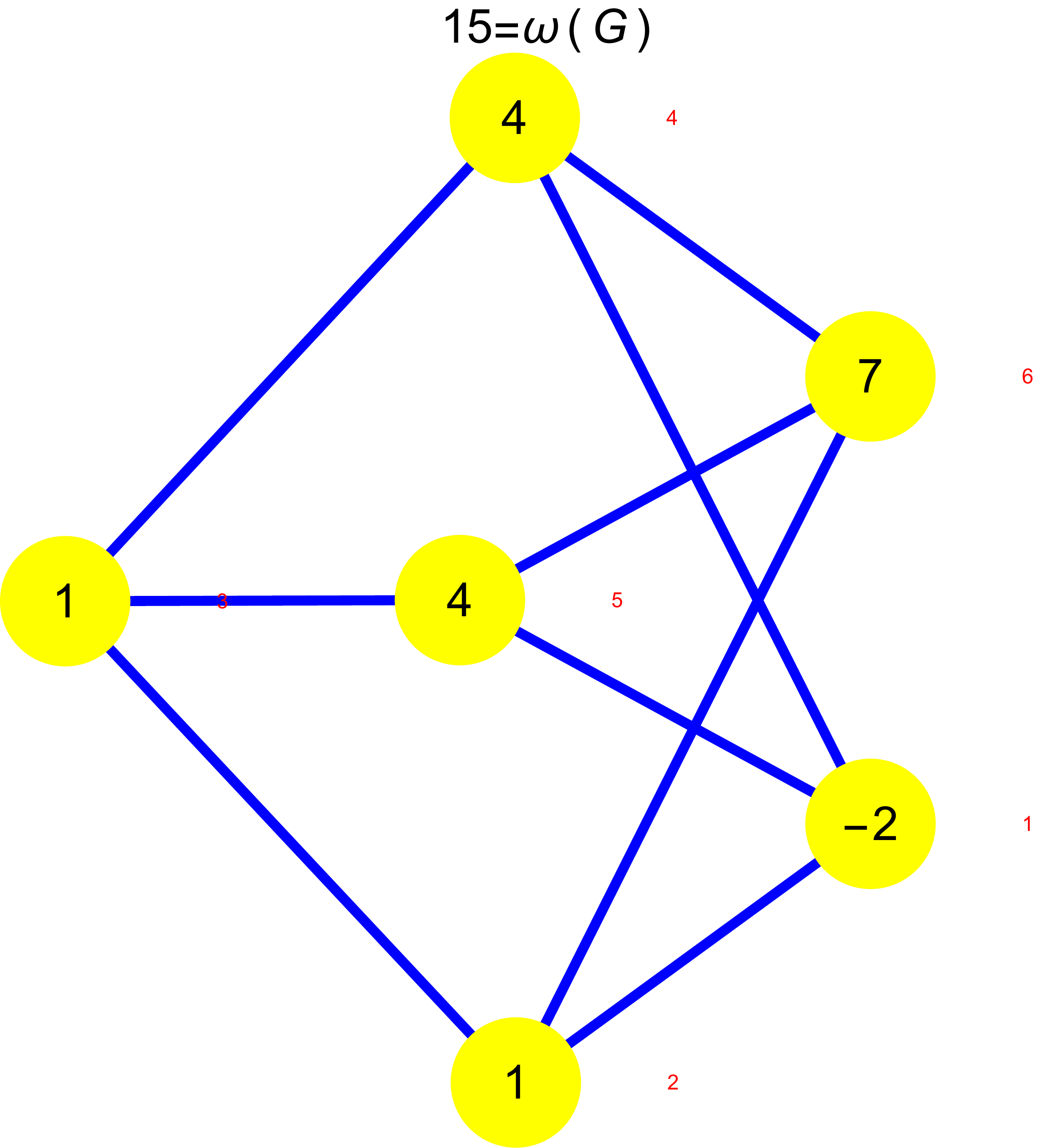}}
\scalebox{0.12}{\includegraphics{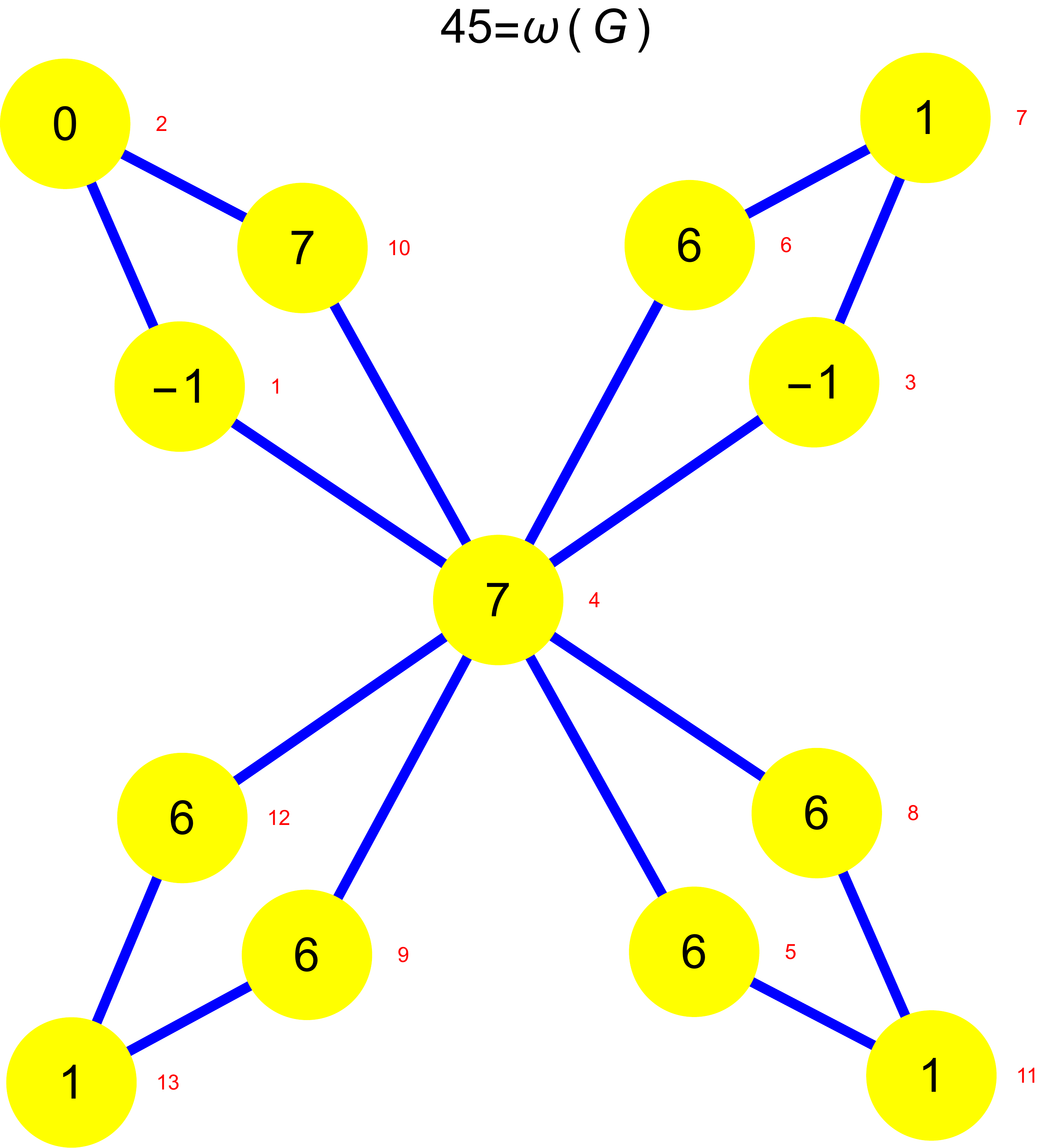}}
\scalebox{0.12}{\includegraphics{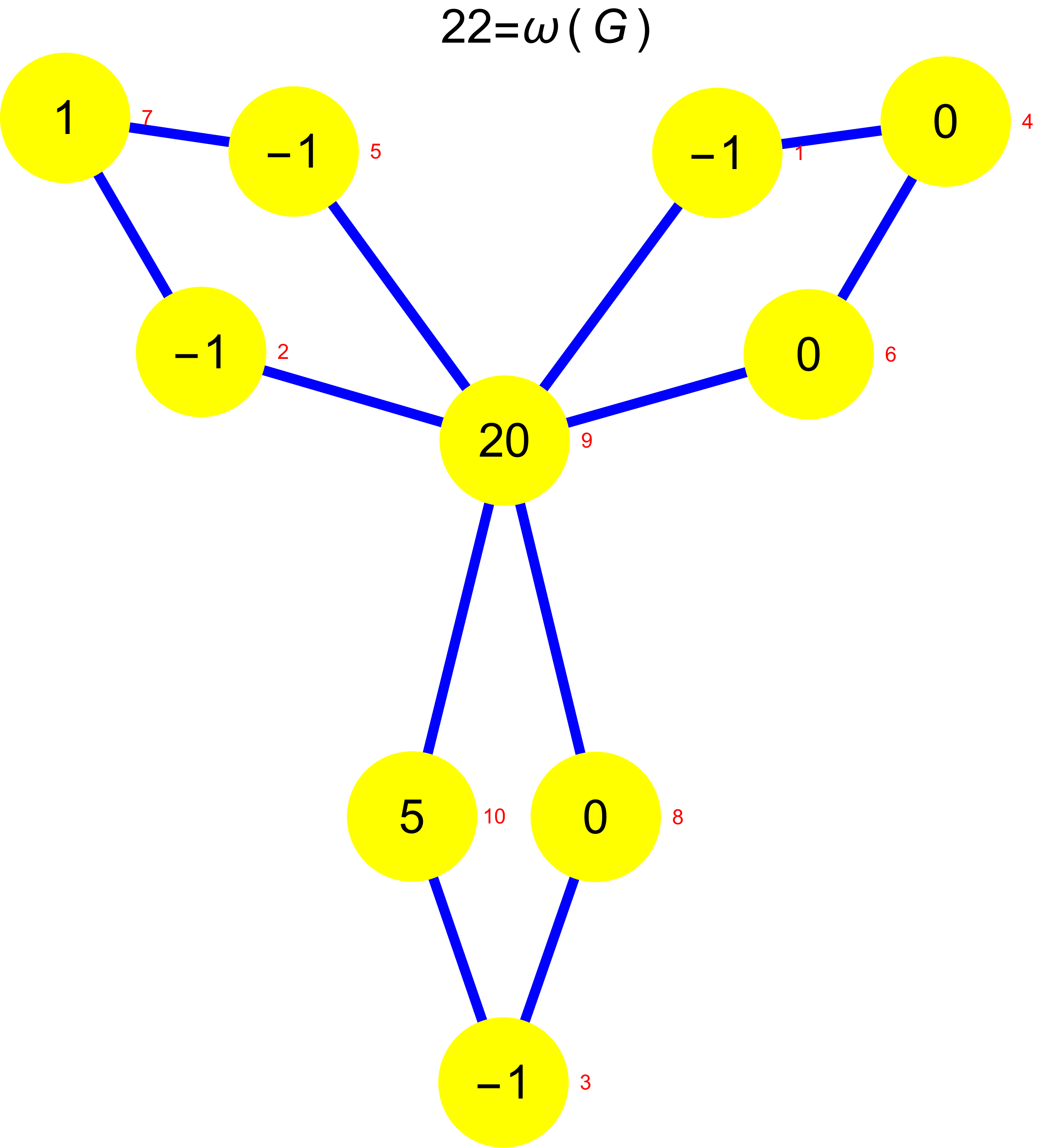}}
\scalebox{0.12}{\includegraphics{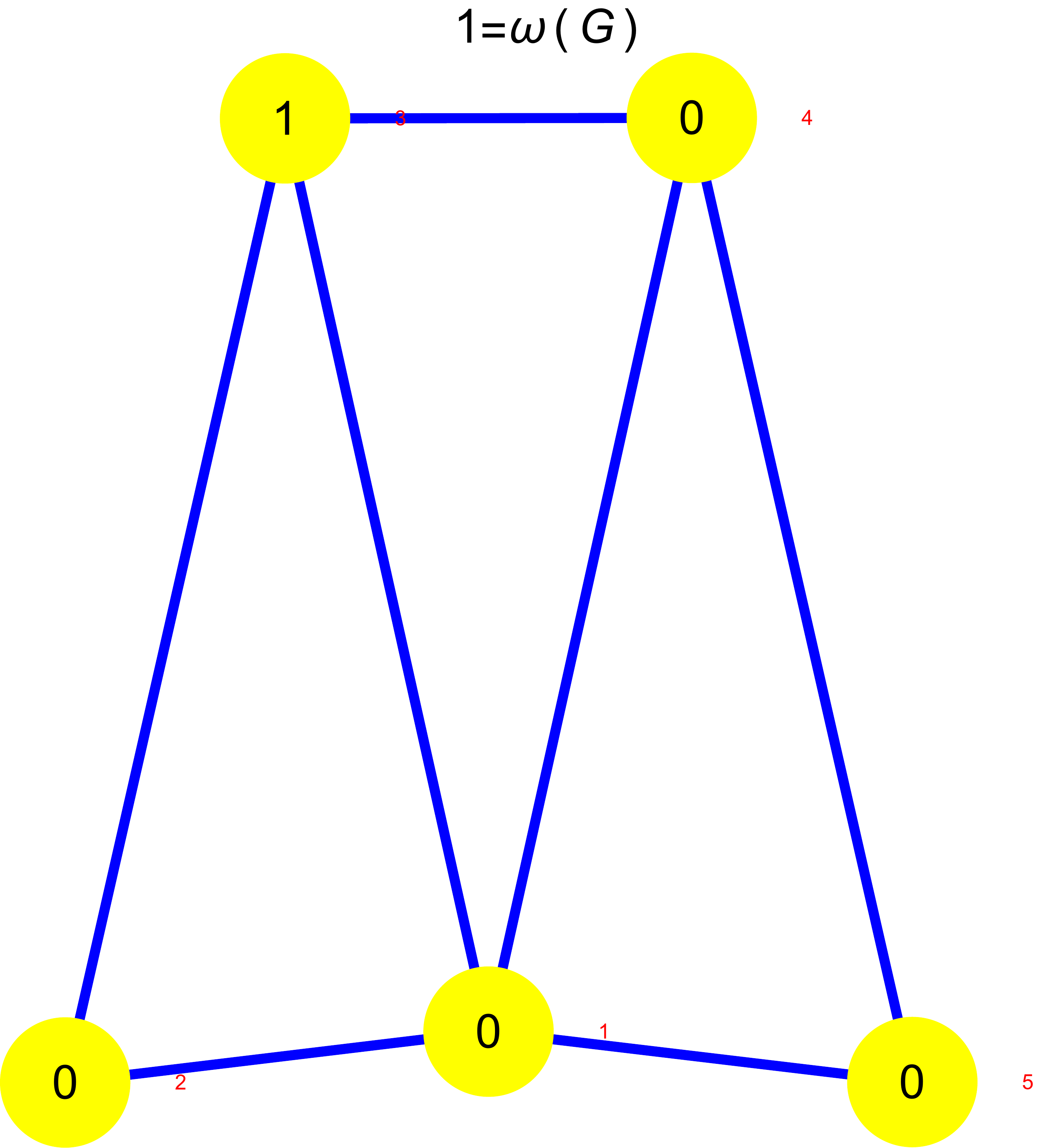}}
\scalebox{0.12}{\includegraphics{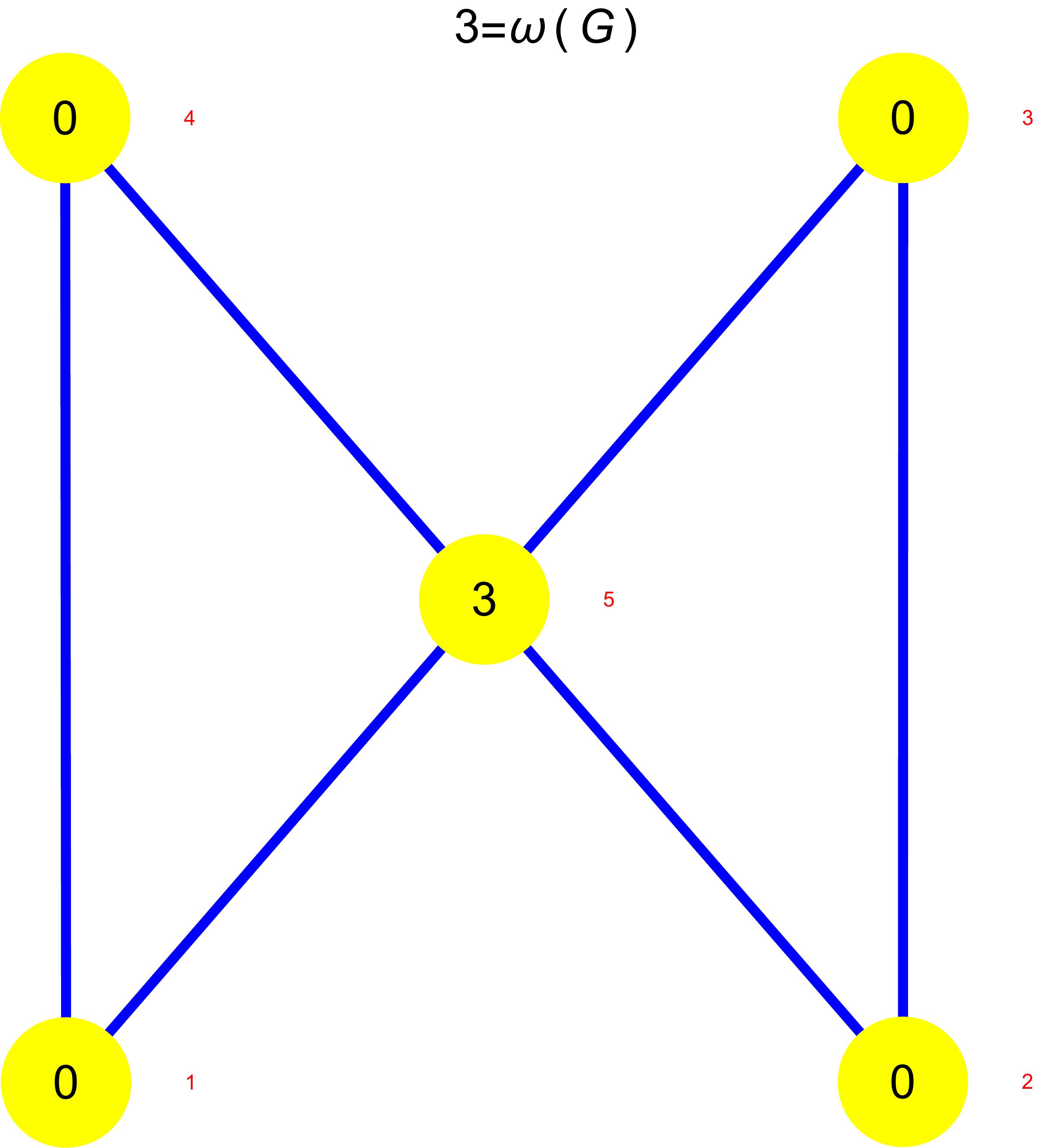}}
\caption{
Examples of graphs with quadratic Wu indices.
}
\end{figure}

\begin{figure}[!htpb]
\scalebox{0.12}{\includegraphics{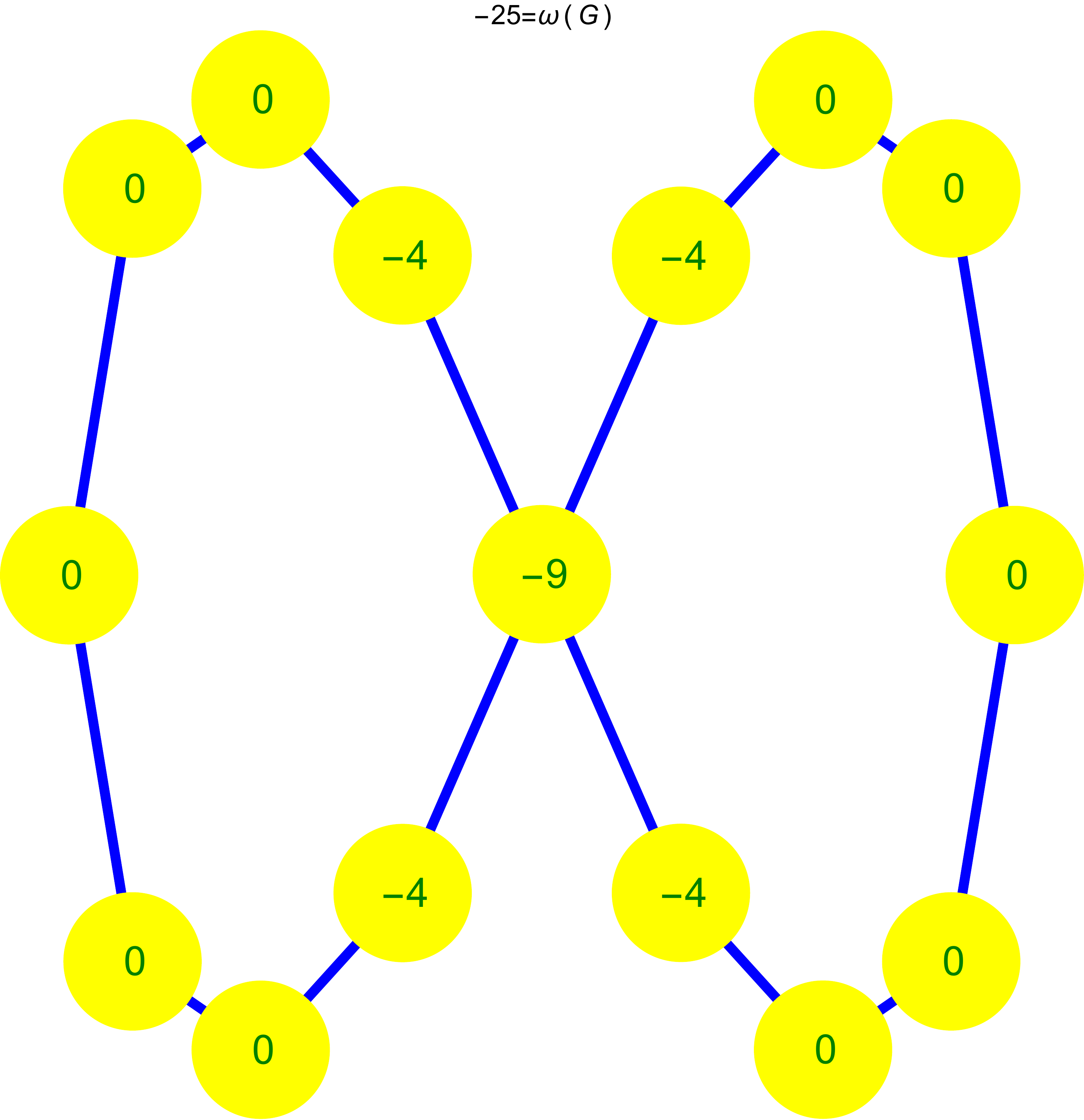}}
\scalebox{0.12}{\includegraphics{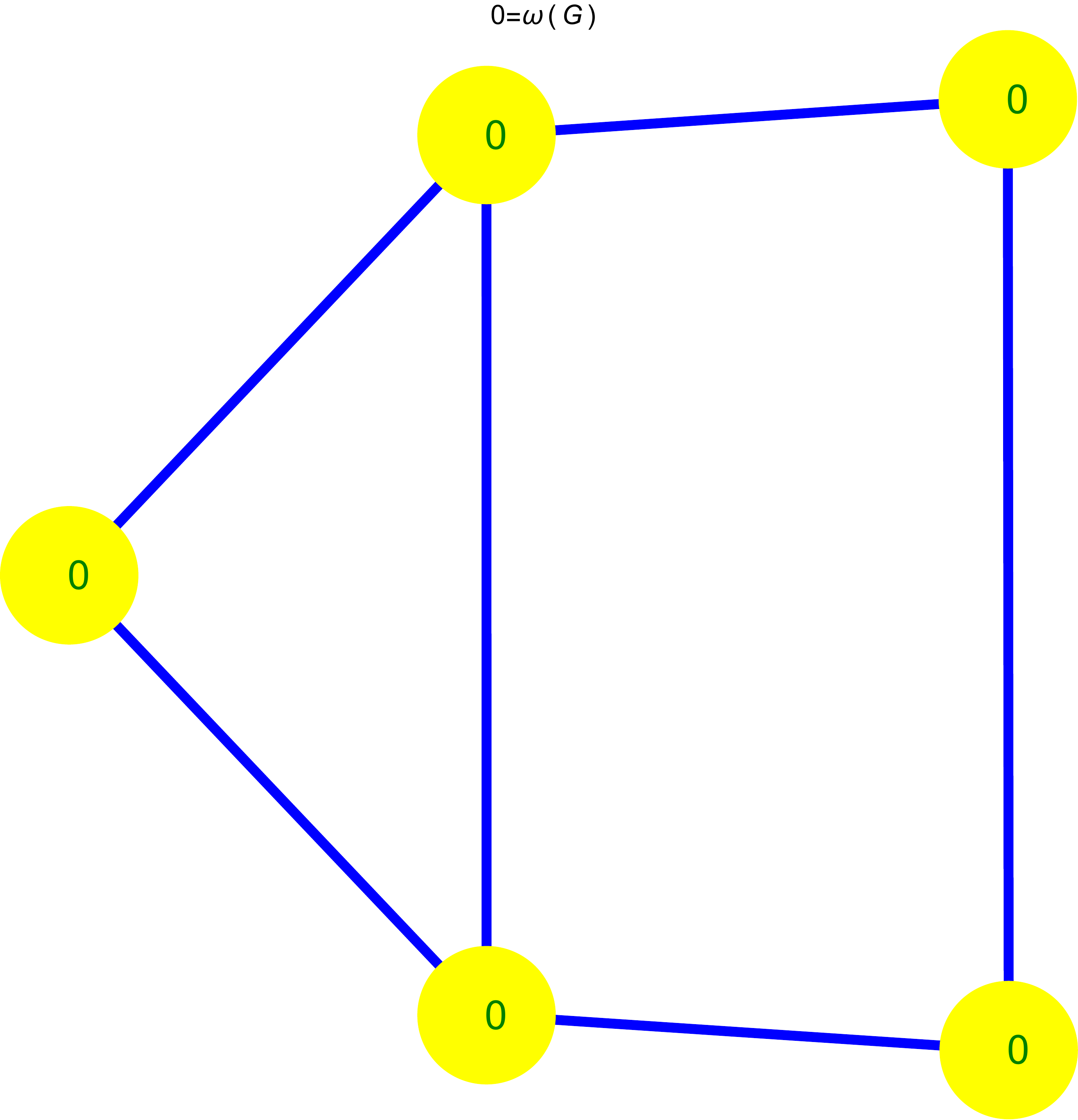}}
\scalebox{0.12}{\includegraphics{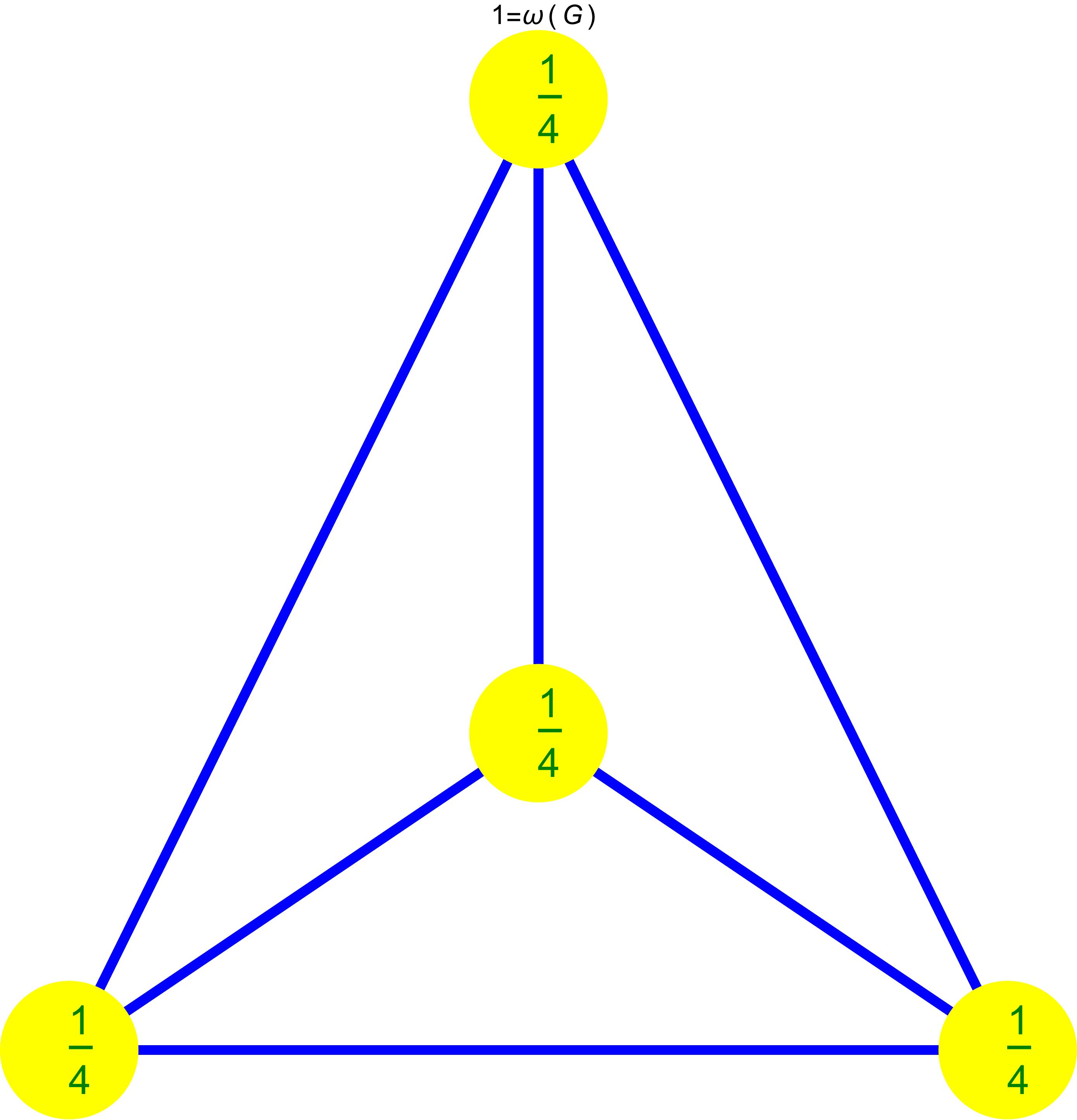}}
\scalebox{0.12}{\includegraphics{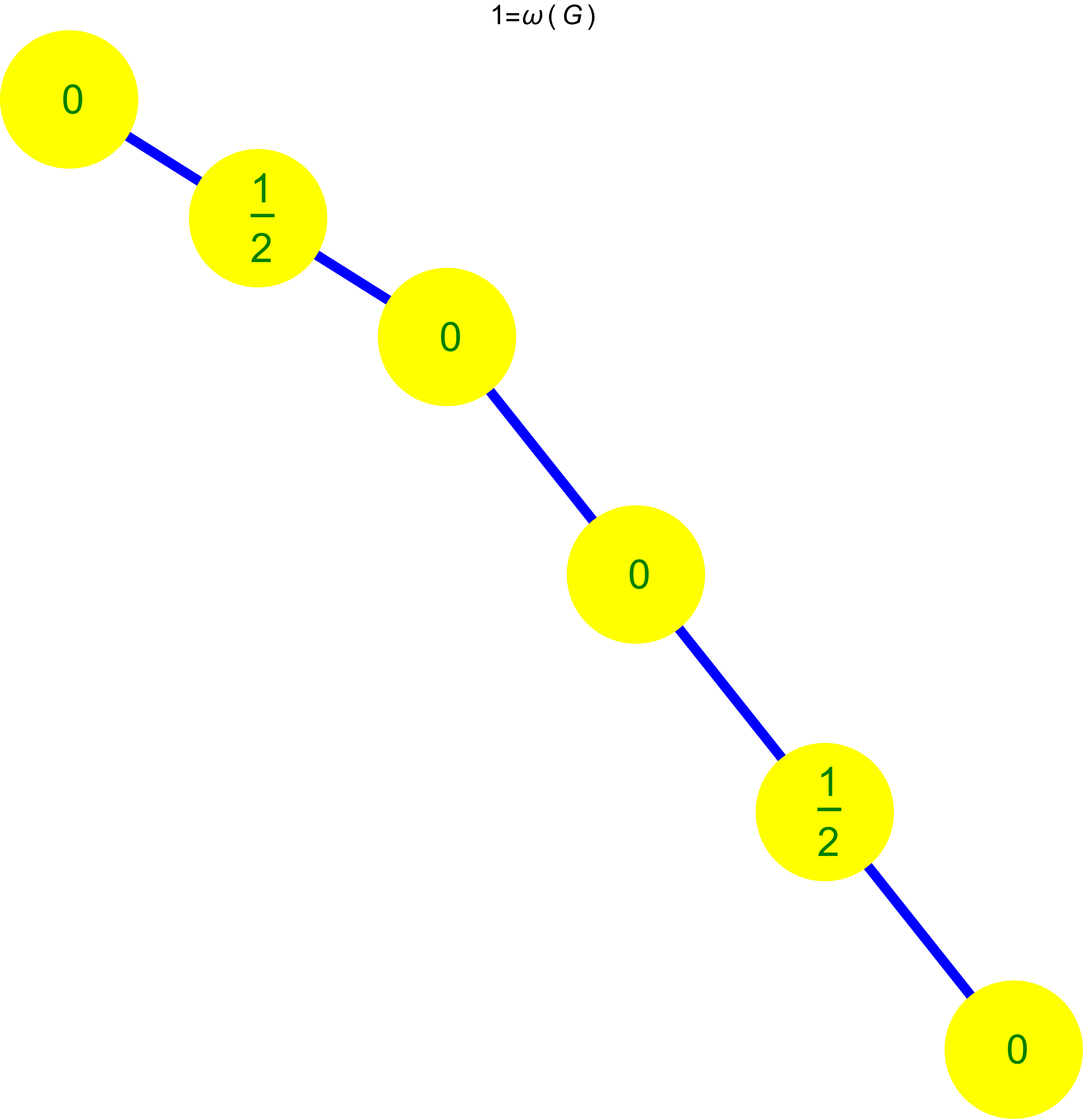}}
\scalebox{0.12}{\includegraphics{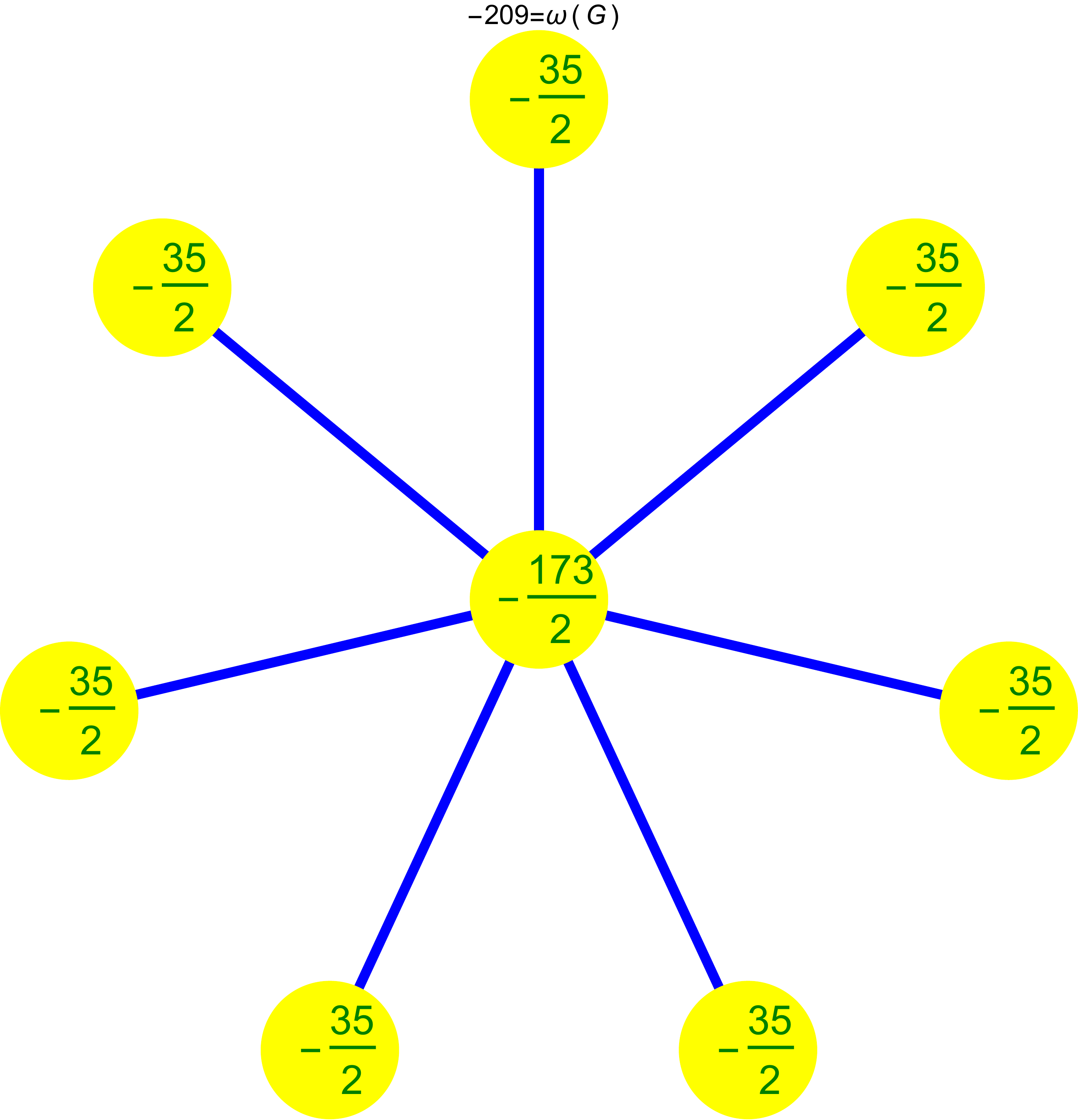}}
\scalebox{0.12}{\includegraphics{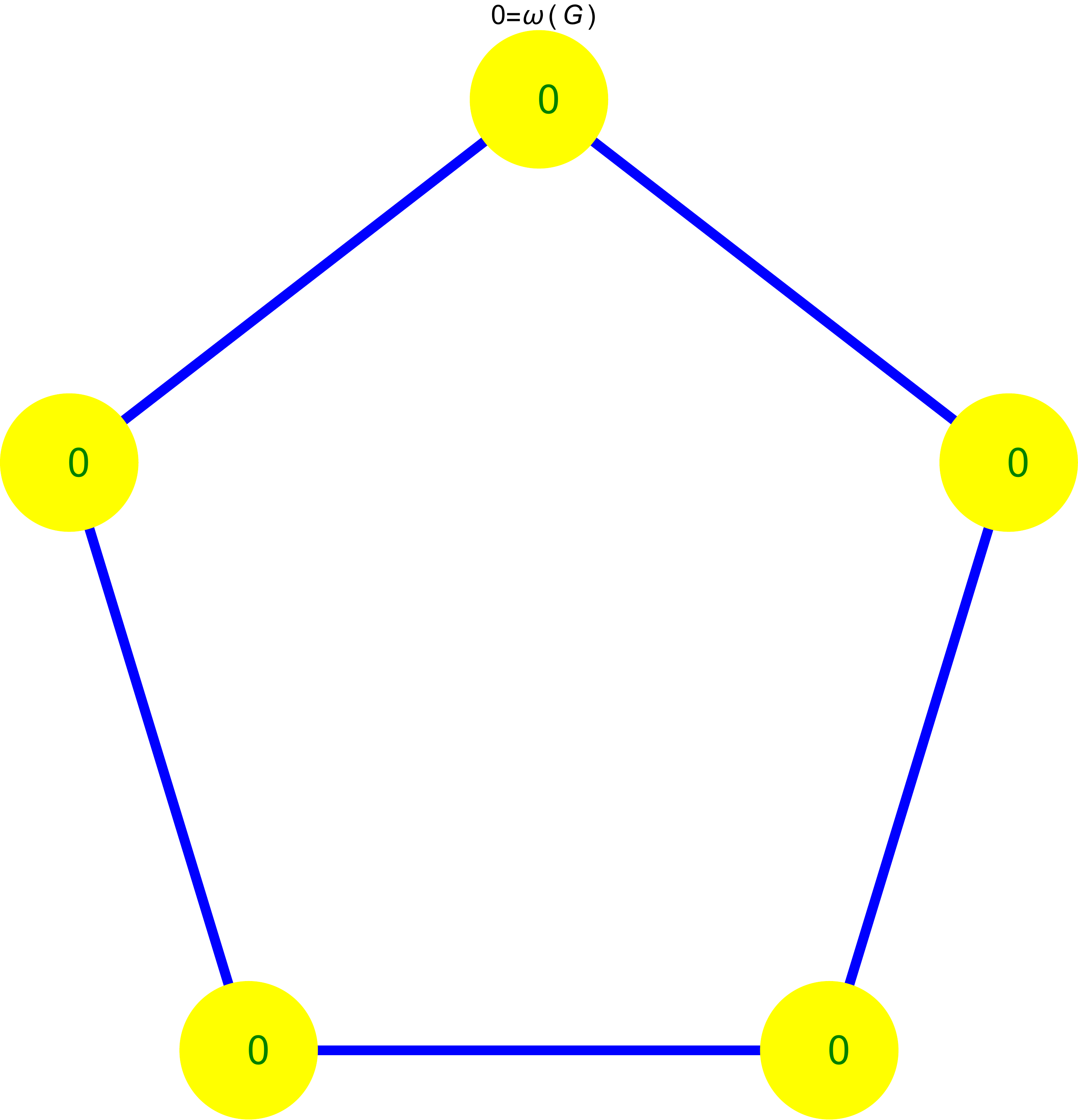}}
\scalebox{0.12}{\includegraphics{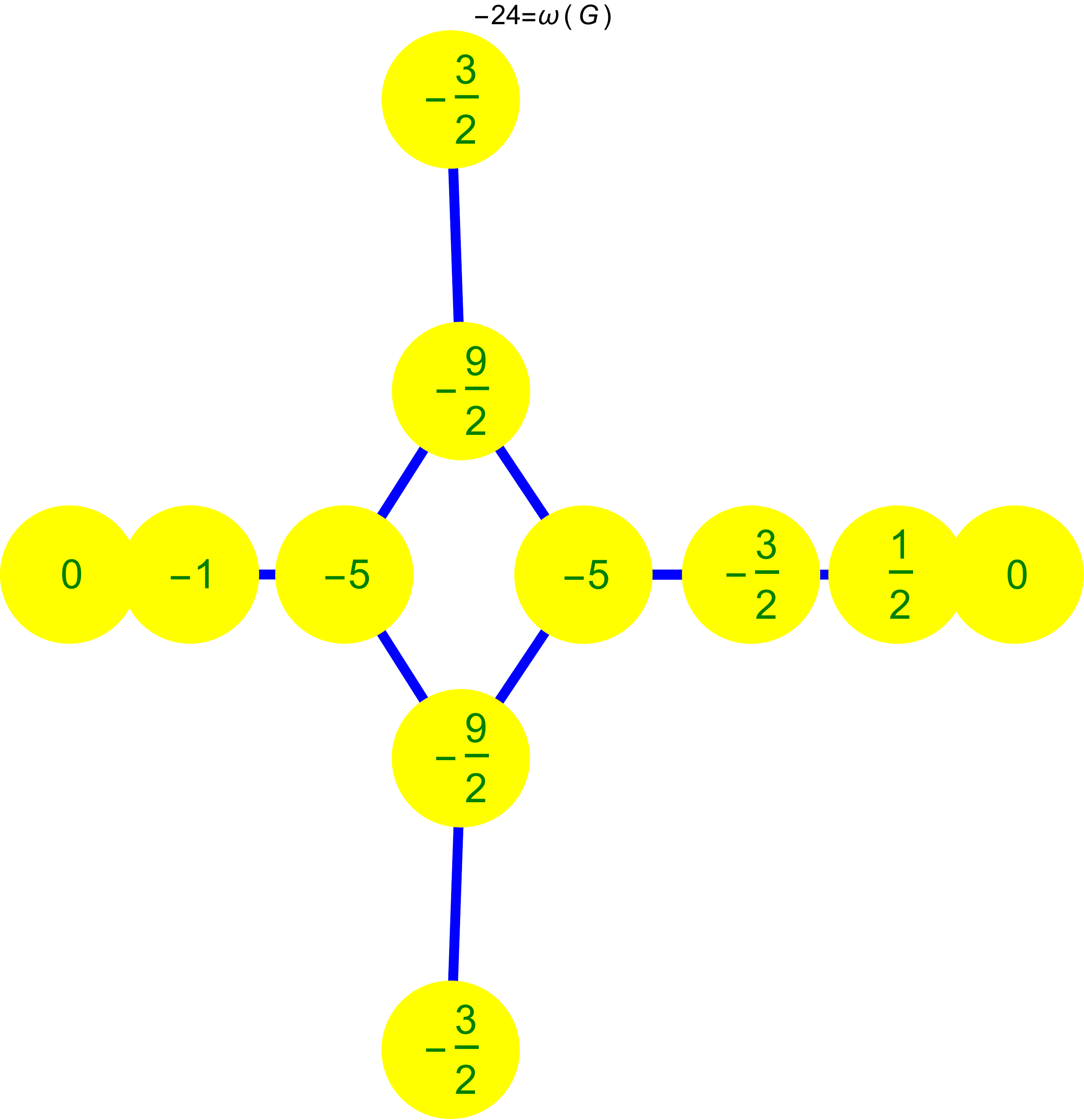}}
\scalebox{0.12}{\includegraphics{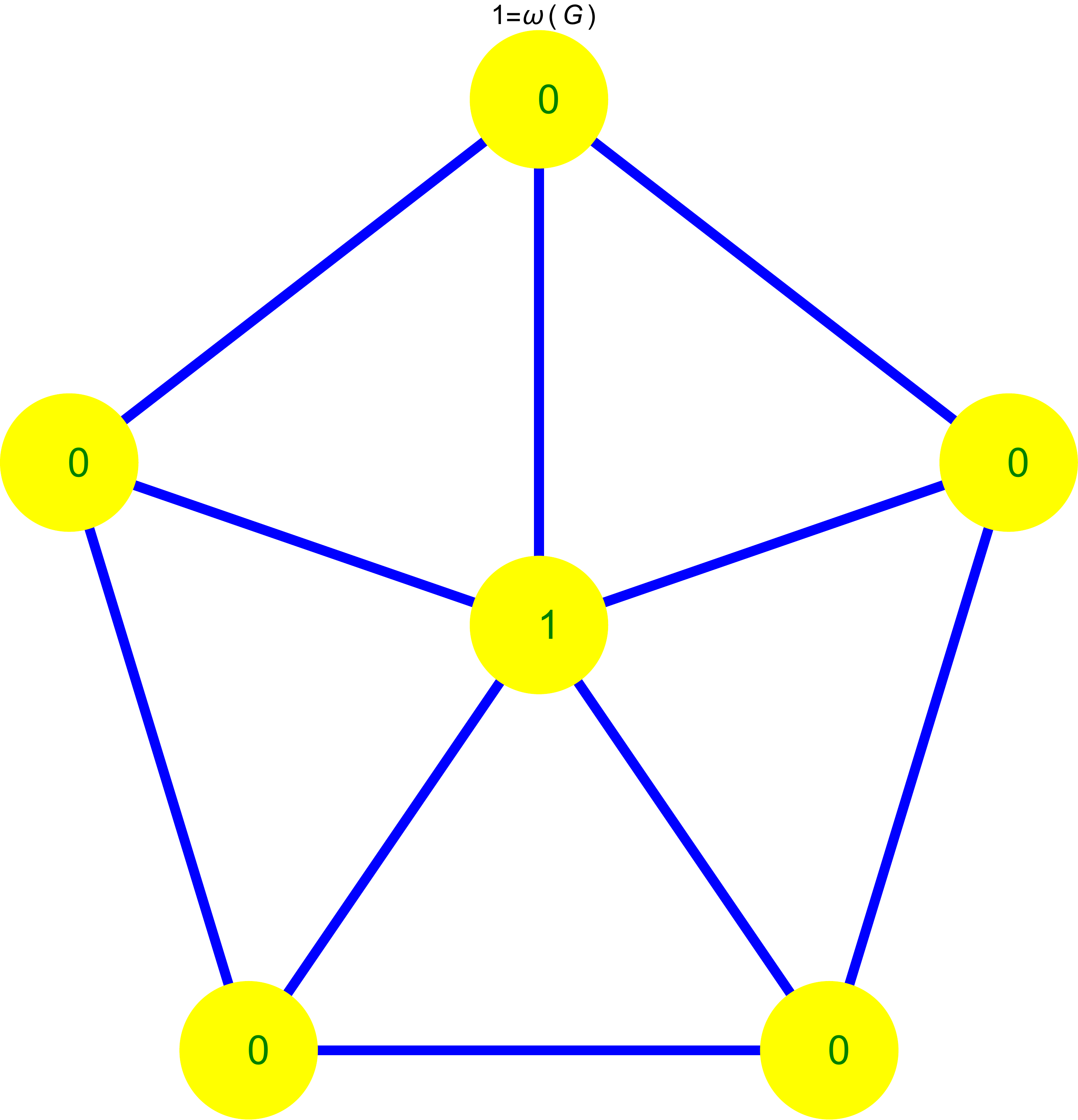}}
\scalebox{0.12}{\includegraphics{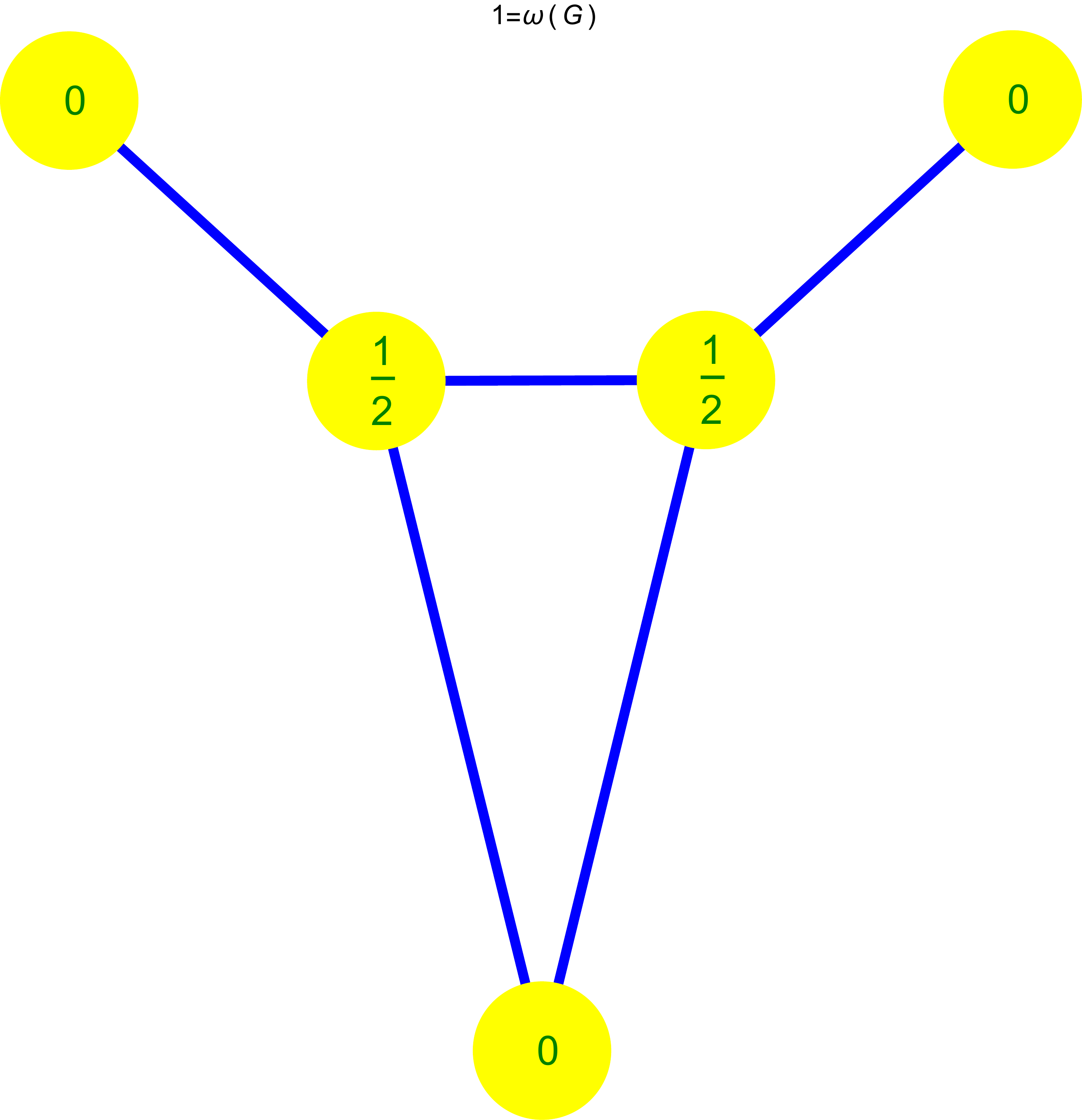}}
\scalebox{0.12}{\includegraphics{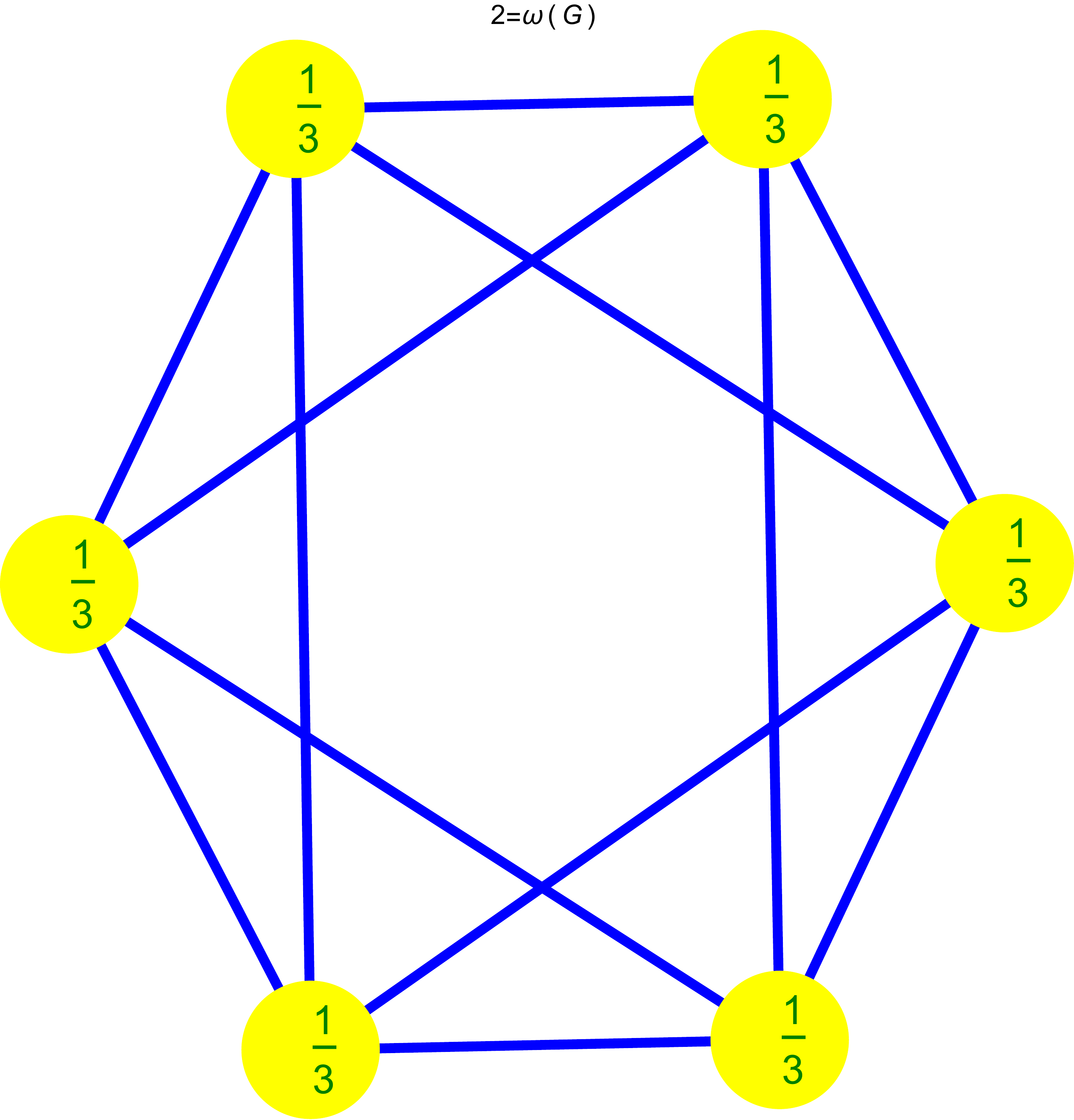}}
\scalebox{0.12}{\includegraphics{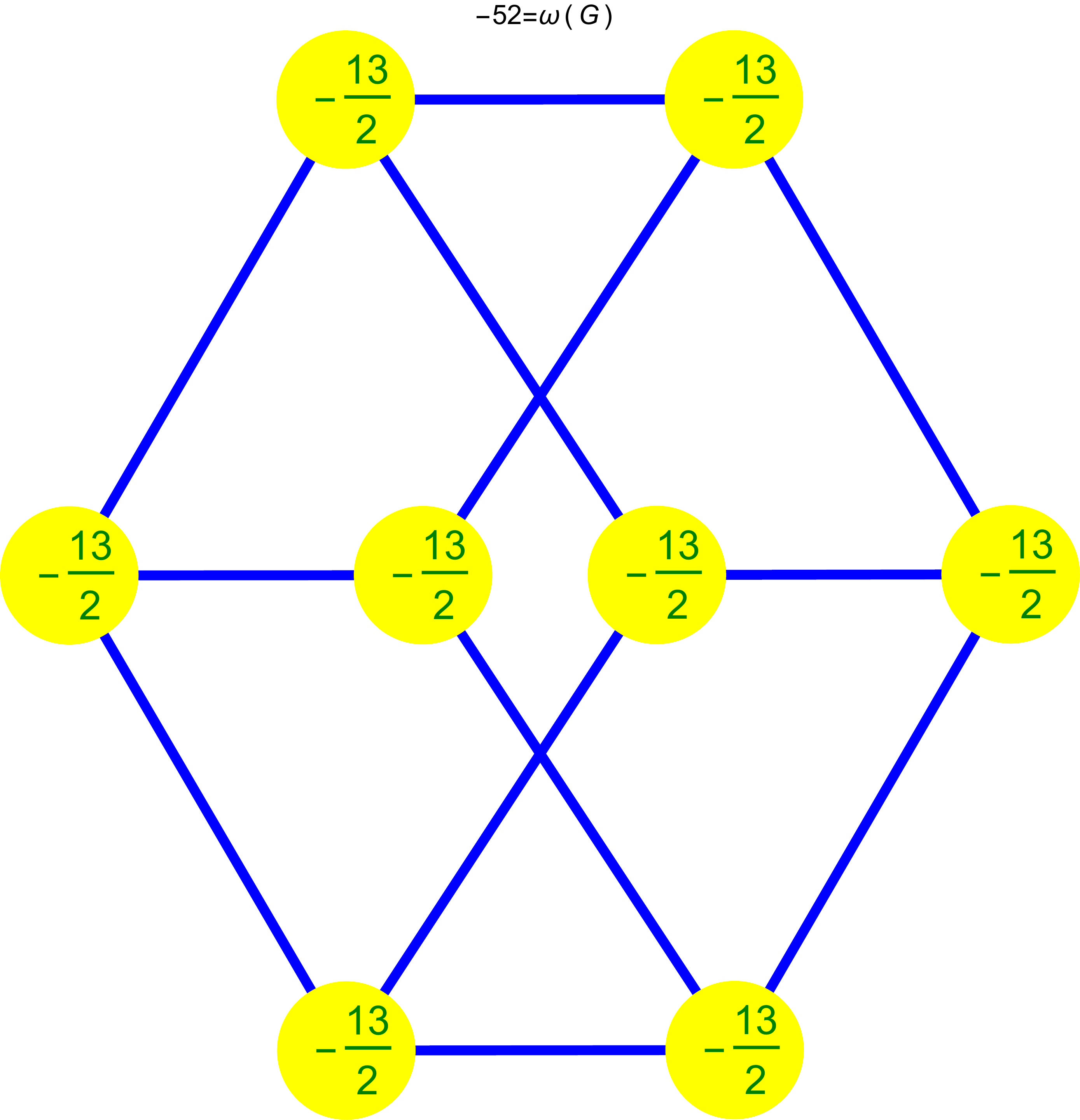}}
\scalebox{0.12}{\includegraphics{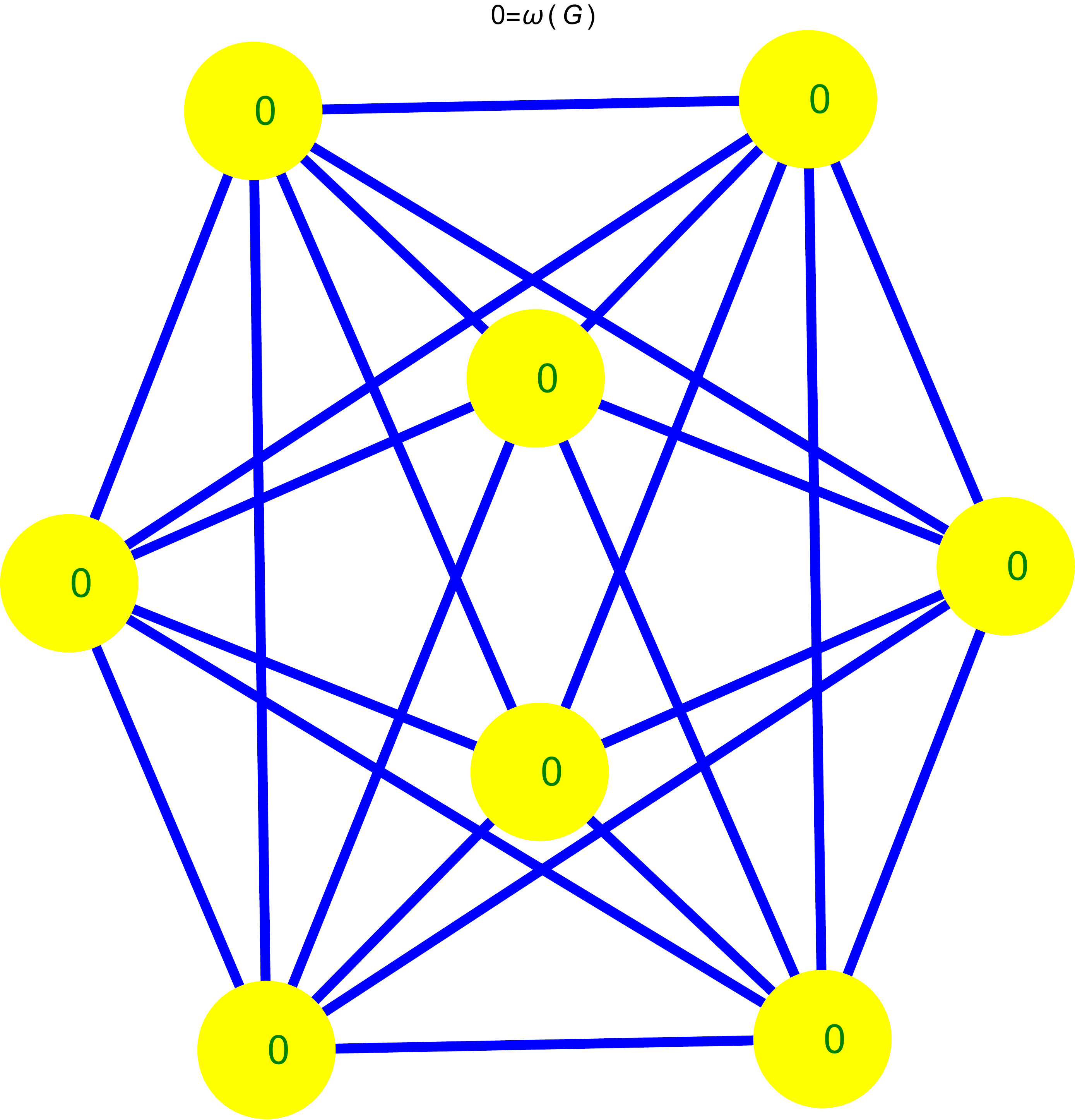}}
\scalebox{0.12}{\includegraphics{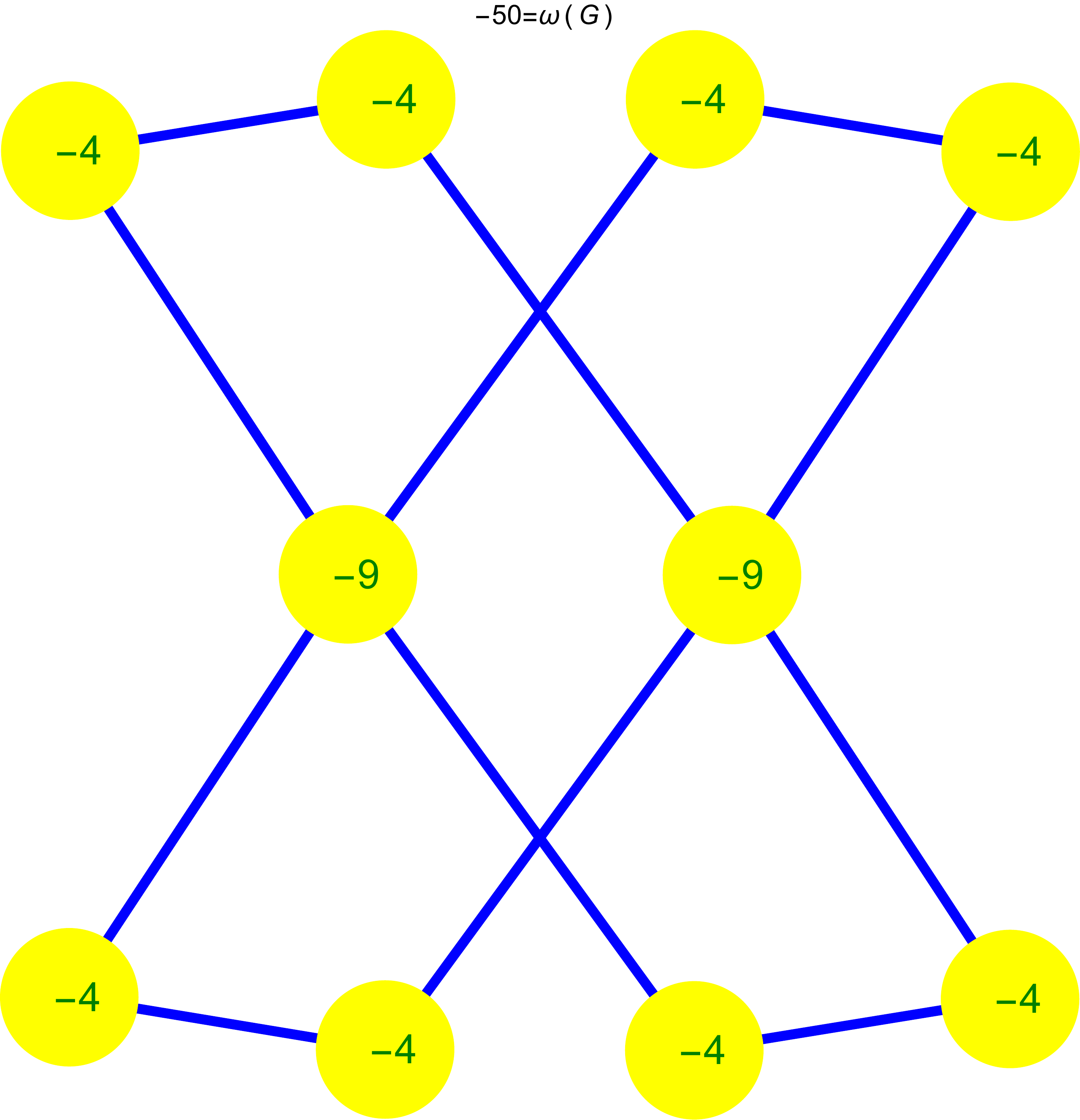}}
\scalebox{0.12}{\includegraphics{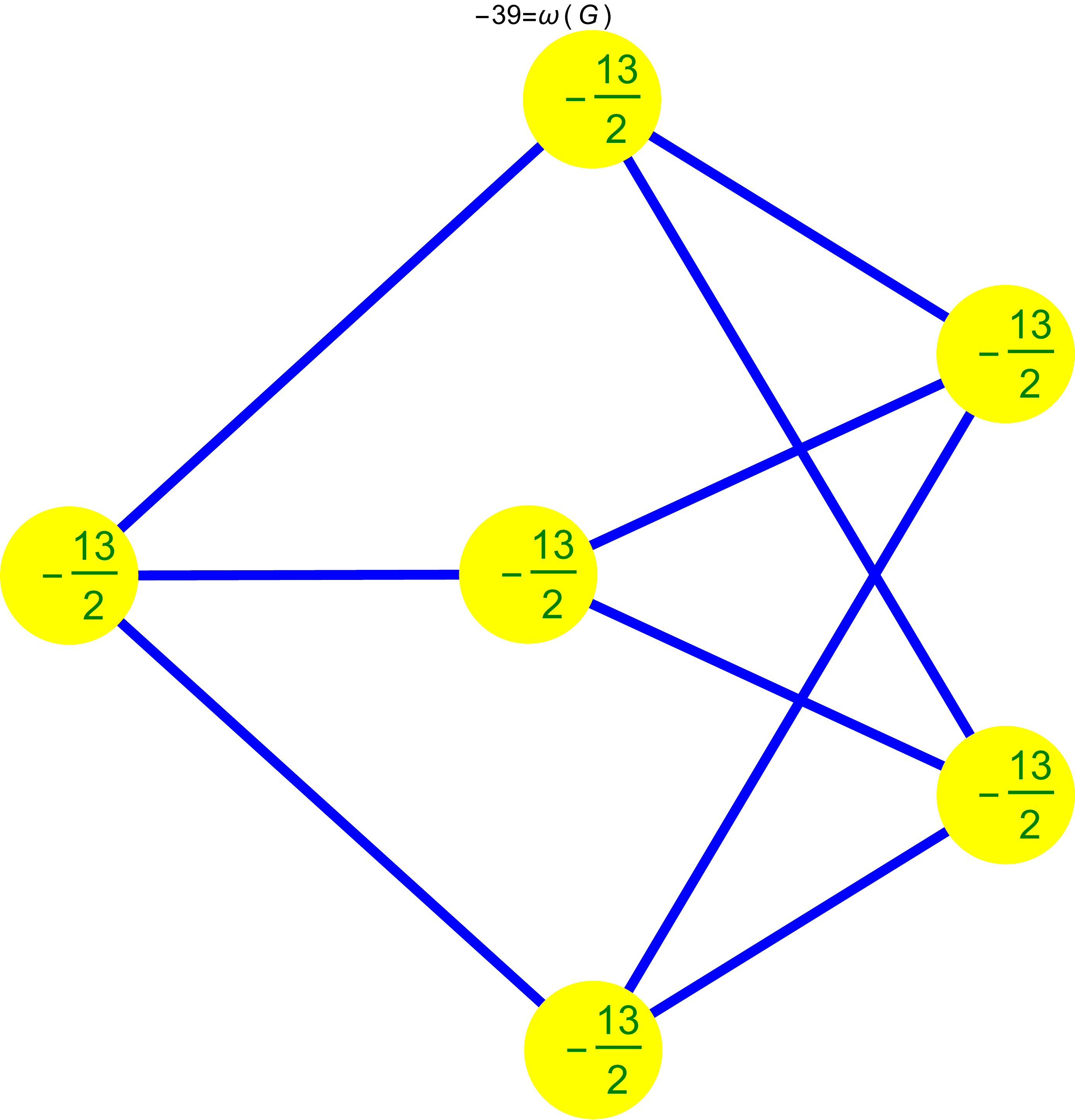}}
\scalebox{0.12}{\includegraphics{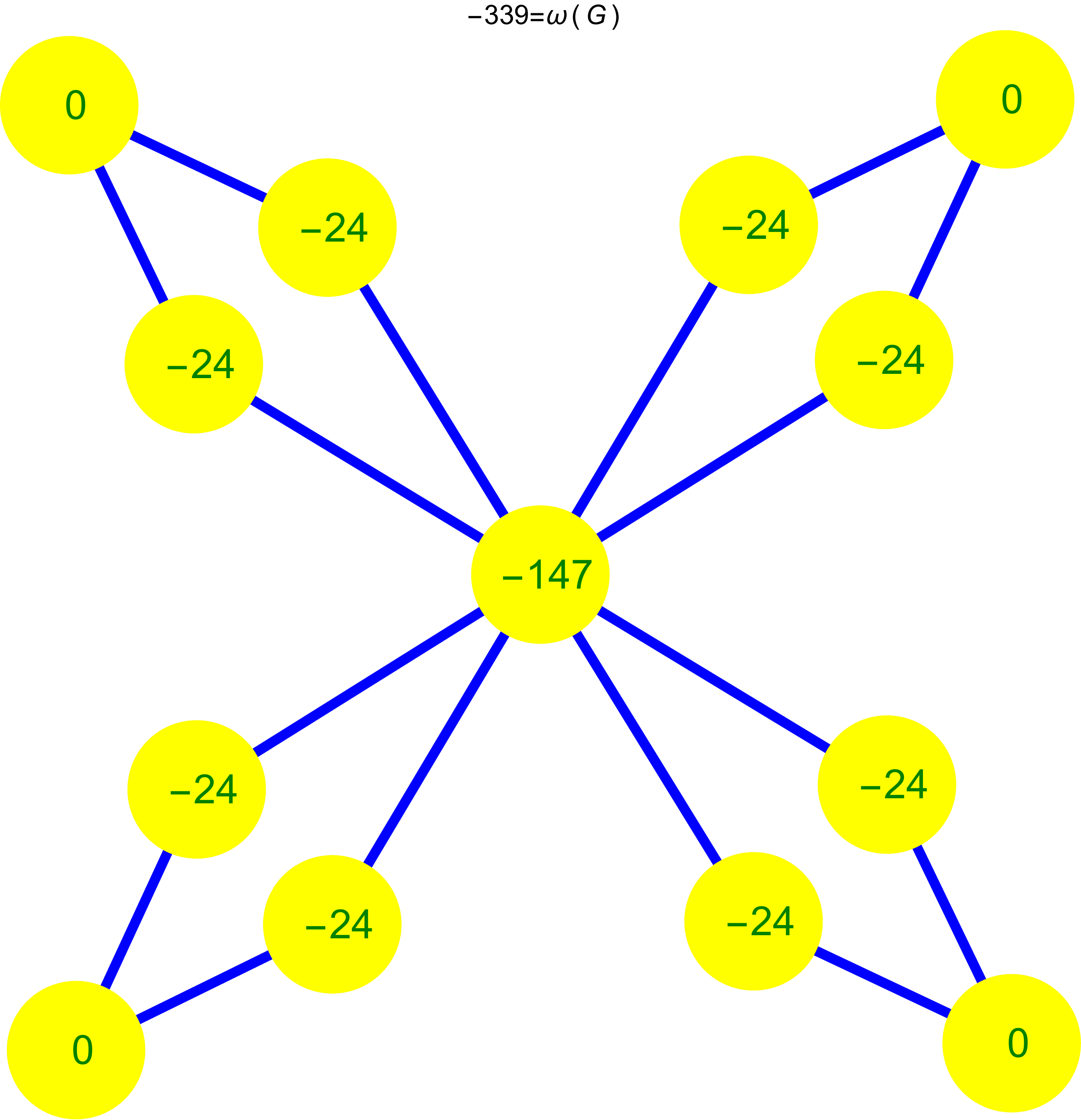}}
\scalebox{0.12}{\includegraphics{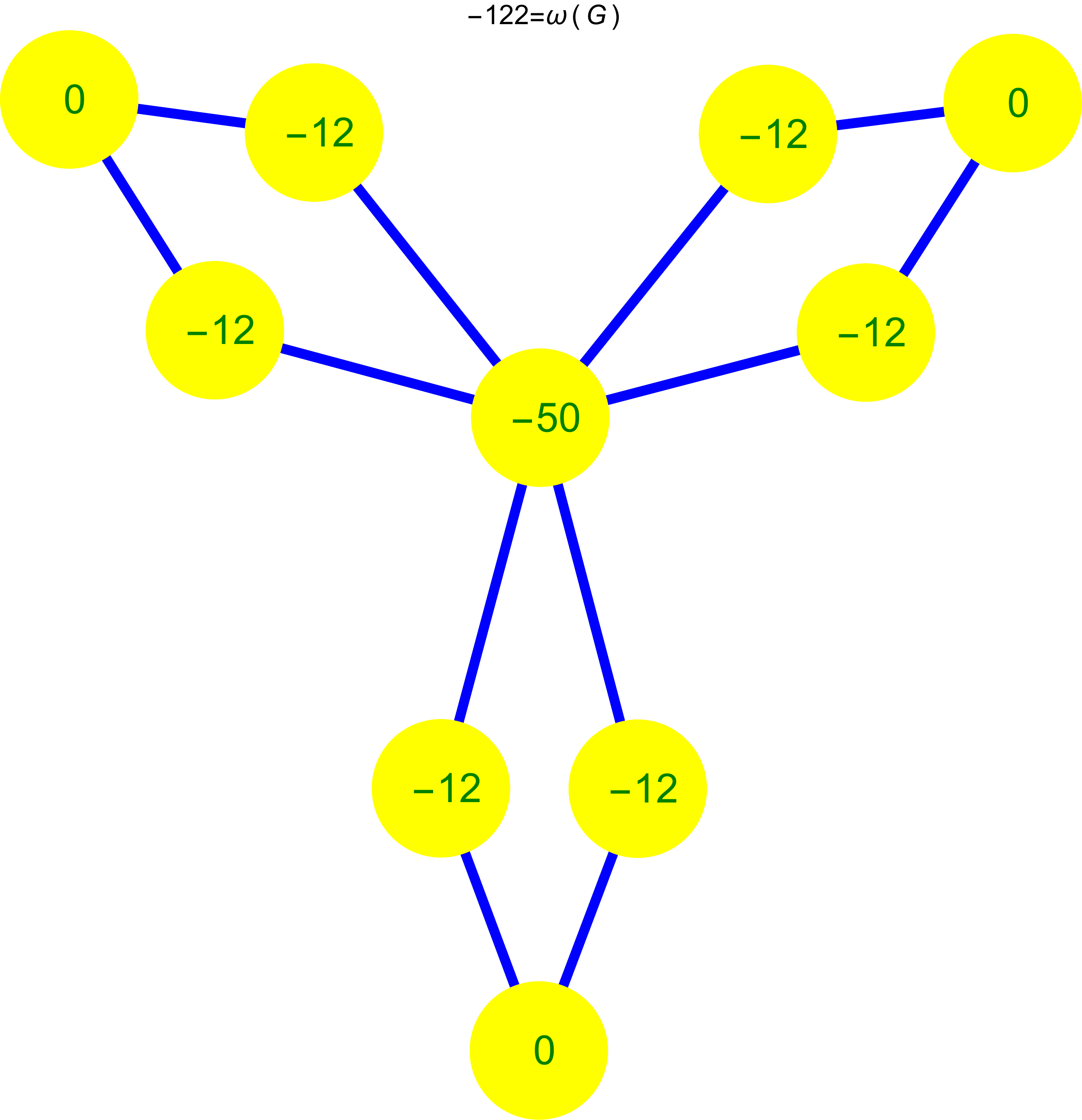}}
\scalebox{0.12}{\includegraphics{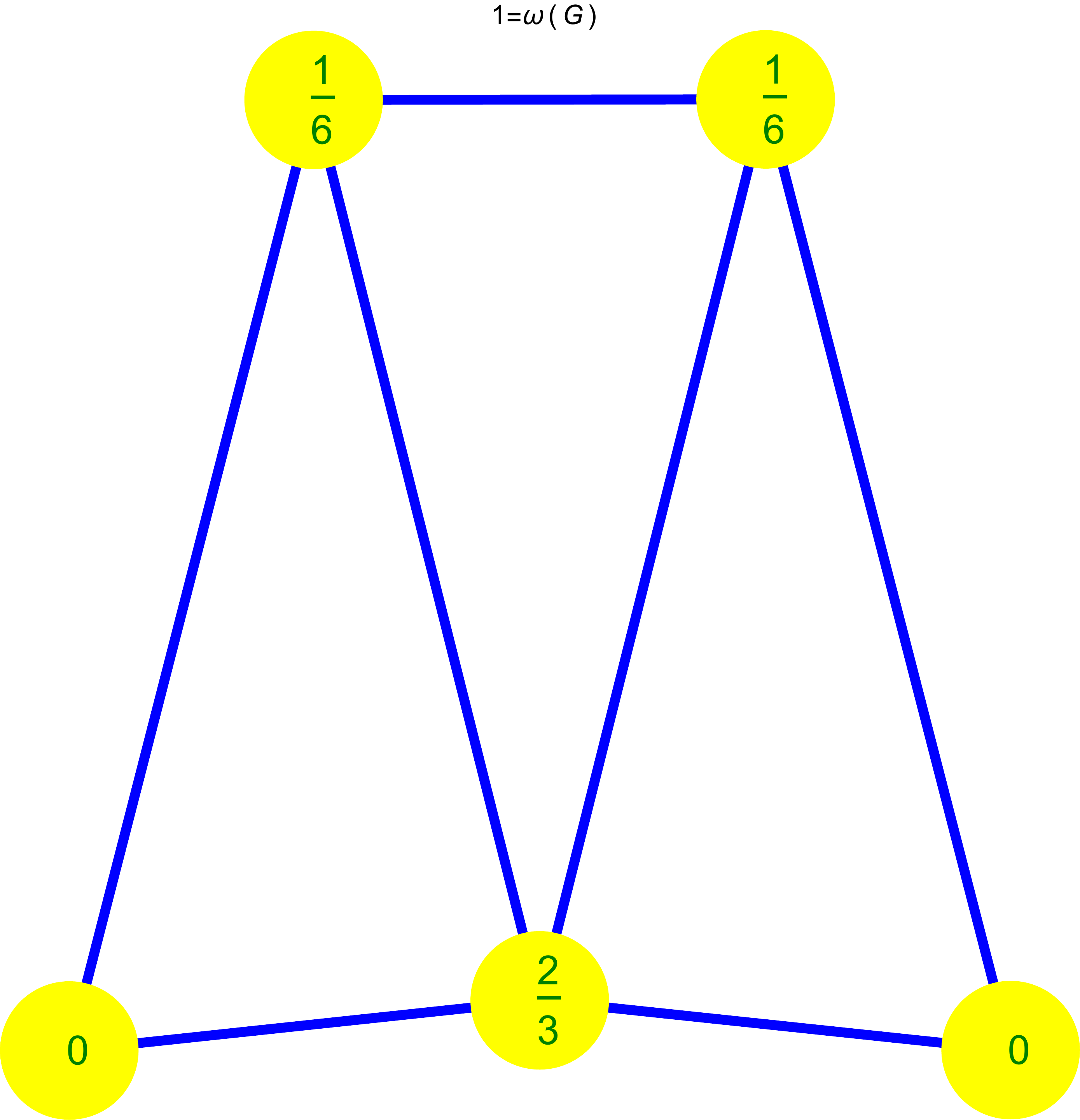}}
\scalebox{0.12}{\includegraphics{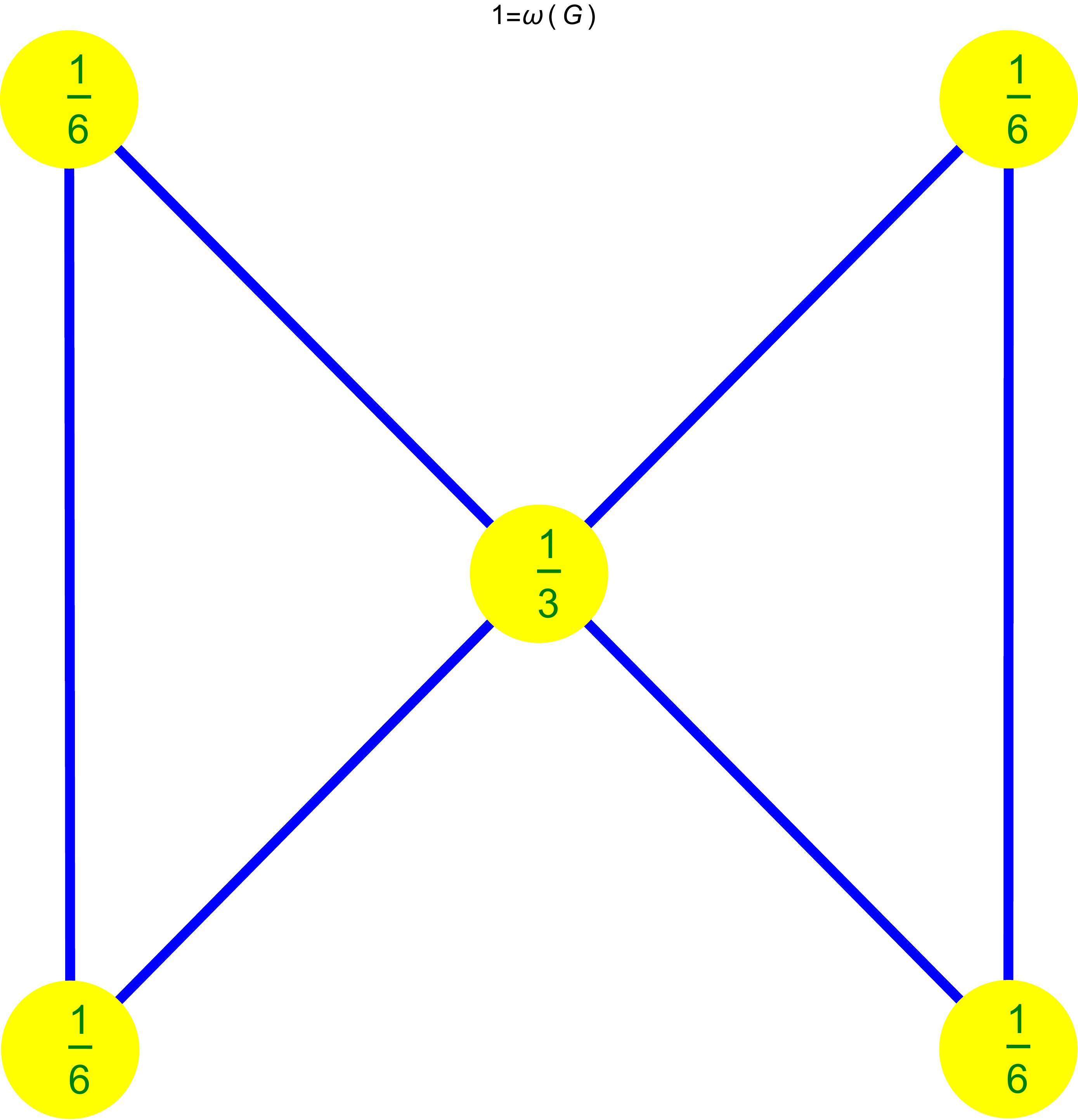}}
\caption{
Examples of graphs with cubic Wu curvatures.
}
\end{figure}

\begin{figure}[!htpb]
\scalebox{0.12}{\includegraphics{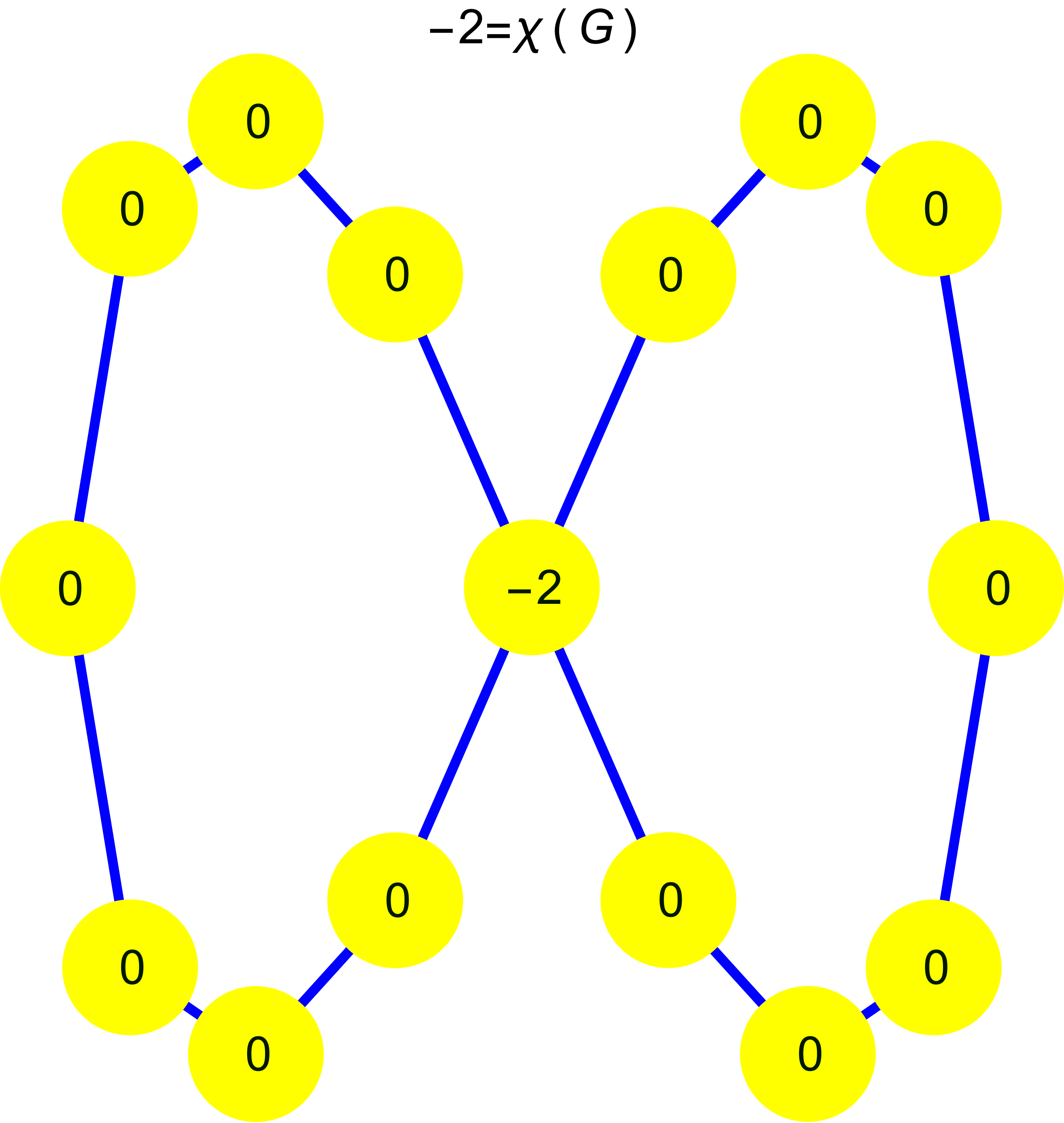}}
\scalebox{0.12}{\includegraphics{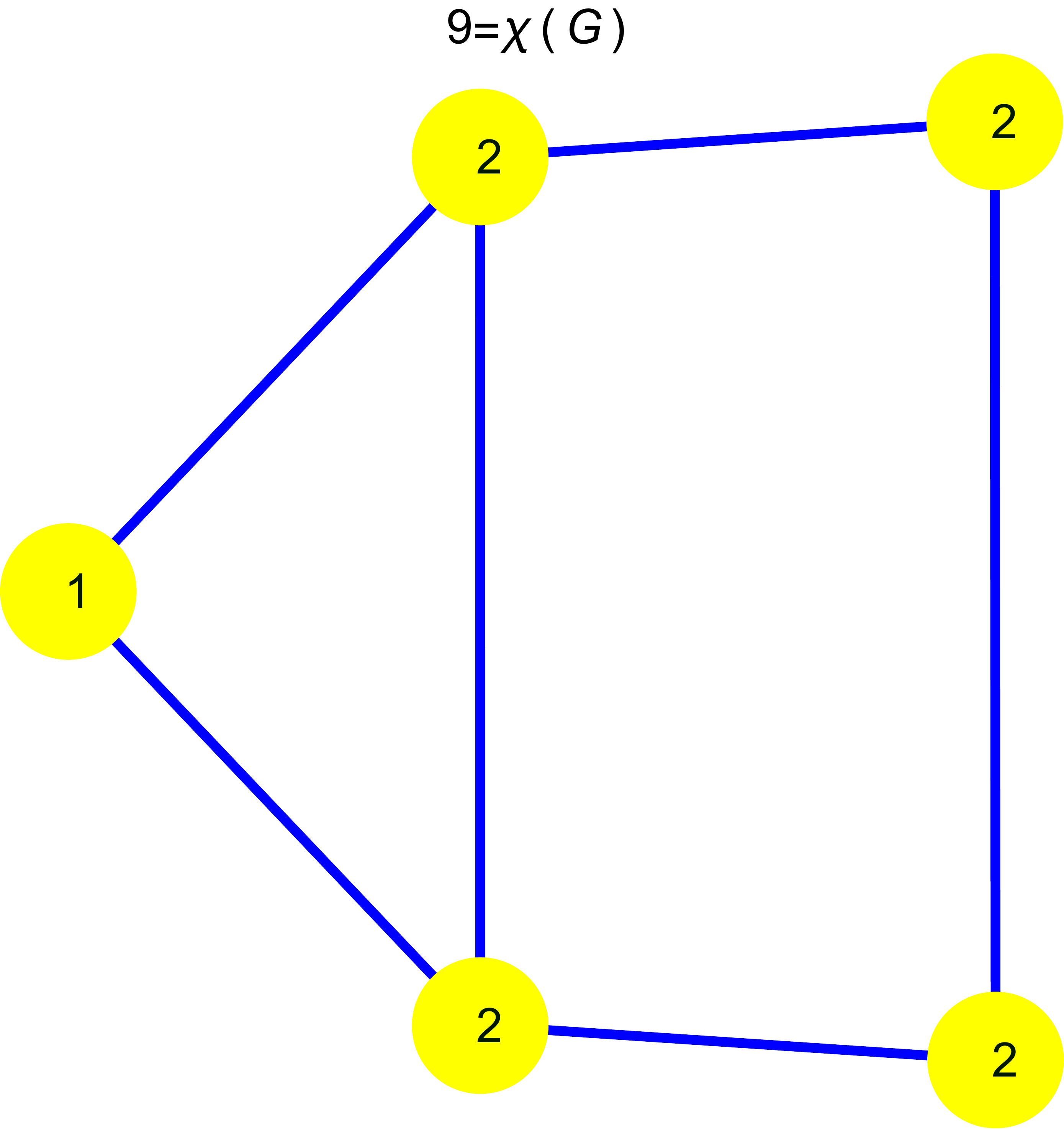}}
\scalebox{0.12}{\includegraphics{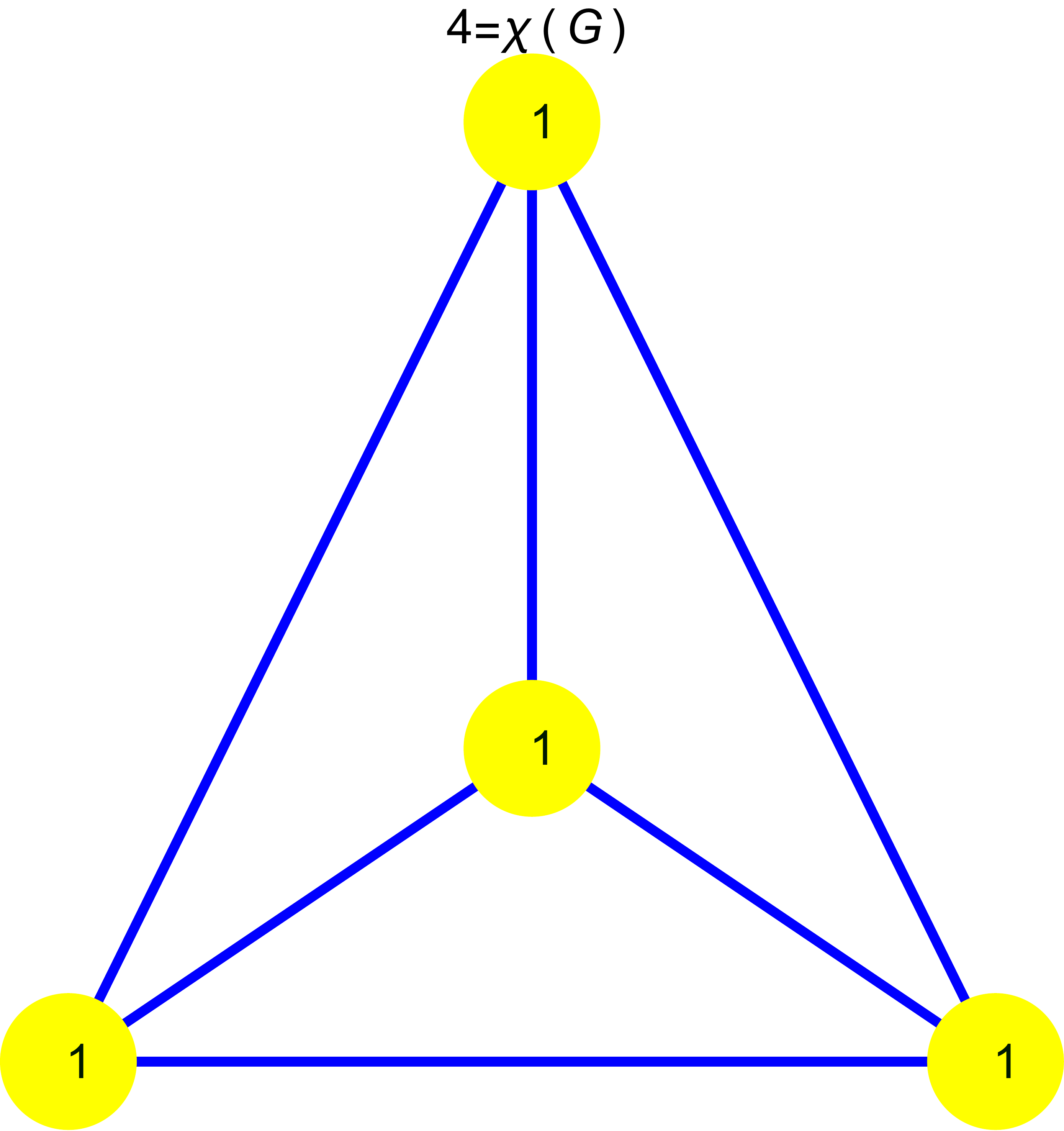}}
\scalebox{0.12}{\includegraphics{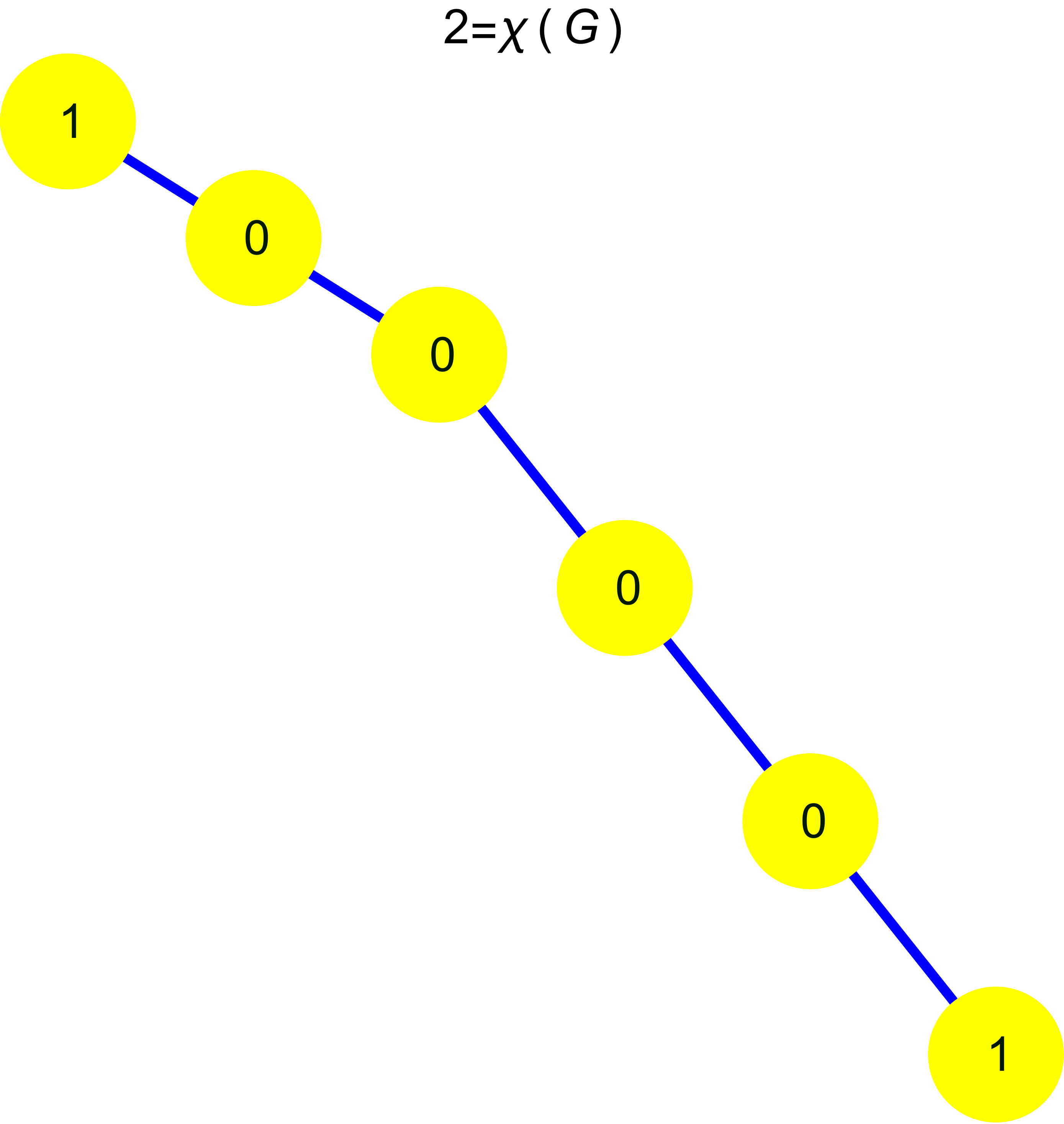}}
\scalebox{0.12}{\includegraphics{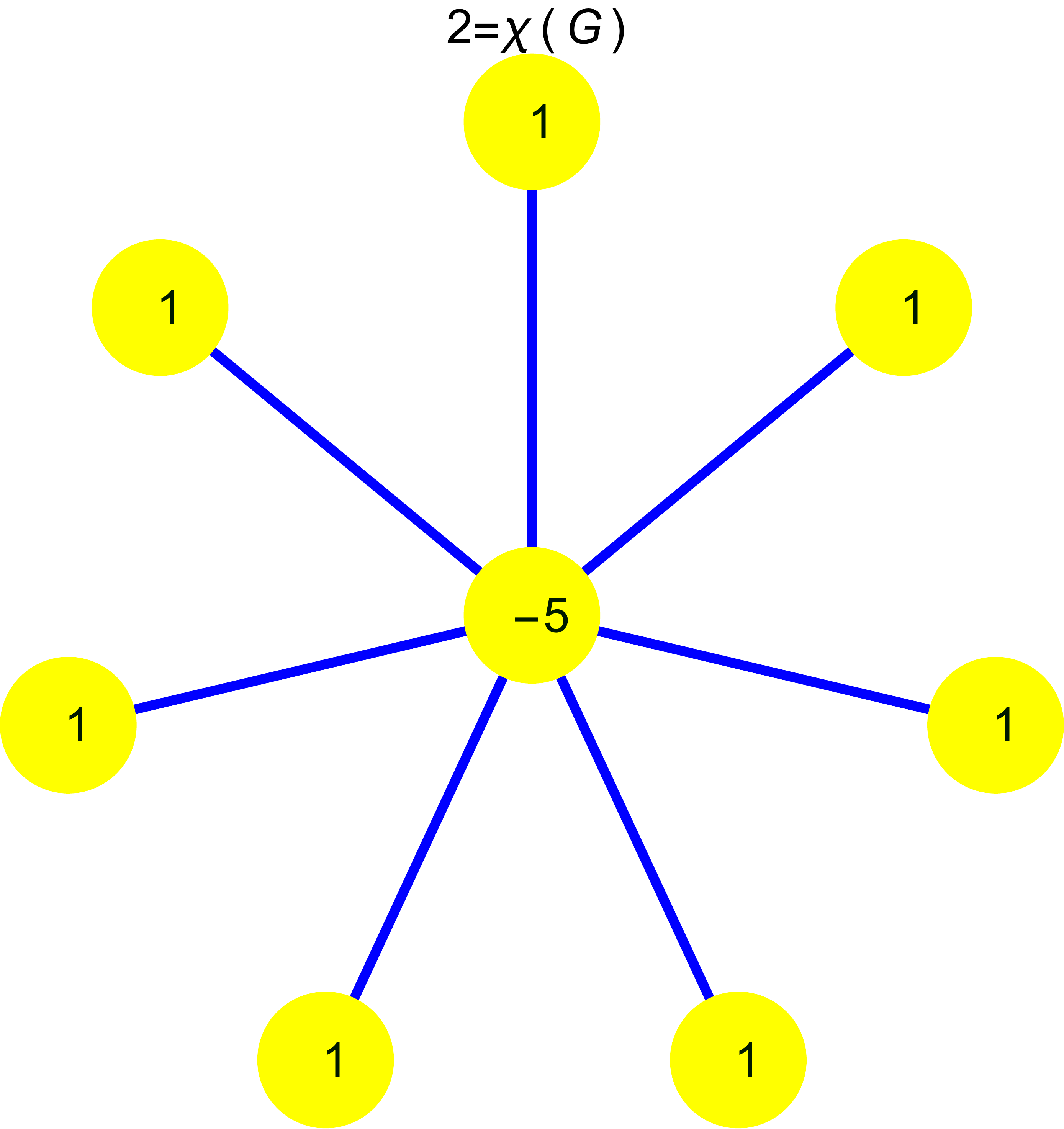}}
\scalebox{0.12}{\includegraphics{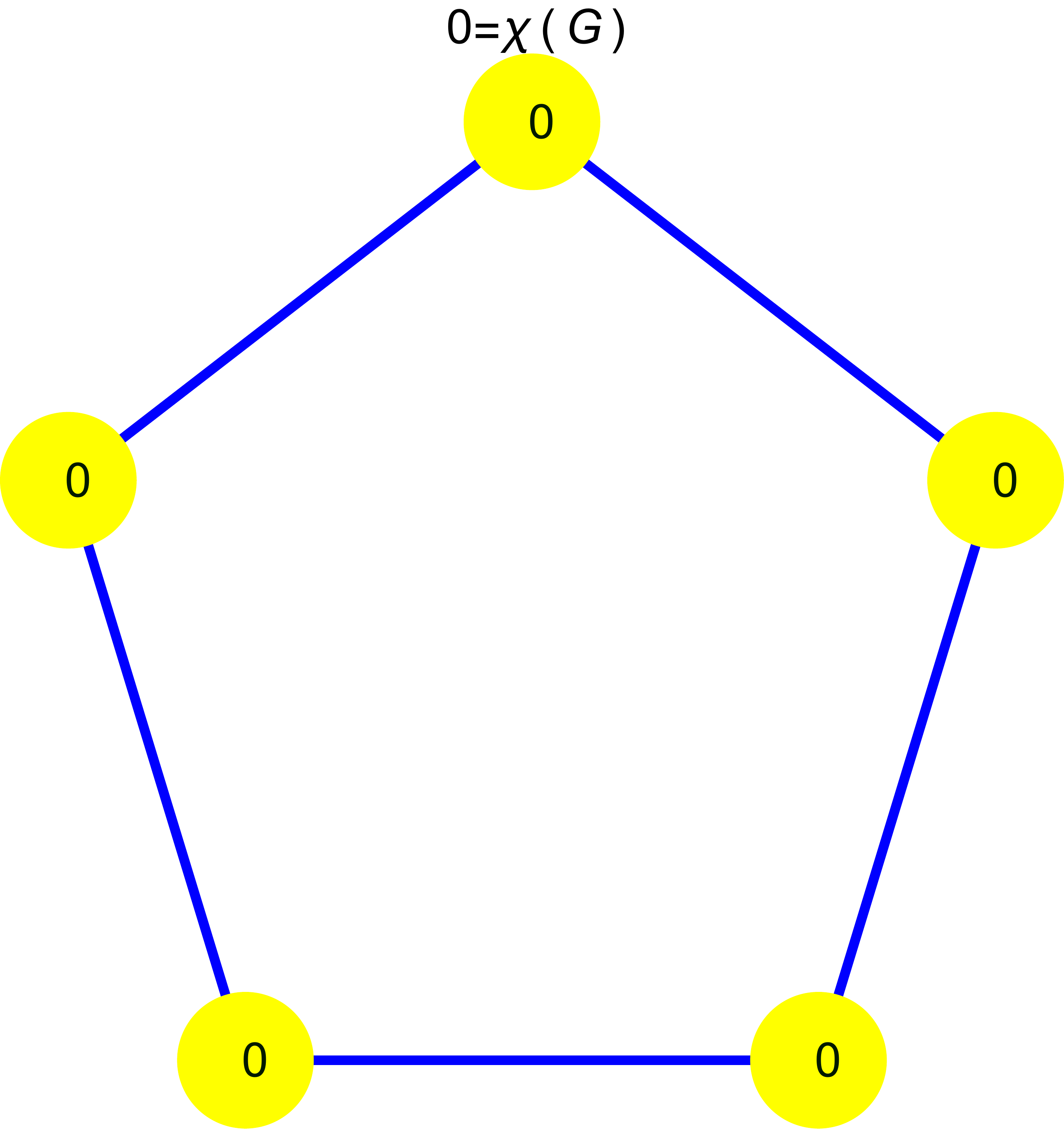}}
\scalebox{0.12}{\includegraphics{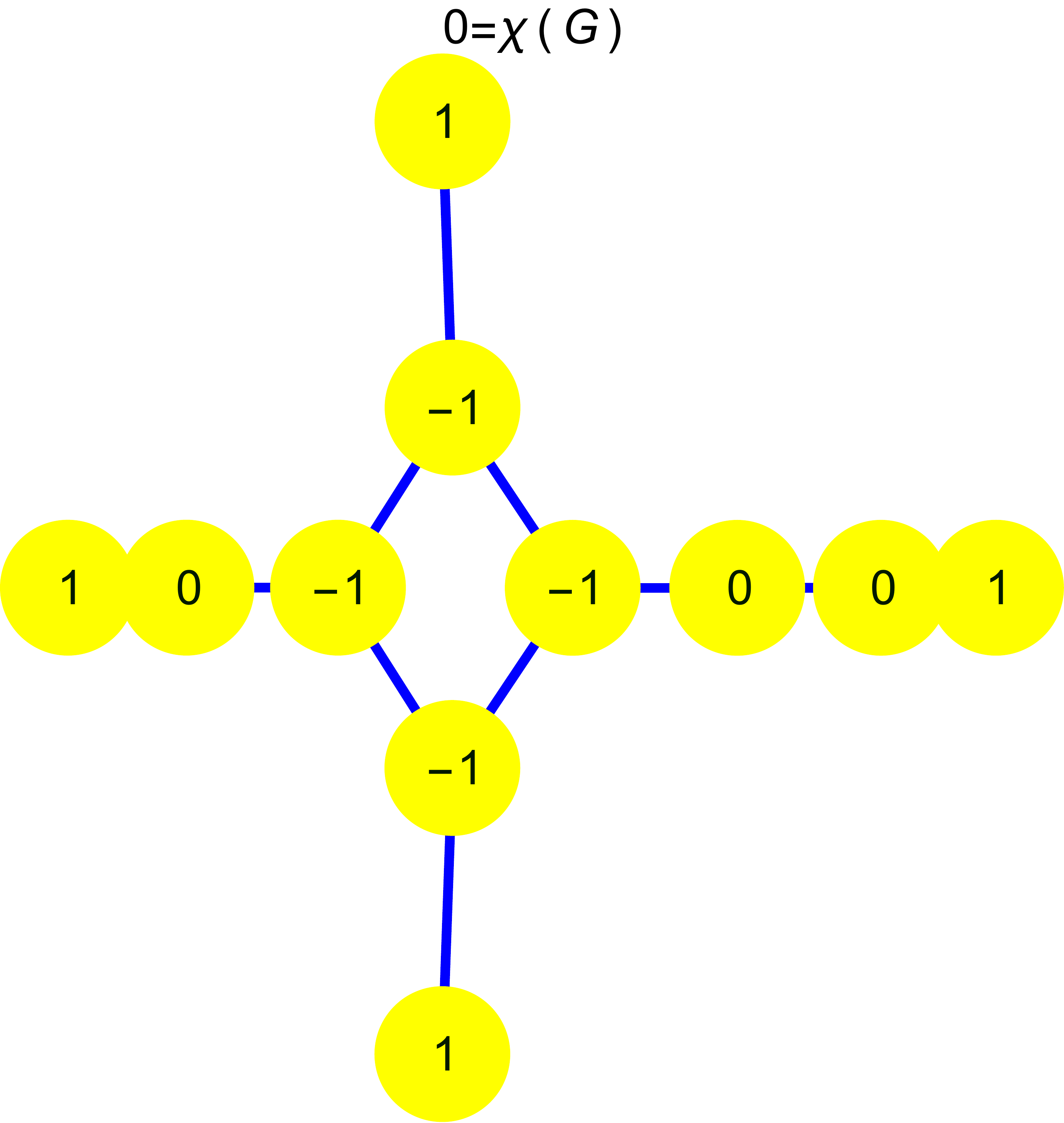}}
\scalebox{0.12}{\includegraphics{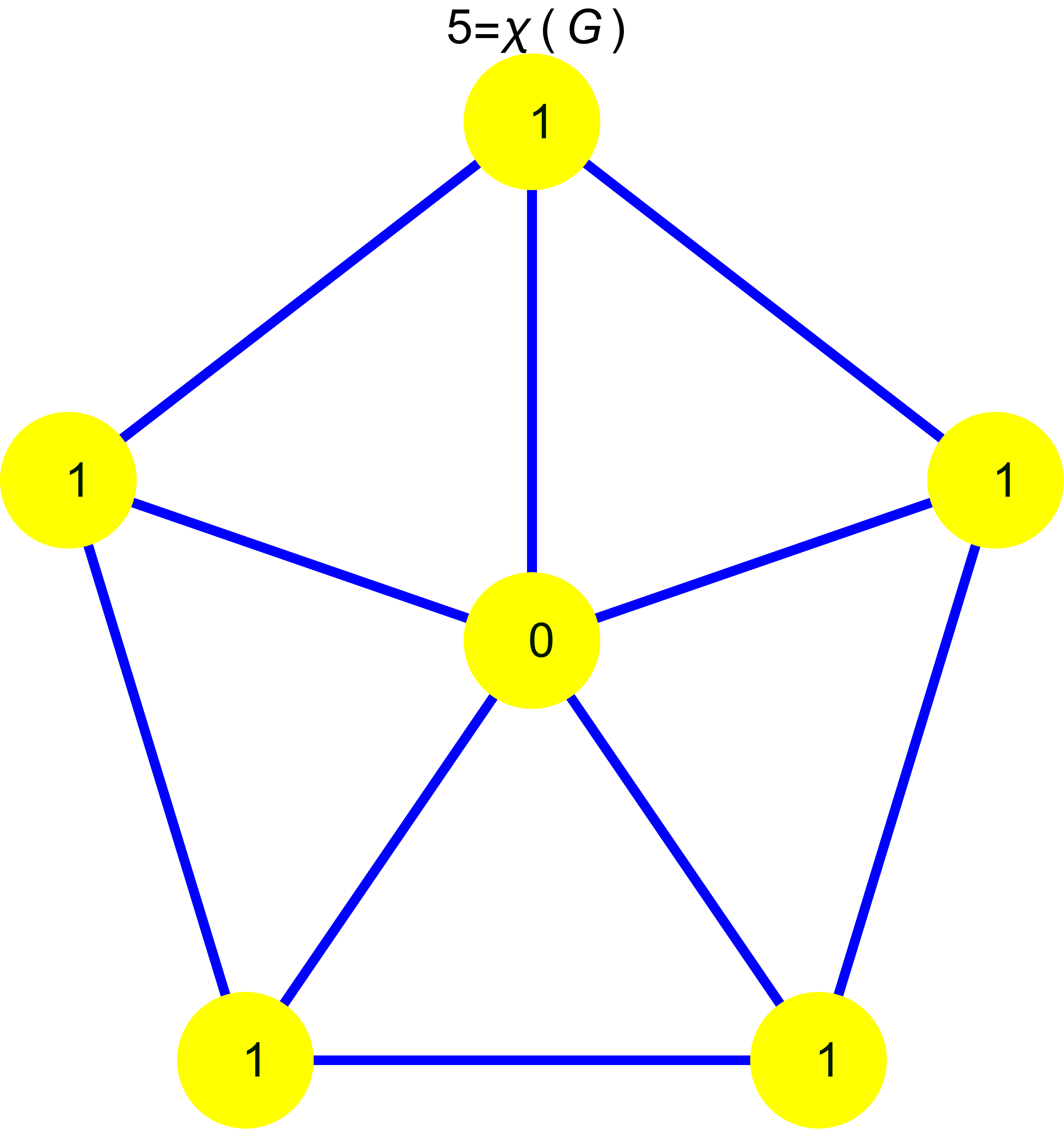}}
\scalebox{0.12}{\includegraphics{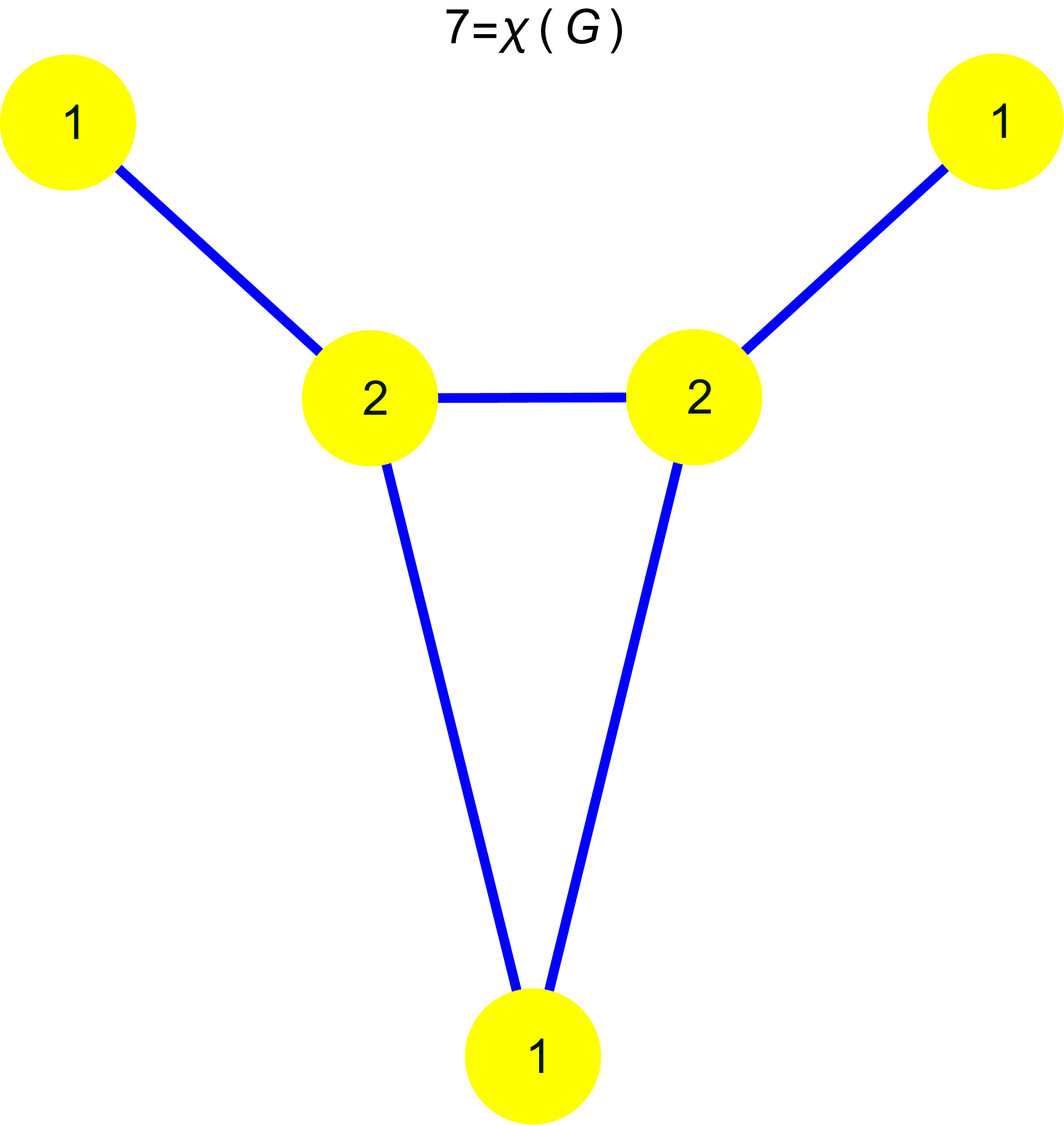}}
\scalebox{0.12}{\includegraphics{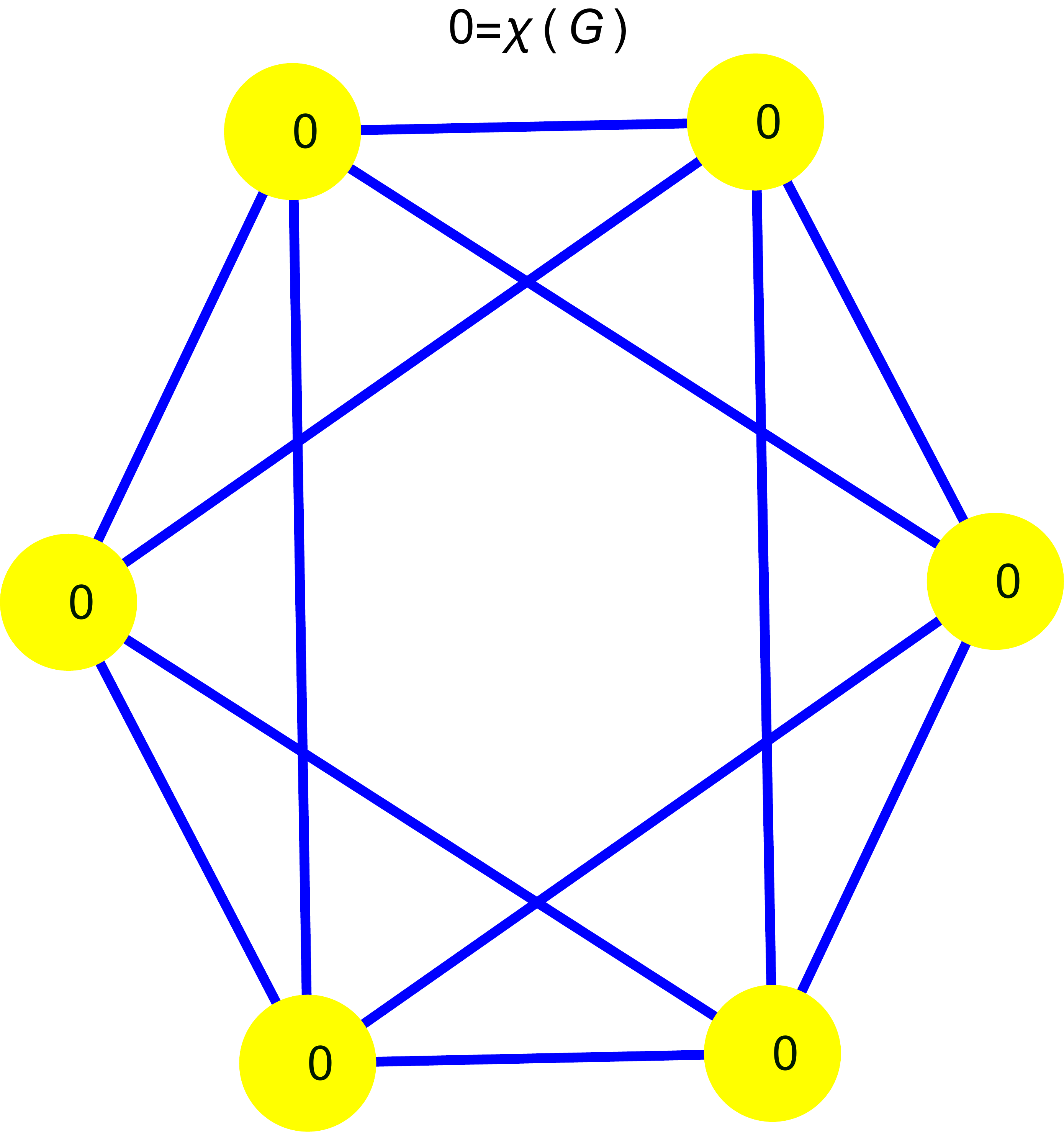}}
\scalebox{0.12}{\includegraphics{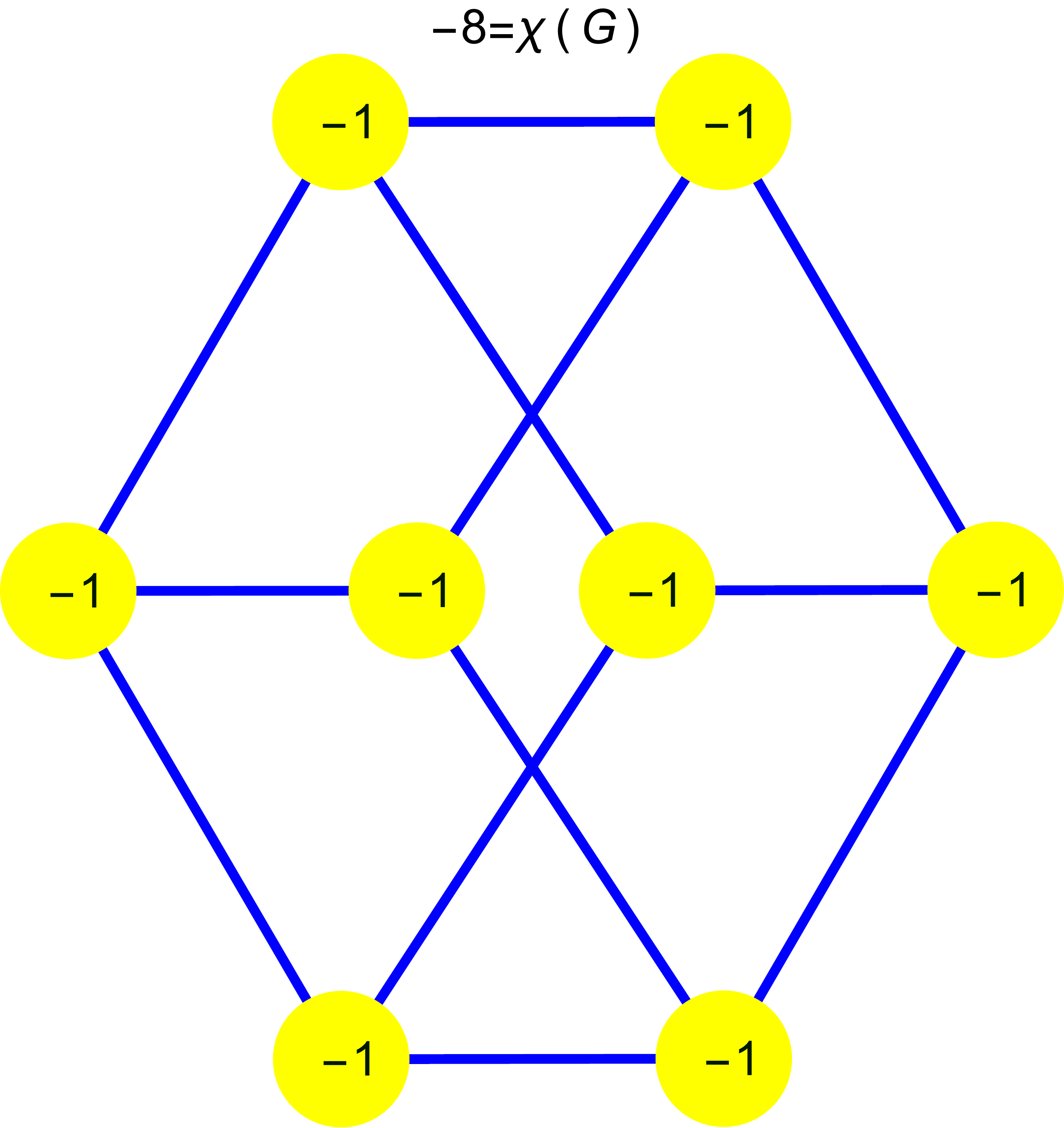}}
\scalebox{0.12}{\includegraphics{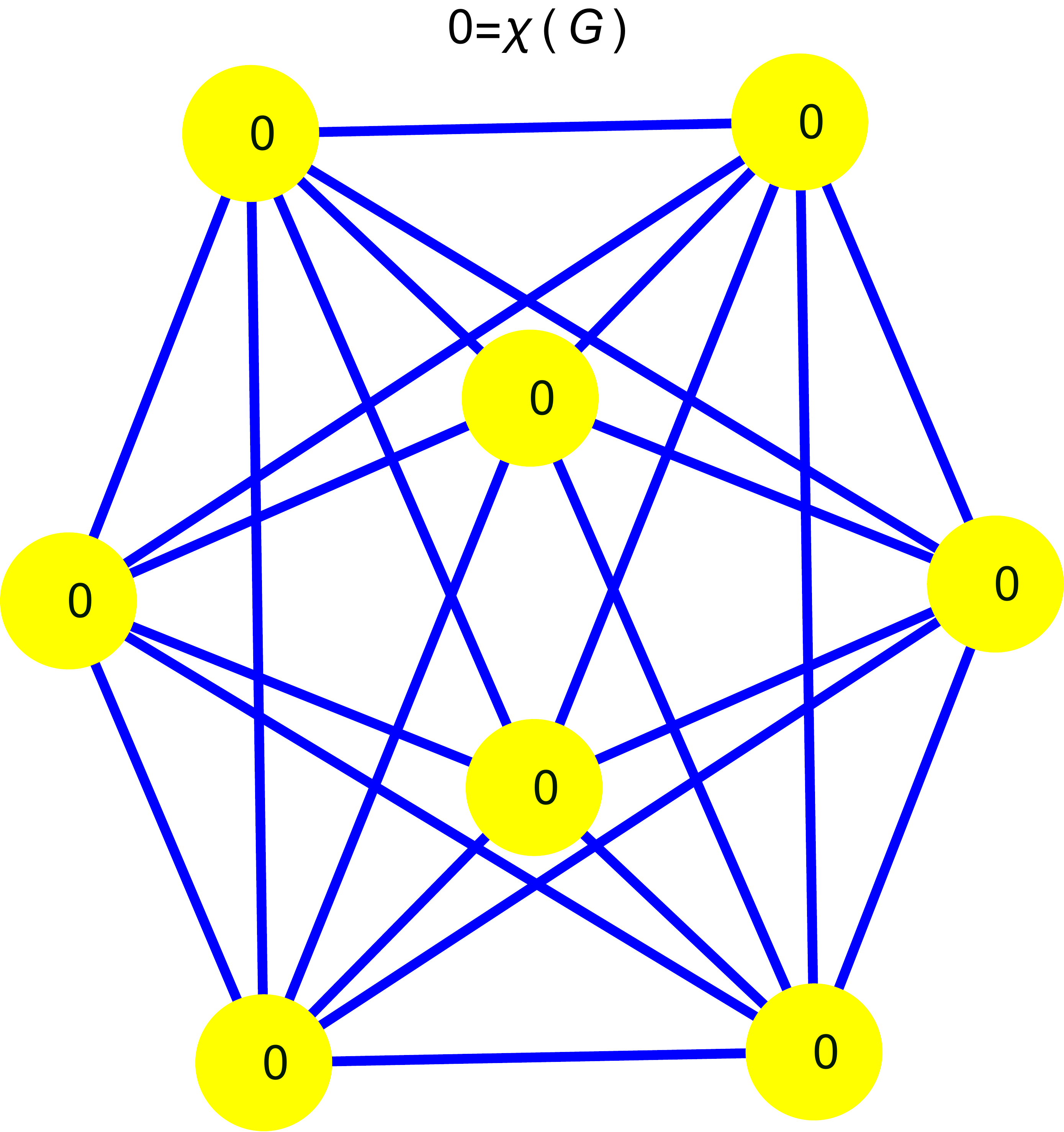}}
\scalebox{0.12}{\includegraphics{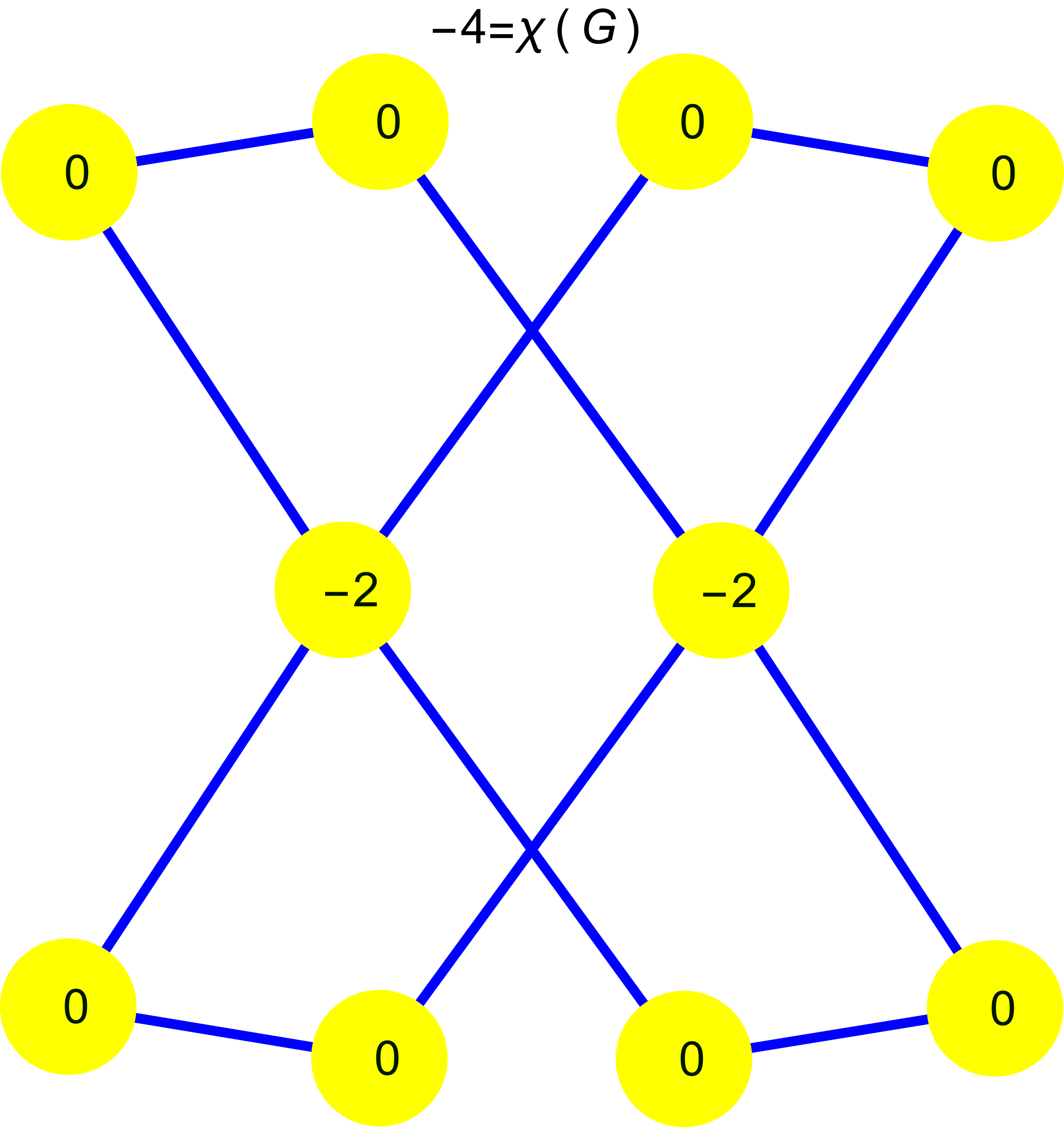}}
\scalebox{0.12}{\includegraphics{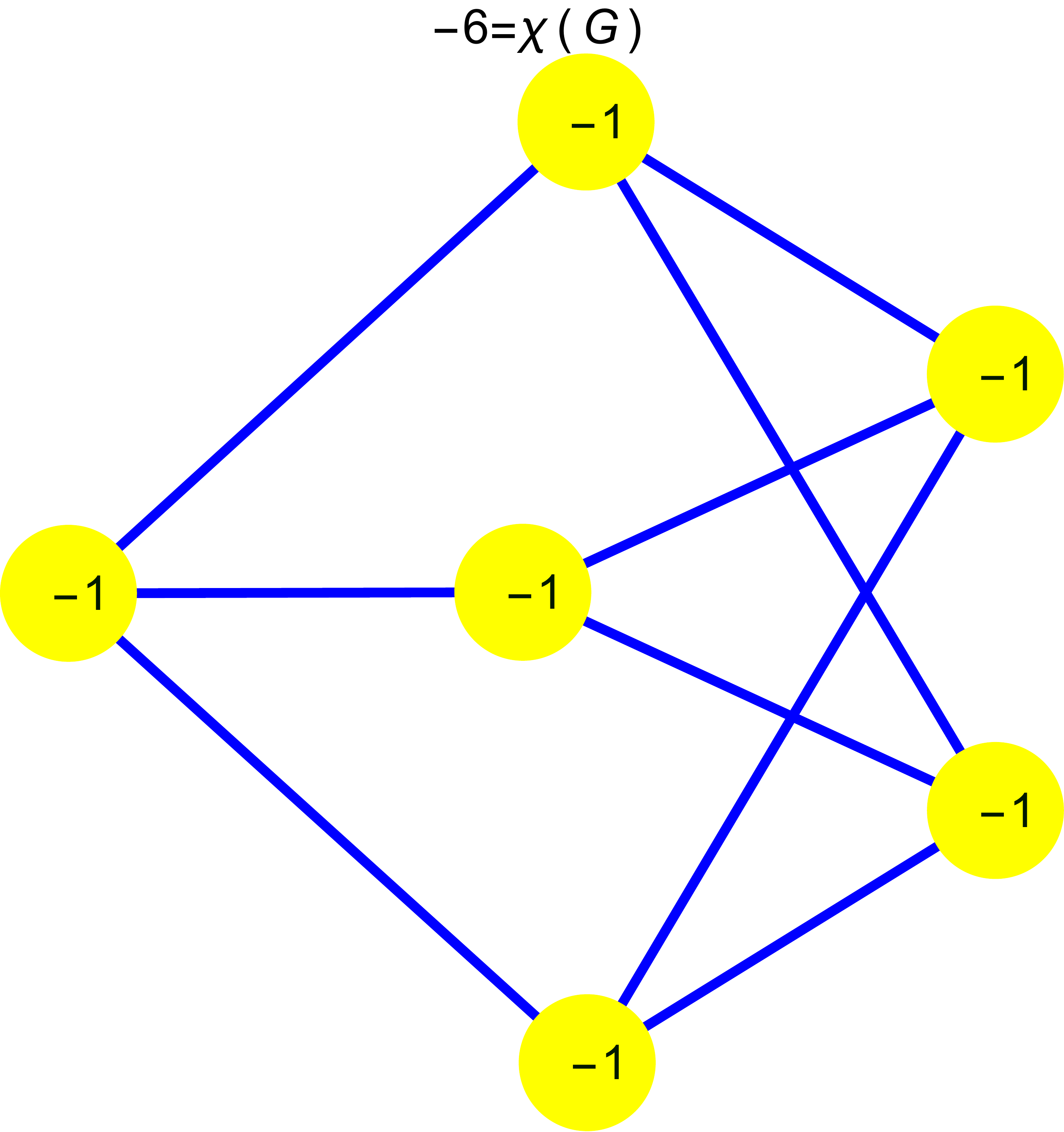}}
\scalebox{0.12}{\includegraphics{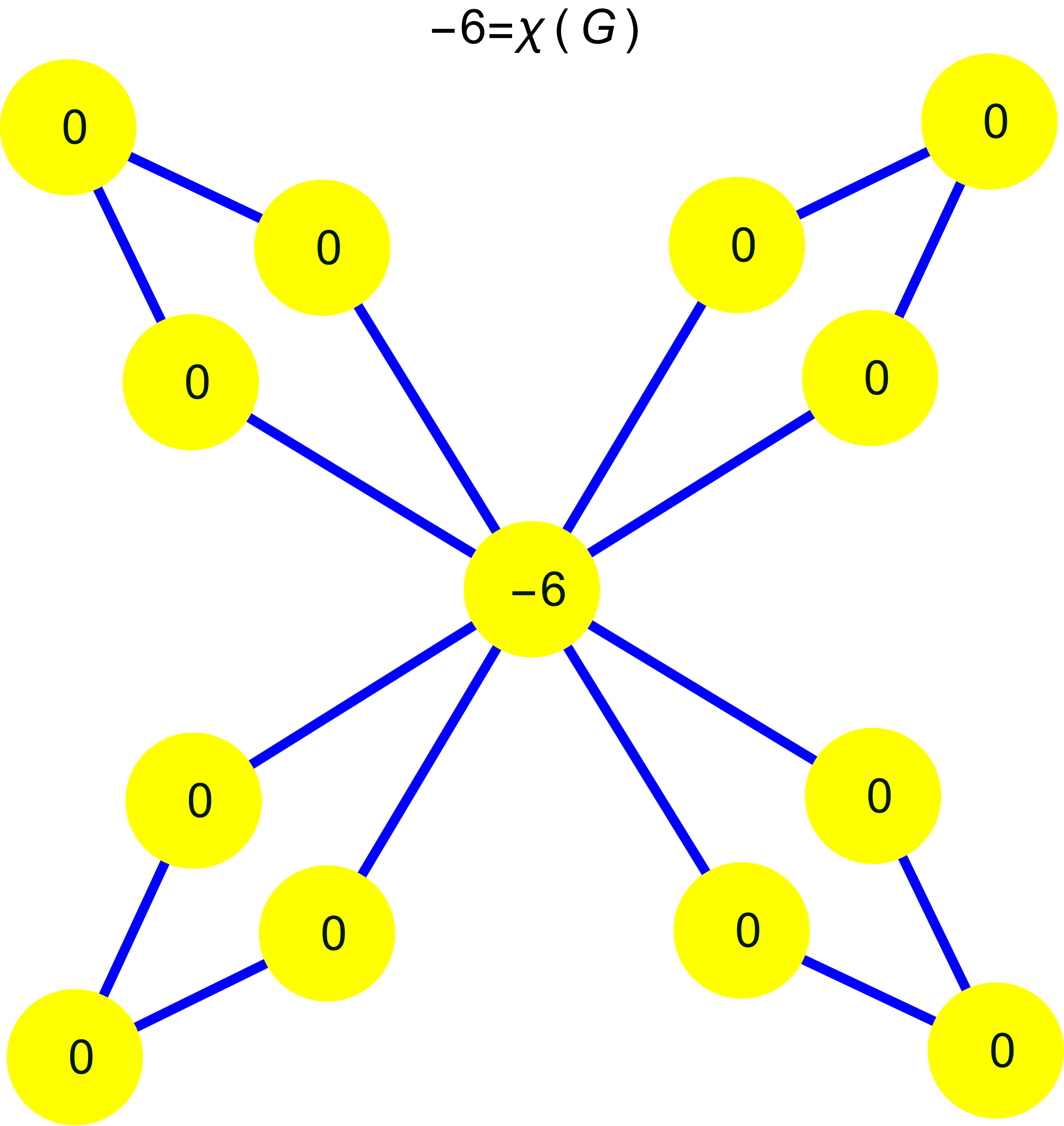}}
\scalebox{0.12}{\includegraphics{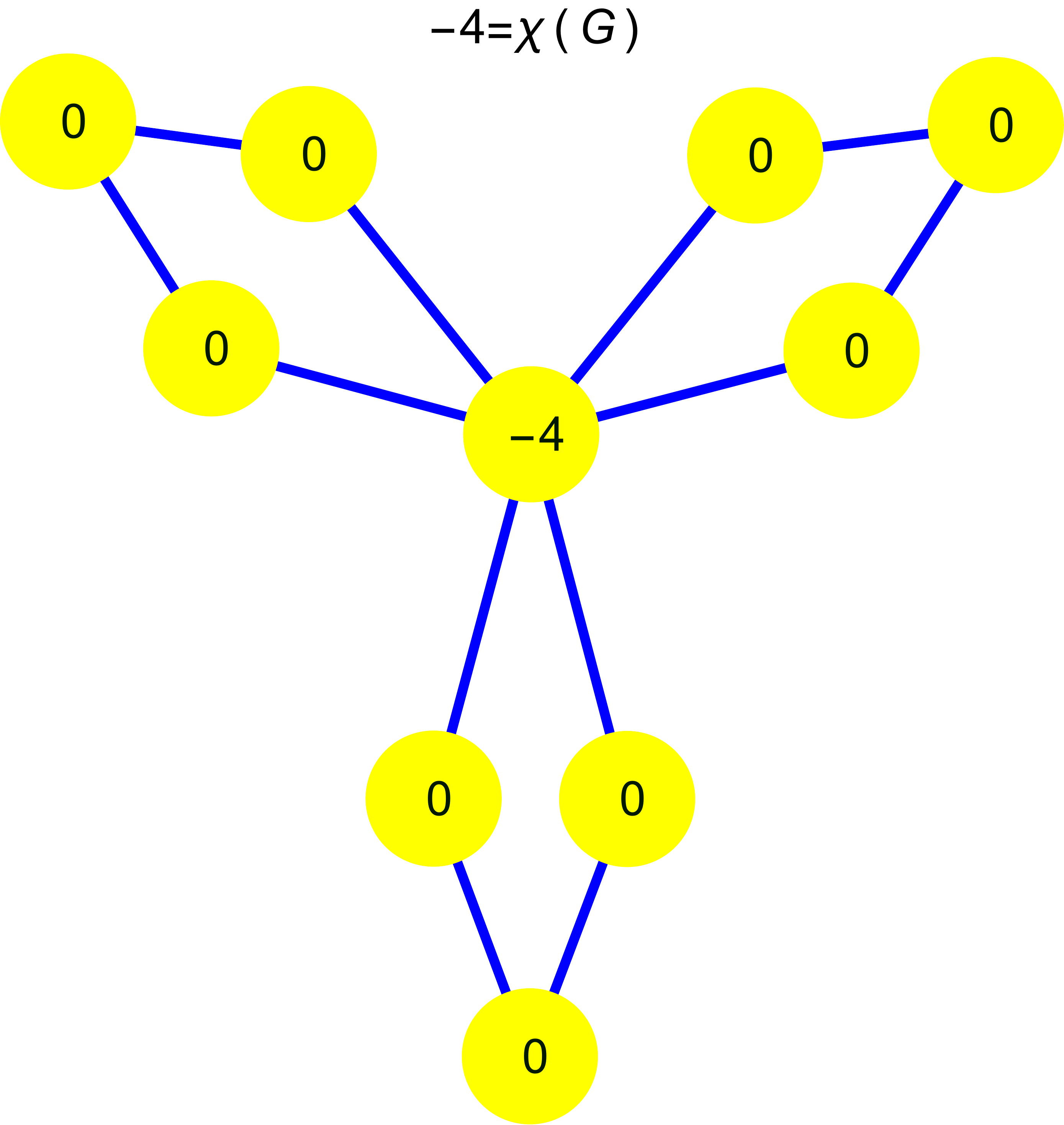}}
\scalebox{0.12}{\includegraphics{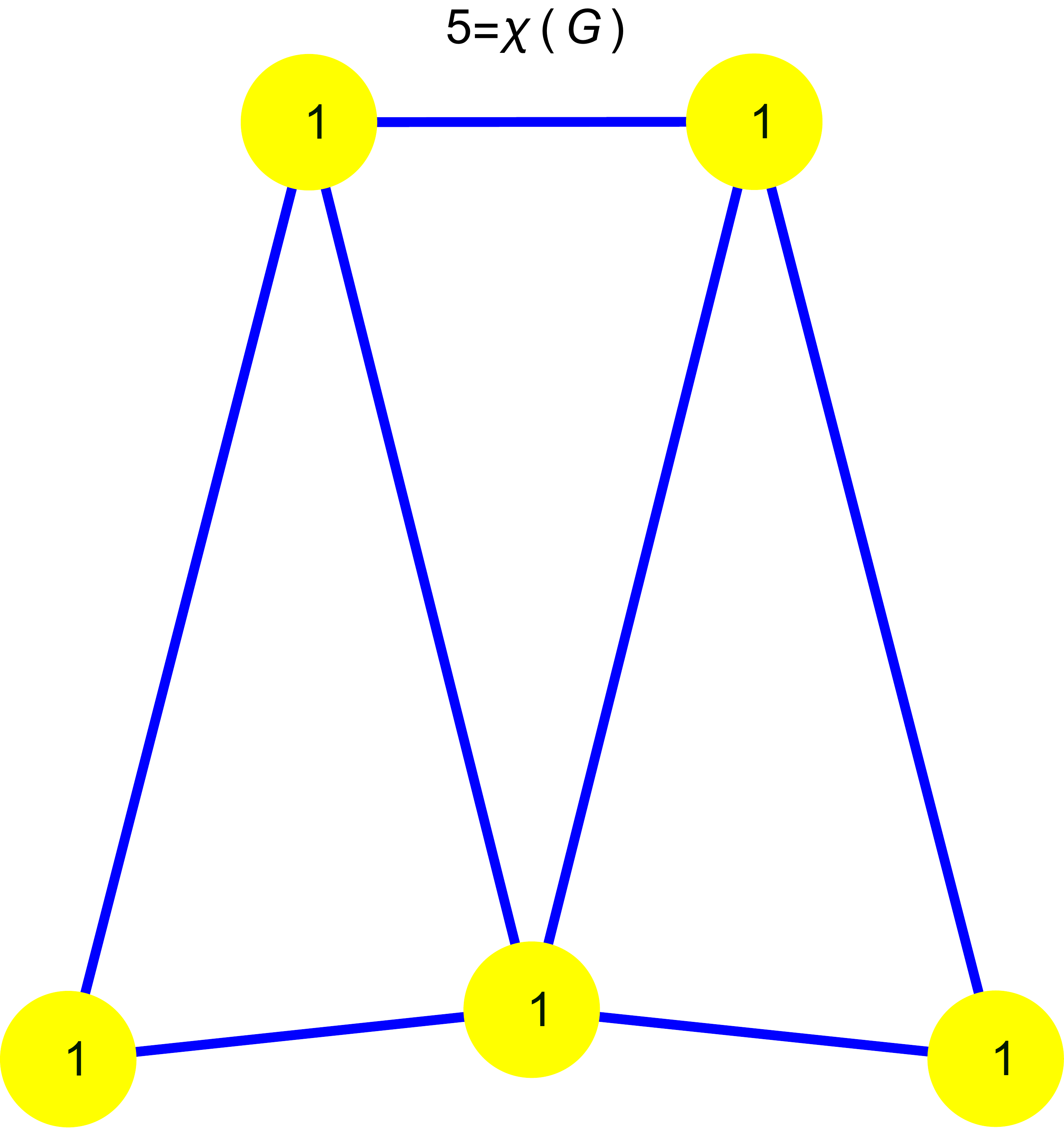}}
\scalebox{0.12}{\includegraphics{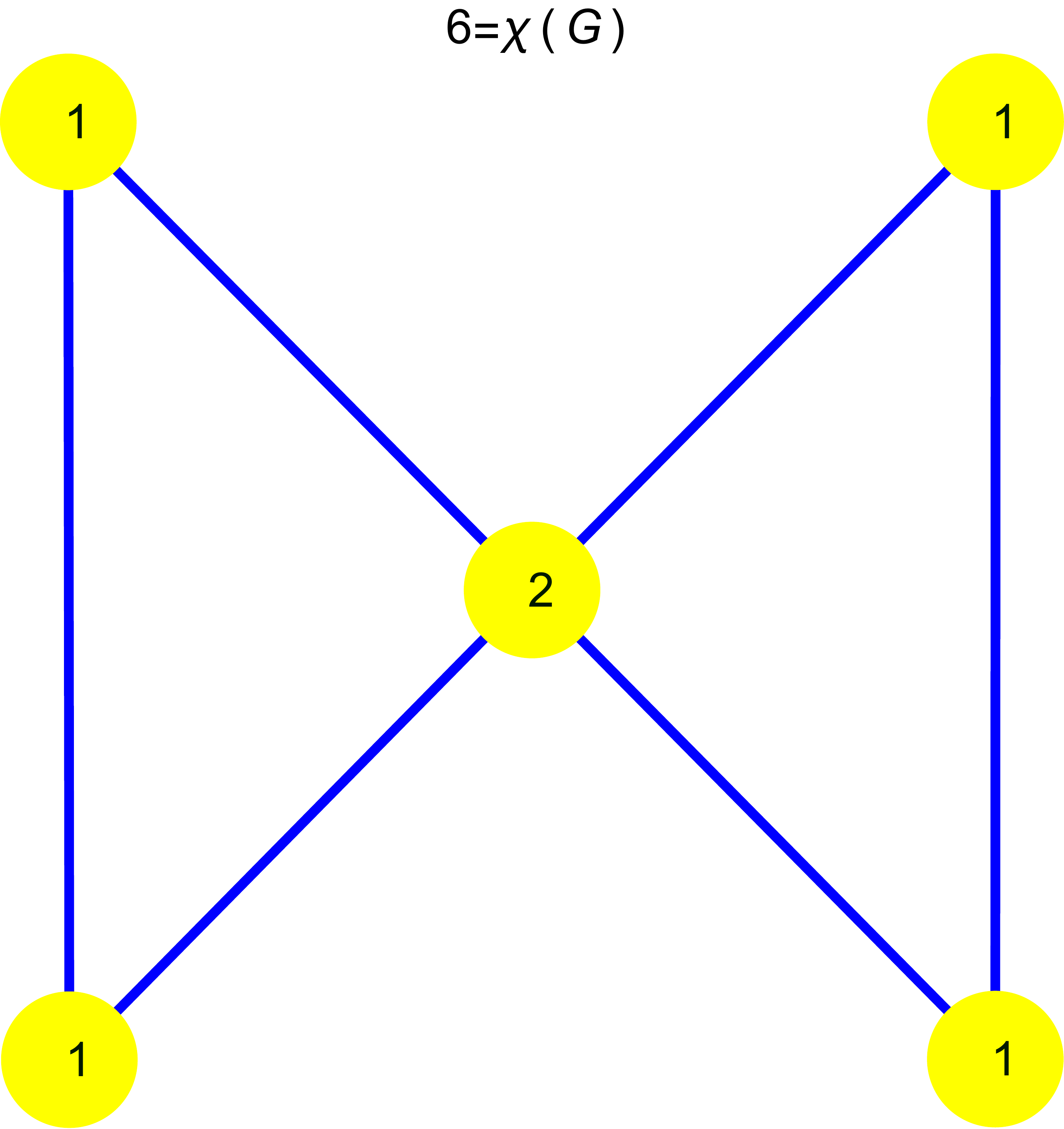}}
\caption{
Examples of graphs with linear Dehn-Sommerville curvatures for $X_{0}$. 
We use notation, where $X_{-1}$ is the Euler characteristic. 
The curvatures are zero on the $1$-graph $C_5$ and the
$2$-graph given as the $2$-sphere, the octahedron.
}
\end{figure}

\begin{figure}[!htpb]
\scalebox{0.12}{\includegraphics{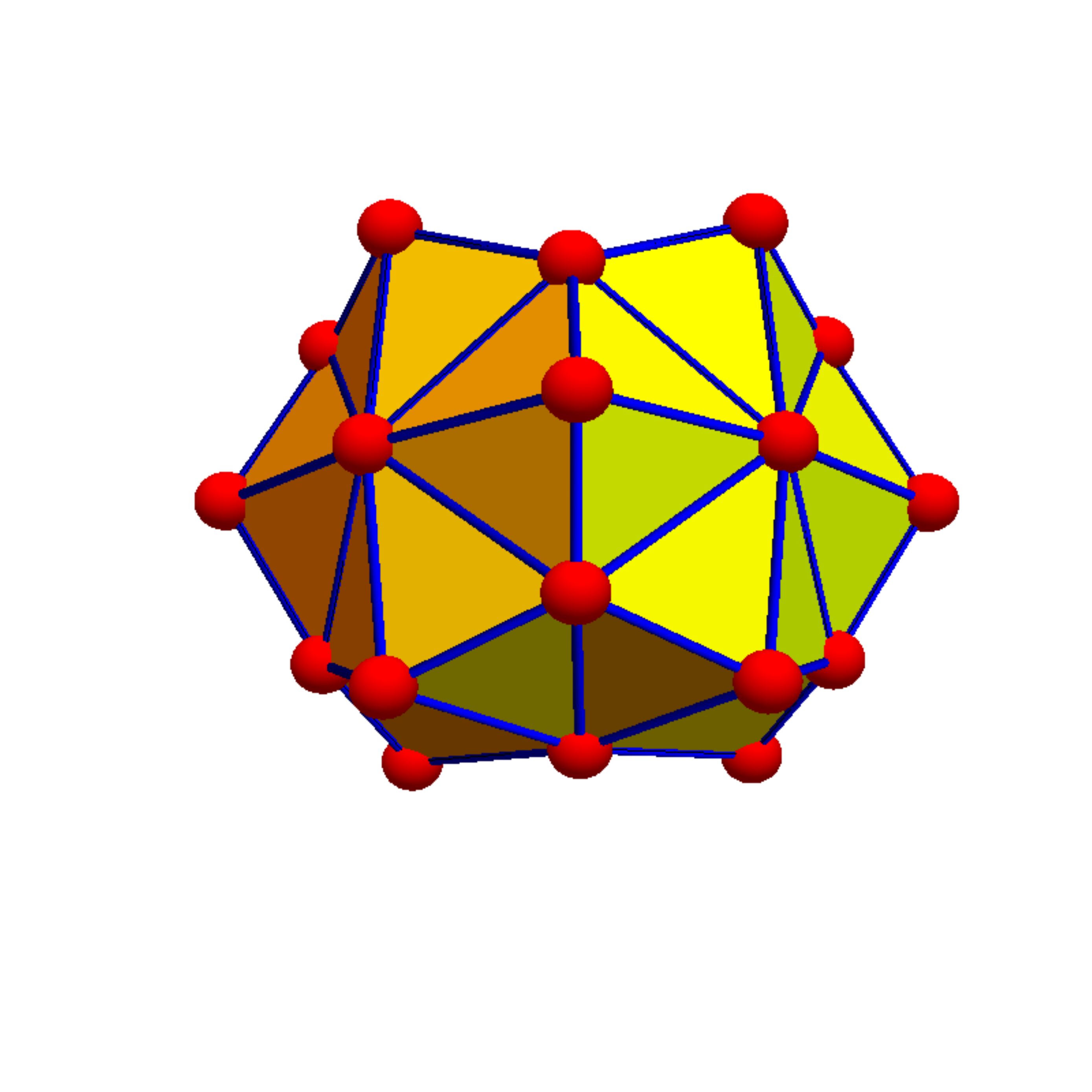}}
\scalebox{0.12}{\includegraphics{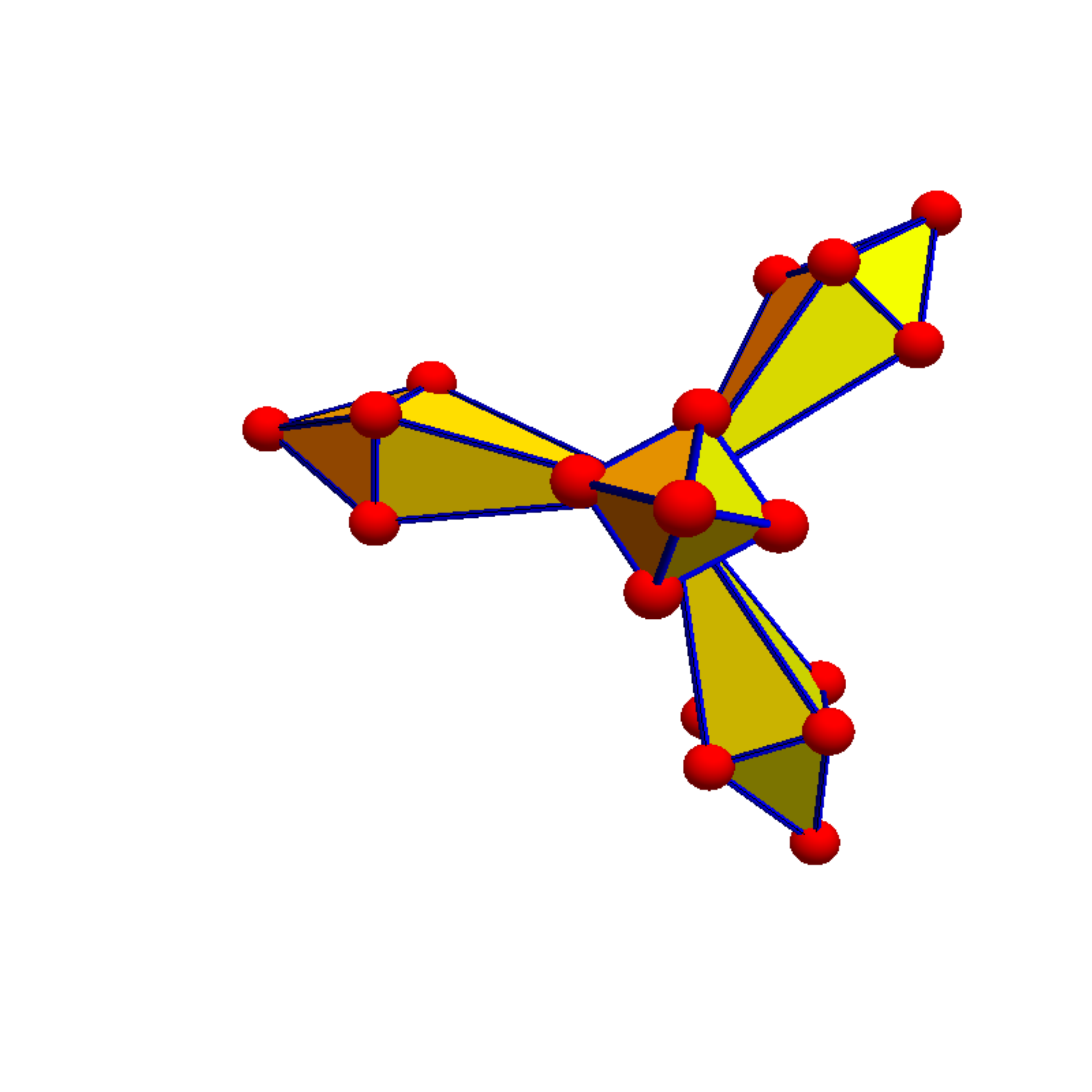}}
\scalebox{0.12}{\includegraphics{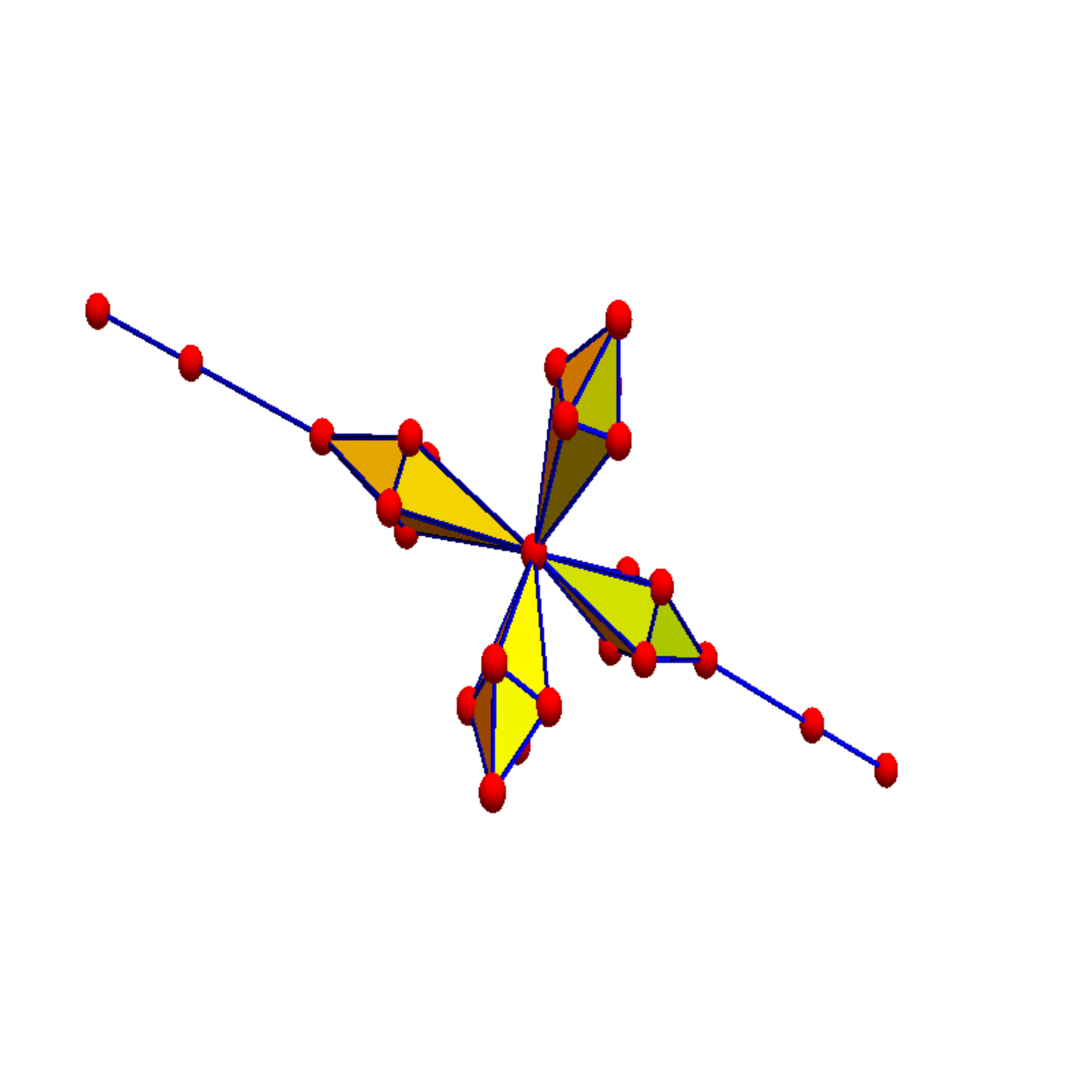}}
\scalebox{0.12}{\includegraphics{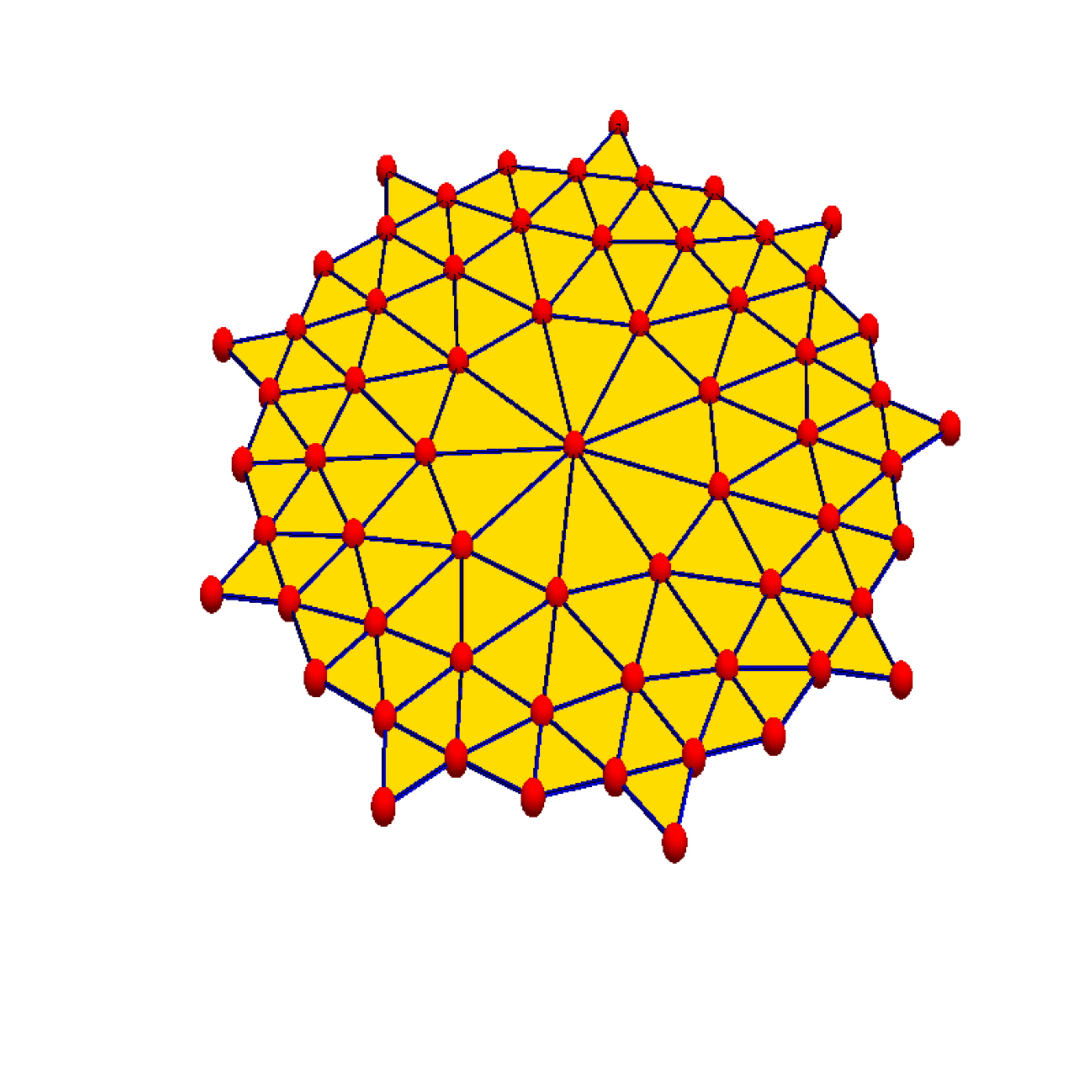}}
\scalebox{0.12}{\includegraphics{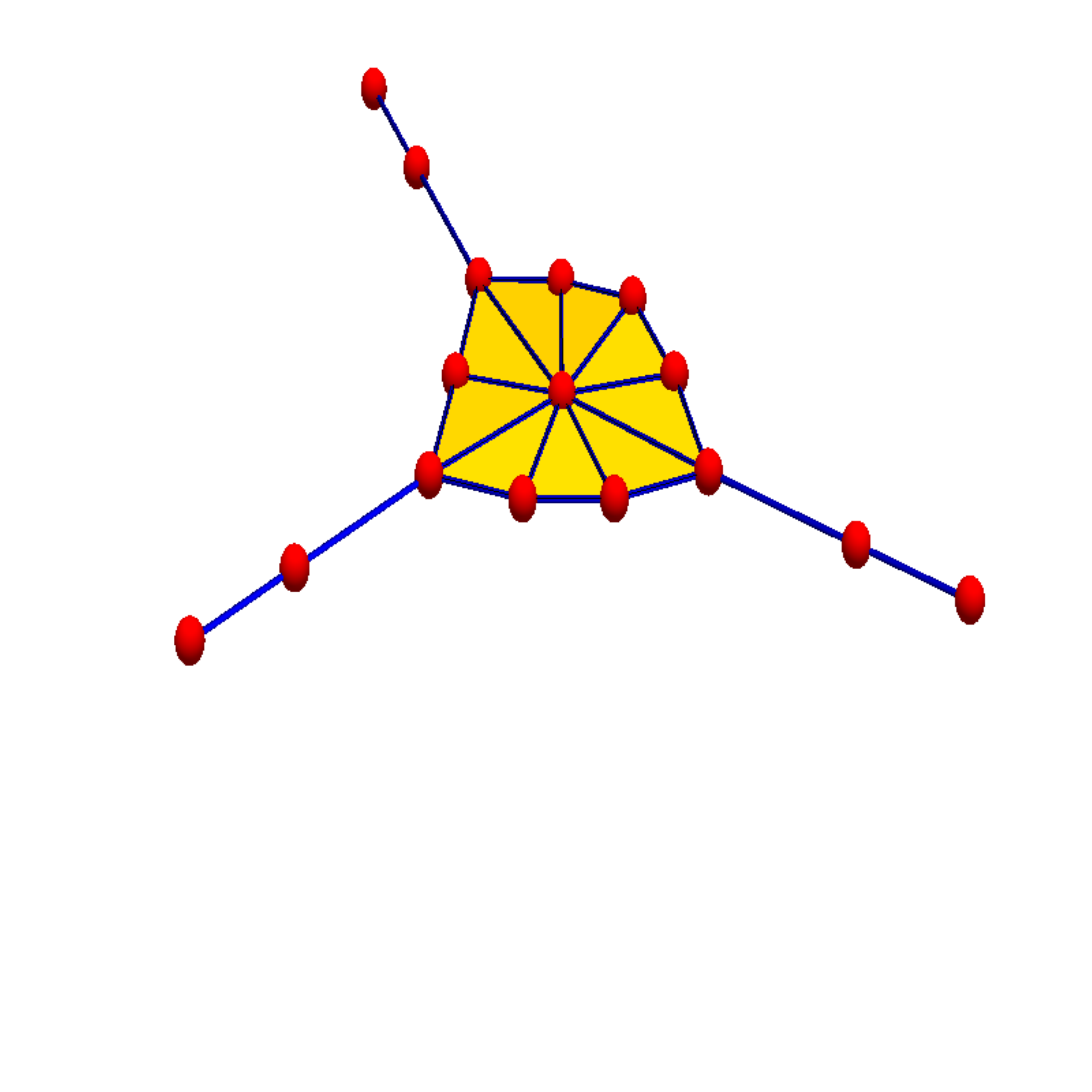}}
\scalebox{0.12}{\includegraphics{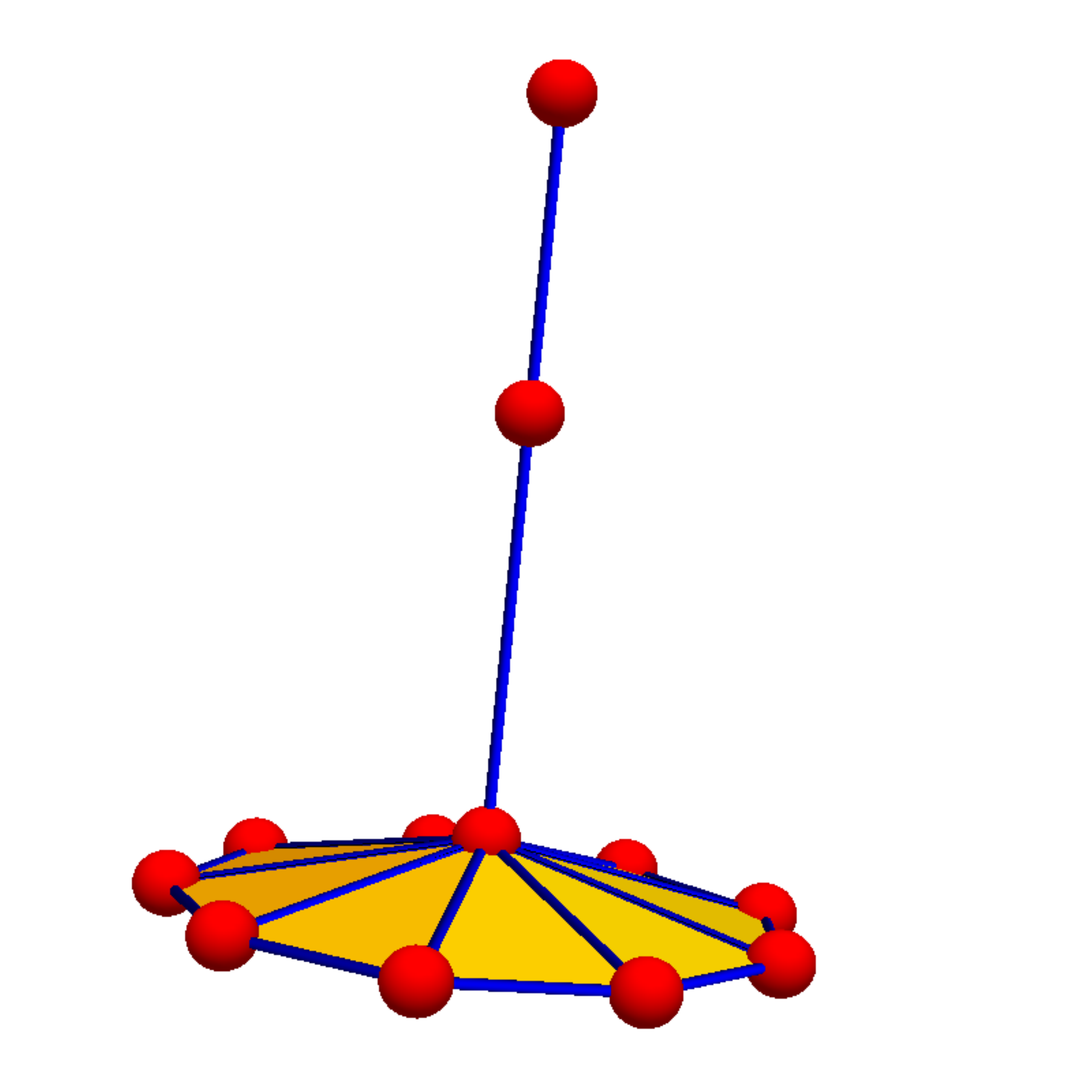}}
\scalebox{0.12}{\includegraphics{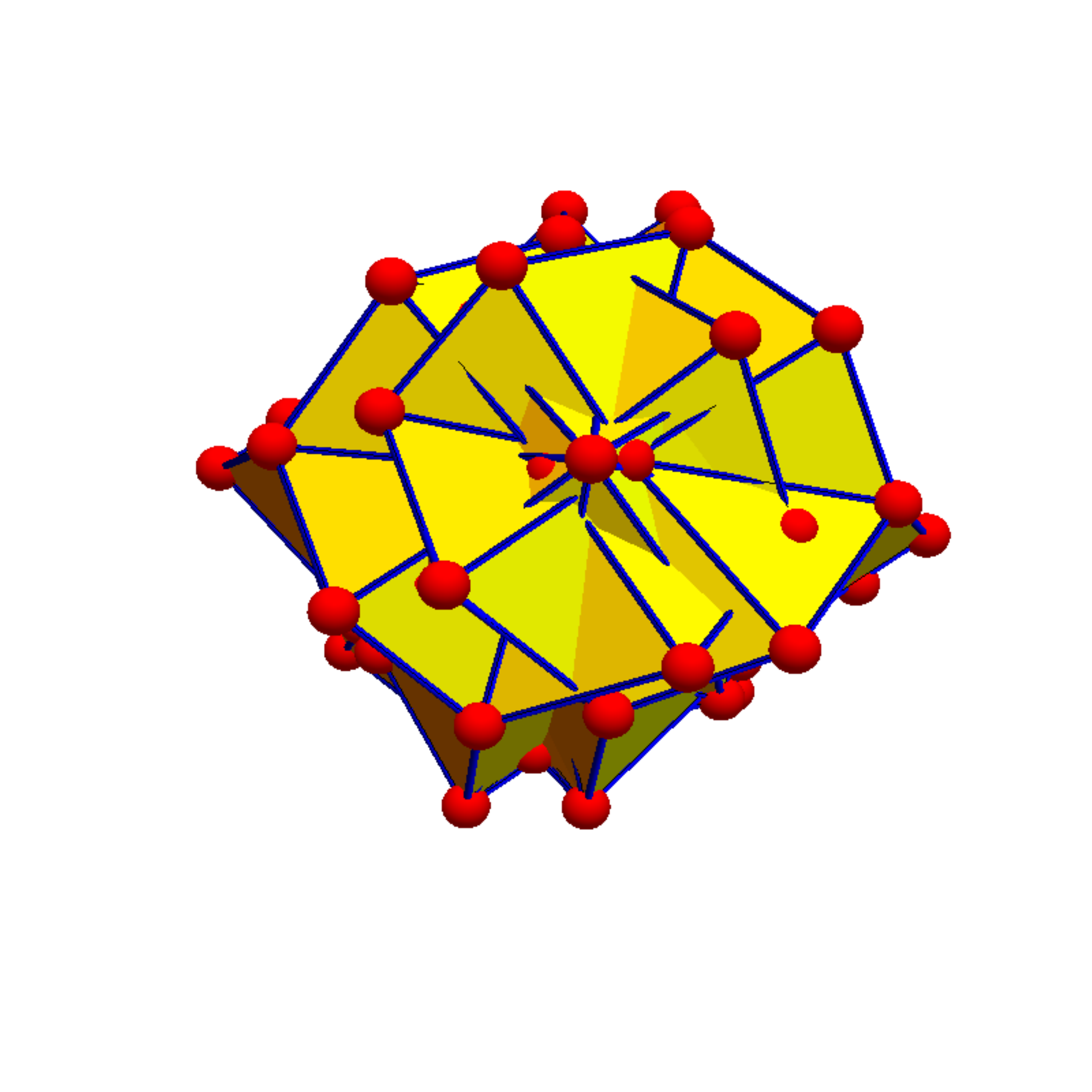}}
\scalebox{0.12}{\includegraphics{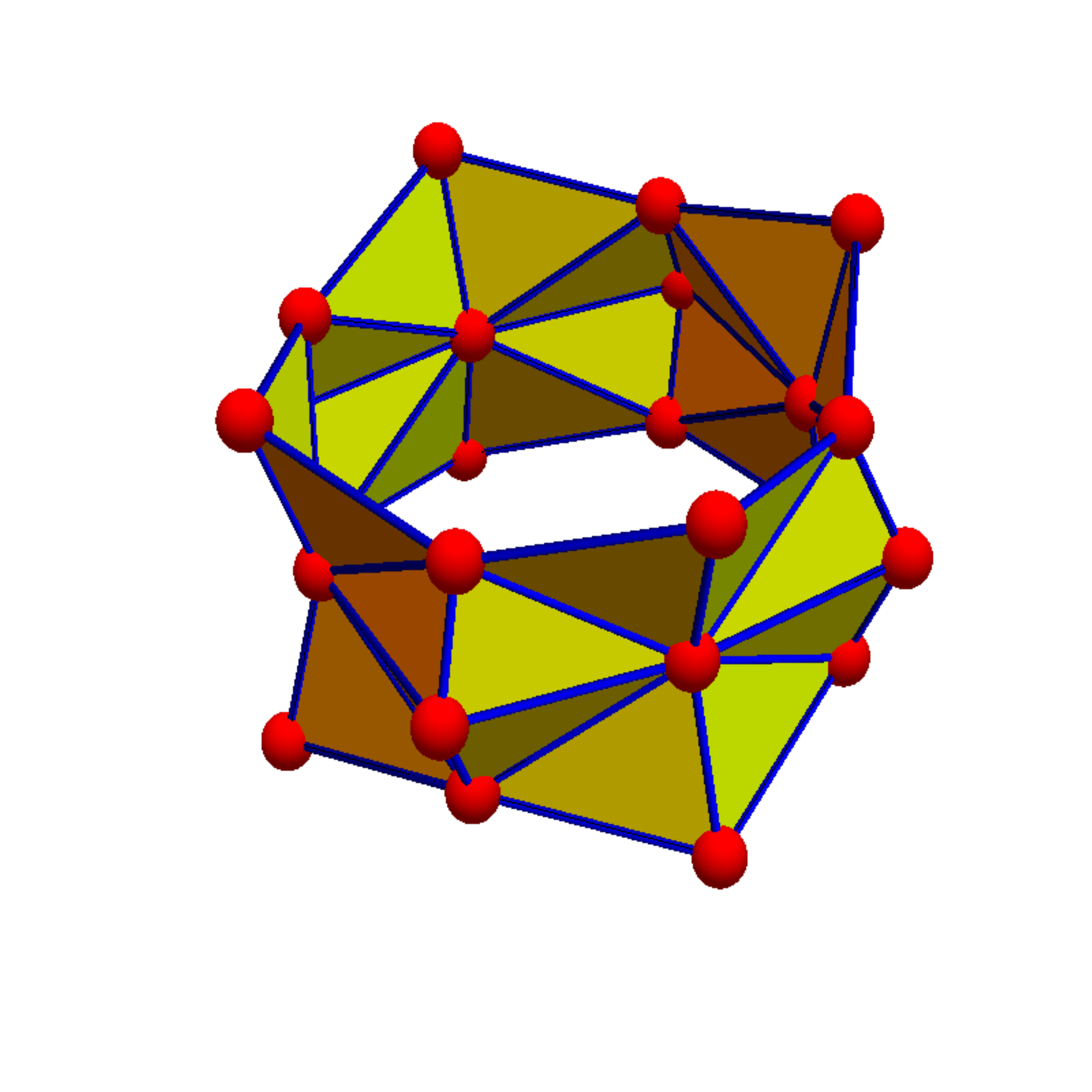}}
\scalebox{0.12}{\includegraphics{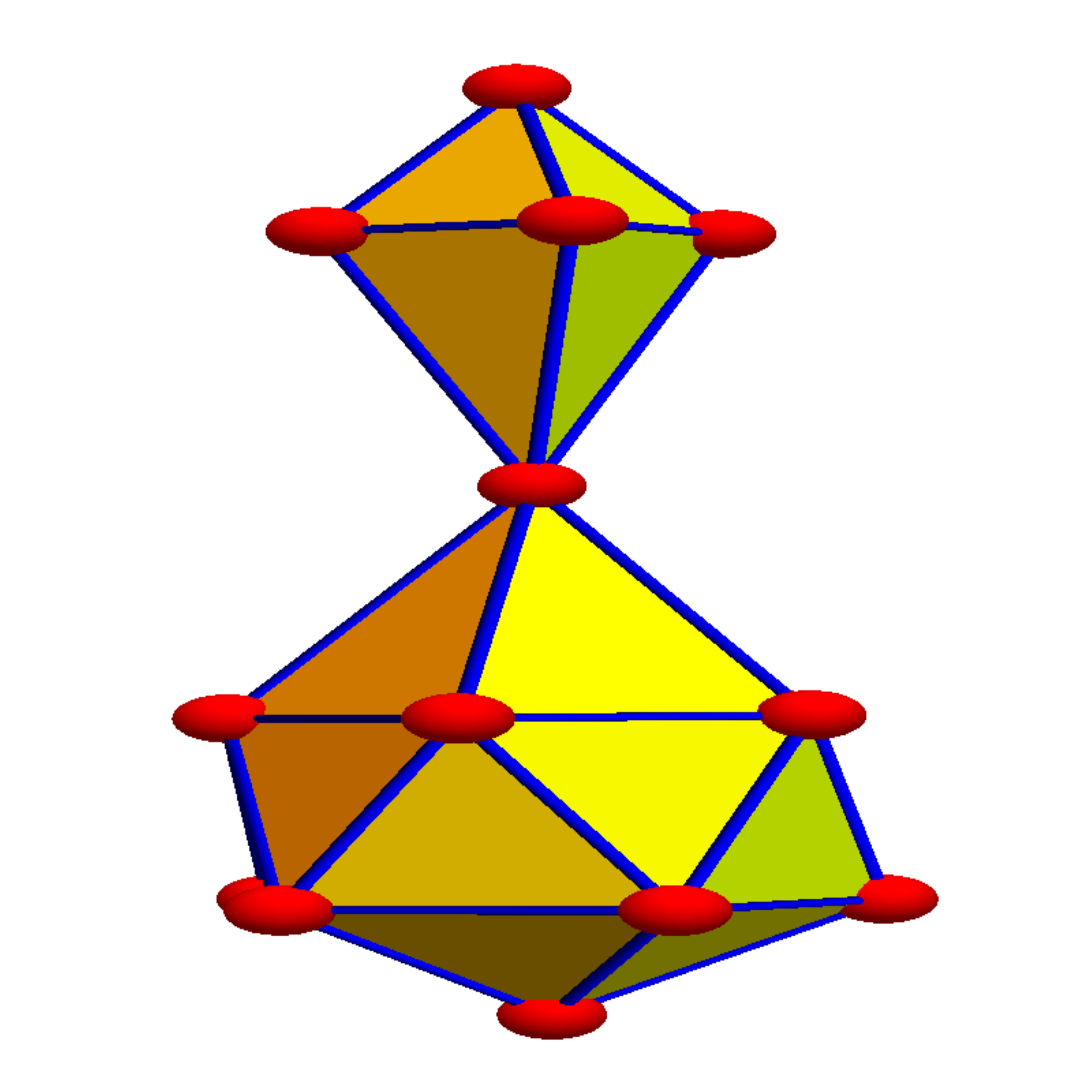}}
\scalebox{0.12}{\includegraphics{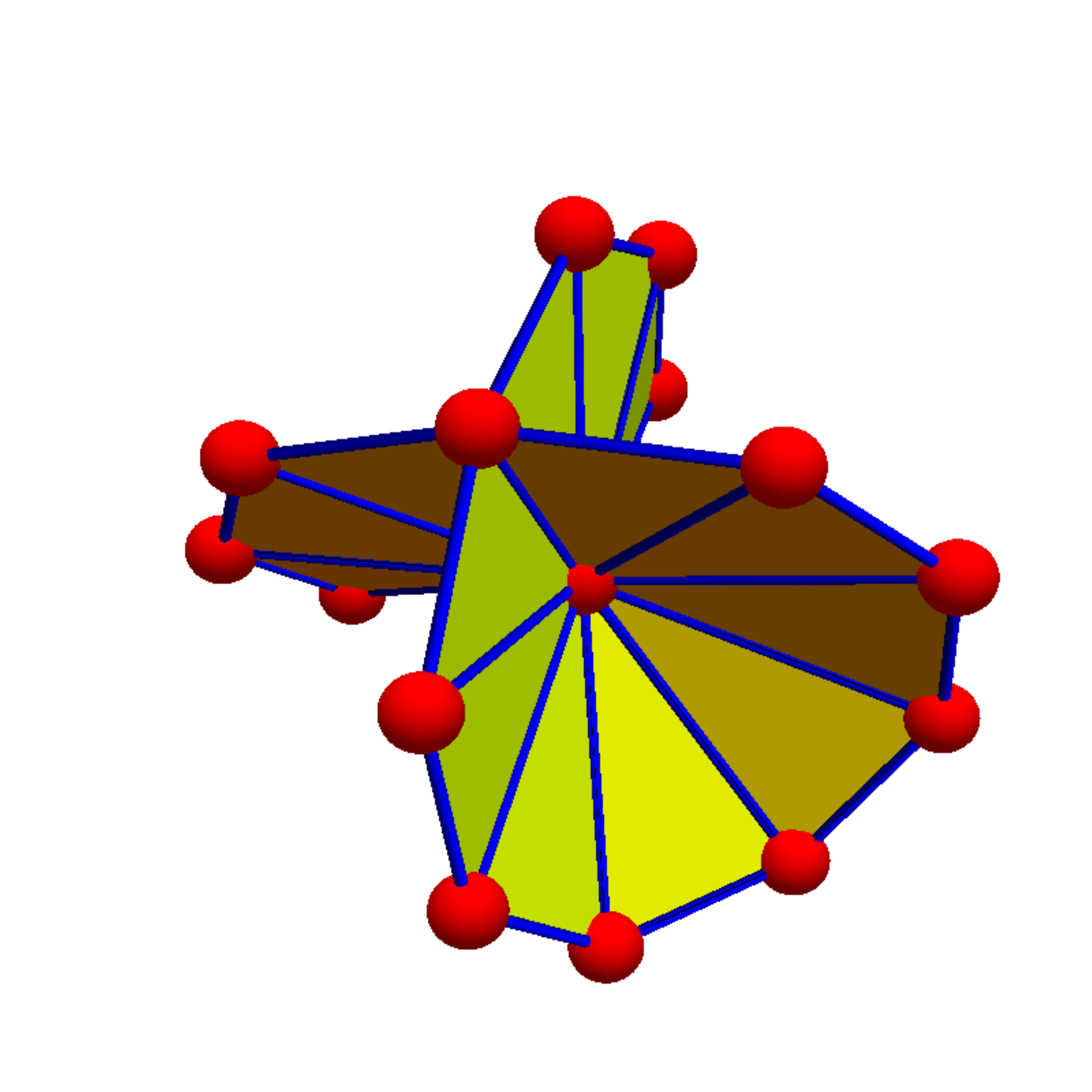}}
\scalebox{0.12}{\includegraphics{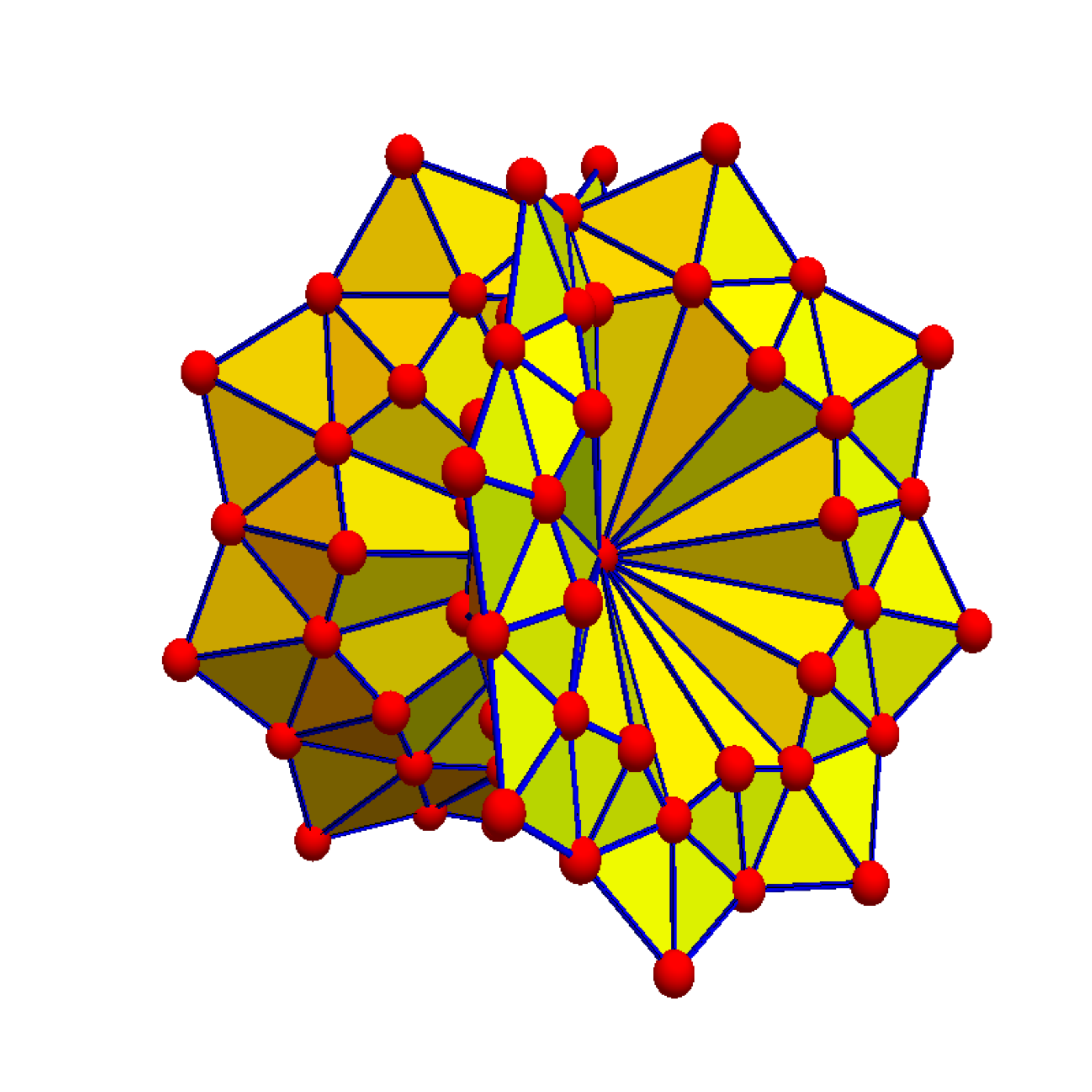}}
\scalebox{0.12}{\includegraphics{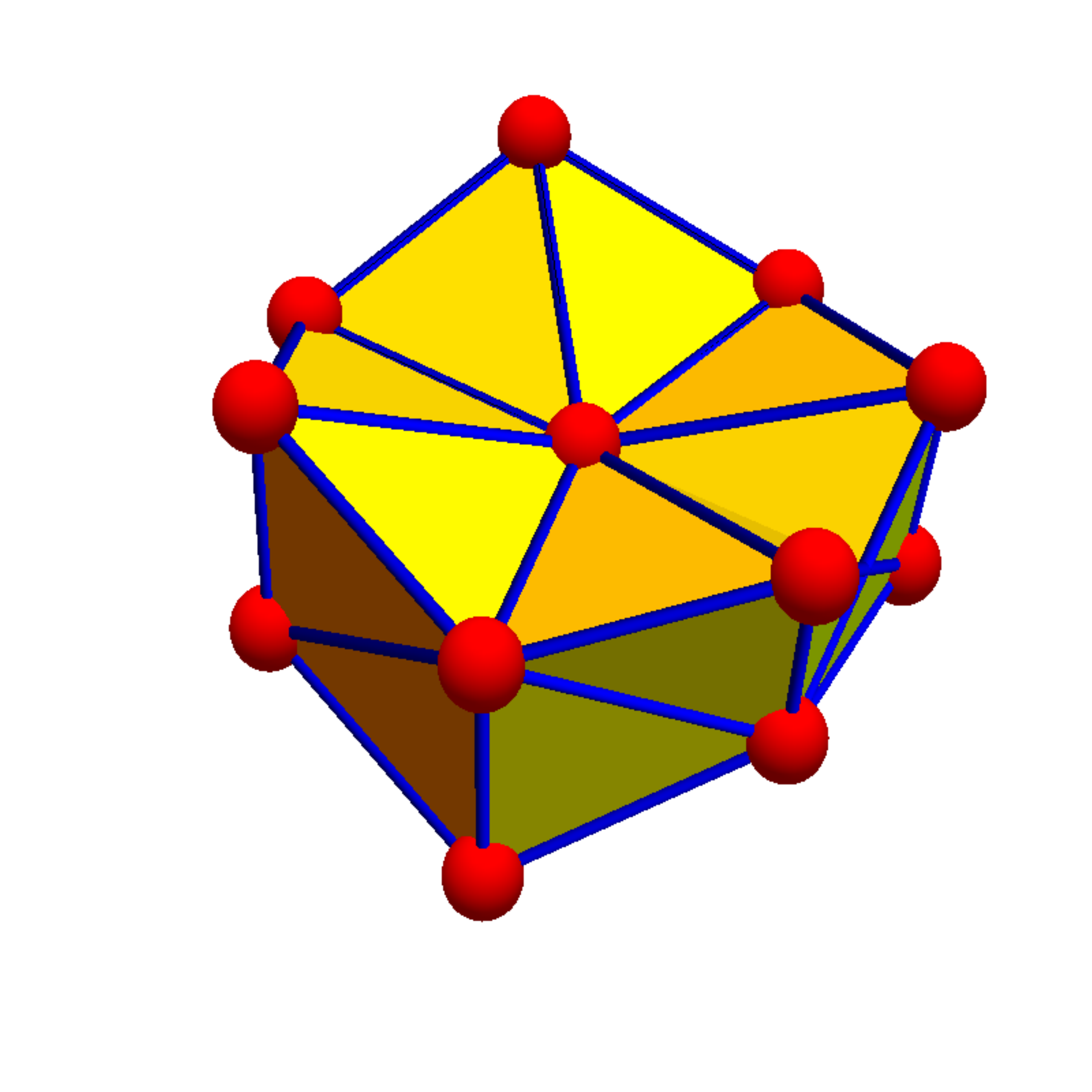}}
\scalebox{0.12}{\includegraphics{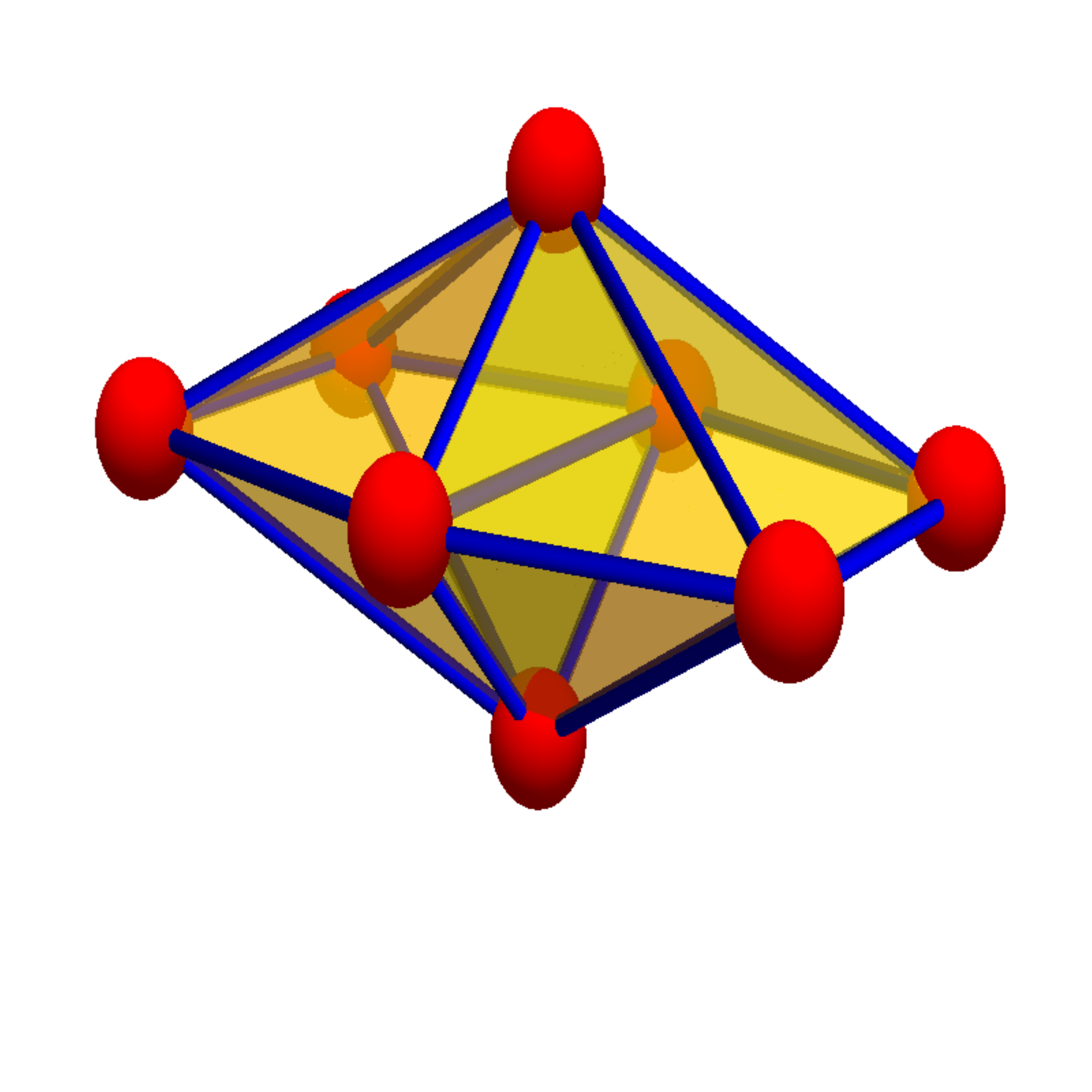}}
\scalebox{0.12}{\includegraphics{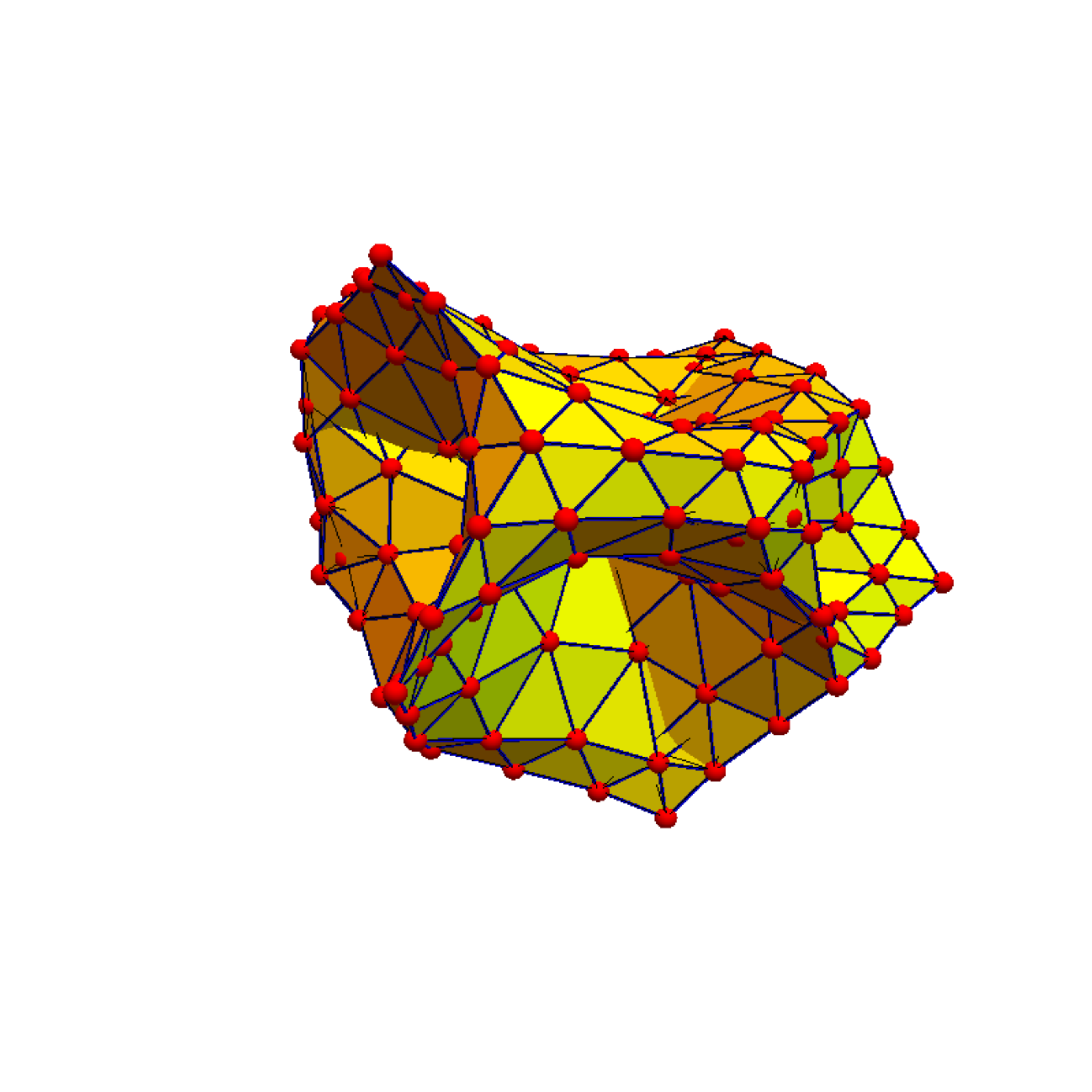}}
\scalebox{0.12}{\includegraphics{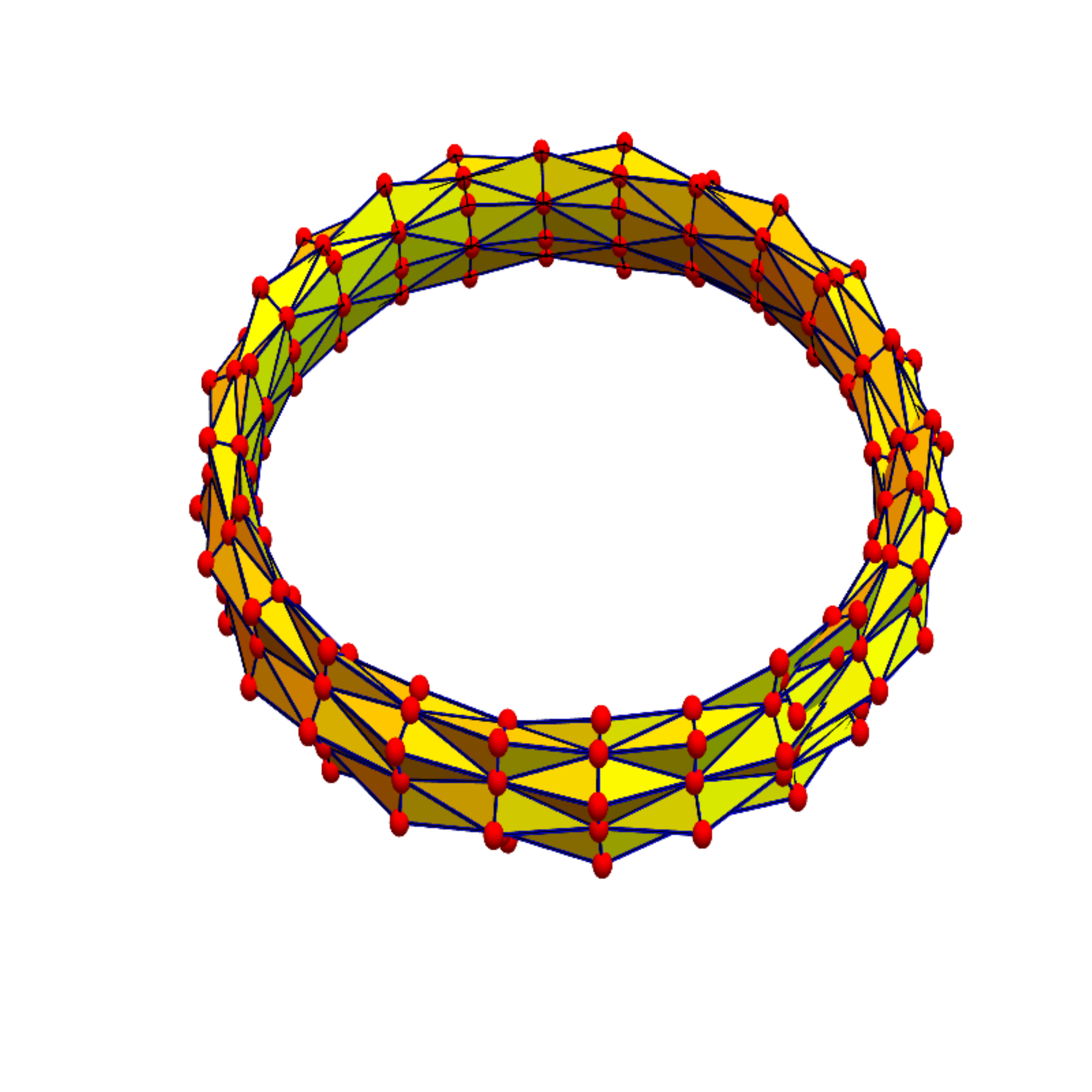}}
\caption{
{\bf 1.} A 3-ball $\omega(G)=-1$. 
{\bf 2.} Bouquet of 2 spheres $\omega(G)=5$. 
{\bf 3.} Dito with hairs $\omega(G)=1$. 4. A disk with $\omega(G)=1$
{\bf 5.} A disk with hairs with $\omega(G)=1$. 6. A flower with $\omega(G)=-1$
{\bf 7.} A cylinder $\omega(G)=0$, 
{\bf 8.} An 2 sphere bouquet $\omega(G)=3$,
{\bf 9.}  A sphere shell $\omega(G)=-2$.
{\bf 10.} A pyramid over figure 8, $\omega(G)=-3$,
{\bf 11.} A figure 8 suspension $\omega(G)=1$,
{\bf 12.} A suspension of a Moebius strip $\omega(G)=2$
{\bf 13.} A fat figure eight suspension $\omega(G)=3$
{\bf 14.} A dunce hat is homotopic to $K_1$ but not contractible. $\omega(G)=1$. 
{\bf 15.} A torus with $\omega(G)=1$. 
}
\end{figure}

\section*{Appendix: Code}

More detailed code and examples can be obtained by downloading the LaTeX 
source file of this preprint on the ArXiv, then copy paste the Mathematica
part from the file.
We compressed the code as much as possible as shortness of code 
benefits both communication and verification.  The Wolfram language
is suitable to serve as pseudo code. Even if the language might
evolve in the future, the following code should remain readable. 
Since Mathematica has not yet built in a procedure to
list of all complete subgraphs in a graph, we do that first by hand.
One can also use the iGraph library which uses compiled code and is faster.
We don't use it by default as this is a complex third party library. 
But the general problem is NP complete, so that one can not expect
a very efficient algorithm.

\begin{tiny}
\lstset{language=Mathematica} \lstset{frameround=fttt}
\begin{lstlisting}[frame=single]
CliqueNumber[s_]:=Length[First[FindClique[s]]];
ListOfCliques[s_,k_]:=Module[{n,t,m,u,r,V,W,U,l={},L},L=Length;
  VL=VertexList;EL=EdgeList;V=VL[s];W=EL[s]; m=L[W]; n=L[V];
  r=Subsets[V,{k,k}];U=Table[{W[[j,1]],W[[j,2]]},{j,L[W]}];
  If[k==1,l=V,If[k==2,l=U,Do[t=Subgraph[s,r[[j]]];
  If[L[EL[t]]==k(k-1)/2,l=Append[l,VL[t]]],{j,L[r]}]]];l];
Whitney[s_]:=Module[{F,a,u,v,d,V,LC,L=Length},V=VertexList[s];
  d=If[L[V]==0,-1,CliqueNumber[s]];LC=ListOfCliques;
  If[d>=0,a[x_]:=Table[{x[[k]]},{k,L[x]}];
  F[t_,l_]:=If[l==1,a[LC[t,1]],If[l==0,{},LC[t,l]]];
  u=Delete[Union[Table[F[s,l],{l,0,d}]],1]; v={};
  Do[Do[v=Append[v,u[[m,l]]],{l,L[u[[m]]]}],{m,L[u]}],v={}];v];
(* Import["IGraphM`"]; Whitney=IGCliques;   *)
\end{lstlisting}
\end{tiny}

Here is example code to compute the Euler characteristic $\chi$
Euler curvature entering Gauss Bonnet and Euler indices $i_f(x)$
entering Poincar\'e-Hopf:

\begin{tiny}
\lstset{language=Mathematica} \lstset{frameround=fttt}
\begin{lstlisting}[frame=single]
Euler[s_]:=Module[{c=Whitney[s]},
  If[c=={},0,Sum[(-1)^(Length[c[[k]]]+1),{k,Length[c]}]]];
EulerC[s_,v_]:=Module[{a,n},
  a=Whitney[VertexDelete[NeighborhoodGraph[s,v],v]];
  1+Sum[n=Length[a[[k]]];(-1)^n/(n+1),{k,Length[a]}]];
EulerCurvatures[s_]:=Module[{V=VertexList[s]},
  Table[EulerC[s,V[[k]]],{k,Length[V]}]];
EulerIndex[f_,s_,v_]:=Module[{a={},t,V},
  V=VertexList[NeighborhoodGraph[s,v]];
  Do[If[f[[V[[k]]]]<f[[v]],a=Append[a,V[[k]]]],{k,Length[V]}];
  t=Subgraph[s,a];1-Euler[t]];
EulerIndices[f_,s_]:=Module[{V=VertexList[s]},
  Table[EulerIndex[f,s,V[[k]]],{k,Length[V]}]];
s=RandomGraph[{40,100}];
{Euler[s],Total[EulerIndices[Range[40],s]]}
\end{lstlisting}
\end{tiny}

The following procedures compute the
Wu characteristic and the cubic Wu Characteristic  as well as
the intersection number $\omega(G,H)$ 

\begin{tiny}
\lstset{language=Mathematica} \lstset{frameround=fttt}
\begin{lstlisting}[frame=single]
Wu[s_]:=Module[{c=Whitney[s],v=0,k},If[c=={},1,
  Do[Do[If[Length[Intersection[c[[k]],c[[l]]]]>0,
  v+=(-1)^(Length[c[[k]]]+Length[c[[l]]]);],
  {k,Length[c]}],{l,Length[c]}];v]];
Wu3[s_]:=Module[{c=Whitney[s],v=0,k,l,m},If[c=={},v=1,
  Do[Do[Do[If[Length[Intersection[c[[k]],c[[l]],c[[m]]]]>0,
  v+=(-1)^(Length[c[[k]]]+Length[c[[l]]]+Length[c[[m]]]+1);],
  {k,Length[c]}],{l,Length[c]}],{m,Length[c]}]];v];
Wu[s1_,s2_]:=Module[{c1=Whitney[s1],c2=Whitney[s2],v=0},
  Do[Do[If[Length[Intersection[c1[[k]],c2[[l]]]]>0,
   v+=(-1)^(Length[c1[[k]]]+Length[c2[[l]]])],
  {k,Length[c1]}],{l,Length[c2]}];v];

Wu[RandomGraph[{40,100}]]
\end{lstlisting}
\end{tiny}

Here are procedures to test the Gauss-Bonnet theorem for
the Wu characteristic as well as the Poincar\'e-Hopf theorem
for the Wu characteristic, two results proven in this paper:

\begin{tiny}
\lstset{language=Mathematica} \lstset{frameround=fttt}
\begin{lstlisting}[frame=single]
WuC[s_,v_]:=Module[{a=Whitney[NeighborhoodGraph[s,v,2]]},
  Sum[Sum[If[Length[Intersection[a[[k]],a[[l]]]]>0 &&
  MemberQ[a[[l]],v],(-1)^(Length[a[[l]]]+Length[a[[k]]]),0],
  {k,Length[a]}]/(Length[a[[l]]]),{l,Length[a]}]];
WuCurvatures[s_]:=Module[{V=VertexList[s]},
  Table[WuC[s,V[[k]]],{k,Length[V]}]];
UnitSphere[s_,a_]:=Module[{b},b=NeighborhoodGraph[s,a];
  If[Length[VertexList[b]]<2,Graph[{}],VertexDelete[b,a]]];
WuIndex[f_,s_,a_,b_]:=Module[{vl,sp,sq,v,w,sa,sb,ba,bb,P,p,S,A,L,V,k},
  P=Position; p[t_,u_]:=P[t,u][[1,1]];S=Subgraph;A=Append;
  L=Length;V=VertexList;sa={};sb={};
  vl=V[s];sp=UnitSphere[s,a];v=V[sp];sq=UnitSphere[s,b]; w=V[sq];
  Do[If[f[[p[vl,v[[k]]]]]<f[[p[vl,a]]],sa=A[sa,v[[k]]]],{k,L[v]}];
  Do[If[f[[p[vl,w[[k]]]]]<f[[p[vl,b]]],sb=A[sb,w[[k]]]],{k,L[w]}];
  ba=A[sa,a];bb=A[sb,b];Wu[S[s,bb],S[s,ba]]-Wu[S[s,bb],S[s,sa]] -
                        Wu[S[s,sb],S[s,ba]]+Wu[S[s,sb],S[s,sa]]];
WuI[f_,s_]:=Module[{V=VertexList[s],L=Length,k,l},
  Table[WuIndex[f,s,V[[k]],V[[l]]],{k,L[V]},{l,L[V]}]];
WuIndices[f_,s_]:=Total[WuI[f,s]];

s=RandomGraph[{10,30}];{Wu[s],Total[WuCurvatures[s]]}
s=RandomGraph[{10,30}];{Wu[s],Total[WuIndices[Range[10],s]]}
\end{lstlisting}
\end{tiny}

Finally, here is an implementation of the graph product and 
Barycentric refinement. Unlike in the \cite{KnillKuenneth}, 
where code for the procedure was algebraically given using the 
Stanley-Reisner ring, we perform the product here 
directly. We also have added an example line computing the
Euler and Wu characteristics for two graphs and its product.

\begin{tiny}
\lstset{language=Mathematica} \lstset{frameround=fttt}
\begin{lstlisting}[frame=single]
GraphProduct[s1_,s2_]:=Module[{a,b,S,L=Length,A=Append,v,e,p,q},
a=Whitney[s1];b=Whitney[s2];S=SubsetQ;R=Sort;p=Range[L[a]];
q=Range[L[b]];v={};Do[v=A[v,{k,i}],{k,L[p]},{i,L[q]}];e={};Do[
If[(S[a[[k]],a[[l]]]&&S[b[[i]],b[[j]]]&&{k,i}!={l,j}) ||
   (S[a[[l]],a[[k]]]&&S[b[[j]],b[[i]]]&&{k,i}!={l,j}),
 e=A[e,R[{k,i}->{l,j}]]],{k,L[p]},{l,L[p]},{i,L[q]},{j,L[q]}];
v=Union[v]; e=Union[e]; UndirectedGraph[Graph[v,e]]];
Barycentric[s_]:=GraphProduct[s,CompleteGraph[1]];

s1=StarGraph[4]; s2=CycleGraph[4]; s=GraphProduct[s1,s2];
Print[Wu[s1],"  times   ", Wu[s2],"  is equal to: ",Wu[s]];
Print[Euler[s1],"  times   ",Euler[s2],"  is equal to",Euler[s]]
\end{lstlisting}
\end{tiny}

Here is the inductive dimension implementation for graphs. 
It is discussed in detail in \cite{randomgraph}. It can also
be used to illustrated the proven ${\rm dim}(A \times B) \geq 
{\rm dim}(A) + {\rm dim}(B)$ formula holding for any pair of graphs.

\begin{tiny}
\lstset{language=Mathematica} \lstset{frameround=fttt}
\begin{lstlisting}[frame=single]
Dimension[s_]:=Module[{v,VL=VertexList,u,n,m,e,L},
  L=Length; v=VL[s]; n=L[v];e=EdgeList[s]; m=L[e];
  If[n==0,u={},u=Table[UnitSphere[s,v[[k]]],{k,n}]];
  If[n==0,-1,If[m==0,0,
  Sum[If[L[VL[u[[k]]]]==0,0,1]+Dimension[u[[k]]],{k,n}]/n]]];

{Dimension[Graph[{},{}]],Dimension[WheelGraph[5]]}
\end{lstlisting}
\end{tiny}

Here are procedures to compute the $f$-vector or $f$-matrix
or $f$-tensor in the cubic case of a graph: 

\begin{tiny}
\lstset{language=Mathematica} \lstset{frameround=fttt}
\begin{lstlisting}[frame=single]
Fvector[s_]:=Delete[BinCounts[Map[Length,Whitney[s]]],1];
Fmatrix[s_]:=Module[{c=Whitney[s],v,m,n,V},
  v=Delete[BinCounts[Map[Length,c]],1];d=Length[v];n=Length[c];
  V=Table[0,{d},{d}]; Do[If[Length[Intersection[c[[k]],c[[l]]]]>0,
  V[[Length[c[[k]]],Length[c[[l]]]]]++],{k,n},{l,n}];V];
Fmatrix3[s_]:=Module[{c=Whitney[s],v,d,n,V},
  v=Delete[BinCounts[Map[Length,c]],1];d=Length[v];n=Length[c];
  V=Table[0,{d},{d},{d}]; Do[Do[Do[If[Length[Intersection[c[[k]],c[[l]],c[[m]]]]>0,
  V[[Length[c[[k]]],Length[c[[l]]],Length[c[[m]]]]]++],{k,n}],{l,n}],{m,n}];V];
\end{lstlisting}
\end{tiny}

Finally, lets look at some linear and multi-linear Dehn-Sommerville 
invariants. To test the produced we build a three or four dimensional
sphere: 

\begin{tiny}
\lstset{language=Mathematica} \lstset{frameround=fttt}
\begin{lstlisting}[frame=single]

BarycentricOperator[n_]:=Table[StirlingS2[j,i]i!,{i,n+1},{j,n+1}];
Bvector[k_,d_]:=Module[{A,v},A=BarycentricOperator[d-1]; 
  v=Reverse[Eigenvectors[Transpose[A]]];v[[k]]];
Dvector[k_,d_]:=Module[{f},f[i_]:=Table[If[i==j-1,1,0],{j,d}];
  Sum[(-1)^(j) Binomial[j+1,k+1]f[j],{j,k,d-1}]+(-1)^d f[k]];
Binvariants[s_]:=Module[{v=Fvector[s],n},
  n=Length[v];Table[Bvector[k,n].v,{k,n}]];
Dinvariants[s_]:=Module[{v=Fvector[s],n},
  n=Length[v];Table[Dvector[k,n].v,{k,n}]];
Binvariants2[s_]:=Module[{V=Fmatrix[s],n},
  n=Length[V];Table[Bvector[k,n].(V.Bvector[l,n]),{k,n},{l,n}]];
Dinvariants2[s_]:=Module[{V=Fmatrix[s],n},
  n=Length[V];Table[Dvector[k,n].(V.Dvector[l,n]),{k,n},{l,n}]];
Binvariants3[s_]:=Module[{V=Fmatrix3[s],n},n=Length[V];
  Table[((V.Bvector[l,n]).Bvector[k,n]).Bvector[m,n],{k,n},{l,n},{m,n}]];
Dinvariants3[s_]:=Module[{V=Fmatrix3[s],n},n=Length[V];
  Table[((V.Dvector[l,n]).Dvector[k,n]).Dvector[m,n],{k,n},{l,n},{m,n}]];

Suspension[s_]:=Module[{v,e,a,b,r,n},
  v=VertexList[s]; e=EdgeRules[s]; n=Length[v]; 
  a=n+1; b=n+2; r=Table[v[[k]]->k,{k,Length[v]}]; v=v /. r; e=e /. r; 
  Do[e=Append[e,v[[k]]->a];e=Append[e,v[[k]]->b],{k,n}]; 
  UndirectedGraph[Graph[e]]];
twosphere=Suspension[CycleGraph[4]]; 
threesphere=Suspension[twosphere]; 
foursphere=Suspension[threesphere] 
Binvariants2[threesphere]
Binvariants2[foursphere]
\end{lstlisting}
\end{tiny}

Like the category of sets, the category of graphs has addition and multiplication
given by the symmetric difference and multiplication given by the 
intersection.  But note that the ring
product in the Stanley-Reiner ring uses a different multiplication which corresponds
to a direct product of sets. Here is how to work with the basic Boolean algebra of graphs 
in the Wolfram language. Note that the built-in ``GraphIntersection" procedure
does not take the intersection of the vertex sets. The vertex list of 
$A \cap B$ with $V(A)=\{1,2,3,4,5,6\}$ and $V(B)=\{1,2,3,4\}$ is $\{1,2,3,4,5,6\}$ 
when using ``GraphIntersection". We therefore implement the procedure ``GProduct", which replaces that 
procedure and produces the true graph intersection which has a vertex set $\{1,2,3,4\}$. 
Note that if $G=(V,E)$ and $H=(W,F)$, then $G \cdot H$ is not $G \cap H$.
The graph product of the two graphs $G=a+b+ab$ and $b+c+bc$ is $0$ while the 
intersection is $b$. The sum of the two graphs is $a+b+c+ab+bc$. On the
vertex level this is not the symmetric difference $\{a,c\}$ of the vertex sets 
$\{a,b\},\{b,c\}$. Mathematica gives the zero dimensional three point graph $a+b+c$ as the graph 
intersection. Both the intersection as well as the Mathematica
graph intersection do not honor distributivity. If modifying the product, we can salvage
distributivity, but we still have 
situations where $A$ has a different graph $B$ with $A+B=A+A=0$. 

\begin{tiny}
\lstset{language=Mathematica} \lstset{frameround=fttt}
\begin{lstlisting}[frame=single]
SetProduct[x_,y_]:=Intersection[x,y];
SetAddition[x_,y_]:=Union[Complement[x,y],Complement[y,x]];
GIntersection[A_,B_]:=UndirectedGraph[Graph[
  SetProduct[VertexList[A],VertexList[B]],
  SetProduct[EdgeList[A],EdgeList[B]]]]; 
GProduct[A_,B_]:=UndirectedGraph[Graph[
  SetProduct[EdgeList[A],EdgeList[B]]]]; 
GAddition[A_,B_]:=UndirectedGraph[Graph[
  SetAddition[VertexList[A],VertexList[B]],
  SetAddition[EdgeList[A],EdgeList[B]]]];
EmptyGraph = Graph[{},{}];

A = RandomGraph[{55,15}]; B=RandomGraph[{16,48}]; U=RandomGraph[{107,1000}];
P = GProduct[GAddition[A,B],U]; Q = GAddition[GProduct[A,U],GProduct[B,U]];
Print["Does distributativity hold? ",IsomorphicGraphQ[P,Q] ]
P = GProduct[GProduct[A,B],U]; Q = GProduct[A,GProduct[B,U]];
Print["Does associativity hold? ",IsomorphicGraphQ[P,Q] ]
P = GAddition[GAddition[A,B],U]; Q = GAddition[A,GAddition[B,U]];
Print["Does associativity hold? ",IsomorphicGraphQ[P,Q] ]
P = GAddition[A,EmptyGraph]; Q=A;
Print["Is the empty graph the zero element? ",IsomorphicGraphQ[P,Q] ]
A=Graph[{1->2}]; B=Graph[{2->3}]; EmptyGraphQ[GProduct[A,B]]
\end{lstlisting}
\end{tiny}

\section*{Appendix: Discrete Hadwiger}

The following result of the discrete Hadwiger theorem parallels the statement for 
simplicial complexes proven in \cite{KlainRota} (Theorem 3.2.4).  
The only change is vocabulary as it is formulated for subgraphs rather than general sub simplicial 
complexes. But the graph case is not new because the set of subgraphs of a graph $G=(V,E)$ defines a lattice 
$(L,\emptyset,\cap, \cup)$, where $\emptyset$ is the empty graph and a linear valuation is a
valuation on that lattice. The set of complete subgraphs in $G$ is then a generating set of 
this lattice and $\emptyset$ is the {\bf bottom element} in 
this distributive lattice. 

\begin{thm}[Discrete Hadwiger]
The vector space of invariant valuations on a finite simple graph $G$ is equal to the clique number $d+1$
of $G$. A basis is given by the valuations $v_k(A)$ counting the number of $K_{k+1}$ subgraphs in $A$
and where $k=1,\dots,d+1$.
\end{thm}
\begin{proof}
We have to show that (i) the functionals $v_k$ are linearly independent and (ii) that every valuation  $X$
a linear combination of the functionals $v_k$. \\
(i) Assume that we can find a vector $\vec{a} = (a_0,\dots,a_d)$ such
that $$ \sum_i a_i v_i(H) = 0 $$ for all subgraphs $H$ of $G$.
By taking complete subgraphs $K_k$ we see that $\vec{a}$ is perpendicular to the vectors
$\vec{v}(K_k) = (\B{k+1}{1},\B{k+1}{2},\dots,\B{k+1}{k+1})$.
Since all these $d+1$ vectors are linearly independent for $k=0,\dots,d$, the vector $\vec{a}$ must be zero.  \\
Barycentric refinement gives an algebraic proof of linear independence as the eigenfunctions of $A^T$, 
where $A$ is the {\bf Barycentric refinement operator} are linearly independent as the $(d+1) \times (d+1)$ 
matrix $A$ has $d+1$ different eigenvalues $k!$ for $k=1, \dots, d+1$. \\

(ii) Given an arbitrary invariant valuation $X$ on $G$. By the invariance property in the definition,
there is for every $k=0,\dots,d$ a real number $b_k$ such that $X(x)=b_k$ for all complete subgraphs $x$ of
dimension $k$. Also $X(\emptyset)=0$ by definition.
Define now $Y(A) = \sum_{k=0}^d b_k v_k(A)$. Since every subgraph $A$ of $G$ can be written as a union of
complete subgraphs and both $X(\emptyset)=Y(\emptyset)=0$, we can use the inclusion-exclusion
property to see inductively that $X(A)=Y(A)$.
\end{proof}

This can be generalized to higher dimensions. Lets just discuss the quadratic case:

\begin{thm}[Hadwiger for quadratic valuations]
The vector space of quadratic valuations has dimension $(d+1)(d+2)/2$. 
\end{thm}
\begin{proof} 
The proof is very similar. One has to get an independent set of $(d+1)(d+2)/2$
valuations and then show that every valuation can be written as a linear combination 
of this set. Lets look at the valuation $V_{ij}(G)$ measuring the 
number of $i$-simplices intersecting with $j$-simplices in $G$. 
This number grows like $(i+1)! (j+1)!$ under Barycentric subdivision. These are
growth rates which are independent. To show that these number span, note
that every quadratic valuation can be written as $X(A) = \sum_{i,j} V_{ij}(G) \phi_i \psi_j$ 
with $i\leq j$. 
\end{proof}

For example, for $d=0$, the set of quadratic valuations is $1$, it is the valuation
which counts the discrete set of points. In the one dimensional case $d=1$, there
the valuation space is $3$ dimensional, we look at the number of vertices, the 
number of pairs of edges intersecting and the number of vertex-edge pairs, which is
twice the number of edges. \\

In the continuum, one restricts the theory of valuations to convex sets
or {\bf polyconvex sets}, which is a term introduced in \cite{KlainRota} and  
means a finite union of convex sets in $R^n$. Also polyconvex sets form
a lattice. Things are related of course as any finite simple graph has
a geometric realization in which it is a polyconvex set. The fact that in graph theory,
the values $v_k(G)$ form a basis of all valuations can also be seen as a
manifestation of the Groemer extension theorem from a generating set
to the full lattice. \\

The structure of valuations on discrete spaces is much richer than the structure of measures. 
Note that on a finite set $V$, there is only a $1$-dimensional space of 
measures if we ask the measure to be invariant. It is a multiple of the counting measure,
which is the Haar measure with respect to the symmetric group acting on $V$. 
So, thinking about a valuation as a measure is somewhat misleading. 
On a finite set, there is a natural $\sigma$ algebra, the set of all subsets and the structure
of all signed measures on this algebra is a $|V|$ dimensional space given by all 
$\{p_1,\dots,p_n\}$, a trivial manifestation of the Riesz representation theorem
applied to linear functionals on the set of continuous functions 
$C(V) = {\rm R}^n$ which happens to adopt even a Hilbert space structure in the 
finite case. The fact that set of measures invariant under permutations is 
$1$-dimensional settles the theory of invariant measures on a finite set. 
On the other hand, if the discrete set is quipped with a graph structure and the
corresponding simplicial complex, 
the theory of valuations is more interesting, as the Hadwiger theorem illustrates. 
The theory of measures does not look at the internal structure of the sets
which are measured, unlike the theory of valuations which look inside. In the continuum,
it is a bit harder to explore this internal structure as one needs tomographic methods,
but the language of probability theory like Crofton or kinematic formulas allow to deal with it. 
As we have hoped to demonstrate in this article, some interesting mathematics in discrete differential geometry
like Gauss-Bonnet, or Poincar\'e-Hopf can be seen naturally as results on valuations. We 
also hoped to show that the language of graphs works well also to discover new results. 
Working with subgraphs of a graph is a bit like working with subsets but hides the difficulty 
that the theory of valuations really does more: it deals with the distributive lattice of 
simplicial subcomplexes of a given simplicial complex rather than the lattice of 
subsets of a given set. But working with sets of subsets as the basic structure is
harder to think and write about. The language of graphs is more intuitive similarly as the 
concept of metric topologies is easier to deal with than arbitrary topologies. \\

The analogy between discrete Hadwiger and continuous Hadwiger is so close that one might wonder
whether it is possible to treat them in a unified manner. This is indeed possible as emerged while 
working with Barry Tng \cite{Tng} ind we sketch the
connection: a measure $\mu$ on the set $\Omega$ of linear functions on $R^d$ defines a
length $|\gamma|$ on smooth parametrized curves $\gamma: [0,1] \to R^d$ by the
{\bf Crofton formula}: take smooth random linear function $f_k$ with equidistant level surfaces and
produce the random variable $X_k(\gamma) \in N$ counting the number of intersections with the curve.
The law of large numbers shows that the expectation of this random variable can be explored by a
Monte Carlo computation. It defines a length functional on curves, the property that it is
additive is related to the additivity of probability.
We have now only a {\bf semi-metric} $d(A,B) = \inf_{\gamma} |\gamma|$ on $R^d$ but 
the Kolmogorov quotient is now a metric space. The Kolmogorov quotient just takes equivalence classes of
points for which the semimetric is zero. For the semi metric $d( (x_1,y_1),(x_2,y_2))=x_1-x_2$ in the plane
for exmple, the Kolmogorov quotint is the real line. 
If $\mu$ is a measure which is invariant under translations and
rotations, then it recovers the usual metric on $R^d$ up to a scaling factor. This is the Crofton formula
which for polygons reduces to the Buffon needle computation. If $\mu$ is
a finite point measure, then the Kolmogorov quotient is a finite graph. This setup is rather general
and most conveniently described in the projective situation, where translation is part of the projective group.
Any {\bf compact Riemannian manifold} can be treated like that: Nash embed it into some projective
space $P^d$, look at the Haar measure on all linear functions on $P^d$ invariant under projective
transformations. Then, as the arc length on curves in $P^d$ is the same than arc length on $M$, we
can look at the measure $\mu$ on the set $\Omega$ of Morse functions on $M$ given by the
push forward of the Haar measure on linear functions to the space of Morse functions on $M$.
This defines now a probability space on the class of Morse functions on $M$ which allows to
recover the Riemannian metric {\bf integral-geometrically} within the manifold $M$.
The point is however that by choosing a different measure $\mu$ on Morse functions, we get
different metric spaces. In particular, if we take a discrete finite point measure, we get
a discrete space. Going back to the case $R^d$, the probability space $\Omega$ is a set of linear functions,
we have convex sets obtained by inequalities $|f_i -k |\leq c_i$.
Why does the dimension of valuations in $R^d$ agree with the dimension of valuations on graphs?
For any finite measures $\mu_m$ we have a $(d+1)$-dimensional set of valuations by
discrete Hadwiger. Now approximate the original Haar measure with finite measures leading to $d$-dimensional
graphs by doing Barycentric subdivisions, we can get to the Haar measure. Since valuations go over to the
limit and the dimension of the set of valuations is upper semi continuous, the discrete approximation argument
shows that we have at least a $(d+1)$-dimensional set of valuations in the continuum. We can also go backwards and
approximate a given discrete finite measure by absolutely continuous measures. Assume that the vector space of valuations
is $k$-dimensional with $k>d+1$. We would get $k$ different valuations in the discrete limit which is not the
case. On simplices we can compute the $k$'th valuation by adding up the $k$-dimensional measures of
$k$-dimensional subsimplices.  On
convex subsets $K$, we can make a triangulation and sufficiently many Barycentric subdivisions
in Euclidean space allowing to compute the valuation numerically by adding up the valuations on subsimplices.
Smooth manifolds allowing an approximation by polytopes, the valuation is still defined as a limit.
The setup is much more than just a unification the theory of valuations on Euclidean space or graphs.
We can see that for any measure $\mu$, which even {\bf might be singular continuous}, we get to geometric spaces
which have a Hadwiger theorem. These spaces can be objects with fractal dimension.
They can be seen as approximations of graphs or then as limiting cases of smooth manifolds.
The more general setup can use to make sense of curvature also on more general spaces as we can
define curvature integral geometrically. We can use the
function $f$ to define an index $i_f(x)$ at a point and use the measure $\mu$ to average this index
to get a curvature function $K(x)$. While we have not gone into this general framework here and stayed strictly
within graph theory, we hope that the prospect of a much more general geometry which includes both
graph theory as well as Riemannian geometry, makes the graph theoretical setup more relevant. 

\section*{Appendix: About the literature}

The results are formulated in the language of graph theory 
\cite{BM,bollobas1,bollobas} which itself has various topological graph theory \cite{TuckerGross} 
or algebraic graph theory \cite{Biggs}. There is overlap with work on polytopes 
\cite{Ziegler,gruenbaum}, simplicial complexes \cite{Stanley86,Stanley1996,JohnsonSimplicial} or
combinatorial topology like \cite{Hatcher,spanier}. See \cite{Dieudonne1989} for history.
Various flavors of discrete topologies have emerged:
digital topology \cite{NeumannLaraWilson,HermanDigitalSpaces,Evako2013},
discrete calculus \cite{GradyPolimeni},
Fisk theory \cite{AlbertsonStromquist,Fisk1978,Fisk1977a,Fisk1977b,Fisk1980}
to which we got in the context of graph colorings \cite{knillgraphcoloring,
knillgraphcoloring2} leading to the notion of spheres which appeared in
\cite{Evako1994} which is based on homotopy \cite{I94,CYY},
based on notions put forward in \cite{Alexander15,Whitehead},
networks \cite{newman2010,nbw2006,vansteen,CohenHavlin,newman2010,vansteen,ibe},
physics \cite{GSD,Regge,Misner,Fro81,MarsdenDesbrun,Zakopane},
computational geometry \cite{Devadoss,ComputationalGeometry,Bobenko,CompElectro2002,YJLG}
discrete Morse theory \cite{forman95,forman98,Forman1999,Forman2003},
eying classical Morse theory \cite{Mil63},
discrete differential geometry in relation to classical differential geometry
\cite{BergerPanorama,Gallot,Jost}.
We got to into the subject through \cite{elemente11}
and generalized it to \cite{cherngaussbonnet} and summarized in
\cite{knillcalculus} after \cite{poincarehopf}.
The general Gauss-Bonnet-Chern result appeared in
\cite{cherngaussbonnet} but was predated in \cite{Levitt1992}.
It seems Gauss-Bonnet for graphs has been rediscovered a couple of times like
\cite{I94a,forman2000}. We noticed the first older appearance
\cite{I94a} in \cite{KnillJordan} and found \cite{Levitt1992,forman2000}
while working on the present topic. Various lower dimensional versions of curvature
have appeared \cite{Gromov87,Presnov1990,Presnov1991,Higuchi,
NarayanSaniee,RetiBitayKosztolanyi}.
The first works on Gauss-Bonnet in arbitrary
dimensions include \cite{hopf26,Allendoerfer,Fenchel,AllendoerferWeil,Chern44,Chern1990}.
For modern proofs, see \cite{Cycon,Rosenberg}.
The story of Euler characteristic is told in \cite{FranceseRicheson,Richeson,Malkevich}.
A historical paper is \cite{EulerGoldbach}.
For uniqueness of Euler characteristic as a linear valuation  see 
\cite{Klee63,Rota71,I94a,forman2000,Roberts2002,Yu2010,LuzonMoron}.
For Hadwiger's theorem of 1957,  see \cite{KlainRota,Klain1995}. 
The first works on Poincar\'e-Hopf
were \cite{poincare85,hopf26,Morse29}. For more history, see
\cite{Spivak1999,Mil65,Neill,Gottlieb,Hirsch}. Poincar\'e-Hopf indices are
central in discrete approaches to Riemann-Roch \cite{BakerNorine,BakerNorine2007}.
The index expectation result is \cite{indexexpectation}. The closest related
work is Banchoff \cite{Banchoff1967,Banchoff1970}.
For integral geometry and geometric probability, see \cite{KlainRota,Santalo}. It is a 
popular topic for REU research projects like \cite{csarjohnsonlamberty}.
The index formula for Euler characteristic appeared in \cite{indexformula}
which proves a special case of the result here along the same lines.
See also \cite{eveneuler} for the recursion.
The Sard approach in \cite{KnillSard} simplified this. That paper gives a discrete
version of \cite{Sard42}.
We got into the Barycentric invariants through \cite{KnillBarycentric,
KnillBarycentric2} after introducing a graph product
\cite{KnillKuenneth} which was useful in \cite{KnillJordan}, a paper
exploring topology of graphs \cite{KnillTopology,josellisknill,brouwergraph}.
Originally we were interested in the spectral theory of graphs
\cite{Chung97,Mieghem,Brouwer,VerdiereGraphSpectra,Post,BLS}
which parallels the continuum
\cite{Chavel,Rosenberg,BergerPanorama}.
The linear algebra part of networks was explored in \cite{knillmckeansinger}
which is a discrete version of \cite{McKeanSinger}. See also
\cite{DiracKnill,KnillBaltimore,IsospectralDirac,IsospectralDirac2,KnillZeta,cauchybinet}
and \cite{DanijelaJost,Mantuano} for discrete combinatorial Laplacians.
For the Dehn-Sommerville relations, see
\cite{Klee1964,NovikSwartz,MuraiNovik,LuzonMoron,BrentiWelker,Hetyei,Klain2002}.
In \cite{NovikSwartz,Klain2002} appear Dehn-Sommerville-Klee
equations for discrete for manifolds with boundary.
In \cite{LuzonMoron}, it was noted that the Euler characteristic is the only invariant,
using the operator $A$. The combinatorics of the Barycentric operator was
studied in\cite{hexacarpet} in the case $d=2$. The explicit formula using
Stirling numbers appeared in \cite{BrentiWelker}.
For Polytopes \cite{Devadoss,Richeson, Schlafli,coxeter,gruenbaum,symmetries,lakatos,cromwell}.
The theory of valuations on distributive lattices has been pioneered by Klee \cite{Klee63}
and Rota \cite{Rota71} who proved that there is a unique
valuation such that $X(x)=1$ for any join-irreducible element. This is the Euler characteristic.
The example of the lattice of subgraphs of a graph fits within this framework. 

\section*{Appendix: Summaries}  

The next two pages are from October 2, 2015. The new results are then added where indicated.
Originally we used the central manifolds to argue the vanishing of the Barycentric invariant numbers. 
There were too statements which overreached in the October 2 summary:
the symmetric index $j_G(x)$ is not always constant zero for Dehn-Sommerville valuations.
Also, the index expectation does not generalize without modifications from Euler characteristic 
to general valuations. The index formula still produces probabilistic statements about random 
geometric subgraphs of spheres, but these results still need to be harvested. \\

If $\G$ is the category of {\bf finite simple graphs} $G=(V,E)$, 
the linear space $\V$ of {\bf valuations} on $\G$ has a basis
given by the {\bf $f$-numbers} $v_k(G)$ counting complete subgraphs $K_{k+1}$ in $G$. 
The {\bf barycentric refinement} $G_1$ of $G \in \G$ is the graph with $K_l$ subgraphs as vertex set
where new vertices $a \neq b$ are connected if $a \subset b$ or $b \subset a$. Under refinement, 
the clique data transform as $\vec{v} \to A \vec{v}$ with the upper triangular matrix 
$A_{ij} = i! S(j,i)$ with {\bf Stirling numbers} $S(j,i)$. 
The eigenvectors $\chi_k$ of $A^T$ with eigenvalues $k!$ form an other basis in $\V$. The $\chi_k$ 
are normalized so that the first nonzero entry is $>0$ and all entries are in $\Z$ 
with no common prime factor.
$\chi_1$ is the {\bf Euler characteristic} $\sum_{k=0}^{\infty} (-1)^k v_k$, the homotopy and 
so cohomology invariant on $\G$. 
Half of the $\chi_k$ will be zero {\bf Dehn-Sommerville-Klee invariants} 
like half the {\bf Betti numbers} are redundant under {\bf Poincar\'e duality}.
On the set $\G_d \subset \G$ with clique number $d$, the valuations $\V$
have dimension $d+1$ by {\bf discrete Hadwiger}. A basis is the eigensystem 
$\vec{\chi}$ of the  $(d+1) \times (d+1)$ submatrix matrix $A_d^T$ of $A^T$. 
The functional $\chi_{d+1}(G)$ is {\bf volume}, counting the {\bf facets} of $G$.
For $x \in V$, define $V_{-1}(x)=1$ and $V_k(x)$ as the number of complete subgraphs 
$K_{k+1}$ of the unit sphere $S(x)$, the graph generated by the neighbors of $x$. 
The {\bf fundamental theorem} of graph theory is the formula
$\sum_{x \in V} V_{k-1}(x) = (k+1) v_k(G)$. For $k=1$, it is the {\bf Euler's handshake}.
For a valuation $X(G) = \sum_{l=0}^\infty a(l)v_l(G)$, define 
{\bf curvature} $K_X(x)=\sum_{l=0}^\infty a(l)V_{l-1}(x)/(l+1)$. 
Generalizing the fundamental theorem:

\begin{satz}[{\bf Gauss-Bonnet}]
$\sum_{x \in V} K_X(x) = X(G)$.
\end{satz}

{\bf Example.} For an icosahedron with $\vec{v}=(12,30,20)$ and $\vec{v}(S(x))=(5,5)$, we have 
$a_1=(1,-1,1)$, $K_1(x)=1-5/2+5/3=1/6$, $\chi_1=2$, 
$a_2=(0,2,-3)$,  $K_2(x)=10/2-15/3=0$, $\chi_2=0$, 
$a_3=(0,0,1)$,  $K_3(x)=5/3$, $\chi_3=20$.  \\

Let $\Omega(G)$ be the set of {\bf colorings} of $G$, locally injective function $f$ on $V(G)$. 
The {\bf unit ball} $B(x)$ at $x$ is the graph generated by the union of $\{x\}$ and
the {\bf unit sphere} $S(x)=\{ y \in V \; | \; (x,y) \in E \;\}$ which
is the {\bf boundary} $\delta B(x)$. For $f \in \Omega$ and $X \in \V$
define the {\bf index} $i_{X,f}(x) = X(B^-(x))-X(S^-(x))$,
where $B^-(x) = S^-(x) \cup \{x\} = \{ y \in B(x) \; | \; f(y) \leq f(x) \}$
and $S^-(x) = \{f(y)< f(x)\}$. It is local and a {\bf divisor}.
Inductive attaching vertices gives:

\begin{satz}[{\bf Poincar\'e-Hopf}]
$\sum_{x \in V} i_{X,f}(x) = X(G)$.
\end{satz}

Let $P$ be a {\bf Borel probability measure} on $\Omega(G) = \R^{v_0(G)}$ and 
$E[\cdot]$ its {\bf expectation}. Let $c(G)$ be the {\bf chromatic number} of $G$. 
Assume either that $P$ is the counting measure on the finite set of colorings of $G$ with 
$c \geq c(G)$ real colors or that $P$ is a product measure on $\Omega$ for which
functions $f \to f(y)$ with $y \in V$ are independent identically distributed
random variables with continuous probability density function. For all $G \in \G$ and 
$X \in \V$:

\begin{satz}[{\bf Banchoff index expectation}]
For Euler characteristic $E[i_{X,f}(x)] = K_X(x)$. 
\end{satz}

The {\bf empty graph} $\emptyset$ is a $(-1)$-graph and $(-1)$-sphere. Inductively,
a {\bf $d$-graph} is a $G \in \G$ for which the unit spheres are $(d-1)$-spheres.
An {\bf Evako $d$-sphere} is a $d$-graph which when punctured becomes contractible. 
Inductively, $G$ is {\bf contractible} if there exists $x \in V(G)$ such that both $S(x)$ and the graph 
without $x$ are contractible. The graph $K_1$ is contractible.
Given $f \in \Omega$ and $c \notin f(V)$, define the graph $\{ f=c \}$ in the refinement of
$G$ consisting of vertices, where $f-c$ changes sign. In that case, at every vertex $x$, there is a 
$(d-2)$-graph $S_f(x)$ defined as the {\bf level surface} $\{f(y)=f(x)\}$ in $S(x)$.
The next {\bf Sard} result belongs to {\bf discrete multivariable calculus}:

\begin{satz}[{\bf Implicit function theorem}]
For a $d$-graph and $f \in \Omega$ and $c \notin f(V)$, the hyper surface $\{ f=c \}$ is a $(d-1)$-graph.
\end{satz}

The {\bf symmetric index} of $f$ at $x$ is defined as $2 j_{X,f}(x) = i_{X,f}(x)+i_{X,-f}(x)$

\begin{satz}[Index formula] For $G \in \G$ and Euler characteristic, then \\
\parbox{1cm}{\hspace{1cm}} $2 j_{f}(x) (1-\chi(S(x))/2) - X(R_f(x))/2$
\end{satz}

Poincar\'e-Hopf allows fast recursive computation of $X$ for most $G$ in $\G_d$ quantified using
{\bf Erd\"os-Renyi} measures.
For $d$-graphs, the symmetric index is is $-\chi(R_f(x))/2$ if $d$ is odd and $1-\chi(R_f(x))/2$ if $d$ is even.
If $k+d$ is even, we have $\chi_k(B(x))=\chi_k(S(x))$, as curvature
is supported on $\delta G$. Furthermore, we have $\chi_k(\{x \})=0$ if $k>1$.

\begin{satz}[{\bf Dehn-Sommerville-Klee}]
For a $d$-graph and even $d+k$, the functions $K_k$ are supported on $\delta G$. 
If $\delta G=\emptyset$, then $\chi_k(G) =0$.
\end{satz}

$\chi_k$ with even $k+d$ span classical invariants. 
Zero curvature $K_k(x)=0$ for all $x \in V$ also follows $\chi_k=0$ 
from Gauss-Bonnet and suspension.
Curvature functionals are linear combination of Barycentric functionals for $d-1$. 
The classical DS-invariants in dimension $d$ can be derived from Gauss-Bonnet and
the fact that curvatures are DS-invariants in dimension d-1.\\

{\bf Illustration}: 
$\chi_2(G)=0$ on 4-graphs shows that a 4-graph {\bf triangulation}
$G$ of a compact 4-manifold with $v$ vertices, $e$ edges, 
$f$ triangles, $t$ tetrahedra, and $p$ pentatopes satisfies $22e+40t=33f+45p$. Examples:
for the 4-crosspolytop $G$, a 4-sphere with $\vec{v}=(10,40,80,80,32)$, we get 
$\vec{\chi}(G)=(2,0,240,0,32)$. For a discrete $G=S^2 \times T^2$, a product graph
constructed using the {\bf Stanley-Reisner ring}, with
$\vec{v}=(1664,23424,77056,92160,36864)$, we get $\vec{\chi}(G)=(0,0,-10496,0,36864)$. 
For a discrete $G=P^2 \times S^2$ with $\vec{v}=(1898, 26424$, $86736, 103680, 41472)$
we get $\vec{\chi}(G)=(2,0,-10896,0,41472)$. \\

This ends the summary from October 2. Here is is a summary of what has been found since
written in a compressed form. \\

A {\bf $k$-linear} valuation is a real valued map $X$ on ordered $k$ tuples of subgraphs 
such that each $A \to X(A_1,\dots,A,\dots,A_k)$ is a linear valuation, a map from the
set of subgraphs satisfying $X(A \cup B) + X(A \cap B)=X(A)+X(B)$ and $X(\emptyset)=0$. 
Given two valuations $X_1,X_2$, the quadratic valuation $X_1(A) X_2(B)$ is an example.
We assume them to be {\bf localized} in the sense $X(A,B)=0$ if $A \cap B=\emptyset$. 
A $k$-linear valuation $X$ defines a {\bf nonlinear valuation} $X(A)=X(A,A,\dots,A)$. 
Of special interest are quadratic valuations $X(A,B)$ which can be seen as intersection numbers.
Every linear valuation which is invariant in the sense that $X(A)=X(B)$ for isomorphic graphs
defines a linear map on $f$-vectors $v(A)=(v_0,\dots,v_k)$. The linear map is represented by
a $(d+1)$-vector like $\chi = (1,-1,\dots,\pm 1)$. We have then $X(A) = \chi \cdot v(A)$.
If $V(A,B)$ is the {\bf quadratic $f$-form}, where $V_{ij}(A,B)$ counts how many $x$ simplces in $A$ 
intersect in a non-empty set with a $y$ simplex in $B$. We especially have the $f$-matrix 
$V_{ij}(A)$ which encodes the intersections of simplices in $A$. A quadratic valuation can
now be given by two $(d+1)$ vectors $\chi,\psi$: one has $X(A,B) = \phi \cdot V(A,B) \psi$
or $X(A)=\phi \cdot V(A) \psi$. An example is the {\bf Wu characteristic} 
$\omega(A) = \chi V(A) \chi$. Similarly, one can define $k$-linear valuations and have
the {\bf higher Wu characteristic} $\omega_k$. By Hadwiger, the space of linear valuations
is $d+1$ dimensiona, the space of quadratic valuations $\leq (d+1)(d+2)/2$ dimensional
with not yet known dimension. \\

Given two finite simple graphs $A,B$, define a new graph $A \times B$ as follows:
assume the vertex sets of $A,B$ are disjoint. The vertex set of $A \times B$ 
consists of all pairs $(x,y)$ with $x$ being a simplex in $A$ and 
$y$ being a simplex in $B$. Two such elements are connected by an edge, if one
is contained in the other. For a finite simple graph $G$, the product $G \times K_1$
is called the {\bf Barycentric refinement} of $G$. It has the set of simplices of $G$
as vertices and two simplices connected if one is contained in the other. The addition
of two graphs $A + B$ is defined as the disjoint union of the two graphs. Note
that $A \times B$ is only associative, when the product is describe algebraically
as the product in the Stanley-Reisner ring.
The following four theorems hold in the class of all finite simple graphs:

\begin{satz}[{\bf Gauss-Bonnet}]
Every $k$-linear valuation has a curvature $K$ defined on vertices
of $G$ such that $X(A) = \sum_v K(v)$. 
\end{satz}

Lets elaborate a bit in the case of quadratic valuations:
By definition $X(A) = \sum_{x,y \subset A} a(x,y)$, where $a(x,y)$ 
only depend on the dimensions of $x$ and $y$. There are two vectors such that using
the quadratic $f$-form $V(G)$, we have $X(A) = \phi \cdot V(G) \psi$. 
To get the curvature, write $X(A) = \sum_x \kappa(x)$ with $\kappa(x)=\sum_y a(x,y)$. 
Now distribute the value of $\kappa(x)$ equally to all of the ${\rm dim}(x)+1$ 
vertices of $x$. This gives the curvature function $K(v)$ on vertices.

\begin{satz}[{\bf Poincar\'e-Hopf}]
For every $k$-linear valuation and any 
locally injective scalar function $f$ on the vertices of
$G$, there is an index $i_f(x)$ such that 
$X(A) = \sum_v i_f(v)$.
\end{satz}

Given a valuation $X$, its Poincar\'e-Hopf index is defined as
$$ i_f(v) = \omega(B^-_f(v))-\omega(S^-_f(v)) \; , $$
where $B^-_f(x)$ is the graph generated by $\{ y \; | \; f(y) \leq f(x) \; \}$ and
$S^-_f(x)$ is the graph generated by $\{ y \; | \; f(y) < f(x) \; \}$.
For the same type of probability measures as in the linear case, we have:

\begin{satz}[{\bf Index expectation}] 
The expectation of $i_f(v)$ over the probability space of functions is the 
curvature $K(v)$. 
\end{satz}

As for linear valuations, a deformation of the probability measure on functions 
(like for example given by a wave evolution) changes the curvature but keeps
Guass-Bonnet intact. The deformation of the probability measure allows for
other type of curvatures. 

\begin{satz}[{\bf Topological Invariants}]
Every Wu characteristic is invariant under Barycentric refinements.
\end{satz}

Since $\omega(G) = \sum_x \omega(x)$, we can see $\omega$ as a sum of values of 
a function on the vertex set of the Barycentric refinement $G_1$. This function is an 
index for a natural ordering of the vertices of $G_1$. 

\begin{satz}[{\bf Multiplicative function}]
Every Wu characteristic is multiplicative 
$\omega_k(A \times B) = \omega_k(A) \omega_k(B)$.
\end{satz}

To illustrate the quadratic case, first show $\omega(x \times y) = \omega(x) \omega(y)$ for simplices
then write $\omega(G) = \sum_{x,y} \omega(x) \omega(y)$, finally uses the Barycentric 
invariance. \\

The following two theorems hold only for $d$-graphs or $d$-graphs with 
boundary: 

\begin{satz}[{\bf Boundary formula}]
For every $d$-graph $G$ with boundary $\delta G$, we have
$\omega(G) = \chi(G) - \chi(\delta G)$.
\end{satz}

For $d$-graphs without boundary this is shown by verifying that the 
Wu curvature and Euler curvatures agree. What happens at the boundary 
is that for simplices hitting the  boundary there is an additional
contribution one or minus one. The corresponding sum is the Euler characteristic
of a thickend boundary which is homotopic to the actual boundary. 

\begin{satz}[{\bf Gr\"unbaum question}]
For every $d$-graph $G$ without boundary and every linear valuation
$X(A)=v(A) \dots \psi$, the quadratic valuation $Y(A) = \chi_1 V(A) \psi$
satisfies $X(G)=Y(G)$. 
\end{satz}

This could be generalized to $k$-linear cases. Like that $X(A)$ and
the cubic valuation $Y(A) = V(A) \chi_1 \chi_1 \psi$ agree up to a sign 
on the $d$ graph $G$. 
The proof goes by writing the valuation as a sum over pairs of intersecting 
simplices and then partition according to which simplex $z$ they intersect. 
When looking at a sum for fixed $z$, then the part of the interacting
simplices is zero. This uses that any lower dimensional spheres have the 
correct Euler characteristic. 

\begin{satz}[{\bf Dehn Sommerville space}]
The Dehn-Sommerville space of $k$-valuations has the same dimension
as the Dehn-Sommerville space of linear valuations $[(d+1)/2]$. 
\end{satz}

The fact that the $[(d+1)/2]$ Dehn-Sommerville valuations are linearly
independent follows from the fact that there is a basis given by 
eigenvectors to different eigenvalues. An alternative basis are the classical
Dehn-Sommerville valuations
$$ X_{k,d}(v) = \sum_{j=k}^{d-1} (-1)^{j+d} \B{j+1}{k+1} v_j(G) + v_k(G)  \; . $$
There can not be a larger dimensional
space of valuations since a simple example of cross polytopes obtained by 
doing multiple suspension of a circular graph $C_4$ shows that only $[(d+1)/2]$
can be zero in general. For some graphs of course, the space of valuations
which vanish can be larger. An example is the 2-torus $G=C_4 \times C_4$ for
which the $f$-matrix is 
$$ V=\left[ \begin{array}{ccc} 0 & 0 & 128 \\ 0 & 640 & 1408 \\ 128 & 1408 & 1920 \\ \end{array} \right] \; .$$
The Dehn-Sommerville space of the 2-torus is $2$-dimensional. 
For $d=2$ dimensional graphis we have only a $[3/2]=1$-dimensional Dehn-Sommerville space. 
By the way, with $\chi=(1,-1,1)$, then $V \chi=(128, 768, 640)$ and the Wu invariant is
$\chi \cdot (V \chi) = 0$ as the Wu characteristic is the Euler characteristic, which is 
$0$ for the 2-torus.  \\

Finally, Gauss-Bonnet for linear valuations immediately shows that

\begin{satz}{\bf Flatness of D-S}
For any $d$-graph and Dehn-Sommerville relation, $X(G)=0$ and $G$ is
flat in the sense that all curvatures are constant zero. 
\end{satz}

The proof uses that curvature at a vertex $v$ is a Dehn-Sommerville valuation
of the unit sphere $S(v)$ of the graph at $v$. 

\bibliographystyle{plain}

\end{document}